\documentclass[american,english]{scrreprt}
\usepackage[T1]{fontenc}
\usepackage[latin9]{inputenc}
\usepackage{array}
\usepackage{verbatim}
\usepackage{rotfloat}
\usepackage{booktabs}
\usepackage{units}
\usepackage{multirow}
\usepackage{amsmath}
\usepackage{amsthm}
\usepackage{amssymb}
\usepackage{stmaryrd}
\usepackage{graphicx}
\usepackage{wasysym}
\usepackage{esint}

\makeatletter

\providecommand{\tabularnewline}{\\}

\makeatother

\PassOptionsToPackage{utf8}{inputenc}
	\usepackage{inputenc}

\PassOptionsToPackage{eulerchapternumbers,listings,linedheaders,pdfspacing,floatperchapter,subfig,parts}{classicthesis}                                        

\newcommand{\myTitle}{Lindbladians with multiple steady states\xspace}
\newcommand{\mySubtitle}{theory and applications\xspace}

\newcommand{\myName}{Victor V. Albert\xspace}

\newcommand{\mySupervisor}{Liang Jiang\xspace}
\newcommand{\myFaculty}{Steven M. Girvin, Leonid I. Glazman, Michel H. Devoret, Robert J. Schoelkopf\xspace}

\newcommand{\myUni}{Yale University\xspace}

\newcommand{\myVersion}{version 1.0}

\newcounter{dummy} 
\providecommand{\mLyX}{L\kern-.1667em\lower.25em\hbox{Y}\kern-.125emX\@}

	\usepackage{babel}

\usepackage{csquotes}
\PassOptionsToPackage{%
    backend=bibtex8,
    language=auto,%
	style=phys,%
	pageranges=false,%
	eprint=true,%
	biblabel=brackets,%
	giveninits=true,%
    sorting=nyt, 
    maxbibnames=100, 
	backrefstyle=three,%
    natbib=true 
	}{biblatex}
    \usepackage{biblatex}
	
\DeclareFieldFormat[article]{title}{\emph{#1}}
\DeclareFieldFormat[book,misc]{title}{\emph{\href{\thefield{url}}{#1}}}
	
\usepackage{thmtools} 
	
\PassOptionsToPackage{fleqn}{amsmath}       
    \usepackage{amsmath}

\PassOptionsToPackage{T1}{fontenc} 
    \usepackage{fontenc}     
\usepackage{textcomp} 
\usepackage{scrhack} 
\usepackage{xspace} 
\usepackage{mparhack} 
\usepackage{fixltx2e} 
\PassOptionsToPackage{printonlyused,smaller}{acronym}
    \usepackage{acronym} 
	
\usepackage{tabularx} 
\usepackage{caption}
\usepackage{subfig}  

\usepackage{listings} 
\PassOptionsToPackage{pdftex,hyperfootnotes=false,pdfpagelabels}{hyperref}
    \usepackage{hyperref}  
\PassOptionsToPackage{pdftex}{graphicx}
    \usepackage{graphicx}

\hypersetup{%
    colorlinks=true, linktocpage=true, pdfstartpage=1, pdfstartview=FitV,%
    breaklinks=true, pdfpagemode=UseNone, pageanchor=true, pdfpagemode=UseOutlines,%
    plainpages=false, bookmarksnumbered, bookmarksopen=true, bookmarksopenlevel=0,%
    hypertexnames=true, pdfhighlight=/O,
    urlcolor=webbrown, linkcolor=RoyalBlue, citecolor=webgreen, 
    pdftitle={\myTitle},%
    pdfauthor={\textcopyright\ \myName, \myUni, \myFaculty},%
    pdfsubject={},%
    pdfkeywords={Lindblad master equation, open quantum systems, quantum dynamical semigroup},%
    pdfcreator={pdfLaTeX},%
    pdfproducer={LaTeX with hyperref and classicthesis}%
}   

\makeatletter
\@ifpackageloaded{babel}%
    {%
       \addto\extrasamerican{%
                }%
       \addto\extrasngerman{%
                }%
    }{\relax}
\makeatother

\usepackage{classicthesis} 

\areaset
  [0.0in]
  {\dimexpr\the\paperwidth-2.0in\relax}
  {\dimexpr\the\paperheight-2.0in\relax}
\usepackage{setspace}

\makeatletter

\ifcsname input@path\endcsname 
  \edef\@basepath{\expandafter\@firstofone\input@path} 
  \def\rm@quotes#1"#2"#3\@nul{\ifx\relax#2\relax #1\else #2\fi}
  \edef\@basepath{\expandafter\rm@quotes\@basepath""\@nul} 
\else\edef\@basepath{./}\fi

\let\orig@addbibresource\addbibresource
\renewrobustcmd*{\addbibresource}[2][type=file]{\orig@addbibresource[#1]{\@basepath#2}}

\theoremstyle{plain}
\newtheorem{thm}{\protect\theoremname}
\ifx\proof\undefined
\newenvironment{proof}[1][\protect\proofname]{\par
\normalfont\topsep6\p@\@plus6\p@\relax
\trivlist
\itemindent\parindent
\item[\hskip\labelsep\scshape #1]\ignorespaces
}{%
\endtrivlist\@endpefalse
}
\providecommand{\proofname}{Proof}
\fi

\usepackage{layout}
\usepackage{ytableau}

\AtBeginDocument{
\renewcommand*{\AC@hyperlink}[2]{%
\begingroup       
\hypersetup{hidelinks}%
\hyperlink{#1}{#2}%
\endgroup   
}%
}

\makeatother

\usepackage{babel}
\addto\captionsamerican{\renewcommand{\theoremname}{Theorem}}
\addto\captionsenglish{\renewcommand{\theoremname}{Theorem}}
\providecommand{\theoremname}{Theorem}

\begin{document}
\global\long\def\half{\frac{1}{2}}
\global\long\def\a{\alpha}
\global\long\def\b{\beta}
\global\long\def\g{\gamma}
\global\long\def\c{\chi}
\global\long\def\d{\delta}
\global\long\def\o{\omega}
\global\long\def\m{\mu}
\global\long\def\s{\sigma}
\global\long\def\n{\nu}
\global\long\def\z{\zeta}
\global\long\def\l{\lambda}
\global\long\def\e{\epsilon}
\global\long\def\k{\kappa}
\global\long\def\x{\chi}
\global\long\def\r{\rho}
\global\long\def\t{\theta}
\global\long\def\G{\Gamma}
\global\long\def\D{\mathcal{D}}
\global\long\def\O{\Omega}
\global\long\def\pr{\prime}

\global\long\def\oo{\mathcal{O}}
\global\long\def\L{\mathcal{L}}
\global\long\def\H{\mathcal{H}}
\global\long\def\R{\mathcal{P}}
\global\long\def\A{\mathcal{A}}
\global\long\def\F{\mathcal{F}}
\global\long\def\dfunc{\mathfrak{J}}
\global\long\def\met{\mathcal{M}}
\global\long\def\geom{\mathcal{Q}}
\global\long\def\rest{\mathcal{R}}
\global\long\def\ppp{\mathcal{P}_{\!\!\!{\scriptscriptstyle \infty}}}
\global\long\def\pppd{\dot{\mathcal{P}}_{\!\!\!{\scriptscriptstyle \infty}}}
\global\long\def\qqq{\mathcal{Q}_{{\scriptscriptstyle \infty}}}
\global\long\def\N{\mathcal{N}}
\global\long\def\E{\mathcal{E}}
\global\long\def\GG{\mathcal{G}}

\global\long\def\lind{\kappa}
\global\long\def\bra{\langle}
\global\long\def\ket{\rangle}
\global\long\def\bb{\langle\!\langle}
\global\long\def\kk{\rangle\!\rangle}
\global\long\def\ash{\textnormal{As(\ensuremath{\mathsf{H}})}}
\global\long\def\oph{\textnormal{Op(\ensuremath{\mathsf{H}})}}
\global\long\def\h{\textnormal{\ensuremath{\mathsf{H}}}}
\global\long\def\dg{\dagger}
\global\long\def\dgt{\ddagger}
\global\long\def\tr{\textnormal{Tr}}
\global\long\def\Tr{\textnormal{\textsc{Tr}}}
\global\long\def\id{\mathcal{I}}

\global\long\def\pp{P}
\global\long\def\qq{Q}
\global\long\def\rin{\rho_{\textnormal{\textsf{in}}}}
\global\long\def\rout{\rho_{{\scriptscriptstyle \infty}}}
\global\long\def\rss{\rho_{\mathsf{ss}}}

\global\long\def\la{\varDelta}
\global\long\def\sl{(}
\global\long\def\sr{)}
\global\long\def\al{[}
\global\long\def\ar{]}
\global\long\def\kef{\kappa_{\textsf{eff}}}

\global\long\def\idfs{\mathcal{P}_{\!\textsf{dfs}}}
\global\long\def\idfsd{\dot{\mathcal{P}}_{\!\textsf{dfs}}}
\global\long\def\iidfs{P_{\!\textsf{dfs}}}
\global\long\def\iidfsd{\dot{P}_{\!\textsf{dfs}}}
\global\long\def\hout{H_{{\scriptscriptstyle \infty}}}
\global\long\def\sout{\mathcal{H}_{{\scriptscriptstyle \infty}}}
\global\long\def\sdfs{\mathcal{H}_{\textsf{dfs}}}

\global\long\def\Ga{\varGamma}
\global\long\def\St{\varPsi}
\global\long\def\stdfs{\varPsi^{\textnormal{\textsf{dfs}}}}
\global\long\def\J{J}
\global\long\def\kkk#1{\left.\left|#1\right\rangle \!\!\!\right\rangle }
\global\long\def\as{\varrho_{\textnormal{\textsf{ax}}}}
\global\long\def\du{d_{\varrho}}
\global\long\def\ai{P_{\!\textnormal{\textsf{ax}}}}
\global\long\def\ot{\otimes}
\global\long\def\an{n_{\textnormal{\textsf{ax}}}}
\global\long\def\aso{\mathcal{L}_{\textnormal{\textsf{ax}}}}
\global\long\def\asi{\mathcal{P}_{\textnormal{\textsf{ax}}}}
\global\long\def\hha{\textnormal{\ensuremath{\mathsf{H_{\textnormal{\textsf{ax}}}}}}}
\global\long\def\hhdfs{\textnormal{\ensuremath{\mathsf{H_{dfs}}}}}
\global\long\def\op{\textnormal{Op}}
\global\long\def\opa{\textnormal{Op(\ensuremath{\mathsf{H_{\textnormal{\textsf{ax}}}}})}}
\global\long\def\opdfs{\textnormal{Op(\ensuremath{\mathsf{H_{dfs}}})}}
\global\long\def\da{d_{\textnormal{\textsf{ax}}}}

\global\long\def\p{\partial}
\global\long\def\hpert{V}
\global\long\def\spert{\mathcal{V}}
\global\long\def\reac{W^{\textsf{dfs}}}
\global\long\def\rreac{\mathcal{W}^{\textsf{dfs}}}
\global\long\def\rrreac{\mathcal{W}}
\global\long\def\hpertdfs{V}
\global\long\def\spertdfs{\mathcal{V}}
\global\long\def\apert{V_{\textsf{ax}}}
\global\long\def\X{\mathcal{X}}
\global\long\def\xdfs{Y_{\textsf{dfs}}}
\global\long\def\xxdfs{\mathcal{Y}_{\textsf{dfs}}}
\global\long\def\Y{\mathcal{Y}}

\global\long\def\pad{P_{0}}
\global\long\def\padd{\dot{P}_{0}}
\global\long\def\qad{Q_{0}}
\global\long\def\uad{U_{\textsf{ad}}}
\global\long\def\uadd{U}
\global\long\def\sad{S}
\global\long\def\uk{U_{\textsf{dfs}}}
\global\long\def\ukdfs{\mathcal{U}_{\textsf{dfs}}}
\global\long\def\U{\mathcal{U}}
\global\long\def\hk{K}
\global\long\def\hol{\psi}
\global\long\def\xx{\mathbf{x}}
\global\long\def\vv{\dot{\mathbf{x}}}
\global\long\def\path{\mathbb{P}}
\global\long\def\dd{d}
\global\long\def\ber{\mathcal{B}}
\global\long\def\rdfs{\mathcal{R}_{\mathsf{dfs}}}
\global\long\def\dax{S_{\textsf{ax}}}
\global\long\def\ddax{\mathcal{S}_{\textsf{ax}}}
\global\long\def\uax{R_{\textsf{ax}}}
\global\long\def\uuax{\mathcal{R}_{\textsf{ax}}}
\global\long\def\tang{\mathsf{T_{M}}(\xx)}
\global\long\def\pt{\mathbf{w}}

\global\long\def\adfs{A^{\textsf{dfs}}}
\global\long\def\aadfs{\mathcal{A}^{\textsf{dfs}}}
\global\long\def\aax{\mathcal{A}^{\textsf{ax}}}
\global\long\def\adfst{\widetilde{A}^{\textsf{dfs}}}
\global\long\def\aadfst{\widetilde{\mathcal{A}}^{\textsf{dfs}}}
\global\long\def\bdfs{B^{\textsf{dfs}}}
\global\long\def\bbdfs{\mathcal{B}^{\textsf{dfs}}}
\global\long\def\bbdfst{\widetilde{\mathcal{B}}^{\textsf{dfs}}}
\global\long\def\fdfs{F^{\textsf{dfs}}}
\global\long\def\ffdfs{\mathcal{F}^{\textsf{dfs}}}
\global\long\def\gdfs{M^{\textsf{dfs}}}
\global\long\def\ggdfs{\mathcal{M}^{\textsf{dfs}}}
\global\long\def\qdfs{Q^{\textsf{dfs}}}
\global\long\def\psp{\mathsf{M}}
\global\long\def\qqdfs{\mathcal{Q}^{\textsf{dfs}}}
\global\long\def\qax{\mathcal{Q}^{\textsf{ax}}}
\global\long\def\gax{\mathcal{M}^{\textsf{ax}}}
\global\long\def\ddfs{S_{\textsf{dfs}}}
\global\long\def\dddfs{\mathcal{S}_{\textsf{dfs}}}
\global\long\def\dist{\mathcal{S}}
\global\long\def\ps{\mathcal{P}_{\St}}
\global\long\def\psd{\dot{\mathcal{P}}_{\St}}
\global\long\def\psb{\overline{\psi}}
\global\long\def\stb{\overline{\varPsi}^{\textsf{dfs}}}
\global\long\def\idfsb{\overline{\mathcal{P}}_{\!\textsf{dfs}}}
\global\long\def\iidfsb{\overline{P}_{\!\textsf{dfs}}}
\global\long\def\pdfs{G^{\textsf{dfs}}}
\global\long\def\ppdfs{\mathcal{G}^{\textsf{dfs}}}
\global\long\def\ooo{\mathcal{O}_{\textsf{ax}}}

\global\long\def\intt#1#2{\int_{#1}^{#2}\!\!\!}
\global\long\def\Re{\text{Re}}
\global\long\def\hdg{H_{\textsf{edg}}}
\global\long\def\aa{a}
\global\long\def\ph{\hat{n}}
\global\long\def\adg{\Delta_{\textnormal{\textsf{edg}}}}
\global\long\def\dgg{\Delta_{\textnormal{\textsf{dg}}}}

\global\long\def\rr{\mathbf{r}}
\global\long\def\mo{\mathbf{p}}

\global\long\def\hi{\textnormal{\ensuremath{\mathsf{H_{\textnormal{\textsf{in}}}}}}}
\global\long\def\ho{\textnormal{\ensuremath{\mathsf{H_{out}}}}}
\global\long\def\dii{d_{\textnormal{\textsf{in}}}}
\global\long\def\doo{d_{\textnormal{\textsf{out}}}}

\global\long\def\ran{\textnormal{ran}}
\global\long\def\T{{\cal T}}
\global\long\def\S{{\cal S}}
\global\long\def\cat{\textnormal{\textsf{Cat}}}

\global\long\def\ct#1{\left|#1\right\rangle }
\global\long\def\cb#1{\left\langle #1\right|}

\global\long\def\nn{\pi}

\global\long\def\bl{b}
\global\long\def\pb{\hat{m}}
\global\long\def\de{\hat{\Delta}}
\global\long\def\ra{\Delta}

\global\long\def\LE{{\cal L}_{\textnormal{\textsf{eff}}}}
\global\long\def\HE{H_{\textnormal{\textsf{eff}}}}
\global\long\def\FE{F_{\textnormal{\textsf{eff}}}}
\global\long\def\K{{\cal K}}

\global\long\def\NU{{\cal Z}}

\DeclareRobustCommand{\lrff}{\rotatebox[origin=c]{45}{$\Leftrightarrow$}}
\DeclareRobustCommand{\lrtf}{\rotatebox[origin=c]{315}{$\Leftrightarrow$}}
\DeclareRobustCommand{\lff}{\rotatebox[origin=c]{45}{$\Leftarrow$}}
\DeclareRobustCommand{\rtf}{\rotatebox[origin=c]{315}{$\Rightarrow$}}

\DeclareRobustCommand{\ul}{{\raisebox{2pt}{\ytableaushort{ {*(black)} {} , {} {} }}}} 
\DeclareRobustCommand{\ur}{{\raisebox{2pt}{\ytableaushort{ {} {*(black)} , {} {} }}}} 
\DeclareRobustCommand{\ll}{{\raisebox{2pt}{\ytableaushort{ {} {} , {*(black)} {} }}}} 
\DeclareRobustCommand{\lr}{{\raisebox{2pt}{\ytableaushort{ {} {} , {} {*(black)} }}}} 
\DeclareRobustCommand{\di}{{\raisebox{2pt}{\ytableaushort{ {*(black)} {} , {} {*(black)} }}}} 
\DeclareRobustCommand{\of}{{\raisebox{2pt}{\ytableaushort{ {} {*(black)} , {*(black)} {} }}}}
\DeclareRobustCommand{\thr}{{\raisebox{2pt}{\ytableaushort{ {*(black)} {*(black)} , {} {*(black)} }}}}
\DeclareRobustCommand{\thu}{{\raisebox{2pt}{\ytableaushort{ {*(black)} {*(black)} , {*(black)} {} }}}}
\DeclareRobustCommand{\tho}{{\raisebox{2pt}{\ytableaushort{ {} {*(black)} , {*(black)} {*(black)} }}}}
\DeclareRobustCommand{\lmt}{{\raisebox{2pt}{\ytableaushort{ {*(black)} {}, {*(black)} {} }}}}
\DeclareRobustCommand{\rmt}{{\raisebox{2pt}{\ytableaushort{ {}{*(black)} ,{} {*(black)}  }}}}
\DeclareRobustCommand{\emp}{{\raisebox{2pt}{\ytableaushort{ {} {} , {} {} }}}}
\ytableausetup{boxsize = 2pt}
\DeclareRobustCommand{\one}{\hyperref[l:one]{\raisebox{.5pt}{\textcircled{\raisebox{-.9pt}{1}}}}}
\DeclareRobustCommand{\two}{\hyperref[l:two]{\raisebox{.5pt}{\textcircled{\raisebox{-.9pt}{2}}}}}
\DeclareRobustCommand{\three}{\hyperref[l:thr]{\raisebox{.5pt}{\textcircled{\raisebox{-.9pt}{3}}}}}
\DeclareRobustCommand{\four}{\hyperref[l:fou]{\raisebox{.5pt}{\textcircled{\raisebox{-.9pt}{4}}}}}
\DeclareRobustCommand{\five}{\hyperref[l:fiv]{\raisebox{.5pt}{\textcircled{\raisebox{-.9pt}{5}}}}}
\DeclareRobustCommand{\six}{\hyperref[l:six]{\raisebox{.5pt}{\textcircled{\raisebox{-.9pt}{6}}}}}
\newcommand{\ulbig}{\ytableausetup{boxsize = 3pt}
\raisebox{3pt}{\ytableaushort{ {*(black)} {} , {} {} }}
\ytableausetup{boxsize = 2pt}}
\newcommand{\ofbig}{\ytableausetup{boxsize = 3pt}
\raisebox{3pt}{\ytableaushort{ {} {*(black)} , {*(black)} {} }}
\ytableausetup{boxsize = 2pt}}
\newcommand{\urbig}{\ytableausetup{boxsize = 3pt}
\raisebox{3pt}{\ytableaushort{ {} {*(black)} , {} {} }}
\ytableausetup{boxsize = 2pt}}
\newcommand{\llbig}{\ytableausetup{boxsize = 3pt}
\raisebox{3pt}{\ytableaushort{ {} {} , {*(black)} {} }}
\ytableausetup{boxsize = 2pt}}
\newcommand{\lrbig}{\ytableausetup{boxsize = 3pt}
\raisebox{3pt}{\ytableaushort{ {} {} , {} {*(black)} }}
\ytableausetup{boxsize = 2pt}}
\newcommand{\thubig}{\ytableausetup{boxsize = 3pt}
\raisebox{3pt}{\ytableaushort{ {*(black)} {*(black)} , {*(black)} {} }}
\ytableausetup{boxsize = 2pt}}
\newcommand{\thobig}{\ytableausetup{boxsize = 3pt}
\raisebox{3pt}{\ytableaushort{ {} {*(black)} , {*(black)} {*(black)} }}
\ytableausetup{boxsize = 2pt}}
\newcommand{\empbig}{\ytableausetup{boxsize = 3pt}
\raisebox{3pt}{\ytableaushort{ {} {} , {} {} }}
\ytableausetup{boxsize = 2pt}}
\newcommand{\dibig}{\ytableausetup{boxsize = 3pt}
\raisebox{3pt}{\ytableaushort{ {*(black)} {} , {} {*(black)} }}
\ytableausetup{boxsize = 2pt}}

\selectlanguage{american}%
\begingroup 
\let\clearpage\relax 
\let\cleardoublepage\relax

\inputencoding{latin9}
\thispagestyle{empty}
\begin{center}
\spacedlowsmallcaps{Abstract}\\
\myTitle: \mySubtitle\\
Victor V. Albert\\
2017
\end{center}\foreignlanguage{english}{Markovian master equations, often called
Liouvillians or Lindbladians, are used to describe decay and decoherence
of a quantum system induced by that system's environment. While a
natural environment is detrimental to fragile quantum properties,
an engineered environment can drive the system toward exotic phases
of matter or toward subspaces protected from noise. These cases often
require the Lindbladian to have more than one steady state, and such
Lindbladians are dissipative analogues of Hamiltonians with multiple
ground states. This thesis studies Lindbladian extensions of topics
commonplace in degenerate Hamiltonian systems, providing examples
and historical context along the way.}

\selectlanguage{english}%
An important property of Lindbladians is their behavior in the limit
of infinite time, and the first part of this work focuses on deriving
a formula for the asymptotic projection \textemdash{} the map corresponding
to infinite-time Lindbladian evolution. This formula is applied to
determine the dependence of a system's steady state on its initial
state, to determine the extent to which decay affects a system's linear
or adiabatic response, and to determine geometrical structures (holonomy,
curvature, and metric) associated with adiabatically deformed steady-state
subspaces. Using the asymptotic projection to partition the physical
system into a subspace free from nonunitary effects and that subspace's
complement (and making a few other minor assumptions), a Dyson series
is derived to all orders in an arbitrary perturbation. The terms
in the Dyson series up to second order in the perturbation \foreignlanguage{american}{are
shown to reproduce quantum Zeno dynamics and the effective operator
formalism.}\selectlanguage{english}%

\inputencoding{latin9}\begin{titlepage}
\large
\hfill
\null 
\vspace{1in}
\begin{center}%

\begingroup \color{Maroon} 
\LARGE
\textls{\myTitle:}
\\ 
\textls{\mySubtitle}
\endgroup

\vfil     
\begin{singlespace}
{\large                
A Dissertation\\       
Presented to the Faculty of the Graduate School\\                 
of\\               
Yale University\\            
in Candidacy for the Degree of\\          
Doctor of Philosophy
\par\vfil\vskip 6ex%
by\\               
\myName\par\vskip 1.5em%
Dissertation Director: 
\mySupervisor\par     
}\vskip 1.5em%
May 2017

\newpage{}

\thispagestyle{empty}
\vfil
\null
\vspace{4in}
© 2017 by Victor V. Albert \\
All rights reserved.\end{singlespace}
\end{center}
\end{titlepage}
\inputencoding{latin9}
\refstepcounter{dummy}
\pdfbookmark[1]{\contentsname}{tableofcontents} 
\setcounter{tocdepth}{2} 
\setcounter{secnumdepth}{3} 
\manualmark 
\markboth{\spacedlowsmallcaps{\contentsname}}{\spacedlowsmallcaps{\contentsname}} 
\tableofcontents  
\automark[section]{chapter} 
\renewcommand{\chaptermark}[1]{\markboth{\spacedlowsmallcaps{#1}}{\spacedlowsmallcaps{#1}}} 
\renewcommand{\sectionmark}[1]{\markright{\thesection\enspace\spacedlowsmallcaps{#1}}} 

\newpage{}

\refstepcounter{dummy} 
\addcontentsline{toc}{chapter}{\tocEntry{Acronyms}}

\chapter*{Acronyms}
\begin{acronym}[UML]
\acro{OPH}[Op(\textsf{H})]{space of operators that act on a Hilbert space \textsf{H}}
\acro{ASH}[As(\textsf{H})]{Asymptotic subspace --- the subspace of Op(\textsf{H}) for which evolution under a given Lindbladian is unitary}
\acro{DFS}{Decoherence-free subspace --- specific type of As(\textsf{H})}
\acro{NS}{Noiseless subspace --- specific type of As(\textsf{H}) that factorizes into a tensor product of a DFS and a unique auxiliary mixed state}
\acro{QGT}{Quantum geometric tensor --- a quantity providing a metric  on and encoding the geometric phase properties of a subspace}
\end{acronym}              \vspace{8ex}
\refstepcounter{dummy} 
\addcontentsline{toc}{chapter}{\tocEntry{List of Theorems}}
\listoftheorems

\newpage{}\begin{center}
\begin{singlespace}
\thispagestyle{empty}
\vfil
\null
\vspace{4in}
\textit{
For my parents\\
Tatyana \& Thomas\\
and my wife\\
Olga}
\end{singlespace}
\end{center}

\inputencoding{latin9}\refstepcounter{dummy} 
\addcontentsline{toc}{chapter}{\tocEntry{Acknowledgments}}
\addtocontents{toc}{\protect\vspace{\beforebibskip}} 

\newpage{}

\chapter*{Acknowledgments}

First I would like to thank my adviser, my wife, and my parents, all
of whom created the environment and provided the opportunities for
a work such as this to be possible. Their contributions however were
on different timescales, but they are all nevertheless important to
the final cumulative outcome. Leading by example, my parents always
made me believe that whatever abilities I have would be sufficient
for ``success'' as long as they are complemented with hard work.
My wife has always supported my professional life by listening to
me when I needed to talk and letting me devote as much time to work
as I wanted. (Luckily, work hasn't amounted to all my time, so I continue
to be happily married!) Of course, a proper upbringing and a supportive
better half would be for nothing if one is in a stressful relationship
with their adviser. Luckily, it has been the opposite: Liang has been
supportive, thoughtful, patient, and kind and has demonstrated what
it takes to be a good physicist. As I developed my own tastes, he
let me pursue other directions while at the same time letting me know
what he thought was important.

I also acknowledge the support of my committee of collaborators. Despite
all of them being extremely busy, I know that I can always count on
Profs. Devoret, Girvin, Glazman, and Schoelkopf to lend an ear and
help sort out scientific and even personal problems. Their combination
of accomplishments and humility is truly inspiring. Among my other
collaborators, I want to thank especially Martin Fraas and Barry Bradlyn
for carefully checking a large portion of this work and helping overcome
a negative referee report.

Several other faculty, both at Yale and elsewhere, have contributed
to the supportive scientific environment I have enjoyed so far: Lorenza
Viola (whom I want to also acknowledge for being an external reader
for this work), Zaki Leghtas, Mazyar Mirrahimi, Barbara Terhal, Alexey
Gorshkov, Nick Read, Steve Flammia, Francesco Ticozzi, Marlan Scully,
Dariusz Chruscinski, Oscar Viyuela, Kirill Velizhanin, and Frank Harris.
Members of my adviser's group have provided a welcoming, active, and
collaborative learning atmosphere: Stefan Krastanov, Sreraman Muralidharan,
Arpit Dua, Kyungjoo Noh, Chao Shen, Linshu Li, Chang-ling Zou, Chi
Shu, Marios Michael, and Jianming Wen. Throughout the years, I was
also fortunate to have fruitful discussions with Richard Brierley,
Matti Silveri, Alex Petrescu, Aris Alexandradinata, Andrey Gromov,
Judith Höller, \foreignlanguage{english}{Phil Reinhold, Reinier Heeres},
Volodymyr Sivak, Ion Garate, Zaheen Sadeq, and many of the members
of Qlab and RSL.

Last and definitely not least, I want to thank my long-time friends
at Yale \textemdash{} Brian Tennyson, Jukka Väyrynen, Staff Sheehan,
Zlatko Minev, Kevin Yaung, and Alexei Surguchev \textemdash{} for
being there during all these years.

\begin{flushleft}
\inputencoding{latin9}\newpage{}\foreignlanguage{english}{}%
\begin{minipage}[t]{0.5\textwidth}%
\selectlanguage{english}%
\begin{flushleft}
\begin{singlespace}``\textit{Give me a Hamiltonian and I will move
the world.}''\end{singlespace}
\par\end{flushleft}
\begin{flushleft}
\hfill{}\textendash{} Leonid I. Glazman
\par\end{flushleft}\selectlanguage{english}%
\end{minipage}
\par\end{flushleft}

\chapter{Introduction and motivation\label{ch:1}}
\selectlanguage{english}%

\section{The field of study}
\selectlanguage{american}%

\subsection{Applied quantum physics}

\selectlanguage{english}%
\inputencoding{latin9}%
We are often taught in school that physics consists of the never-ending
interplay between theorist and experimentalist. Theorists propose
models for experimentalists to verify, while unexpected experimental
results cause theorists to adjust their models. While this type of
science is currently happening in high-energy areas of physics (such
as dark matter detection), other areas have become more crystallized
and their theories more verified. One such area can be called \textit{low-energy
}or \textit{applied quantum physics}. In this broad field, the theory
\textemdash{} quantum mechanics \textemdash{} is used to describe
almost all processes and is not generally questioned. In other words,
enough evidence has been gathered showing that the vast majority of
processes on nanometer length scales and at nanoKelvin temperatures
are most effectively and to a high degree of precision described by
quantum mechanics. Instead of testing theories, some primary aims
of this field are to notice interesting quantum mechanical effects,
to demonstrate (or \textit{realize}) these effects using various currently
available quantum technologies, and to develop devices based on these
effects which have potential applications to the ``real world''.
In the past, examples of such devices include the laser, the transistor,
and nuclear magnetic resonance imagers. Although these devices can
be described using mostly classical physics \cite{siegman,ernstnmr},
it is difficult to do so without resorting to the discretized nature
of quantum energy levels and even more difficult to claim that the
rise of quantum mechanics did not directly contribute to their development.

Currently, a primary potential application of this field is the development
of a \textit{quantum computer} \cite{maninbook,Feynman1982} \textemdash{}
a device that has been theoretically proven to perform certain computational
tasks significantly faster than any ordinary computer. By ``significantly'',
we mean that the time taken by a quantum computer to perform a task
is a couple of days while a classical computer would take at least
until the sun burns out. Since one of the tasks \textemdash{} integer
factorization \textemdash{} can be used to crack an often-used encryption
scheme used to communicate securely over the internet and since quantum
technologies are improving very quickly, quantum computation continues
to steadily gain attention from the non-scientific community. 

Another important task that ``quantum computers can do in their sleep''
\cite{Feynman1982} is \textit{quantum simulation} \textemdash{} the
tailoring of one quantum system, in this case the quantum computer,
to simulate another, less well-understood quantum system. Therefore,
the development of a quantum computer should in principle lead to
a qualitatively improved understanding of quantum processes in other
fields such as chemistry or biology.

A closely related field is \textit{quantum metrology and sensing}
\textemdash{} the use of quantum devices to perform precise measurements
and microscopy which are not possible with ``classical devices''.
This field's associated quantum technologies have promising applications
ranging from table-top measurements of nature's fundamental constants
to imaging of biological systems.

Developments in applied quantum physics should also potentially allow
one to better understand and even synthesize exotic \textit{quantum
phases of matter}. ``For a large collection of similar particles,
a \textit{phase} is a region in some parameter space in which the
thermal equilibrium states possess some properties in common that
can be distinguished from those in other phases'' \cite{Read2012}.
Naturally, a quantum phase of matter is one which is most efficiently
described by quantum mechanics. More such phases have been predicted
and theoretically studied than realized in the laboratory, and one
of the goals of this field is to bridge this gap.

Another primary application of this field is \textit{quantum communication}
\textemdash{} using the properties of quantum mechanics to communicate
in such a way that any eavesdropper is readily noticed as soon as
they try to intercept any communication. Such secure quantum communication
channels not only have obvious applications in society, but also have
lead to the development of the field of \textit{quantum information
theory} \textemdash{} a merging of applied quantum physics and information
theory which analyzes, among other things, how much information can
securely go through quantum communication channels \cite{wildebook}.

While the applications are clear, it is currently unclear which technology
will \textit{actually} build the first practical (i.e., scalable!)
quantum computer. Consequently, this lack of a precise focus on one
technology has stimulated a broad and thrilling theoretical investigation
into all areas of quantum mechanics which are even remotely useful
in building quantum devices. The area that characterizes this thesis
is \textit{open quantum systems} \textemdash{} the study of quantum
systems which are in contact with a larger environment or reservoir.

\subsection{Open quantum systems}

Introductory physics courses devote much time to studying systems
which are isolated from their environments. Examples include an object
falling without air resistance or, more informally, a ``spherical
cow in vacuum''. Besides being necessary for a complete understanding
of the systems in question, environmental effects can also steer the
systems in favorable directions. For example, including air resistance
in the calculation of a falling object leads to the understanding
that a parachute can prevent said object from falling too quickly. 

Environmental effects are even more pronounced in quantum systems.
On the one hand, environment or \textit{quantum reservoir engineering}
is poised to synthesize longer-lasting quantum memories, faster quantum
computers, and exotic and previously inaccessible quantum phases of
matter. On the other hand, commonplace environmental effects are known
to destroy delicate quantum states, preventing us from observing them
in our everyday lives. One aspect of the field of open quantum systems
deals with understanding the limitations and possibilities of using
the environment of a quantum system to further the aforementioned
goals of applied quantum physics.

While ordinary or closed quantum system dynamics is generated by a
Hamiltonian, open quantum system dynamics is generally not. We now
derive the most general form of evolution of a system coupled to an
environment (see, e.g., \cite{breuer}, Sec.~3.2.1). Assuming that
the dynamics of the universe is generated by a Hamiltonian, the reduced
dynamics of any open quantum system $S$ coupled to an environment
$E$ can be derived starting from the joint Hamiltonian equation of
motion (in units of $\hbar=1$)
\begin{equation}
\frac{d\r_{SE}}{dt}=-i\left[H_{SE},\r_{SE}\right]\,,\label{eq:first}
\end{equation}
where $\r_{SE}$ is the quantum-mechanical density matrix of $S$
and $E$, $H_{SE}$ is the Hamiltonian governing the dynamics of $S$
and $E$, and $[A,B]\equiv AB-BA$ is the commutator of $A$ and $B$.
The reduced density matrix of the system at time $t$ is then
\begin{equation}
\r\left(t\right)\equiv\tr_{E}\left\{ \r_{SE}\left(t\right)\right\} =\tr_{E}\left\{ e^{-iH_{SE}t}\r_{SE}\left(0\right)e^{iH_{SE}t}\right\} \,,
\end{equation}
where $\tr_{E}\{A\}\equiv\sum_{\ell}\bra\ell|A|\ell\ket$ is a tracing
out of the environment and $\{|\ell\ket\}_{\ell}$ is a basis of states
for the environment. Closed-system evolution preserves the \textit{purity
}of quantum states, i.e., $\r(t)$ can be written as a rank-one projection\footnote{We remind the reader that the \textit{rank} of a diagonalizable matrix
is the number of its not necessarily distinct nonzero eigenvalues.} onto some state $|\psi(t)\ket$ {[}$\r(t)=|\psi(t)\ket\bra\psi(t)|$,
assuming that $\r(0)=|\psi(0)\ket\bra\psi(0)|${]}. In the open scenario,
the system and the environment may become entangled under $H_{SE}$
and the resulting initially pure reduced density matrix may become
\textit{mixed} (i.e., not pure). Therefore, throughout this thesis,
we will always denote a quantum state by its density matrix, which
happens to also be an operator on the Hilbert space.

For simplicity, let us now make the assumption that the initial state
factorizes. In other words, $\r_{SE}\left(0\right)=\rin\ot|0\ket\bra0|$,
where $\rin$ is an arbitrary initial state of the system and $|0\ket=|\ell=0\ket$.\footnote{The following derivation is easily extendable to an environment in
an arbitrary initial state, $\rho_{E}=\sum_{\ell}c_{\ell}|\ell\ket\bra\ell|$,
as long as the joint initial state is factorizable, $\r_{SE}\left(0\right)=\rin\ot\rho_{E}$
\cite{Pechukas1994}.} Explicitly writing out $\tr_{E}$, we can massage $\r\left(t\right)$
into the alternative form
\begin{equation}
\r\left(t\right)=\sum_{\ell}E^{\ell}\left(t\right)\rin E^{\ell\dg}\left(t\right)\,,\label{eq:kraus}
\end{equation}
where the \textit{Kraus operators} $E^{\ell}\left(t\right)\equiv\bra\ell|e^{-iH_{SE}t}|0\ket$
operate only on the system Hilbert space. This \textit{Kraus map}
\cite{Kraus1971}, \textit{quantum channel}, or \textit{c}ompletely
\textit{p}ositive \textit{t}race-\textit{p}reserving \textit{(CPTP)}
map is the most general map from density matrices to density matrices
that respects the laws of quantum mechanics, namely (\cite{preskillnotes},
Ch.~3), 
\begin{enumerate}
\item it preserves the trace of $\r$: $\tr\{\rho\left(t\right)\}=1$ for
all $t$, 
\item it preserves positivity: $\r\left(t\right)\geq0$ for all $t$, and 
\item it preserves positivity when acting on a part of a larger system. 
\end{enumerate}
The unitarity of $e^{-iH_{SE}t}$ and completeness of $\{|\ell\ket\}_{\ell}$
can be used to derive the constraint
\begin{equation}
\sum_{\ell}E^{\ell\dg}\left(t\right)E^{\ell}\left(t\right)=I\,,\label{eq:krausnorm}
\end{equation}
where $I$ is the identity on the system. This constraint is equivalent
to Property 1 above. Equation~(\ref{eq:kraus}) can arguably be taken
as the starting point of the entire field of open quantum systems
and there are many types of such quantum channels \cite{Wolf2008}.
In the next Section, we simplify it to derive the Lindbladian \textemdash{}
the most restrictive form of open-system evolution but also the \textit{simplest}
non-trivial extension of Hamiltonian-based quantum mechanics.

\section{What is a Lindbladian?\label{sec:What-is-a}}

Let us study the time evolution of the system density matrix $\r\left(t\right)$
over an infinitesimal time increment $\dd t$, following Ch.~3 of
Ref.~\cite{preskillnotes}. We assume that the time evolution over
only this increment takes the same form as eq.~(\ref{eq:kraus}),
namely,
\begin{equation}
\r\left(t+dt\right)=\sum_{\ell}E^{\ell}\left(dt\right)\r\left(t\right)E^{\ell\dg}\left(dt\right)\equiv\E_{dt}\left[\r\left(t\right)\right]\,.\label{eq:kraus-1}
\end{equation}
This assumption implies that the behavior of $\r\left(t+dt\right)$
depends only on $\r\left(t\right)$ and not any previous times $\tau<t$.
Such a statement is commonly known as the \textit{Markov approximation},
and it is necessary to obtain a Lindbladian from the more general
form (\ref{eq:kraus}). We proceed to approximate $\E_{dt}$ and construct
a bona fide linear differential equation for $\r$, but first let
us sketch what this will accomplish. Namely, we show that evolution
due to any quantum channel $\E_{t}$ is generated by a Lindbladian
when $t$ is ``small'':
\begin{equation}
\E_{dt}=\id+dt\L+\cdots\,\,\,\,\,\,\,\,\,\,\,\,\,\,\text{and}\,\,\,\,\,\,\,\,\,\,\,\,\,\,\L\equiv\lim_{dt\rightarrow0}\frac{\E_{dt}-\id}{dt}\,,
\end{equation}
where $\id$ is the identity channel, $\id\left(\r\right)=\r$, and
$\L$ is a Lindbladian. To determine the precise form of $\L$, we
expand $E^{\ell}\left(dt\right)$ and keep only the terms up to order
$O(dt)$. If evolution of $\r$ is governed by a Hamiltonian $H$
only, then $E^{0}=I-iHdt$ and all other $E^{\ell>0}=0$, yielding
the von Neumann equation $\frac{d\r}{dt}=-i[H,\r]$ analogous to eq.~(\ref{eq:first}).
However, in this more general case, we include an order $O(\sqrt{dt})$
dependence of $E^{\ell}$, which contributes another piece of order
$O(dt)$ since $E^{\ell}$ acts on both sides of $\r\left(t\right)$
simultaneously. Without loss of generality, let us write 
\begin{equation}
E^{0}\left(dt\right)\sim I+\left(-iH+V\right)dt\,\,\,\,\,\,\,\,\,\,\,\,\,\,\,\,\,\text{and}\,\,\,\,\,\,\,\,\,\,\,\,\,\,\,\,\,E^{\ell>0}\left(dt\right)\sim\sqrt{\lind_{\ell}dt}F^{\ell}\,,
\end{equation}
where $\lind_{\ell}$ are real nonzero rates and $V$ is a to-be-determined
Hermitian operator. To determine $V$, we plug these leading-order
forms into eqs.~(\ref{eq:krausnorm}-\ref{eq:kraus-1}) and keep
only terms up to $O(dt)$. Since we want eq.~(\ref{eq:krausnorm})
to be satisfied to this order, we require $V=-\half\sum_{\ell>0}F^{\ell\dg}F^{\ell}$.
Plugging this into eq.~(\ref{eq:kraus-1}) and dividing both sides
by $dt$ yields the \textit{Lindbladian}
\begin{equation}
\frac{d\r}{dt}=\L\left(\r\right)\equiv-i[H,\r]+\half\sum_{\ell>0}\lind_{\ell}\left(2F^{\ell}\r F^{\ell\dg}-F^{\ell\dg}F^{\ell}\r-\r F^{\ell\dg}F^{\ell}\right)\,,\label{eq:def}
\end{equation}
with Hamiltonian $H$,\textit{ jump operators} $F^{\ell}$, and rates
$\lind_{\ell}>0$. The \textit{recycling}, \textit{sandwich}, or \textit{jump
term} $F^{\ell}\cdot F^{\ell\dg}$ acts non-trivially on the state
from both sides simultaneously and is the reason one cannot reduce
the above equation to one involving only a ket-state $|\psi\ket$.
The remaining terms $F^{\ell\dg}F^{\ell}$ can combined with $H$
to form the non-Hermitian operator 
\begin{equation}
K\equiv H-\frac{i}{2}\sum_{\ell>0}\lind_{\ell}F^{\ell\dg}F^{\ell}\,.\label{eq:kkkkkkk}
\end{equation}
This operator generates the \textit{deterministic} or \textit{no-jump}
part of the evolution and its contribution can be expressed as a modified
von Neumann equation using the redefined commutator $[K,\r]^{\#}\equiv K\r-\r K^{\dg}$.
Together, the recycling and deterministic terms conspire to make the
evolution of $\r$ preserve properties 1-3 listed in the previous
Section, \textit{unlike} systems whose time evolution is generated
by a non-Hermitian operator alone.\footnote{Systems where the evolution is governed by $K$ only are also sometimes
called ``open quantum systems'' \cite{Rotter2015} despite not preserving
Properties 1-3 listed above.} In particular, using the cyclic property of the trace, it is straightforward
to show that $\tr\{\r\}$ remains conserved throughout the evolution:
$\tr\{\nicefrac{d\r}{dt}\}=0$. In this way, the jump term compensates
for the decay of probability caused by the deterministic term.

Solving eq.~(\ref{eq:def}) yields the system density matrix for
all times,
\begin{equation}
\r\left(t\right)\equiv e^{t\L}\left(\rin\right)\,,\label{eq:evol}
\end{equation}
where $e^{t\L}$ is a formal power series in $t\L$ and $\rin$ is
the initial state of the system. The exponential can be done directly
by re-expressing the $N\times N$ density matrix as an $N^{2}\times1$
vector (assuming the system Hilbert space is $N$-dimensional), which
in turn allows one to re-express $\L$ as an $N^{2}\times N^{2}$
matrix acting from the left on the vector version of $\r$ (see Sec.~\ref{app:Preliminaries}
for further details). However, there is another method to solve the
above equation for initially pure $\rin=|\psi\ket\bra\psi|$ that
allows one to work with $N\times N$ matrices instead of $N^{2}\times N^{2}$
ones. This method is based on an average over multiple instances (or
\textit{quantum trajectories}) of a procedure applied to the initial
ket-state $|\psi\ket$ \cite{Plenio1998}. During a trajectory of
a typical version of this \textit{unraveling} procedure, the system
is acted on by one instance of the jump term (and then renormalized)
at a discrete set of certain randomly generated times $\{\tau_{n}\}$
and otherwise evolved deterministically under $K$ during the increments
of time between neighboring $\tau_{n}$. More specifically, at time
$t=0$ and assuming one jump operator, a random number $r\in[0,1]$
is generated and the system evolves under $K$ until its norm decreases
to $r$, i.e., until a time $\tau_{1}$ at which $|e^{i\tau_{1}K}|\psi\ket|^{2}=r$.
Then, the evolved state is acted on by the jump term and renormalized,
resulting in the state $|\psi_{1}\ket=Fe^{i\tau_{1}K}|\psi\ket/|Fe^{i\tau_{1}K}|\psi\ket|$.
The same procedure is then repeated with $|\psi_{1}\ket$ until the
desired time of evolution $t$ is reached. It turns out that, in the
limit of an infinite number of such trajectories, the average over
the final states of the trajectories is exactly $\r(t)$ from eq.~(\ref{eq:evol}).
Unraveling is not only useful from a numerical standpoint, but it
also has a physical analogy to the evolution of a continuously measured
quantum state \textit{conditioned} on the results of the measurement.
For example, if the measured quantity is a photon detector and the
sole jump operator represents photon loss, clicks of the detector
are associated with applications of the jump term to the density matrix
while periods of time for which the detector is quiet are associated
with evolution under $K$.
\selectlanguage{american}%

\section{A brief historical context\label{subsec:A-brief-historical}}

\begin{figure}[t]
\includegraphics[width=1\textwidth]{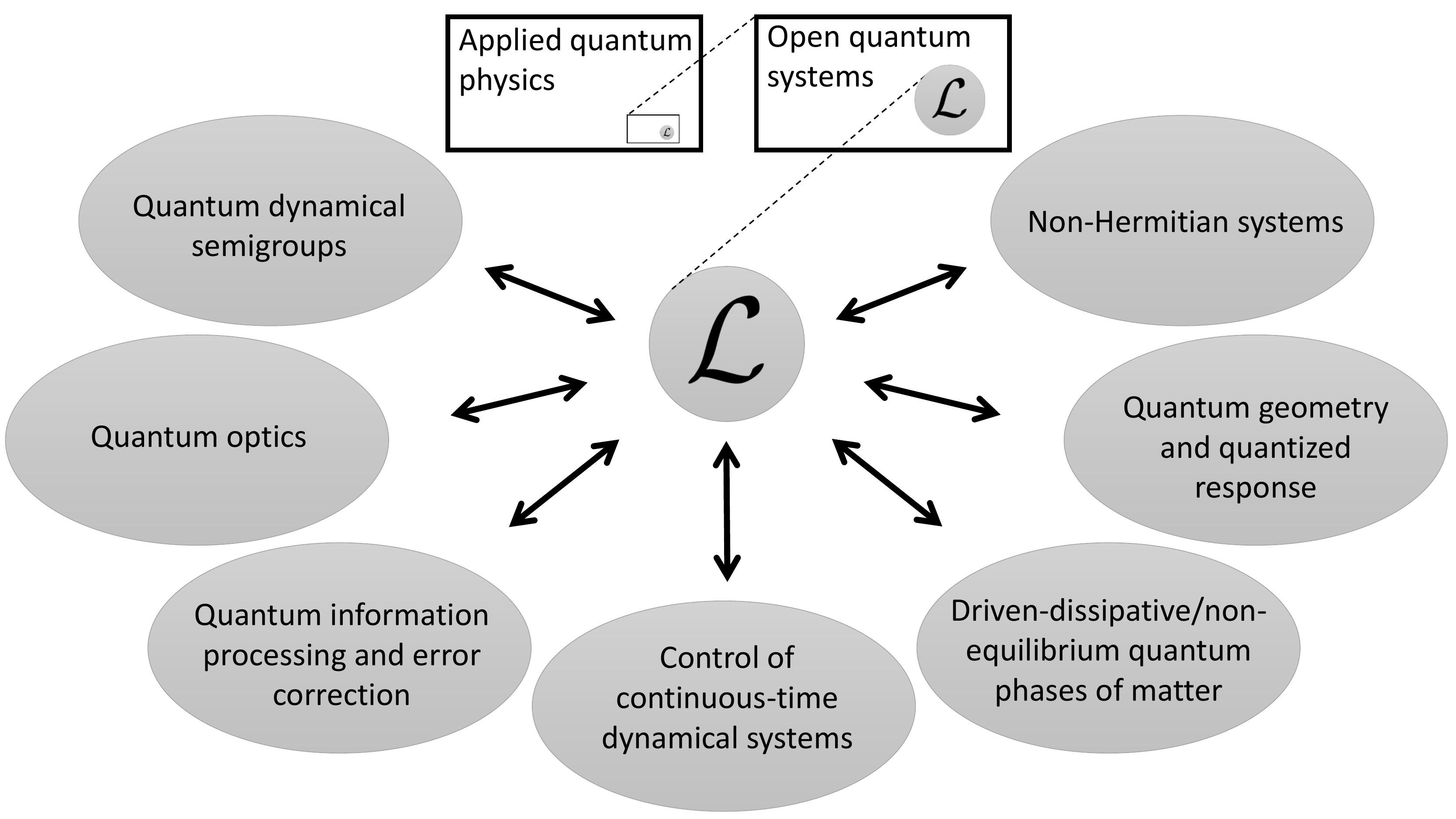}\caption{\label{fig:intro}Sketch of some of the ways in which Lindbladians
(\(\L\)) fit into the field of quantum physics and of connections
between them and other subfields. Details and references are provided
in Sec.~\ref{subsec:A-brief-historical}.}
\end{figure}

\selectlanguage{english}%
Lindbladians (\ref{eq:def}) are also called Louivillians, Lindblad
master equations, Lindblad-Kossakowski differential equations, or
GKS-L equations. Those names mostly stem from two nearly simultaneous
papers, one by Gorini, Kossakowski, and Sudarshan \cite{Gorini1976a}
that derives the equation for finite-dimensional systems and the other
by Lindblad \cite{Lindblad1976} for the infinite-dimensional case.
Variations of eq.~(\ref{eq:def}) were written down earlier \cite{Belavin1969,Lehmberg1970},
although complete positivity of $\L$ was not proven (see \cite{Chruscinski2017}
for a history). Note also a later chain of interest from the high-energy
community (e.g., \cite{Bertlmann2006}), started by Ref.~\cite{Banks1984}.
Equation~(\ref{eq:def}) is most general form for time-independent
Hamiltonian and jump operators, but time-dependent extensions are
also possible \cite{violacitation}. 

The heuristic derivation of $\L$ in the previous section has not
covered all of the conditions on a system and reservoir for which
Lindbladian evolution captures the dynamics of the system. The standard
treatment \cite{alicki_book} covers three common types of derivations
of $\L$ starting from $H_{SE}$: the weak coupling limit, the low
density approximation, and the singular coupling limit. Each of these
relies on specific physical assumptions regarding, e.g., correlation
functions of the environment. For example, in the weak coupling limit
derivation, one typically assumes that (a) correlations of the system
with the environment develop slowly, (b) excitations of the environment
caused by system decay quickly, and (c) terms which are fast-oscillating
when compared to the system timescale of interest can be neglected.
These three approximations are called Born, Markov (as already discussed
above), and rotating wave, respectively. The weak coupling limit is
common in quantum optical scenarios (e.g., Ch.~1 of Ref.~\cite{carmichael1}),
and there are comparisons to other models \cite{Kohen1997} and analyses
studying its validity in said scenarios \cite{rivas2010,Boyanovsky2017}
as well as in heat or electron transport \cite{Wichterich2007,Rivas2016,Elenewski2017}
and many-body systems \cite{Hofer2017,Gonzalez2017}.

Lindbladians lie at the nexus of a multitude of topics in applied
quantum physics and beyond, some of which are listed in Fig.~\ref{fig:intro}.
We conclude this Section by commenting about each of them in approximately
chronological order, noting that this is by no means a complete survey.
\begin{enumerate}
\item Studies of general Lindbladians arose in mathematical physics under
the umbrella of \textit{quantum dynamical semigroups} \cite{Kossakowski1972,Ingarden1975}
\textemdash{} continuous, one-parameter families of dynamical maps
satisfying the (homogeneous) semigroup property $\T_{t}\T_{s}=\T_{t+s}$
(clearly satisfied when $\T_{t}=e^{t\L}$). Early work (e.g., \cite{Alicki1976,Spohn1977,Gorini1978})
often focused on Lindbladians with a unique steady state, on systems
stabilizing states in thermal equilibrium, and on systems satisfying
detailed balance. Unbounded generators, a topic outside the scope
of this thesis, have been rigorously studied in this field (e.g.,
\cite{Dhahri2010,Siemon2017}). Overviews are provided \foreignlanguage{american}{in
the books }\cite{alicki_book,rivas_book,tarasov_book} and an important
relatively recent work is Ref. \cite{baum2}. 
\selectlanguage{american}%
\item Physicists' interest in Lindbladians initiated in \textit{quantum
optics}, where the approximations needed to derive the Lindblad form
are often well-justified. Expositions of this literature are provided
in the books \cite{carmichael1,zoller_book}, interesting methods
to tackle such problems are given in Refs.~\cite{klimov_book,Buchleitner2002},
and some specific older references are discussed in the beginning
of Sec.~\ref{sec:Many-photon-absorption}. The term ``reservoir
engineering'' was coined in this context \foreignlanguage{english}{\cite{Poyatos1996}.}
\item At the turn of the millennium, researchers began to realize the potential
for Lindbladians to stabilize subspaces for \textit{quantum information
and error correction} \cite{Zanardi1997}. In other words, Lindbladians
with multiple steady states are the focus of this field. Their steady-state
subspaces can be decoherence-free subspaces \cite{Lidar1998}, noiseless
subsystems \cite{Knill2000}, and, most generally, information-preserving
structures \cite{BlumeKohout2008,robin}. \foreignlanguage{english}{These
structures }are covered thoroughly in this thesis\foreignlanguage{english}{
and have been realized in various quantum technologies, including
quantum optics \cite{Kwiat2000,Bourennane2004}, liquid-state NMR
\cite{Viola2001,Fortunato2002,Boulant2005}, trapped ions \cite{Kielpinski2001,Roos2006,Barreiro2010,zoller_stabilizers,Schindler2013,Pruttivarasin2015},
and circuit QED \cite{Leghtas2014,S.Touzard}.} \foreignlanguage{english}{Such
Lindbladians can also model an autonomous version of quantum error-correction
\cite{Paz1998,Barnes2000,Ahn2002}, where syndromes are measured and
corrections applied in continuous fashion using $\L$'s jump operators.
Holonomic quantum computation \cite{Zanardi1999} with }such information-preserving
structures\foreignlanguage{english}{ and its associated literature
are discussed in the beginning of Ch.~\ref{ch:5}.}
\item Some researchers began to take a look at Lindbladians from the mathematical
perspective of \textit{control/systems theory} with the goal of controllability
\cite{Altafini2003,Altafini2004} and state/subspace stabilization
\cite{Dirr2008,Dirr2009,Mirrahimi2008,schirmer}, with earlier discussions
focused on more general quantum channels \cite{Lloyd2001}. We discuss
some of this literature in Sec.~\ref{sec:Ground-state-subspaces}.
Researchers also applied the well-developed theory of continuous monitoring
and feedback to state stabilization, at first for the simplest case
of a two-level system \cite{wangwiseman,Wiseman2002} and later more
generally (e.g., \cite{Ticozzi2008}); see the books \cite{wisemanmilburn,gregoratti2009}
for an introduction and further literature. The effect of a generic
measurement on a system can also be described by a Lindbladian, collapsing
the system to a convex combination of pointer states in a quantum-to-classical
transition \cite{Zurek2003}; see Sec.~\ref{sec31} for such a Lindbladian.
\item Physicists started thinking about using many-body Lindbladians to
stabilize \textit{non-equilibrium} (i.e., not in the form of a Gibbs
ensemble) \textit{steady states} (NESS). Most of the initial theory
\cite{Diehl2008,Kraus2008,Garcia-Ripoll2009} was motivated by cold
atom physics, but another early work \cite{Carusotto2009} focused
on optical cavities. Besides state generation, researchers found they
could engineer phase transitions \cite{Prosen2008,Prosen2008a} and
encode the outcome of a quantum computation into the steady state
\cite{Verstraete2009}. These early works initiated the burgeoning
field of \textit{driven-dissipative open systems} \textemdash{} an
extension of the study of Hamiltonian-stabilized phases and phase
transitions to Lindbladians; we discuss some of this literature in
Sec.~\ref{sec:Ground-state-subspaces}. The reader should consult
Refs.~\cite{Noh2017,Noh2017a} for more thorough reviews of driven-dissipative
systems.
\item A series of works \cite{Avron2011,Avron2012a,Avron2012b} thoroughly
investigated Lindbladian generalizations of geometric/quantized adiabatic
response. We investigate these and related efforts in Chs.~\ref{ch:4}-\ref{ch:5}
and provide a review and generalization of the quantum geometric tensor
(\acs{QGT}), a related geometric quantity \foreignlanguage{english}{\cite{provost1980,BerryQGT}},
in Ch.~\ref{ch:6}. These works rely on a rigorous formulation of
the adiabatic theorem for Lindbladians, which is reviewed in Ch.~\ref{ch:5}.
\item A simple ``classical'' open system with loss and gain is one where
time evolution is generated by a \textit{``non-Hermitian Hamiltonian,''}
which is similar to the deterministic term $K$ (\ref{eq:kkkkkkk}).
However, the dependence of the eigenvalues and eigenvectors of $\L$
on $K$, and therefore the connection between Lindbladians and systems
with loss and gain, is not completely understood. We show in Sec.~\ref{subsec:non-Hermitian}
that, for certain types of $\L$, a subspace of the system undergoes
exactly the evolution generated by $K$ and the recycling term merely
takes one out of that subspace. \foreignlanguage{english}{There also
exist methods to extend a given $K$ into full Lindblad form \cite{Bertlmann2006,Selsto2012}.
Such methods, and Lindbladians in general, may be useful as phenomenological
models of resonance decay \cite{Genkin2008}.}
\end{enumerate}

\section{Which Lindbladians are the focus of this work?}

Initial states $\rin$ undergoing Lindbladian evolution (\ref{eq:evol})
evolve into infinite-time or \textit{asymptotic states} $\rout$ for
sufficiently long times \cite{baum2}, 
\begin{equation}
\rin\xrightarrow{t\rightarrow\infty}\rout\equiv\lim_{t\rightarrow\infty}e^{t\L}\left(\rin\right)=e^{-i\hout t}\ppp\left(\rin\right)e^{i\hout t}\,.\label{eq:te}
\end{equation}
The non-unitary effect of Lindbladian time evolution is encapsulated
in the \textit{asymptotic projection} superoperator $\ppp$ (with
$\ppp^{2}=\ppp$). The extra Hamiltonian $\hout$ quantifies any residual
unitary evolution, which persists for all time and, of course, does
not cause any further decoherence. The various asymptotic states $\rout$
are elements of an \textit{asymptotic subspace} \acs{ASH} \textemdash{}
a subspace of the space of operators \acs{OPH} acting on the system
Hilbert space $\h$,
\begin{equation}
\ash\equiv\ppp\oph\,.\label{eq:ashdef}
\end{equation}
The asymptotic subspace attracts all initial states, is free from
the non-unitary effects of $\L$, and any remaining time evolution
within \acs{ASH} is exclusively unitary. The asymptotic subspace
can thus be thought of as a Hamiltonian-evolving subspace embedded
in a larger Lindbladian-evolving space. If \acs{ASH} has no residual
unitary evolution, then $\hout$ is zero and all $\rout$ are \textit{stationary}
or \textit{steady}.

The results presented in this thesis, with the notable exception of
Ch.~\ref{ch:7}, are novel for Lindbladians which (a) admit multiple
steady states and (b) cause one (or more) state \textit{population}
to decay. We clarify these notions in this Section, introducing convenient
notation along the way.

\subsection{Multiple steady states}

When $\dim\ash=1$, only one asymptotic state exists and all $\rin$
converge to it. This is what happens \textit{generically}, i.e., if
one were to pick a Lindbladian at random. Therefore, for the same
reason as Hamiltonians with degenerate ground states, the set of Lindbladians
with multiple steady states is ``small'' (i.e., of measure zero)
compared to the set of all $\L$. In \textit{general} however, \acs{ASH}
may be multi-dimensional and, in that case, the resulting asymptotic
state \textit{will} depend on the initial condition $\rin$. Such
Lindbladians are interesting to study (once again) for similar reasons
as Hamiltonians with degenerate ground states, although one most likely
needs a properly engineered environment to synthesize such Lindbladians.
Subspace stabilization is also the default scenario from the systems
theory perspective \cite{Ticozzi2008,Ticozzi2009,Ticozzi2014}; this
work however focuses on asymptotics and response instead of stabilizability
of particular states or subspaces.

On one hand, \acs{ASH} which can support quantum information are
promising candidates for storing, preserving, and manipulating such
information, particularly when their states can be engineered to possess
favorable features (e.g., topological protection \cite{Diehl2011,bardyn}).
With many experimental efforts (see Sec.~\ref{subsec:A-brief-historical})
aimed at engineering environments admitting nontrivial asymptotic
subspaces, it is important to gain a comprehensive understanding of
any differences between the properties of these subspaces and analogous
subspaces of Hamiltonian systems (e.g., subspaces spanned by degenerate
energy eigenstates).

On the other hand, response properties of \acs{ASH} which do not
necessarily support quantum information can help model experimental
probes into driven-dissipative open systems. Due to, for example,
symmetry \cite{pub011,prozen} or topology \cite{Diehl2011}, the
asymptotic subspace can be degenerate yet not support a qubit. For
example, an \acs{ASH} spanned by two orthogonal pure state projections
$\St_{0}=|\psi_{0}\ket\bra\psi_{0}|$ and $\St_{1}=|\psi_{1}\ket\bra\psi_{1}|$
(see Sec.~\ref{sec31}) only consists of density matrices which are
their convex superpositions, $\rout=c\St_{0}+(1-c)\St_{1}$ with $0\leq c\leq1$,
so no off-diagonal coherences between $\St_{0}$ and $\St_{1}$ are
present. For these and similar cases, standard thermodynamical concepts
\cite{Alicki1976,DEROECK2006,Jaksic2013,Strasberg2013} (see Ref.
\cite{Liu2017} for a review) may not apply and steady states may
no longer be thermal or even full-rank. The work here is directly
tailored to such systems, i.e., those possessing one or more non-equilibrium
steady states whose rank is less than $\dim\h$. The rank constraint
implies the presence of population decay, which is the remaining feature
that we now address.

\subsection{Presence of population decay\label{subsec:Presence-of-decay}}

Unlike Hamiltonians, Lindbladians have the capacity to model decay.
As a result, Lindbladians are often used to describe commonplace non-Hamiltonian
processes (e.g., cooling to a ground state). We define the presence
of \textit{population decay} as the disappearance of at least one
population component of all possible $\rin$ in the infinite-time
limit. In other words, there exist one or more states $|\psi\ket$
such that 
\begin{equation}
\bra\psi|e^{t\L}(|\psi\ket\bra\psi|)|\psi\ket\xrightarrow{t\rightarrow\infty}0\,.\label{eq:popdec}
\end{equation}
Before proceeding, it is important to make a clear distinction between
the decaying and non-decaying parts of the $N$-dimensional system
Hilbert space $\h$. Let us group all non-decaying parts of \acs{OPH},
the space of operators on $\h$, into the upper left corner {[}of
the matrix representation of \acs{OPH}{]} and denote them by the
``upper-left'' block $\ulbig$. Thereby, any completely decaying
parts will be in the complementary $\lrbig$ block, and coherences
between the two will be in the ``off-diagonal'' blocks $\ofbig$.
We can discuss such a decomposition in the familiar language of NMR
\cite{ernstnmr}: the $\ulbig$ block consists of a degenerate ground
state subspace $\{|\psi_{k}\ket\}_{k=0}^{d-1}$ immune to nonunitary
effects, the $\lrbig$ block contains the set of populations decaying
with rates commonly known as $1/T_{1}$, and the $\ofbig$ block is
the set of coherences dephasing with rate $1/T_{2}$. While we have
implicitly assumed only one pair of decay times $T_{1,2}$ for simplicity,
in general every population in $\lrbig$ and coherence in $\ofbig$
has its own decay time.

For the NMR case above, $\ash=\ulbig$ forms a $d^{2}$-dimensional
\textit{decoherence-free subspace (}\foreignlanguage{american}{\acs{DFS}}\textit{)}
\cite{Lidar1998}. More generally, there can be further dephasing
\textit{within} $\ulbig$ without population decay. While we postpone
the discussion of the (many!) types of \acs{ASH} until Ch.~\ref{ch:2},
let us briefly mention one illustrative example. In the NMR case,
\acs{ASH} is spanned by $\{|\psi_{k}\ket\bra\psi_{l}|\}_{k,l=0}^{d-1}$.
If we add a dephasing process which makes all coherences $|\psi_{k}\ket\bra\psi_{l\neq k}|$
between the degenerate ground states decay, \acs{ASH} will then reduce
to the $d$-dimensional subspace spanned by only the state populations
$\{|\psi_{k}\ket\bra\psi_{k}|\}_{k=0}^{d-1}$. This is a case of $\ash\subset\ulbig$,
but $\ash\subseteq\ulbig$ in general. An example of \acs{ASH} which
includes dephasing is shown in the gray region in Fig.~\ref{fig:decomp}(a).

\begin{figure}[t]
\begin{centering}
\includegraphics[width=0.6\columnwidth]{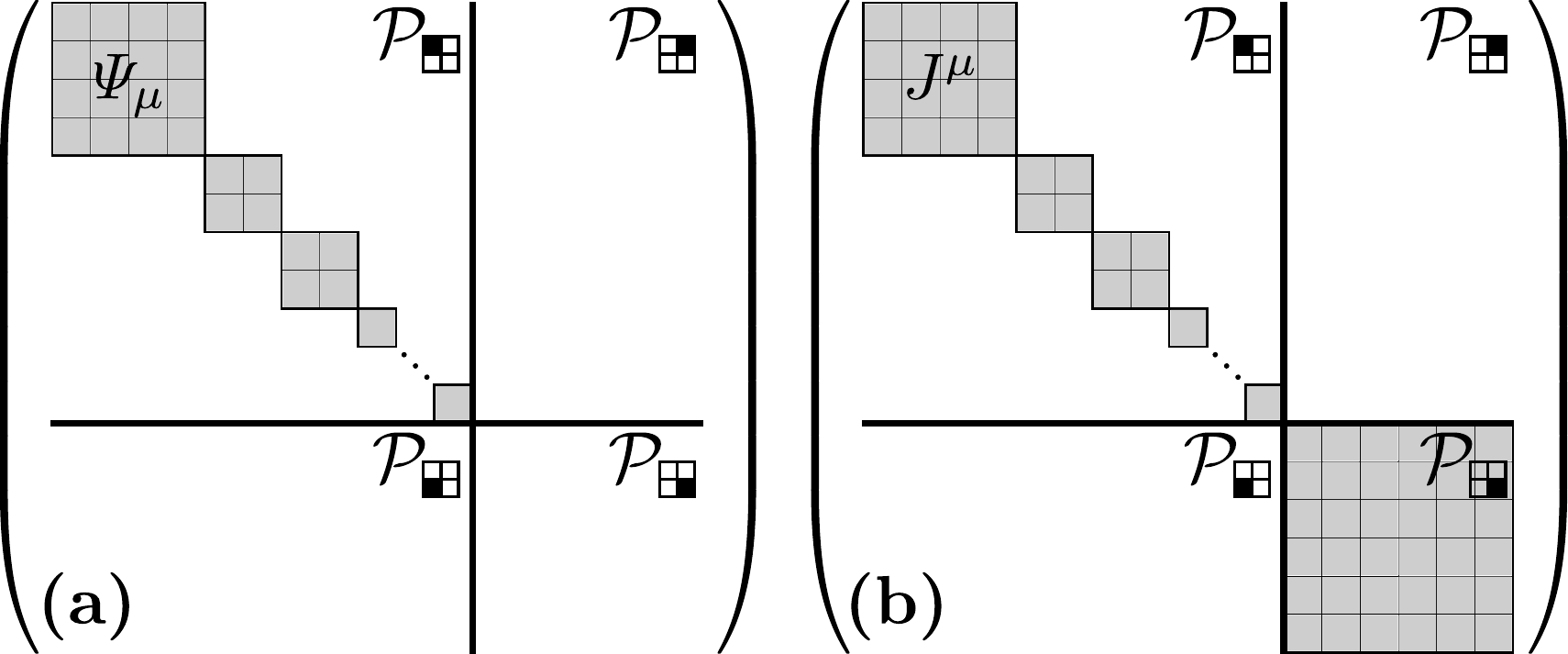}
\par\end{centering}
\caption{\label{fig:decomp}Decompositions of the space of matrices \(\oph\)
acting on a Hilbert space \(\h\) using the projections \(\{\pp,\qq\}\)
defined in (\ref{eq:cond}) and their corresponding superoperator
projections \(\{\R_{\ul},\R_{\ur},\R_{\ll},\R_{\lr}\}\) defined in
(\ref{eq:superproj}). Panel \textbf{(a)} depicts the block diagonal
structure of the asymptotic subspace \(\ash\), which is spanned by
steady-state basis elements \(\Psi_\mu\) (cf. \cite{robin}, Fig.~3).
Panel \textbf{(b)} depicts the subspace of \(\oph\), spanned by conserved
quantities \(J_\mu\). These quantities determine dependence of the
final (asymptotic) state \(\rout\) on the initial state \(\rin\).}
\end{figure}

Let us now formally define the superoperator projections on the blocks.
The subspaces $\ulbig$, $\urbig$, $\llbig$, and $\lrbig$ are just
four blocks (or corners) making up a matrix, so a nonspecialist reader
may simply visualize them without focusing too much on their technical
definitions in this paragraph. Let $\pp$ be the orthogonal operator
projection ($\pp=\pp^{2}=\pp^{\dg}$) on \textit{and only on} the
non-decaying subspace of $\h$ (also, the maximal invariant subspace).
This projection is uniquely defined by the following conditions:
\begin{equation}
\begin{aligned}\rout & =\pp\rout\pp\,\,\,\,\,\,\forall\,\,\,\,\,\,\rout\in\ash\,,\\
\tr\{\pp\} & =\max_{\rout}\{\text{rank}(\rout)\}\,.
\end{aligned}
\label{eq:cond}
\end{equation}
The first condition makes sure that $\pp$ projects onto all non-decaying
subspaces while the second guarantees that $\pp$ does not project
onto any decaying subspace. Naturally, the orthogonal projection onto
the maximal decaying subspace of $\h$ is $\qq\equiv I-\pp$ (with
$\pp\qq=\qq\pp=0$). Therefore, population decay (\ref{eq:popdec})
occurs for any part of $\rin$ that is in $\lrbig$:
\begin{equation}
\qq\r(t)\qq\rightarrow0\,\,\,\,\,\,\,\,\,\,\,\,\,\,\,\text{as}\,\,\,\,\,\,\,\,\,\,\,\,\,\,\,t\rightarrow\infty\,.
\end{equation}
The projection $\pp$ is also the projection on the range or \textit{support}\footnote{\selectlanguage{american}%
\label{fn:The-support-(kernel)}The \textit{support} (\textit{kernel})
of an operator $A$ is the set of all vectors $|x\ket$ which are
not mapped (are mapped) to zero under $A$. The \textit{range} of
$A$ is the set of all vectors $|y\ket$ which are mapped to under
action of $A$ on some other vector $|x\ket$: $A|x\ket=|y\ket$.
If $A|y\ket\notin\ker(A)$ if and only if $|y\ket\in\ran(A)$, then
the range and the support of $A$ are equal (up to the zero vector).
Whenever we use the word support, we will be describing operators
for which this is true.\selectlanguage{english}%
} of $\ppp\left(I\right)$. We define the \textit{four-corners projections}
acting on $A\in\oph$ as follows:
\begin{equation}
\begin{aligned}A_{\ul} & \equiv\R_{\ul}(A)\equiv\pp A\pp\\
A_{\ur} & \equiv\R_{\ur}(A)\equiv\pp A\qq\\
A_{\ll} & \equiv\R_{\ll}(A)\equiv\qq A\pp\\
A_{\lr} & \equiv\R_{\lr}(A)\equiv\qq A\qq\,.
\end{aligned}
\label{eq:superproj}
\end{equation}
By our convention, taking the conjugate transpose of the upper-right
part places it in the lower-left subspace (projection acts \textit{before}
adjoint): $A_{\ur}^{\dg}\equiv(A_{\ur})^{\dg}=(A^{\dg})_{\ll}$. The
superoperators $\R_{\emp}$ (with $\empbig\in\{\ulbig,\urbig,\llbig,\lrbig\}$)
are projections ($\R_{\emp}=\R_{\emp}^{2}$) which partition the identity
$\id$ on \acs{OPH},
\begin{eqnarray}
\R_{\ul}+\R_{\ur}+\R_{\ll}+\R_{\lr} & = & \id\,,\label{eq:superdecomp}
\end{eqnarray}
analogous to $\pp+\qq=I$. They conveniently add, e.g.,
\begin{equation}
\R_{\of}\equiv\R_{\ur}+\R_{\ll}\,\,\,\,\,\,\,\,\,\text{and}\,\,\,\,\,\,\,\,\,\R_{\di}\equiv\R_{\ul}+\R_{\lr}\,.\label{eq:off}
\end{equation}
The subspace $\ofbig\equiv\R_{\of}\oph$ consists of all coherences
between $\pp\h$ and $\qq\h$, and the ``diagonal'' subspace $\dibig\equiv\R_{\di}\oph$
consists of all operators which do not contain any such coherences.
While most of the cases we focus on have $\lrbig\neq0$, for completeness
we mention two cases which \textit{do not} contain decaying subspaces.
The four-corners projections are different from those of the Nakajima-Zwanzig
method \cite{Nakajima1958,Zwanzig1960} and are instead somewhat related
to the Feshbach projection method \cite{Feshbach1958,Rotter2015}
(see Sec. \ref{subsec:non-Hermitian}).

\paragraph{Hamiltonian case:}

If $\L=-i[H,\cdot]$ for some Hamiltonian, any state written in terms
of the $N$ eigenstate projections $|E_{k}\ket\bra E_{k}|$ of $H$
($H|E_{k}\ket=E_{k}|E_{k}\ket$) is a steady state. Therefore, there
is no decaying subspace in Hamiltonian evolution ($\pp=I$).

\paragraph{Unique state case (full-rank):}

In the case of a one-dimensional \acs{ASH}, $\pp$ is the projection
on the rank of the unique steady state $\rout\equiv\varrho$. If the
state's spectral decomposition is $\varrho=\sum_{k=0}^{\du-1}\l_{k}|\psi_{k}\ket\bra\psi_{k}|$
(with $\du$ being the number of nonzero eigenvalues of $\varrho$),
then $\pp=\sum_{k=0}^{\du-1}|\psi_{k}\ket\bra\psi_{k}|$. If all $N$
eigenvalues are nonzero, then $\varrho$ is full-rank (e.g., in a
Gibbs state) and there is no decaying subspace ($\pp=I$).
\selectlanguage{american}%

\section{Summary and reading guide\label{sec:Questions-addressed-and}}

\selectlanguage{english}%
Nontrivial decaying subspaces $\lrbig$ are ubiquitous in actively
researched quantum information schemes (e.g., \cite{cats,Paulisch2015,Reiter2017}).
For example, consider \textit{driven two-photon absorption} \textemdash{}
a bosonic Lindbladian with jump operator 
\begin{equation}
F=\aa^{2}-\a^{2}\,,\label{eq:jump-for-cats}
\end{equation}
bosonic lowering operator $\aa$ ($[\aa,\aa^{\dg}]=I$), and real
non-negative parameter $\a$. The steady states of such a Lindbladian
are the two coherent states $|\a\ket$ and $\left|-\a\right\rangle $,
since both are annihilated by $F$. (In this Section, we always consider
the $\a\gg1$ limit, meaning that the overlap between the two states
is negligible.) In this pseudo double-well system (recently realized
experimentally \cite{Leghtas2014,S.Touzard}), the asymptotic subspace
$\ash=\ulbig$ is spanned by the outer products $|\a\ket\bra\a|$,
$|\a\ket\cb{-\a}$, $\ct{-\a}\bra\a|$, and $\ct{-\a}\cb{-\a}$, and
the projection $\pp$ (\ref{eq:cond}) generating the four-corners
decomposition (in the large $\a$ limit) is
\begin{equation}
\pp=|\a\ket\bra\a|+\ct{-\a}\cb{-\a}\,.
\end{equation}
All states orthogonal to $\left|\pm\a\right\rangle $ constitute $\qq\h$
and the decaying subspace $\lrbig$ is spanned by outer products of
those states. Similarly, the coherences $\ofbig$ are spanned by all
states of the form $\ct{\pm\a}\bra\psi|$ and $\ct{\psi}\cb{\pm\a}$,
where $\cb{\pm\a}\psi\ket=0$. For simplicity, we phrase the concepts
tackled in this work in terms of questions about this particular example.
While many of the questions we consider have already been answered
for this type of \acs{ASH}, it is simpler to first state them in
this context and detail the extensions to the various types of \acs{ASH}
later. We summarize the remaining chapters below, explicitly mention
all collaborators directly related to this thesis, and end with a
brief reading guide.
\begin{description}
\item [{Chapter\,Two}] contains an application of the four-corners decomposition
to Lindbladians and a derivation of when Lindbladians admit a decaying
subspace $\lrbig$ (Thm.~\ref{prop:2}). In process, we present various
types of $\L$, making contact with quantum channel simulation and
non-Hermitian systems. It turns out that, while not all quantum channels
can be expressed as $e^{t\L}$ for any finite $t$ (since $\L$ generates
only Markovian channels), all channels can be embedded in some asymptotic
projection $\ppp=\lim_{t\rightarrow\infty}e^{t\L}$ (assuming $\hout\neq0$).\\
Building on previous results \cite{Jakob2004,pub011}, we then proceed
to derive an analytical formula for the asymptotic projection $\ppp$
from eq.~(\ref{eq:te}) in terms of conserved quantities of $\L$
(Thms.~\ref{thm:dual}-\ref{prop:3}). For driven two-photon absorption
and other examples where $\ash=\ulbig$, the formula states
\begin{equation}
\rout=\ppp\left(\rin\right)=\R_{\ul}\left(\rin\right)-\R_{\ul}\L\R_{\lr}\L^{-1}\R_{\lr}\left(\rin\right)\,.\label{eq:maindecomp}
\end{equation}
One can see that there are two terms. The first ($\R_{\ul}$) term
states the obvious: if one starts in \acs{ASH}, then nothing happens
for all time. The second term shows how an initial state in $\lrbig$
is transferred to an asymptotic state in $\ulbig$. Since there are
no other terms, it turns out that $\rout$ does not depend on any
initial coherences in $\ofbig$. We summarize the ramifications of
a more general version of this formula in terms of the no-leak (\ref{eq:no-leak})
and the clean-leak (\ref{eq:clean-leak}) properties. We then overview
the various types of \acs{ASH}, summarize the results derived in
Chs.~\ref{ch:4}-\ref{ch:5} for each \acs{ASH} type, and build
notation for the various types that is used to derive said results.\\
We conclude the chapter in Sec.~\ref{sec:Relation-of-conserved}
with some statements about the relation between conserved quantities
and symmetries of Lindbladians. Symmetries of a Hamiltonian $H$ are
powerful tools since they can be used to block-diagonalize $H$,
\begin{equation}
H=\bigoplus_{k}H_{k}\,,
\end{equation}
such that all states in a given block $H_{k}$ have the same eigenvalue
of the symmetry operator. Each block can then be further diagonalized
by finding its eigenstates $\{|E_{k}\ket\}$, and each projection
$|E_{k}\ket\bra E_{k}|$ is a steady state ($[H,|E_{k}\ket\bra E_{k}|]=0$).
So if each block is $d_{k}$ dimensional, then $H_{k}$ admits at
least $d_{k}$ steady states. Turning to Lindbladians, one can also
use symmetries \textit{on the superoperator level} to block-diagonalize
$\L$,
\begin{equation}
\L=\bigoplus_{\k}\mathcal{O}_{\k}\,,
\end{equation}
where each superoperator $\mathcal{O}_{\k}$ is not generally in Lindblad
form. Therefore, $\mathcal{O}_{\kappa}$ may not admit any steady
states at all! We will learn in the correspondence from Thm. \ref{thm:dual}
that there are as many steady-state basis elements as there are conserved
quantities. Thus, if $\mathcal{O}_{\k}$ does not admit any steady
states, it also does not admit any conserved quantities. This is a
key difference between Hamiltonians and Lindbladians that breaks the
usual duality between symmetries and conserved quantities known as
Noether's theorem. We explore this idea in Sec.~\ref{sec:Relation-of-conserved}.\\
This chapter is technical and heavily based on \citet*{pub011} and
\citet*{ABFJ}, although some parts have been expanded and further
clarified. The connection to non-Hermitian systems in Sec.~\ref{subsec:non-Hermitian}
is new.
\item [{Chapter\,Three}] reviews examples of conserved quantities from
few-qubit systems and undriven ($\a=0$) two-photon absorption. We
begin with the simplest possible example \textemdash{} a single two-level
system admitting a two-dimensional \acs{ASH} and expound on its relation
to more complicated systems. We continue with a two-qubit example,
studying (among other things) the effect of residual unitary evolution
on $\rout$ {[}$\hout\neq0$ in eq.~(\ref{eq:maindecomp}){]}. In
turns out that in general $\ppp$ depends on $\hout$, leading to
additional dephasing of $\rout$ caused by a ``misalignment'' of
the driving inside \acs{ASH} with the decay coming from outside.
We conclude with a many-body example which allows stabilization of
the ground-state subspace of any frustration-free Hamiltonian (Thm.~\ref{thm:ff}),
making contact with and reviewing earlier work on state stabilization.\\
This chapter is non-technical and is a collection of single-body examples
from Refs.~\cite{pub011,ABFJ}. The many-body example in Sec.~\ref{sec:Ground-state-subspaces}
is synthesized from Refs.~\cite{Ticozzi2009,Ticozzi2012} (see also
\cite{Ticozzi2014}).
\item [{Chapter\,Four}] studies Lindbladian perturbation theory. We begin
with the effect of Hamiltonian perturbations on $\rout$. One can
show numerically \cite{cats} that applying a perturbation of the
type 
\begin{equation}
\hpert=\e(\aa+\aa^{\dg})\label{eq:catpert}
\end{equation}
to the two-photon absorption Lindbladian generates, to linear order,
motion within \acs{ASH} due to the effective Hamiltonian $\hpert_{\ul}\equiv\pp\hpert\pp$.
Related results \cite{Zanardi2014,Zanardi2015} also show that Hamiltonian
perturbations and perturbations to the jump operators of $\L$ generate
unitary evolution within some \acs{ASH} to linear order. Do these
results hold in general? In Sec.~\ref{subsec:hams}, we apply our
formula for $\ppp$ to prove that such perturbations induce unitary
evolution within \textit{all} \acs{ASH} to linear order. This result
also holds for perturbations to the jump operators, $F\rightarrow F+f$,
extending the capabilities of environment-assisted quantum computation
and quantum Zeno dynamics (\cite{Facchi2002,Schafer2014,Signoles2014,Arenz2016};
see also \cite{Anandan1988,Beige2000a}).\\
Extending Ref.~\cite{Oreshkov2010}, we determine the energy scale
governing leakage out of \acs{ASH} due to Hamiltonian perturbations,
jump operator perturbations, and adiabatic evolution (the latter is
shown in the next chapter). Contrary to popular belief, this scale
is not always the dissipative gap of $\L$ \textemdash{} the nonzero
eigenvalue with the smallest real part {[}see eq.~(\ref{eq:dissipative-gap-def}){]}.
On the contrary, this leakage scale is the dissipative gap of $\R_{\thu}\L\R_{\thu}\equiv\L_{\thu}$.
The derivation is given in Sec.~\ref{subsec:Leakage-out-of} for
Hamiltonian/jump operator perturbations and in Sec.~\ref{subsec:adiabatic-response}
for non-adiabatic corrections.\\
More generally, Ch.~\ref{ch:4} contains an application of the four-corners
($\empbig$) partition to the Kubo formula, splitting the formula
into a part within \acs{ASH} which closely corresponds to the ordinary
Hamiltonian-based formula and parts which cause leakage out of \acs{ASH}
and contain non-unitary effects. Theorem \ref{thm:Dyson} provides
an all-order Dyson expansion for cases where $\dim\ash\geq1$, given
a slowly ramping-up perturbation and assuming the initial state is
already in \acs{ASH}.\foreignlanguage{american}{ }An important distinction
from most previous work is that we do not assume anything about $\L$
or its steady states (their number, detailed balance, thermodynamic
equilibrium, etc.), making this analysis applicable to thermodynamic
and quantum computational systems alike. It turns out that the number
of terms to each order in this expansion is equal to a Catalan number.
We finish the chapter by making contact with \foreignlanguage{american}{dark
states }\cite{Kraus2008}\foreignlanguage{american}{, geometric response
}\cite{Avron2012a}\foreignlanguage{american}{, and the effective
operator formalism} \cite{Reiter2012}.\\
All first-order perturbation theory results are from Ref.~\citep{ABFJ}.
The exact Dyson series for all higher orders and connections to previous
work in Sec.~\ref{sec:Exact-all-order-Dyson} are new and will be
studied further in a future publication \cite{pert}.
\item [{Chapter\,Five}] studies the geometric ``phase'' acquired by
$\rout$ after cyclic adiabatic deformations of $\L$. Adiabatically
changing the value of $\a$ of a coherent state $|\a\ket$ in our
two-photon absorption example over a closed path produces a Berry
phase (more generally, a \textit{holonomy}) proportional to the area
enclosed by the path \cite{Chaturvedi1987}. However, does this result
still hold when the coherent state is part of an \acs{ASH} of an
open system? Can Lindbladians induce any additional undesirable effects
in the adiabatic limit for the various types of \acs{ASH}? We extend
previous results \cite{Carollo2006,Sarandy2006,Oreshkov2010,Avron2012a,Albert2015}
to show in Sec.~\ref{subsec:adiabatic-response} that cyclic Lindbladian-based
\cite{Avron2012b} adiabatic evolution of states in \acs{ASH} is
\textit{always} unitary. This result extends the capabilities of holonomic
quantum computation \cite{Zanardi1999,pachos1999,lidarbook_zanardi}
via reservoir engineering. This chapter also contains an application
of the four-corners partition of the leading non-adiabatic corrections
to adiabatic evolution.\\
This chapter is a reshuffled version of the results from Ref.~\citep{ABFJ}.
\item [{Chapter\,Six}] introduces a Lindbladian version of the quantum
geometric tensor (\ac{QGT}) \cite{provost1980,BerryQGT} which encodes
both the curvature associated with the aforementioned adiabatic deformations
and a metric associated with distances between adiabatically connected
steady states. We also construct other geometric tensors and discuss
why these are not always relevant to adiabatic deformations.\\
This chapter is a reshuffled version of the results from Ref.~\citep{ABFJ}.
\item [{Chapter\,Seven}] applies what the results derived from the previous
chapters to the driven two-photon absorption system. We discuss in
detail the steady states for all parameters $\a$ and derive the system's
conserved quantities. Using said conserved quantities, we find out
$\rout$ for various initial states. We then study the leading-order
effect of the aforementioned Hamiltonian perturbation (\ref{eq:catpert})
and two types of noise \textemdash{} dephasing and loss. The perturbation
theory formalism of Ch.~\ref{ch:4} reveals that the system's steady
states are resilient to dephasing noise. We then apply the adiabatic
results of Ch.~\ref{ch:5} and review a way to induce holonomic quantum
computation within \acs{ASH} by adiabatically varying the state parameters
$\ct{\pm\a}$.\\
\foreignlanguage{american}{The conserved quantities and all perturbative
effects are from \citet*{cats} while the holonomic quantum computation
results are from \citet*{Albert2015}. However, }application of perturbation
theory to eq.~(\ref{eq:catpert}) in Sec.~\ref{subsec:A-Hamiltonian-based-gate}
is new, providing a theoretical underpinning to the numerical results
of Ref.~\cite{cats}.
\item [{Chapter\,Eight}] discusses single- and multi-mode extensions of
the driven two-photon absorption example and their relation to cat
codes \foreignlanguage{american}{\cite{Cochrane1999,Leghtas2013b,cats,Albert2015}}
in coherent state quantum information processing.\\
The single-mode extensions were first introduced in the Supplement
of Ref.~\cite{zaki} and further discussed in Refs.~\cite{Albert2015,Bergmann2016}
and Li, Zou, Albert, Muralidharan, Girvin, and Jiang \cite{li2016}.
The multi-mode extensions in Secs.~\ref{sec:Two-mode-cat-codes}-\ref{sec:-mode-cat-codes}
are studied in Albert, Mundhada, Grimm, Touzard, Devoret, and Jiang
\cite{paircat}.
\end{description}
The only recommended prerequisites for understanding about 80\% of
this work are a course in introductory quantum mechanics and a course
in linear algebra (taught by a mathematician). Readers unfamiliar
with Lindbladians are welcome to continue reading the introduction
below and are encouraged to consult the simpler examples in the first
part of Ch.~\ref{ch:3} if things become too general in Chs.~\ref{ch:2},
\ref{ch:4}, or \ref{ch:5}. Applied physicists and quantum computing
researchers familiar with open systems are encouraged to skip ahead
to Chs.~\ref{ch:7}-\ref{ch:8} to enjoy the fruits of the labor
of the previous chapters applied to concrete examples; links to the
general derivations are provided there for convenience. Mathematicians
are encouraged to read the Theorems below and in Ch.~\ref{ch:2}.
An outlook is presented in Ch.~\ref{ch:9}.

\section{A technical introduction\label{app:Preliminaries}}

\subsection{The playground and its features}

This subsection contains a standard introduction into superoperators
and double-ket notation \cite{mukamel,Caves1999,Kosut2009,ernstnmr}.
Lindbladians operate on the space of (linear) operators on $\h$,
or $\oph\equiv\h\ot\h^{\star}$ (also known as Liouville, von Neumann,
or Hilbert-Schmidt space). This space is also a Hilbert space when
endowed with the Hilbert-Schmidt inner product and Frobenius norm
(for $N\equiv\dim\h<\infty$). An operator $A$ in quantum mechanics
is thus both in the space of operators \textit{acting on} ordinary
states and in the space of vectors \textit{acted on} by superoperators.
We denote the two respective cases as $A|\psi\ket$ and $\oo|A\kk$
(for $|\psi\ket\in\h$ and for a superoperator $\oo$). Strictly speaking,
$|\r\kk$ is an $N^{2}$-by-1 vector and $\r$ is an $N$-by-$N$
matrix, and superoperators acting on the vector version of $\r$ are
constructed according to the conversion rule
\begin{align}
A\r B & \leftrightarrow(A\otimes B^{T})|\r\kk\,,
\end{align}
where $B^{T}$ is the transpose of $B$. This extra transpose is necessary
because the bra (row) part of the outer products making up $\r$ is
flipped when $\r$ is written as a (column) vector, $|\psi\ket\bra\phi|\rightarrow|\psi\ket(|\phi\ket)^{\star}$
(see Sec.~2.1.4.5 of \cite{ernstnmr} for details). The double-ket
notation differentiates between superoperators acting on the matrix
or vector versions of operators: $\oo|A\kk$ and $|\oo(A)\kk$ are
written in vector form while $\oo(A)$ is a matrix.

For $A,B\in\oph$, the Hilbert-Schmidt inner product and Frobenius
norm are respectively
\begin{equation}
\bb A|B\kk\equiv\tr\{A^{\dg}B\}\,\,\,\,\,\,\,\,\,\,\text{and}\,\,\,\,\,\,\,\,\,\,\left\Vert A\right\Vert \equiv\sqrt{\bb A|A\kk}\,.\label{eq:inprod}
\end{equation}
The inner product allows one to define an adjoint operation $\dgt$
which complements the adjoint operation $\dg$ on matrices in \acs{OPH}:
\begin{equation}
\bb A|\oo(B)\kk=\bb A|\oo|B\kk=\bb\oo^{\dgt}(A)|B\kk\,.\label{eq:adj}
\end{equation}
Writing $\oo$ as an $N^{2}$-by-$N^{2}$ matrix, $\oo^{\dgt}$ is
just the conjugate transpose of that matrix. For example, if $\oo(\cdot)=A\cdot B$,
then one can use eq.~(\ref{eq:inprod}) to verify that 
\begin{equation}
\oo^{\dgt}(\cdot)=A^{\dg}\cdot B^{\dg}\,.
\end{equation}
Similar to the Hamiltonian description of quantum mechanics, $\oo$
is Hermitian if $\oo^{\dgt}=\oo$. For example, all projections $\R_{\emp}$
from eq.~(\ref{eq:superproj}) are Hermitian.

\begin{figure}
\centering{} \includegraphics[width=0.45\linewidth]{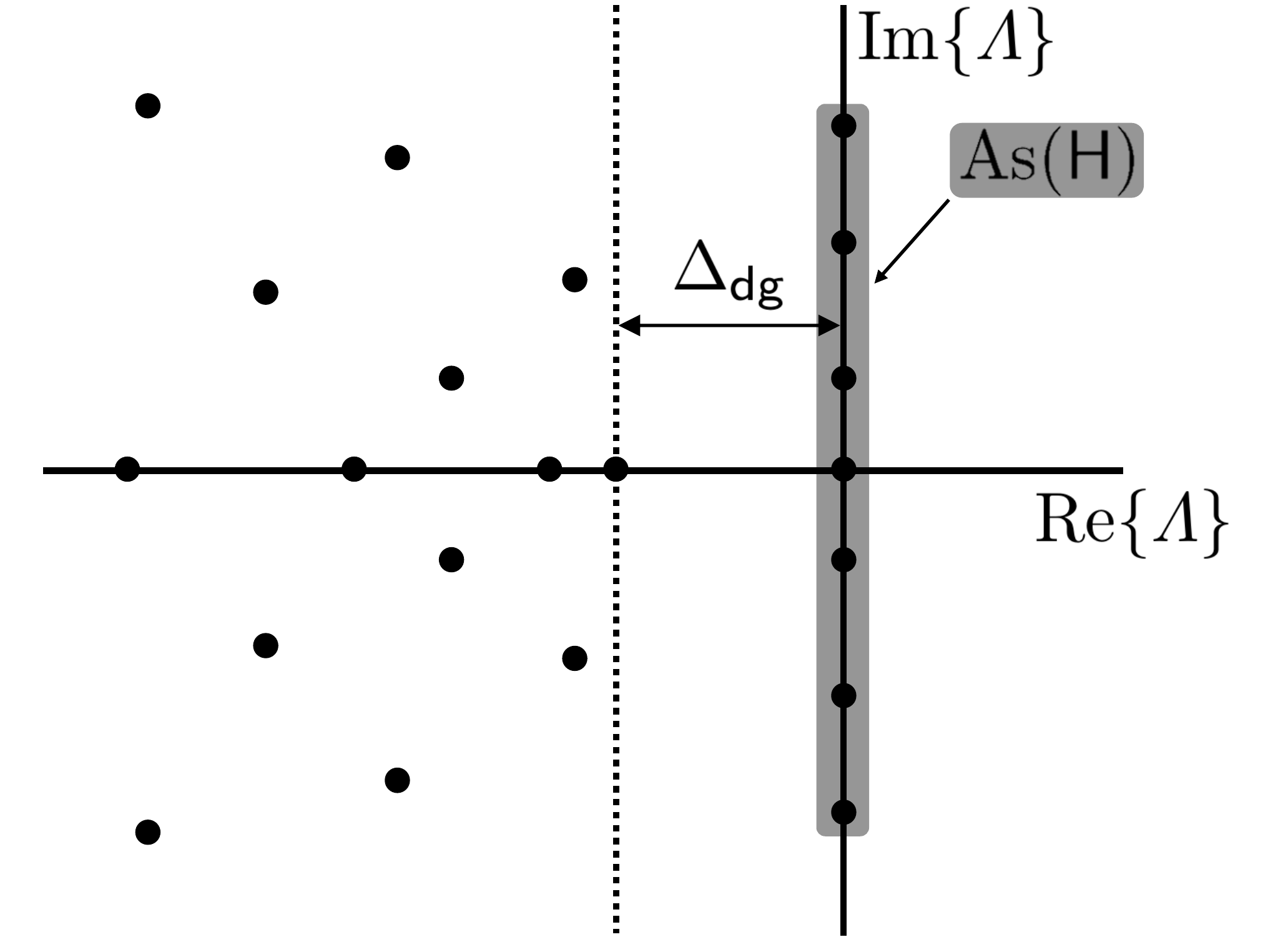}\caption{\label{fig:While--may}A plot of a spectrum of an example \(\L\)
with 21 eigenvalues \(\varLambda\) in the complex plane.}
\label{f1-1} 
\end{figure}

\subsection{More on Lindbladians\label{subsec:More-on-Lindbladians}}

The form of the Lindbladian (\ref{eq:def}) is not unique due to the
following ``gauge'' transformation (for complex $g_{\ell}$), 
\begin{equation}
\begin{aligned}H & \rightarrow H-\frac{i}{2}\sum_{\ell}\lind_{\ell}(g_{\ell}^{\star}F^{\ell}-g_{\ell}F^{\ell\dg})\\
F^{\ell} & \rightarrow F^{\ell}+g_{\ell}I\,,
\end{aligned}
\label{eq:gauge}
\end{equation}
that allows parts of the Hamiltonian to be included in the jump operators
(and vice versa) while keeping $\L$ (\ref{eq:def}) invariant. Note
that there exists a unique ``gauge'' in which $F^{\ell}$ are traceless
(\cite{Gorini1976a}, Thm.~2.2). The Lindbladian is also invariant
under unitary transformations on the jumps: for any unitary matrix
$u_{\ell\ell^{\pr}}$,
\begin{equation}
\sqrt{\lind_{\ell}}F^{\ell}\rightarrow\sum_{\ell^{\prime}}u_{\ell\ell^{\pr}}\sqrt{\lind_{\ell^{\pr}}}F^{\ell^{\pr}}\,.
\end{equation}

It is easy to determine how an observable $A\in\oph$ evolves (in
the Heisenberg picture) using the definition of the adjoint (\ref{eq:adj})
and cyclic permutations under the trace:
\begin{equation}
\L^{\dgt}(A)=-\H(A)+\half\sum_{\ell}\lind_{\ell}\left(2F^{\ell\dg}AF^{\ell}-\left\{ F^{\ell\dg}F^{\ell},A\right\} \right).\label{eq:adjl}
\end{equation}
The superoperator $\H(\cdot)\equiv-i[H,\cdot]$ corresponding to the
Hamiltonian (more precisely, the adjoint representation of $H$) is
therefore \textit{anti-Hermitian} because we have absorbed the ``$i$''
in its definition.

The norm of a wavefunction corresponds to the trace of $\r$ ($\tr\{\r\}=\bb I|\r\kk$);
we have already seen in Sec.~\ref{sec:What-is-a} that it is preserved
under both Hamiltonian and Lindbladian evolution. It is easy to check
that the exponential of any superoperator of the above form preserves
both trace {[}$\bb I|\L|\r\kk=0$ with $I$ the identity operator{]}
and Hermiticity \{$\L(A^{\dg})=[\L(A)]^{\dg}$ as can be verified
from eq.~(\ref{eq:def})\}. However, the norm/purity of $\r$ ($\bb\r|\r\kk=\tr\{\r^{2}\}$)
is not always preserved under Lindbladian evolution.

While $\L$ may not be diagonalizable, one can still obtain information
about the dynamics by observing its eigenvalues (see Fig.~\ref{fig:While--may}).
All eigenvalues $\varLambda$ lie on the non-positive plane and non-real
eigenvalues exist in complex conjugate pairs (hence the symmetry under
complex conjugation). The dots with zero real part in the Figure represent
\acs{ASH}, whose eigenstates have \foreignlanguage{american}{pure
imaginary eigenvalues ($\varLambda=i\la$ for real $\la$) and thus}
survive in the infinite-time limit. The value 
\begin{equation}
\dgg\equiv\min_{\Re\varLambda\neq0}\left|\Re\varLambda\right|\label{eq:dissipative-gap-def}
\end{equation}
is the \textit{dissipative/dissipation/damping/relaxation gap }(also,
\textit{asymptotic decay rate} \cite{Kessler2012}) \textendash{}
the slowest non-zero rate of convergence toward \acs{ASH}.

One can eigendecompose $\L$ to obtain, in principle, the evolution
for all time. Let us first assume that $\L$ is diagonalizable with
eigenvalues $\varLambda$, an additional index $\m$ which labels
any degeneracies for each $\varLambda$, right eigenmatrices $R_{\varLambda\m}$
($\L|R_{\varLambda\m}\kk=\varLambda|R_{\varLambda\m}\kk$), and left
eingematrices $L^{\varLambda\m}$ ($\L^{\dgt}|L^{\varLambda\m}\kk=\varLambda^{\star}|L^{\varLambda\m}\kk$).
Then, the evolution superoperator (\ref{eq:evol}) can be written
as
\begin{equation}
e^{t\L}|\rin\kk=\sum_{\varLambda,\m}e^{\varLambda t}|R_{\varLambda\m}\kk\bb L^{\varLambda\m}|\rin\kk\,.\label{eq:expdeg}
\end{equation}
If $\L$ is not diagonalizable, there exists at least one Jordan block
of $\L$ (in Jordan normal form) which has only one eigenmatrix, with
the remaining matrix basis elements in the support of the block making
up the block's \textit{generalized eigenmatrices} (e.g., \cite{puri},
Sec.~10.2). Exponentiating such an $\L$ brings about extra powers
of $t$ in front of the exponent $e^{\varLambda t}$ above as well
as off-diagonal elements of the form $|R_{\varLambda\m}\kk\bb L^{\varLambda,\n\neq\m}|$.
For example, if $\varLambda$ has a two-dimensional Jordan block with
right eigenmatrix $|R_{\varLambda0}\kk$ ($\L|R_{\varLambda0}\kk=\varLambda|R_{\varLambda0}\kk$)
and generalized right eigenmatrix $|R_{\varLambda1}\kk$ ($\L|R_{\varLambda1}\kk=|R_{\varLambda0}\kk+\varLambda|R_{\varLambda1}\kk$),
then $e^{t\L}$ on that block is
\begin{equation}
e^{\varLambda t}\left(|R_{\varLambda0}\kk\bb L^{\varLambda0}|+t|R_{\varLambda1}\kk\bb L^{\varLambda0}|+|R_{\varLambda1}\kk\bb L^{\varLambda1}|\right)\equiv e^{\varLambda t}\begin{pmatrix}1 & t\\
0 & 1
\end{pmatrix}\,.\label{eq:expjor}
\end{equation}
Let us partition the Jordan normal form of $\L$ into blocks that
are either diagonal or have an upper diagonal of ones. Since there
could exist blocks associated with a particular $\varLambda$ which
contain both diagonal and off-diagonal sub-blocks, the sum over $\varLambda$
has to include each sub-block separately. Generalizing eq.~(\ref{eq:expdeg}),
the full expansion of $e^{t\L}$ is then
\begin{equation}
e^{t\L}|\rin\kk=\sum_{\varLambda,\m}e^{\varLambda t}|R_{\varLambda\m}\kk\sum_{\n\geq\m}\frac{(\d_{\varLambda}t)^{\n-\m}}{(\n-\m)!}\bb L^{\varLambda\n}|\rin\kk\,,\label{eq:expansion}
\end{equation}
with $\m,\n\in\{0,1,\cdots\}$ indexing either the degeneracy of the
eigenspace of each $\varLambda$ if the Jordan block associated with
$\varLambda$ is diagonal ($\d_{\varLambda}=0$; here $0^{0}=1$)
or indexing the generalized eigenmatrices of $\varLambda$'s Jordan
block if the block is not diagonal ($\d_{\varLambda}=1$).

\selectlanguage{american}%
Equations (\ref{eq:expjor}-\ref{eq:expansion}) immediately reveal
that all Jordan blocks with pure imaginary eigenvalues $\varLambda=i\la$
are diagonal (\cite{schirmer}, Sec.~5; \cite{baum2}, Thm.~18;
\cite{wolf2010}, Prop. 6.2). By contradiction, if one assumes that
$\L$ is not diagonalizable in the subspace of the Jordan normal form
with diagonals of zero real part, then exponentiating those Jordan
blocks causes the dynamics to diverge as $t\rightarrow\infty$.

\selectlanguage{english}%
\begin{table}
\centering{}%
\begin{tabular}{cc}
\toprule 
\addlinespace
Operator & Superoperator\tabularnewline
Notation & Notation\tabularnewline\addlinespace
\midrule
\addlinespace
$\L(\r)$ & $\L|\r\kk$\tabularnewline\addlinespace
$\tr\{A^{\dg}\oo(\r)\}$ & $\bb A|\oo|\r\kk$\tabularnewline\addlinespace
$-i[H,\r]$ & $\H|\r\kk$\tabularnewline\addlinespace
$-i[\hpert,\r]$ & $\spert|\r\kk$\tabularnewline\addlinespace
$U\r U^{\dg}$ & $\U|\r\kk$\tabularnewline\addlinespace
$S\r S^{\dg}$ & $\dist|\r\kk$\tabularnewline\addlinespace
\bottomrule
\end{tabular}~~%
\begin{tabular}{cc}
\toprule 
\addlinespace
Operator & Superoperator\tabularnewline
Notation & Notation\tabularnewline\addlinespace
\midrule
\addlinespace
$A_{\ul}\equiv PAP$ & $|A_{\ul}\kk\equiv\R_{\ul}|A\kk$\tabularnewline\addlinespace
$A_{\ur}\equiv PAQ$ & $|A_{\ur}\kk\equiv\R_{\ur}|A\kk$\tabularnewline\addlinespace
$A_{\ll}\equiv QAP$ & $|A_{\ll}\kk\equiv\R_{\ll}|A\kk$\tabularnewline\addlinespace
$A_{\lr}\equiv QAQ$ & $|A_{\lr}\kk\equiv\R_{\lr}|A\kk$\tabularnewline\addlinespace
$P\L(Q\r Q)P$ & $\R_{\ul}\L\R_{\lr}|\r\kk$\tabularnewline\addlinespace
$i\tr\{H[\St_{\m},\St_{\n}]\}$ & $\bb\St_{\m}|\H|\St_{\n}\kk$\tabularnewline\addlinespace
\bottomrule
\end{tabular}\caption{\label{tab:glossary}\protect\ytableausetup{boxsize = 2pt}Comparison
of operator and superoperator notations for symbols used throughout
the text (cf. Table 3.2 in \cite{mukamel}). \(\L\) is a Lindbladian
superoperator (\ref{eq:def}), \(\oo\) is a superoperator, \(A\)
is an operator, and \(\r\) is a density matrix. Hamiltonians \(H\)
and \(\hpert\) have corresponding Hamiltonian superoperators \(\H\)
and \(\spert\), respectively. Unitary operators \(U\) and \(S\)
have corresponding unitary superoperators \(\U\) and \(\dist\),
respectively. The projection \(\pp\) (\ref{eq:cond2}) projects onto
the largest subspace whose states do not decay under \(\L\) and \(\qq\equiv I-\pp\)
with \(I\) the identity. The last two entries respectively represent
the part \(\R\ytableaushort{ {*(black)} {} , {} {} }\L\R\ytableaushort{ {} {} , {} {*(black)} }\)
of the projection decomposition of \(\L\) (\ref{eq:def}) acting
on \(\r\) and a (superoperator) matrix element of \(\H\) in terms
of a Hermitian matrix basis \(\{\St_{\m}\}\).}
\end{table}

\subsection{Double-bra/ket basis for steady states}

We now bring in intuition from Hamiltonian-based quantum mechanics
by building bases for \acs{OPH} from those for $\h$. Given any orthonormal
basis $\{|\phi_{k}\ket\}_{k=0}^{N-1}$ for $\h$, one can construct
the corresponding orthonormal (under the trace) outer product basis
for \acs{OPH}, 
\begin{equation}
\{|\varPhi{}_{kl}\kk\}_{k,l=0}^{N-1}\,,\text{\,\,\,\,\,\,\,\,\,where\,\,\,\,\,\,\,\,\,}\varPhi_{kl}\equiv|\phi_{k}\ket\bra\phi_{l}|\,.\ 
\end{equation}
The analogy with quantum mechanics is that the matrices $\Phi_{kl}\leftrightarrow|\Phi_{kl}\kk$
and $\Phi_{kl}^{\dg}\leftrightarrow\bb\Phi_{kl}|$ are vectors in
the vector space \acs{OPH} and superoperators $\oo$ are linear operators
on those vectors. Furthermore, one can save an index and use properly
normalized Hermitian matrices $\Ga_{\k}^{\dg}=\Ga_{\k}$ to form an
orthonormal basis $\{|\Ga_{\kappa}\kk\}_{\kappa=0}^{N^{2}-1}$: 
\begin{equation}
\bb\Ga_{\kappa}|\Ga_{\lambda}\kk\equiv\tr\{\Ga_{\kappa}^{\dg}\Ga_{\lambda}\}=\tr\{\Ga_{\kappa}\Ga_{\lambda}\}=\d_{\kappa\lambda}\,.
\end{equation}
Each $\G_{\k}$ consists of Hermitian linear superpositions of the
outer products $|\phi_{k}\ket\bra\phi_{l}|$ and is \textit{not} a
density matrix. For example, an orthonormal Hermitian matrix basis
for \acs{OPH} with $\h$ two-dimensional consists of the identity
matrix and the three Pauli matrices, all normalized by $\nicefrac{1}{\sqrt{2}}$.
An example for $N=3$ is the set of properly normalized Gell-Mann
matrices.

It is easy to see that the coefficients in the expansion of any Hermitian
operator in such a matrix basis are real. For example, the coefficients
$c_{\kappa}$ in the expansion of a density matrix, 
\begin{equation}
|\r\kk=\sum_{\kappa=0}^{N^{2}-1}c_{\kappa}|\Ga_{\kappa}\kk\,\,\,\,\,\,\,\,\,\text{with}\,\,\,\,\,\,\,\,\,c_{\kappa}=\bb\Ga_{\kappa}|\r\kk\,,
\end{equation}
are clearly real and represent the components of a generalized Bloch/coherence
vector \cite{schirmer,alicki_book}. Furthermore, defining 
\begin{equation}
\oo_{\kappa\lambda}\equiv\bb\Ga_{\kappa}|\oo|\Ga_{\lambda}\kk\equiv\tr\{\Ga_{\kappa}^{\dg}\oo(\Ga_{\l})\}
\end{equation}
for any superoperator $\oo$, one can write
\begin{eqnarray}
\oo & = & \sum_{\kappa,\lambda=0}^{N^{2}-1}\oo_{\kappa\lambda}|\Ga_{\kappa}\kk\bb\Ga_{\lambda}|\,.
\end{eqnarray}
There are many physical $\oo$ for which the ``matrix'' elements
$\oo_{\k\l}$ are real. For example, we define the superoperator equivalent
of a Hamiltonian $H$ acting on a state $\r$ as $\H(\r)\equiv-i[H,\r]$
{[}so that if $H$ generates time evolution, $\p_{t}\r=\H(\r)${]}.
For this case, it is easy to show that matrix elements $\H_{\k\l}$
are real using cyclic permutations under the trace and Hermiticity
of the $\Ga$'s:
\begin{equation}
\H_{\kappa\lambda}^{\star}=\bb\G_{\kappa}|\H|\G_{\lambda}\kk^{\star}=i\tr\{\Ga_{\l}[H,\Ga_{\k}]\}=-i\tr\{\Ga_{\k}[H,\Ga_{\l}]\}=\bb\G_{\kappa}|\H|\G_{\lambda}\kk=\H_{\k\l}\,.\label{eq:hamder}
\end{equation}
This calculation easily extends to all Hermiticity-preserving $\oo$,
i.e., superoperators such that $\oo(A^{\dg})=[\oo(A)]^{\dg}$ for
all operators $A$. 

Given a Lindbladian, one can provide necessary and sufficient conditions
under which it generates Hamiltonian time evolution. This early key
result in open quantum systems can be used to determine whether a
perturbation generates unitary evolution.
\selectlanguage{american}%
\begin{thm}
[When Lindbladians generate unitary evolution \cite{Kossakowski1972}]\foreignlanguage{english}{\label{prop:1}The
matrix $\L_{\kappa\lambda}=\bb\Ga_{\kappa}|\L|\Ga_{\lambda}\kk$ is
real. Moreover, 
\begin{equation}
\L_{\lambda\kappa}=-\L_{\kappa\lambda}\,\,\,\Leftrightarrow\,\,\,\L=-i\left[H,\cdot\right]\text{ \textit{with} \textit{Hamiltonian} }H.\label{eq:prop1}
\end{equation}
}
\end{thm}
\selectlanguage{english}%
\begin{proof}
To prove reality, use the definition of the adjoint of $\L$, Hermiticity
of $\Ga_{\kappa}$, and cyclicity under the trace: 
\begin{eqnarray}
\L_{\kappa\lambda}^{\star} & = & \bb\Ga_{\lambda}|\L^{\dgt}|\Ga_{\kappa}\kk=\bb\L(\Ga_{\lambda})|\Ga_{\kappa}\kk=\bb\Ga_{\kappa}|\L|\Ga_{\lambda}\kk=\L_{\kappa\lambda}\,.
\end{eqnarray}
$\Leftarrow$ Assume $\L$ generates unitary evolution. Then there
exists a Hamiltonian $H$ such that $\L|\Ga_{\kappa}\kk=-i|[H,\Ga_{\kappa}]\kk$
and $\L$ is antisymmetric:
\begin{eqnarray}
\L_{\lambda\kappa} & = & -i\tr\{\Ga_{\lambda}[H,\Ga_{\kappa}]\}=i\tr\{\Ga_{\kappa}[H,\Ga_{\lambda}]\}=-\L_{\kappa\lambda}\,.
\end{eqnarray}
$\Rightarrow$\footnote{An alternative way to prove this part is to observe that all eigenvalues
of $\L$ lie on the imaginary axis and use Thm.~18-3 in \cite{baum2}.} Assume $\L_{\kappa\lambda}$ is antisymmetric, so $\L^{\dgt}=-\L$.
Then the dynamical semigroup $\{e^{t\L};\,t\geq0\}$ is isometric
(norm-preserving): let $t\geq0$ and $|A\kk\in\oph$ and observe that
\begin{equation}
\bb e^{t\L}(A)|e^{t\L}(A)\kk=\bb A|e^{-t\L}e^{t\L}|A\kk=\bb A|A\kk\,.
\end{equation}
Since it is clearly invertible, $e^{t\L}:\oph\rightarrow\oph$ is
a surjective map. All surjective isometric one-parameter dynamical
semigroups can be expressed as $e^{t\L}(\r)=U_{t}\r U_{t}^{\dg}$
with $U_{t}$ belonging to a one-parameter unitary group $\{U_{t};\,t\in\mathbb{R}\}$
acting on $\h$ (\cite{Kossakowski1972}, Thm.~6). Therefore, there
exists a Hamiltonian $H$ such that $U_{t}=e^{-iHt}$ and $\L(\r)=-i[H,\r]$.\selectlanguage{english}%
\end{proof}

\inputencoding{latin9}\newpage{}\foreignlanguage{english}{}%
\begin{minipage}[t]{0.5\textwidth}%
\selectlanguage{english}%
\begin{flushleft}
\begin{singlespace}``\textit{{[}...{]} the orientation in the literature
on semigroups is towards the proof of rigorous mathematical results
and hence the connections to quantum optics applications are somewhat
indirect.}''\end{singlespace}
\par\end{flushleft}
\begin{flushleft}
\hfill{}\textendash{} Howard J. Carmichael
\par\end{flushleft}\selectlanguage{english}%
\end{minipage}

\chapter{The asymptotic projection and conserved quantities\label{ch:2}}
\selectlanguage{english}%

\section{Four-corners partition of Lindbladians, with examples\label{sec:Four-corners-partition-of}\label{app:decomp}}

From the previous chapter, we learned that the four-corners projections
(\ref{eq:superproj}) partition every operator $A\in\oph$ into four
independent parts. Combining this notation with the vectorized or
\textit{double-ket} notation for matrices in \acs{OPH} (see Sec.~\ref{app:Preliminaries}),
we can express any $A$ as a vector whose components are the respective
parts. The following are therefore equivalent,
\begin{equation}
A=\begin{pmatrix}A_{\ul} & A_{\ur}\\
A_{\ll} & A_{\lr}
\end{pmatrix}\,\,\,\,\longleftrightarrow\,\,\,\,|A\kk=\left[\begin{array}{c}
|A_{\ul}\kk\\
|A_{\of}\kk\\
|A_{\lr}\kk
\end{array}\right],\label{eq:mat}
\end{equation}
and $A_{\of}=A_{\ur}+A_{\ll}$. With $A$ written as a block vector,
superoperators can now be represented as 3-by-3 block matrices acting
on said vector. Note that we use square-brackets for partitioning
\textit{superoperators} and parentheses for \textit{operators} in
\acs{OPH} {[}as in Fig.~\ref{fig:decomp} and eq.~(\ref{eq:mat}){]}.
We will do so with the Lindbladian $\L$ (\ref{eq:def}). Recall that
\begin{equation}
\L(\r)=-i[H,\r]+\half\sum_{\ell}\lind_{\ell}\left(2F^{\ell}\r F^{\ell\dg}-F^{\ell\dg}F^{\ell}\r-\r F^{\ell\dg}F^{\ell}\right)\label{eq:def-1}
\end{equation}
with Hamiltonian $H$,\textit{ }jump operators $F^{\ell}\in\oph$,
and positive rates $\lind_{\ell}$. By writing $\L=\id\L\id$ using
eqs.~(\ref{eq:superdecomp}) and (\ref{eq:off}), we find that 
\begin{eqnarray}
\L & = & \left[\begin{array}{ccc}
\L_{\ul}\, & \,\R_{\ul}\L\R_{\of}\, & \,\R_{\ul}\L\R_{\lr}\\
0 & \L_{\of} & \,\R_{\of}\L\R_{\lr}\\
0 & 0 & \,\L_{\lr}
\end{array}\right]\,,\label{eq:gen}
\end{eqnarray}
where $\L_{\emp}\equiv\R_{\emp}\L\R_{\emp}$. Note that $\L_{\ul}$
is a bona fide Lindbladian governing evolution within $\ulbig$. The
reason for the zeros in the first column is the inability of $\L$
to take anything out of $\ulbig$ (stemming from the definition of
the four-corners projections). This turns out to be sufficient for
$\R_{\lr}\L\R_{\of}$ to also be zero, leading to the block upper-triangular
form above. These constraints on $\L$ translate to well-known constraints
on the Hamiltonian and jump operators as follows.
\begin{thm}
[When Lindbladians generate decay \cite{baum2,Shabani2005,Ticozzi2008,ABFJ}]\label{prop:2}Let
$\{\pp,\qq\}$ be projections on $\h$ and $\{\R_{\ul},\R_{\ur},\R_{\ll},\R_{\lr}\}$
be their corresponding projections on $\oph$. Then
\begin{eqnarray}
F_{\ll}^{\ell} & = & 0\text{ for all }\ell\label{eq:fzero}\\
H_{\ur} & = & -\frac{i}{2}\sum_{\ell}\lind_{\ell}F_{\ul}^{\ell\dg}F_{\ur}^{\ell}\,.\label{eq:conham}
\end{eqnarray}
\end{thm}
\begin{proof}
By definition (\ref{eq:cond}), $\ulbig$ is the smallest subspace
of \acs{OPH} containing all asymptotic states. Therefore, all states
evolving under $\L$ converge to states in $\ulbig$ as $t\rightarrow\infty$
(\cite{baum2}, Thm.~2-1). This implies invariance, i.e., states
$\r_{\ul}=\R_{\ul}(\r)$ remain there under application of $\L$:
\begin{equation}
\L(\r_{\ul})=\L\R_{\ul}(\r)=\R_{\ul}\L\R_{\ul}(\r)\,.\label{eq:lp}
\end{equation}
Applying $\R_{\lr}$, we get $\R_{\lr}\L\R_{\ul}(\r)=\sum_{\ell}\lind_{\ell}F_{\ll}^{\ell}\r F_{\ll}^{\ell\dg}=0$
since the projections are mutually orthogonal. Taking the trace, $\bb I|\R_{\lr}\L\R_{\ul}|\r\kk=\sum_{\ell}\lind_{\ell}\tr\{\r F_{\ll}^{\ell\dg}F_{\ll}^{\ell}\}=0$.
If $\r$ is a full rank density matrix ($\text{rank}\{\r\}=\tr\{\pp\}$),
then each summand above is non-negative (since $\lind_{\ell}>0$ and
$F_{\ll}^{\ell\dg}F_{\ll}^{\ell}$ are positive semidefinite). Thus
the only way for the above to hold for all $\r$ is for $F_{\ll}^{\ell\dg}F_{\ll}^{\ell}=0$
for all $\ell$, which implies that $F_{\ll}^{\ell}=0$. Applying
$\R_{\ur}$ to eq.~(\ref{eq:lp}) and simplifying using $F_{\ll}^{\ell}=0$
gives
\begin{equation}
\R_{\ur}\L\R_{\ul}(\r)=\pp\r\left(iH_{\ur}-\half\sum_{\ell}\lind_{\ell}F_{\ul}^{\ell\dg}F_{\ur}^{\ell}\right)=0\,,
\end{equation}
implying the condition on $H_{\ur}$.
\end{proof}
The constraints on $H_{\of}$ and $F_{\ll}^{\ell}$ (due to Hermiticity,
$H_{\ll}=H_{\ur}^{\dg}$) leave only their complements as degrees
of freedom. The four-corners decomposition provides simple expressions
for the surviving matrix elements of $\L$ (\ref{eq:gen}) in terms
of $H_{\di},F_{\thr}^{\ell}$; these are shown below. From eq.~(\ref{eq:conham}),
one can see that $H_{\ur}$ in general depends on the nonzero jump
operator rates $\lind_{\ell}$, demonstrating an intricate cancellation
of Lindbladian effects via a Hamiltonian term. However, $\pp$ is
independent of $\lind_{\ell}$ in many physically relevant ($\k$\textit{-robust}
\cite{Ticozzi2008}) Lindbladians, and in those cases $H_{\of}=0$
and either $F_{\ul}^{\ell}=0$ or $F_{\ur}^{\ell}=0$ for each $\ell$.

We now list all of the matrix elements of $\L$ (\ref{eq:gen}) and
mention important special cases of \foreignlanguage{american}{\acs{DFS}}
type, making contact with previous works and applications. Based on
conditions (\ref{eq:cond}) and after simplifications due to Thm.~\ref{prop:2},
the non-zero elements of eq.~(\ref{eq:gen}) acting on a Hermitian
matrix $\r=\r_{\ul}+\r_{\of}+\r_{\lr}$ are\begin{subequations}
\begin{align}
\L_{\ul}(\r) & =-i\left[H_{\ul},\r_{\ul}\right]+\half\sum_{\ell}\lind_{\ell}\left(2F_{\ul}^{\ell}\r_{\ul}F_{\ul}^{\ell\dg}-F_{\ul}^{\ell\dg}F_{\ul}^{\ell}\r_{\ul}-\r_{\ul}F_{\ul}^{\ell\dg}F_{\ul}^{\ell}\right)\label{eq:term0}\\
\L_{\ur}(\r) & =-i\left(H_{\ul}\r_{\ur}-\r_{\ur}H_{\lr}\right)+\half\sum_{\ell}\lind_{\ell}\left[2F_{\ul}^{\ell}\r_{\ur}F_{\lr}^{\ell\dg}-F_{\ul}^{\ell\dg}F_{\ul}^{\ell}\r_{\ur}-\r_{\ur}(F^{\ell\dg}F^{\ell})_{\lr}\right]\\
\L_{\ll}(\r) & =[\L_{\ur}(\r)]^{\dg}\phantom{+\half\sum_{\ell}}\\
\L_{\lr}(\r) & =-i\left[H_{\lr},\r_{\lr}\right]+\half\sum_{\ell}\lind_{\ell}\left[2F_{\lr}^{\ell}\r_{\lr}F_{\lr}^{\ell\dg}-(F^{\ell\dg}F^{\ell})_{\lr}\r_{\lr}-\r_{\lr}(F^{\ell\dg}F^{\ell})_{\lr}\right]\label{eq:llr}\\
\R_{\ul}\L\R_{\of}(\r) & =\sum_{\ell}\lind_{\ell}\left(F_{\ul}^{\ell}\r_{\of}F_{\ur}^{\ell\dg}-\r_{\of}F_{\ur}^{\ell\dg}F_{\ul}^{\ell}\right)+H.c.\phantom{+\half\sum_{\ell}}\label{eq:term}\\
\R_{\of}\L\R_{\lr}(\r) & =\sum_{\ell}\lind_{\ell}\left(F_{\ur}^{\ell}\r_{\lr}F_{\lr}^{\ell\dg}-F_{\ul}^{\ell\dg}F_{\ur}^{\ell}\r_{\lr}\right)+H.c.\phantom{+\half\sum_{\ell}}\\
\R_{\ul}\L\R_{\lr}(\r) & =\sum_{\ell}\lind_{\ell}F_{\ur}^{\ell}\r_{\lr}F_{\ur}^{\ell\dg}\,.\phantom{+\half\sum_{\ell}}\label{eq:transfer}
\end{align}
\end{subequations}Note that $(F^{\ell\dg}F^{\ell})_{\lr}=F_{\ur}^{\ell\dg}F_{\ur}^{\ell}+F_{\lr}^{\ell\dg}F_{\lr}^{\ell}$,
that $\R_{\ur}\L\R_{\ll}=\R_{\ll}\L\R_{\ur}=0$, i.e., evolution of
coherences decouples under this decomposition, and that the transfer
term $\R_{\ul}\L\R_{\lr}$ is in the form of a quantum channel.

\subsection{DFS case\label{subsec:DFS-case}}

Recall from Sec.~\ref{subsec:Presence-of-decay} that now $\ash=\ulbig$
and $\pp=\sum_{k=0}^{d-1}|\psi_{k}\ket\bra\psi_{k}|$ is the projection
on the \foreignlanguage{american}{\acs{DFS}} states $\{|\psi_{k}\ket\}_{k=0}^{d-1}$.
In the case of a non-steady \foreignlanguage{american}{\acs{DFS}},
any residual evolution within $\ulbig$ is exclusively unitary and
generated by the Hamiltonian superoperator $\sout\equiv\L_{\ul}$.
The jump operators in $\L_{\ul}$ (\ref{eq:term0}) must then act
trivially, and the most general condition for them (for some complex
constants $a_{\ell}$) is 
\begin{equation}
F_{\ul}^{\ell}=a_{\ell}\pp\,.
\end{equation}
This implies that $\R_{\ul}\L\R_{\of}$ (\ref{eq:term}) is zero and
the partition (\ref{eq:gen}) becomes
\begin{eqnarray}
\L & = & \left[\begin{array}{ccc}
\sout\, & 0 & \,\R_{\ul}\L\R_{\lr}\\
0 & \,\,\,\L_{\of}\,\,\, & \,\R_{\of}\L\R_{\lr}\\
0 & 0 & \,\L_{\lr}
\end{array}\right].\label{eq:gen-1}
\end{eqnarray}
This extra zero will prove important when we examine the Kubo formula
for this case in Ch.~\ref{ch:4}.
\selectlanguage{american}%

\subsection{Semisimple DFS case\label{subsec:Semisimple-DFS-case}}

\selectlanguage{english}%
A simplified (\textit{semisimple} \cite{Lidar1998})\textit{ }\foreignlanguage{american}{\acs{DFS}}
case can be obtained by setting $F_{\ul}^{\ell}=0$, i.e., $a_{\ell}=0$.
Then, all \foreignlanguage{american}{\acs{DFS}} states are annihilated
by the jumps, 
\begin{equation}
F^{\ell}|\psi_{k}\ket=0\,.
\end{equation}
If $|\psi_{k}\ket$ are also eigenstates of $H$, they are called
\textit{dark states} \cite{Kraus2008}. We can also set $H=0$, meaning
that the \foreignlanguage{american}{\acs{DFS}} is stationary ($\hout=0$).
Stationarity leads to only one new zero matrix element of $\L$, 
\begin{eqnarray}
\L & = & \left[\begin{array}{ccc}
0 & 0 & \,\R_{\ul}\L\R_{\lr}\\
0 & \,\,\,\L_{\of}\,\,\, & \,\R_{\of}\L\R_{\lr}\\
0 & 0 & \,\L_{\lr}
\end{array}\right].\label{eq:gen-1-2}
\end{eqnarray}
However, the $\ofbig$ part now decays deterministically:
\begin{equation}
\L_{\of}(\r)=-\half\sum_{\ell}\lind_{\ell}\{\r_{\of},(F^{\ell\dg}F^{\ell})_{\lr}\}\equiv-\{\r_{\of},\hdg\}\,,\label{eq:disgap}
\end{equation}
where the ground states of the \textit{decoherence} \cite{Karasik2008}
or\textit{ parent }\cite{Iemini2015}\textit{ Hamiltonian }$\hdg$
are exactly the \foreignlanguage{american}{\acs{DFS}} states. We
see that the excitation gap $\adg$ (where $\textsf{edg}$ stands
for \textit{effective dissipative gap}) of this Hamiltonian is relevant
in Lindbladian first-order perturbation theory in Sec.~\ref{subsec:Example:-decoherence-Hamiltonian}.
All examples considered in Chs.~\ref{ch:7}-\ref{ch:8} are of this
type.

\subsection{Non-Hermitian Hamiltonian systems\label{subsec:non-Hermitian}}

Now let us keep $H\neq0$ and further simplify the jumps to 
\begin{equation}
F=F_{\ur}\,.
\end{equation}
Having $F_{\lr}=0$ leads to significant simplification of the matrix
elements of $\L$ for the \foreignlanguage{american}{\acs{DFS}} case
(\ref{eq:gen-1}). Now, $\R_{\of}\L\R_{\lr}=0$ and evolution within
$\thobig$ has no jump terms and is governed solely by the deterministic
part (\ref{eq:kkkkkkk}) of $\L$, 
\begin{equation}
K=K_{\di}\equiv H_{\di}-i\half\sum_{\ell}\lind_{\ell}F_{\ur}^{\ell\dg}F_{\ur}^{\ell}\,\,\,\,\,\,\text{with superoperator}\,\,\,\,\,\,\K(\r)\equiv-i(K\r-\r K^{\dg})\,.
\end{equation}
Since $K_{\ul}=H_{\ul}$, evolution within $\ulbig$ is unitary and
governed by the Hamiltonian piece $H_{\ul}$. The remaining $\thobig$
parts evolve under $K_{\lr}$, which includes the \textit{positive
definite} ($\sum_{\ell}\lind_{\ell}F_{\ur}^{\ell\dg}F_{\ur}^{\ell}>0$
on $\lrbig$) contribution due to the jumps and thereby guarantees
decay of everything initially in $\thobig$. Due to trace preservation,
the transfer term (\ref{eq:transfer}) $\R_{\ul}\L\R_{\lr}\neq0$
makes sure that populations in $\lrbig$ are transferred into $\ulbig$
in the limit of infinite time, making up the only non-deterministic
part of the evolution. The full decomposition is then
\begin{eqnarray}
\L & = & \left[\begin{array}{ccc}
\K_{\ul} & 0 & \R_{\ul}\L\R_{\lr}\\
0 & \,\,\,\K_{\of}\,\,\, & 0\\
0 & 0 & \K_{\lr}
\end{array}\right].\label{eq:gen-1-1}
\end{eqnarray}
Thus, in the $\thobig$ sector, the evolution is generated by the
non-Hermitian Hamiltonian $K$. In a related work \cite{Venuti2017},
a non-Hermitian Hamiltonian governs evolution in $\ofbig$ but not
in $\lrbig$.

It is useful here to remark on the relation between the four-corners
projections and the Feshbach projection method \cite{Feshbach1958,Rotter2015}.
The former splits the Hilbert space into decaying and non-decaying
parts while the latter splits it into subspaces spanned by bound and
scattering states. Connections between the two become clearer when
one looks at the effect of perturbations which excite states from
$\ulbig$ into $\lrbig$ on $\L$ of the type (\ref{eq:gen-1-1}).
We calculate leading-order corrections to evolution within $\ulbig$
due to these effects in Sec.~\ref{subsec:Effective-Operator-Formalism},
showing that they are identical to those derived by the effective
operator formalism \cite{Reiter2012} \textemdash{} an extension of
the Feshbach method to Lindbladians.
\selectlanguage{american}%

\subsection{Quantum channel simulation\foreignlanguage{english}{\label{subsec:Quantum-channel-simulation}}}

\selectlanguage{english}%
In this most simplified example, \acs{ASH} is a steady \foreignlanguage{american}{\acs{DFS}}
($\hout=H_{\ul}=0$) of dimension $d$ and there is a decaying subspace
$\lrbig$ of dimension $N-d$. We assume that $F=F_{\ur}$, that all
rates $\lind_{\ell}$ of the Lindbladian are equal to one rate $\kef$,
and that 
\begin{equation}
\sum_{\ell}F^{\ell\dg}F^{\ell}=\sum_{\ell}F_{\ur}^{\ell\dg}F_{\ur}^{\ell}=\qq\,.\label{eq:cptp2}
\end{equation}
The last condition guarantees that the transfer term (\ref{eq:transfer}),
$\R_{\ul}\L\R_{\lr}\equiv\kef\E$, is proportional to a bona fide
quantum channel $\E$ from a $N-d$-dimensional input space $\lrbig$
to a $d$-dimensional output space $\ulbig$. The condition (\ref{eq:cptp2})
also makes sure that $\thobig$ decays uniformly:
\begin{eqnarray}
\L & = & \kef\left[\begin{array}{ccc}
0 & 0 & \E\\
0 & \,\,\,-\half\R_{\of}\,\,\, & 0\\
0 & 0 & -\R_{\lr}
\end{array}\right].\label{eq:gen-1-1-1}
\end{eqnarray}
The rate $\kef$ is the inverse of the uniform relaxation time $T_{1}$
for $\lrbig$. Notice how the uniform coherence relaxation time $T_{2}=\half T_{1}$.

The example is constructed such that the piece of the asymptotic projection
$\ppp\R_{\lr}$ (\ref{eq:te}) taking states from $\lrbig$ into $\ulbig$
is exactly the channel $\E$. While we derive the equation for $\ppp$
below, we apply that result (Thm.~\ref{prop:3}) here since the simplifications
are trivial: 
\begin{equation}
\ppp\R_{\lr}=-(\R_{\ul}\L\R_{\lr})\L_{\lr}^{-1}=\frac{1}{\kef}\R_{\ul}\L\R_{\lr}=\E\,.
\end{equation}
In other words, all quantum channels can be embedded in some $\ppp=\lim_{t\rightarrow\infty}e^{t\L}$.
Note that this scheme is different from constructing a Lindbladian
$\L\equiv\E-\id$ out of a channel $\E$, shown in eq.~(\ref{eq:stab}).
Here, the channel $\E$ does not act on all of $\oph$ because the
channel's input and output spaces are different subspaces of $\oph$.
It is impossible to simulate a more general channel that maps all
of $\oph$ to itself using only a Lindbladian \cite{Wolf2008}. 

The immediate application of this construction is quantum error correction
\cite{Terhal2015}, where $\lrbig$ is the subspace into which a protected
state (initially in $\ulbig$) is taken after application of an error
channel. The channel $\ppp\R_{\lr}$ then recovers the information,
and the above construction shows that all such recovery channels can
be implemented to arbitrary accuracy by turning on Lindbladian evolution
for a sufficiently long period of time.

\section{The asymptotic projection}

Armed with the partition of $\L$ from eq.~(\ref{eq:gen}), here
we study cases where \acs{ASH} contains unitarily evolving states
{[}$\hout\neq0$ from eq.~(\ref{eq:te}){]} and formally introduce
the asymptotic projection $\ppp$. The basis for \acs{ASH} consists
of right eigenmatrices of $\L$ with pure imaginary eigenvalues. Recalling
eq.~(\ref{eq:expansion}), we can expand $|\rout\kk$ in such a basis
since exponentials ($e^{t\varLambda}$) of any eigenvalues $\varLambda$
with negative real parts decay to zero for large $t$. We call the
non-decaying eigenmatrices \textit{right asymptotic eigenmatrices}
$|\St_{\la\m}\kk$ or \acs{ASH} \textit{basis elements} with purely
imaginary eigenvalue $\varLambda=i\la$ (used here as an index) and
degeneracy index $\m$ (that depends on $\la\in\mathbb{R}$). By definition,
$|\St_{\la\m}\kk\in\ulbig$ and the eigenvalue equation is 
\begin{equation}
\L|\St_{\la\m}\kk=i\la|\St_{\la\m}\kk\,.\label{eq:reig}
\end{equation}
Since $\L$ is not always diagonalizable, any degeneracy may induce
a non-trivial Jordan block structure for a given $\la$ (see Sec.~\ref{subsec:More-on-Lindbladians}).
However, we show below that all Jordan blocks corresponding to asymptotic
eigenmatrices are diagonal. Therefore, there exists a dual set of
\textit{conserved quantities} or \textit{left asymptotic eigenmatrices}
$\bb\J^{\la\m}|$ such that
\begin{equation}
\bb\J^{\la\m}|\L=i\la\bb\J^{\la\m}|\longleftrightarrow\L^{\dgt}(\J^{\la\m})=-i\la(\J^{\la\m})\,.\label{eq:leig}
\end{equation}
We state and prove this duality in the following correspondence. 
\begin{thm}
[Conserved quantity -- steady state correspondence \cite{pub011}]\label{thm:dual}Let
$\{|\St_{\la\m}\kk\}_{\la,\m}$ be an orthonormal basis for $\ash\subseteq\oph$
of dimension $D$, i.e., $\bb\St_{\la\m}|\St_{\varTheta\n}\kk=\d_{\la\varTheta}\d_{\m\n}$.
Then, for all $\rin\in\oph$, the asymptotic state $\rout$ is expressible
as
\begin{equation}
|\rout(t)\kk=\sum_{\la,\m}c_{\la\m}e^{i\la t}|\St_{\la\m}\kk\label{eq:inf}
\end{equation}
and there exist $D$ conserved quantities $\{J^{\la\m}\}_{\la,\m}$
such that 
\begin{align}
c_{\la\m} & \equiv\bb\J^{\la\m}|\rin\kk=\tr\{(\J^{\la\m})^{\dg}\rin\}\label{eq:coeff-1}\\
\bb\J^{\la\m}|\St_{\varTheta\n}\kk & =\d_{\la\varTheta}\d_{\m\n}\,.\label{eq:norms}
\end{align}
\end{thm}
\selectlanguage{american}%
\begin{proof}
The matrix form of $\L$ can be put into Jordan normal form $\mathcal{C}$
via a non-unitary similarity transformation $\NU$, 
\begin{equation}
\L=\NU^{-1}\mathcal{C}\NU\,.\label{jordan}
\end{equation}
We have already seen from eq.~(\ref{eq:expjor}) that all Jordan
blocks with pure imaginary eigenvalues $i\la$ must be diagonal in
order for the rules of quantum mechanics to not be violated as $t\rightarrow\infty$.
Therefore, the respective transformed left and right eigenmatrices,
$|\tilde{\St}_{\la\m}\kk=\NU|\St_{\la\m}\kk$ and $\bb\tilde{J}^{\la\m}|=\bb J^{\la\m}|\NU^{-1}$,
are linearly independent and orthogonal to all other basis matrices
of \acs{OPH}. Thus they can be made to be \textit{biorthogonal} \cite{brink2001},
$\bb\J^{\la\m}|\St_{\varTheta\n}\kk=\d_{\la\varTheta}\d_{\m\n}$,
\foreignlanguage{english}{by }a procedure similar to the Gram-Schmidt
process. It is clear that once the transformed vectors are biorthogonal,
the original ones are also, satisfying eq.~(\ref{eq:norms}). Taking
the expectation value of $\J^{\la\m}$ with respect to $\rout$ yields
\begin{equation}
\bb\J^{\la\m}|\r(t)\kk=\bb\J^{\la\m}|e^{t\L}|\rin\kk=e^{i\la t}\bb\J^{\la\m}|\rin\kk\equiv c_{\la\m}e^{i\la t}\,.\label{eq:ccdynam}
\end{equation}
Since $\J^{\la\m}$ are dual to $\St_{\la\m}$ in the sense described
above, eq.~(\ref{eq:inf}) holds.
\end{proof}
\selectlanguage{english}%
The coefficients $c_{\la\m}$ determine the footprint that $\rin$
leaves on $\rout$, implying that $\rout$ does not depend on dynamics
at any intermediate times. In general, any part of $|\rin\kk$ not
in the kernel of $\ppp$ imprints on the asymptotic state since, by
definition, that part overlaps with some $\J^{\la\m}$. Equation (\ref{eq:ccdynam})
tells us that the $\J^{\la\m}$ are either conserved in time (when
$\la=0$) or oscillating indefinitely. While the term \textit{non-decaying
quantity} is thus more accurate, we use the term conserved quantity
to describe all $\J^{\la\m}$ since they are are bona fide conserved
quantities in the rotating frame of $\hout$ (as we now show). Since
it was shown that evolution of asymptotic states is exclusively unitary
(\cite{baum2}, Thm.~2), it must be that the eigenvalue set $\{\la\}$
is that of a Hamiltonian superoperator, which we define to be $\sout\equiv-i[\hout,\cdot]$.
In other words, we use the set $\{\la\}$ to construct a Hamiltonian
$\hout\in\ulbig$ (defined up to a constant energy shift) such that
each $\la$ is a difference of the energies of $\hout$ and $|\St_{\la\m}\kk$
are eigenmatrices of $\sout$.\footnote{The eigenvalues $i\la$ can be extracted from $H_{\ul}$, as shown
in eq.~(\ref{eq:diffe}). Note that $\sout$ shares the same eigenvalues
as $\ppp\L\ppp$, but $\sout\neq\ppp\L\ppp$ because the latter is
not anti-Hermitian.} Because of this, the eigenmatrices $\{\St,\J\}$ must come in complex
conjugate pairs: $\St_{-\la\m}=\St_{\la\m}^{\dg}$ (which obstructs
us from constructing a Hermitian basis for $\{\St_{\la\ne0,\m}\}$)
and same for $\J^{\la\m}$. The explicit form of $\hout$ depends
on the structure of \acs{ASH}.

Combining $\ppp$ with the definition of $\sout$ allows us to rearrange
eq.~(\ref{eq:inf}) into
\begin{equation}
|\rout(t)\kk\equiv\lim_{t\rightarrow\infty}e^{t\L}|\rin\kk=e^{t\sout}\sum_{\la,\m}|\St_{\la\m}\kk\bb\J^{\la\m}|\rin\kk=e^{t\sout}\ppp|\rin\kk\,,\label{eq:asymp}
\end{equation}
where the outer products are used to express the asymptotic projection
\begin{equation}
\ppp\equiv\sum_{\varDelta,\m}|\St_{\la\m}\kk\bb\J^{\la\m}|\,.\label{eq:asproj}
\end{equation}
This is indeed a projection ($\ppp^{2}=\ppp$) due to eq.~(\ref{eq:norms}).
This projection is onto the \textit{peripheral spectrum} of $e^{t\L}$
\textemdash{} all eigenvalues of $e^{t\L}$ (for $t>0$) whose modulus
is one. As a result, one can apply standard formulas to express it
in two other ways,
\begin{align}
\ppp & =\lim_{T\rightarrow\infty}\frac{1}{T}\sum_{\la}\intop_{0}^{T}dte^{t(\L-i\la)}=-\frac{1}{2\pi i}\ointop_{\G}dz\left(\L-z\right)^{-1}\,.
\end{align}
The first (ergodic) way is to take a proper limit such that the projections
on all eigenspaces associated with $i\la$ remain {[}\cite{wolf2010},
eq.~(6.15){]}. The second way defines $\ppp$ in terms of Riesz projections
on the relevant eigenspaces: $\G$ is the contour in the complex plane
which encircles all $i\la$ and no other points in the spectrum of
$\L$ {[}\cite{katobook}, Ch.~1, eq.~(5.23){]}. While both ways
exist and are equal to each other since we assume a finite-dimensional
\acs{OPH}, they are also valid for some infinite-dimensional examples
(see Chs.~\ref{ch:7}-\ref{ch:8}).

\section{Analytical formula for conserved quantities}

We proceed to determine $\J^{\la\m}$ by plugging in the four-corners
partition of $\L$ (\ref{eq:gen}) into the eigenvalue equation (\ref{eq:leig}).
The block upper-triangular structure of $\L$ readily implies that
$|\J_{\ul}^{\la\m}\kk$ are left eigenmatrices of $\L_{\ul}$: 
\begin{equation}
\bb\J_{\ul}^{\la\m}|\L_{\ul}=i\la\bb\J_{\ul}^{\la\m}|\,.
\end{equation}
Writing out the conditions on the remaining components $|\J_{\tho}^{\la\m}\kk$
yields an analytic expression for $|\J^{\la\m}\kk$. We state this
formula below, noting that $[\L_{\lr},\R_{\lr}]=0$. Plugging that
result into eq.~(\ref{eq:asproj}) and setting $\la=0$ yields the
formula for $\ppp$ (\ref{eq:maindecomp}) for the case when $\hout=0$.
We finish this Chapter by going through the ramifications of the formula
for $\ppp$ for the various types of \acs{ASH}, introducing notation
used throughout the rest of the thesis, and comparing our conserved
quantities to those of Hamiltonian-based systems. In Ch.~\ref{ch:3},
we apply the decomposition of $\ppp$ into conserved quantities $J^{\la\m}$
to study examples of footprints left on $\rout$ by $\rin$.
\begin{thm}
[Analytical expression for conserved quantities \cite{ABFJ}]\label{prop:3}The
left eigenmatrices of $\L$ corresponding to pure imaginary eigenvalues
$i\la$ are
\begin{equation}
\bb\J^{\la\m}|=\bb\J_{\ul}^{\la\m}|\left(\R_{\ul}-\L\frac{\R_{\lr}}{\L_{\lr}-i\la\R_{\lr}}\right)\,\,\,\longleftrightarrow\,\,\,\J^{\la\m}=\J_{\ul}^{\la\m}-\R_{\ul}\L^{\dgt}\frac{\R_{\lr}}{\L_{\lr}^{\dgt}+i\la\R_{\lr}}(\J_{\ul}^{\la\m})\,,\label{eq:main}
\end{equation}
where $\bb\J_{\ul}^{\la\m}|$ are left eigenmatrices of $\L_{\ul}$.
\end{thm}
\begin{proof}
\footnote{There exist related formulas for the parts of $\ppp\R_{\lr}$ corresponding
to fixed points of discrete-time quantum channels in Lemma 5.8 of
Ref.~\cite{robin} and Prop. 7 of Ref.~\cite{Cirillo2015} and of
Markov chains in Thm.~3.3 of Ref.~\cite{Novotny2012}.}For a left eigenmatrix $\bb\J^{\la\m}|$ with eigenvalue $i\la$,
$\L^{\dgt}|\J^{\la\m}\kk=-i\la|\J^{\la\m}\kk$. Now partition this
eigenvalue equation using the projections $\{\R_{\ul},\R_{\of},\R_{\lr}\}$.
Taking the adjoint of the partitioned $\L$ from eq.~(\ref{eq:gen})
results in 
\begin{equation}
\L^{\dgt}|\J^{\la\m}\kk=\left[\begin{array}{ccc}
\L_{\ul}^{\dgt}\, & 0\vspace{4pt} & 0\\
\,\R_{\of}\L^{\dgt}\R_{\ul}\, & \L_{\of}^{\dgt} & 0\vspace{4pt}\\
\,\R_{\lr}\L^{\dgt}\R_{\ul}\, & \,\R_{\lr}\L^{\dgt}\R_{\of}\, & \,\L_{\lr}^{\dgt}
\end{array}\right]\left[\begin{array}{c}
|\J_{\ul}^{\la\m}\kk\vspace{4pt}\\
|\J_{\of}^{\la\m}\kk\vspace{4pt}\\
|\J_{\lr}^{\la\m}\kk
\end{array}\right].
\end{equation}
The eigenvalue equation is then equivalent to the following three
conditions:\begin{subequations}
\begin{eqnarray}
-i\la\J_{\ul}^{\la\m} & = & \L_{\ul}^{\dgt}(\J_{\ul}^{\la\m})\label{eq:pc1}\\
-i\la\J_{\of}^{\la\m} & = & \R_{\of}\L^{\dgt}\R_{\ul}(\J_{\ul}^{\la\m})+\L_{\of}^{\dgt}(\J_{\of}^{\la\m})\label{eq:pc2}\\
-i\la\J_{\lr}^{\la\m} & = & \R_{\lr}\L^{\dgt}\R_{\ul}(\J_{\ul}^{\la\m})+\L_{\lr}^{\dgt}(\J_{\lr}^{\la\m})+\R_{\lr}\L^{\dgt}\R_{\of}(\J_{\of}^{\la\m})\,.\label{eq:pc3}
\end{eqnarray}
\end{subequations}We now examine them in order.
\begin{enumerate}
\item Condition (\ref{eq:pc1}) implies that $[F_{\ul}^{\ell\dg},\J_{\ul}^{\la\m}]=0$
for all $\ell$.\footnote{This part is essentially the Lindblad version of a similar statement
for quantum channels (\cite{robin}, Lemma 5.2). Another way to prove
this is to apply ``well-known'' algebra decomposition theorems (e.g.,
\cite{Knill2000}, Thm.~5).} To show this, we use the \textit{dissipation function} $\dfunc$
associated with $\L_{\ul}$ \cite{Lindblad1976}. For some $A\in\ulbig$,
\begin{eqnarray}
\dfunc(A) & \equiv & \L_{\ul}^{\dgt}(A^{\dg}A)-\L_{\ul}^{\dgt}(A^{\dg})A-A^{\dg}\L_{\ul}^{\dgt}(A)=\sum_{\ell}\lind_{\ell}[F_{\ul}^{\ell},A]^{\dg}[F_{\ul}^{\ell},A]\,.
\end{eqnarray}
Using (\ref{eq:pc1}) and remembering that $\J_{\ul}^{\la\m\dg}=\J_{\ul}^{-\la\m}$,
the two expressions for $\dfunc(\J_{\ul}^{\la\m})$ imply that
\begin{equation}
\L_{\ul}^{\dgt}(\J_{\ul}^{\la\m\dg}\J_{\ul}^{\la\m})=\sum_{\ell}\lind_{\ell}[F_{\ul}^{\ell},\J_{\ul}^{\la\m}]^{\dg}[F_{\ul}^{\ell},\J_{\ul}^{\la\m}].\label{eq:dissfunc}
\end{equation}
We now take the trace using the full rank steady-state density matrix
\begin{equation}
|\rout\kk=\R_{\ul}|\rout\kk\equiv\sum_{\m}c_{\m}|\St_{0\m}\kk\,.
\end{equation}
Such an asymptotic state is simply that from eq.~(\ref{eq:inf})
with $c_{\la\m}=\d_{\la0}c_{\m}$ and $c_{\m}\neq0$. It is full rank
because it is a linear superposition of projections on eigenstates
of $\hout$ and such projections provide a basis for all diagonal
matrices of $\ulbig$. Taking the trace of the left-hand side of eq.~(\ref{eq:dissfunc})
yields
\begin{equation}
\bb\rss|\L_{\ul}^{\dgt}(\J_{\ul}^{\la\m\dg}\J_{\ul}^{\la\m})\kk=\bb\L_{\ul}(\rss)|\J_{\ul}^{\la\m\dg}\J_{\ul}^{\la\m}\kk=0\,,
\end{equation}
implying that the trace of the right-hand side is zero:
\begin{equation}
\sum_{\ell}\lind_{\ell}\tr\{\rss[F_{\ul}^{\ell},\J_{\ul}^{\la\m}]^{\dg}[F_{\ul}^{\ell},\J_{\ul}^{\la\m}]\}=0\,.
\end{equation}
Each summand above is non-negative (since $\lind_{\ell}>0$, the commutator
products are positive semidefinite, and $\rss$ is positive definite).
Thus the only way for the above to hold is for $[F_{\ul}^{\ell},\J_{\ul}^{\la\m}]^{\dg}[F_{\ul}^{\ell},\J_{\ul}^{\la\m}]=0$,
which implies that $F_{\ul}^{\ell}$ and $\J_{\ul}^{\la\m}$ commute
for all $\ell,\la,\m$. If we once again remember that $\J_{\ul}^{\la\m\dg}=\J_{\ul}^{-\la\m}$
and that the eigenvalues come in pairs $\pm\la$, then
\begin{equation}
[F_{\ul}^{\ell\dg},\J_{\ul}^{\la\m}]=[F_{\ul}^{\ell},\J_{\ul}^{\la\m}]=0\,.\label{eq:comm}
\end{equation}
With the jump operators out of the picture, condition (\ref{eq:pc1})
now becomes 
\begin{equation}
-i[H_{\ul},\J_{\ul}^{\la\m}]=i\la\J_{\ul}^{\la\m}\,,\label{eq:diffe}
\end{equation}
confirming that $\la$ are differences of eigenvalues of a Hamiltonian.
\item Now consider condition (\ref{eq:pc2}). The first term on the right-hand
side can be obtained from eq.~(\ref{eq:term}) and is as follows:
\begin{equation}
\R_{\of}\L^{\dgt}\R_{\ul}(\J_{\ul}^{\la\m})=\sum_{\ell}\lind_{\ell}(F_{\ul}^{\ell\dg}\J_{\ul}^{\la\m}F_{\ur}^{\ell}-\J_{\ul}^{\la\m}F_{\ul}^{\ell\dg}F_{\ur}^{\ell})+\sum_{\ell}\lind_{\ell}(F_{\ur}^{\ell\dg}\J_{\ul}^{\la\m}F_{\ul}^{\ell}-F_{\ur}^{\ell\dg}F_{\ul}^{\ell}\J_{\ul}^{\la\m})\,.
\end{equation}
This term is identically zero due to eq.~(\ref{eq:comm}), reducing
condition (\ref{eq:pc2}) to $\L_{\of}^{\dgt}(\J_{\of}^{\la\m})=-i\la\J_{\of}^{\la\m}$.
We now show that this implies 
\begin{equation}
\R_{\of}|\J^{\la\m}\kk=|\J_{\of}^{\la\m}\kk=0\label{eq:crss}
\end{equation}
for all $\la$ and $\m$. By contradiction, assume $\J_{\of}^{\la\m}\,(\neq0)$
is a left eigenmatrix of $\L_{\of}$. Then there must exist a corresponding
right eigenmatrix $\St_{\la\m}^{\prime}=\R_{\of}(\St_{\la\m}^{\prime})$
since $\St$ and $\J$ are biorthogonal by Thm.~\ref{thm:dual}.
However, all right eigenmatrices are contained in $\ulbig$ by definition
(\ref{eq:cond}), so we have a contradiction and $\J_{\of}^{\la\m}=0$.
\item Finally consider condition (\ref{eq:pc3}). Applying eq.~(\ref{eq:crss})
removes the last term on the right-hand side of that condition and
simplifies it to
\begin{equation}
[\L_{\lr}^{\dgt}+i\la\R_{\lr}](\J_{\lr}^{\la\m})=-\R_{\lr}\L^{\dgt}(\J_{\ul}^{\la\m})\,.
\end{equation}
Now we can show that the operator $\L_{\lr}^{\dgt}+i\la\R_{\lr}$
is invertible when restricted to $\lrbig$ using a proof by contradiction
similar to the one used to prove eq.~(\ref{eq:crss}). Inversion
gives a formula for $\J_{\lr}^{\la\m}$ which, along with eq.~(\ref{eq:crss}),
yields the statement.
\end{enumerate}
\end{proof}
\selectlanguage{american}%

\section{No-leak and clean-leak properties}

\selectlanguage{english}%
Armed with the analytical formula for $J^{\la\m}$ from Thm.~\ref{prop:3},
we can now construct the most general $\ppp$. We state our result
for the $\hout=0$ case and outline the non-trivial consequences of
$\hout\neq0$ in Sec.~\ref{subsec:Coherence-suppressed-by}. We can
split $\ppp$ as follows: 
\begin{equation}
\ppp=\ppp\R_{\ul}+\ppp\R_{\lr}=\ps(\id-\L\L_{\lr}^{-1})\,,\label{eq:maindecomp-1}
\end{equation}
where $\ps$ is the \textit{minimal projection} that further projects
$\ulbig$ onto \acs{ASH}. The form of $\ps$, the asymptotic projection
of $\L_{\ul}$ (which does not admit a decaying subspace), depends
on the details of \acs{ASH} and is already known \cite{baum2,robin}.
Our work therefore extends previous Lindbladian results to cases when
a decaying subspace is present. Of course, any pair $\{\L_{\ul},\ps\}$
can be extended (via Thm.~\ref{prop:2}) to the pair $\{\L,\ppp\}$
that admits an arbitrarily large decaying subspace $\lrbig$. The
first ($\ppp\R_{\ul}$) terms states that if one starts in $\ulbig$,
then one is simple projected into $\ash\subseteq\ulbig$. The second
term ($\ppp\R_{\lr}$) shows that an initial state in $\lrbig$ is
transferred to an asymptotic state in $\ulbig$ via application of
the projected inverse $\L_{\lr}^{-1}$. Since the projection $\R_{\of}$
is not present in the above formula, we can immediately read off that
no coherences between the non-decaying subspace and its counterpart
are preserved (see Sec.~\ref{sec:State-initialization}). Thus, the
above formula allows us to determine which parts of $\rin$ are preserved
in the large-time limit. The \foreignlanguage{american}{\acs{DFS}}
case (when $\ps=\R_{\ul}$) was addressed in Sec.~\ref{sec:Questions-addressed-and}.

Since superoperator perturbation theory requires spectral projections
such as $\ppp$, the above formula is also useful in determining how
states that are \textit{already in} \acs{ASH} respond to perturbations.
We now switch gears and sketch the effect of small perturbations $\oo$
on a state $\rout$ already in \acs{ASH} in order to lay the groundwork
for Chs.~\ref{ch:4}-\ref{ch:6}. The perturbations of interest are
either Hamiltonian perturbations $\spert\equiv-i[V,\cdot]$ (with
Hamiltonian $\hpert$ and small parameter $\e$) or derivatives $\p_{\a}\equiv\p/\p\xx_{\a}$
(with parameters $\xx_{\a}$ and adiabatic evolution time $T$) of
the now parameter-dependent $\rout(\xx_{\a})$ and $\L(\xx_{\a})$:
\begin{equation}
\mathcal{O}\in\left\{ \e\spert,\frac{1}{T}\p_{\a}\right\} \,.
\end{equation}
In Chs.~\ref{ch:4}-\ref{ch:5}, we show that both of these can be
used to induce unitary operations on \acs{ASH}. The latter determine
adiabatic connection(s) and thus help with defining parallel transport
(i.e., adiabatic evolution) of \acs{ASH}. We show in Ch.~\ref{ch:4}
that this analysis holds for jump operator perturbations $F^{\ell}\rightarrow F^{\ell}+f^{\ell}$
as well, but omit discussing those perturbations for now to keep things
simple. Within first order for the case of perturbation theory ($\e\rightarrow0$)
and approaching the adiabatic limit for the case of parallel transport
($T\rightarrow\infty$), two relevant perturbative processes after
the action of $\oo$ on an asymptotic state are (A) subsequent projection
onto \acs{ASH} and (B) leakage out of \acs{ASH} via the perturbation
and $\L^{-1}$: 
\begin{equation}
\rout\rightarrow\ppp\oo(\rout)-\L^{-1}\oo(\rout)\,.\label{eq:outline}
\end{equation}
We study these terms here and show later that they occur both in the
Kubo formula and in adiabatic response.

We first observe that $\oo$ is limited in its effect on $\rout$.
Acting with $\oo$ once does not connect $\ulbig$ with $\lrbig$
because $\oo$ does not act non-trivially on $\rout$ from both sides
simultaneously. This \textit{no-leak property} can be understood if
one observes that Hamiltonian superoperator perturbations $\spert$
act nontrivially on $\rout$ only from one side at a time due to their
commutator form. Likewise, derivatives $\p_{\a}$ act nontrivially
on either the ``ket'' or ``bra'' parts of all basis elements used
to write $\rout$ due to the product rule. Therefore, acting with
$\oo$ once only connects $\ulbig$ to itself and nearest-neighbor
squares $\thubig$ and does not cause ``transitions'' into $\lrbig$:
\begin{equation}
\mathcal{O}(\rout)=\R_{\thu}\mathcal{O}(\rout)\,,\tag{\textnormal{\textbf{LP1}}}\label{eq:no-leak}
\end{equation}
where $\R_{\thu}\equiv\id-\R_{\lr}$. Moreover, despite two actions
of $\oo$ connecting $\ulbig$ to $\lrbig$, (\ref{eq:no-leak}) still
provides some insight into second-order effects (see Sec.~\ref{subsec:Second-order-terms}).

The no-leak property is important in determining the energy scale
governing leakage out of \acs{ASH}. Let us apply this property to
the second term in eq.~(\ref{eq:outline}):
\begin{equation}
\L^{-1}\oo(\rout)=\L^{-1}\R_{\thu}\oo(\rout)=\L_{\thu}^{-1}\oo(\rout)\,,\label{eq:no-leak-inv}
\end{equation}
where $\L_{\emp}^{-1}\equiv(\R_{\emp}\L\R_{\emp})^{-1}$ with $\empbig$
being any block(s). Note that the last step in eq.~(\ref{eq:no-leak-inv})
also uses $\R_{\lr}\L\R_{\thu}=0$ from eq.~(\ref{eq:gen}). Since
the restriction to studying $\L$ on $\thubig$ in linear response
has previously gone unnoticed, it is conventionally believed that
the leakage energy scale is determined by the dissipative gap $\dgg$
(\ref{eq:dissipative-gap-def}) \textemdash{} the nonzero eigenvalue
of $\L$ with smallest real part. As shown in eq.~(\ref{eq:no-leak-inv}),
that energy scale is actually governed by the \textit{effective dissipative
gap} $\adg\geq\dgg$ \textemdash{} the nonzero eigenvalue of $\L_{\thu}$
with smallest real part. In Hamiltonian systems ($\L=-i[H,\cdot]$),
a special case of the no-leak property states that the energy denominator
in the first-order perturbative correction to the $k^{\text{th}}$
eigenstate of $H$ contains only energy differences involving the
energy $E_{k}$ of that eigenstate (and not, e.g., $E_{k-1}-E_{k+1}$).

We now project $\oo(\rout)$ back to \acs{ASH} to examine the first
term in eq.~(\ref{eq:outline}). Applying $\ppp$ to eq.~(\ref{eq:no-leak})
and using $\ppp\R_{\of}=0$ from eq.~(\ref{eq:maindecomp-1}) removes
two more squares: 
\begin{equation}
\ppp\oo(\rout)=(\ppp\R_{\di})(\R_{\thu}\oo)(\rout)=\ppp\R_{\ul}\oo(\rout)=\ps\oo\ps(\rout)\,.\tag{\textnormal{\textbf{LP2}}}\label{eq:clean-leak}
\end{equation}
The \textit{clean-leak property} shows that any leakage of the perturbed
$\rout$ into $\ofbig$ does not contribute to the first-order effect
of $\oo$ within \acs{ASH}. Essentially, the clean-leak property
makes \acs{ASH} resistant to the non-unitary effects of Lindbladian
evolution and allows for a closer analogue between \acs{ASH} and
subspaces of unitary systems. The clean-leak property simplifies calculations
of both Hamiltonian perturbations and adiabatic/Berry connections.
It can be used to show that $\ps$ (instead of $\ppp$) fully governs
adiabatic evolution, so a natural Lindbladian generalization of the
quantum geometric tensor (\ac{QGT}) is 
\begin{equation}
\geom_{\a\b}\equiv\ps\p_{\a}\ps\p_{\b}\ps\ps\,.
\end{equation}
In Ch.~\ref{ch:6}, we show that the part of the \ac{QGT} anti-symmetric
in $\a,\b$ corresponds to the adiabatic curvature $\F_{\a\b}$ (determined
from the Berry connections) and, for most relevant \acs{ASH}, derive
a metric $\met_{\a\b}$ on the parameter space from the part of the
\ac{QGT} symmetric in $\a,\b$.

\begin{figure}[t]
\begin{centering}
\includegraphics[width=0.75\columnwidth]{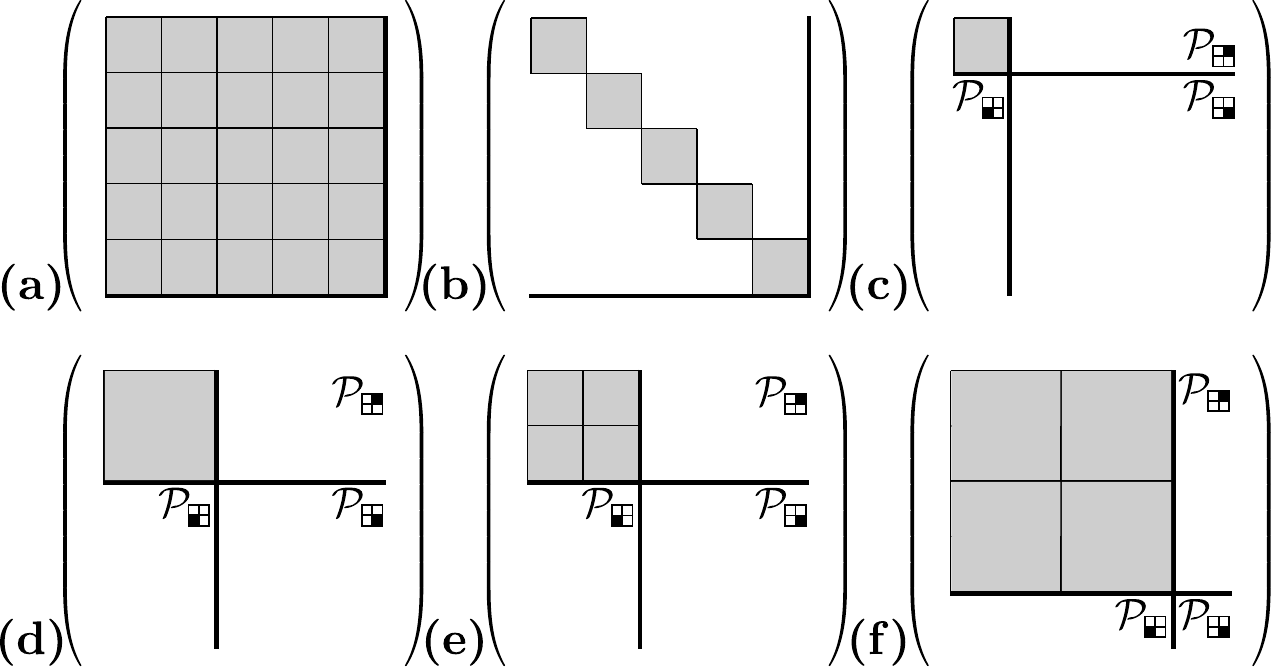}
\par\end{centering}
\caption{\label{fig:specific-ash}\protect\ytableausetup{boxsize = 3pt}Sketches
of various types of \(\ash\) (gray) embedded in a $25$-dimensional
space of matrices \(\oph\) (corresponding to a five-dimensional Hilbert
space \(\h\)). The number of (non-overlapping) gray squares counts
the dimension of \(\ash\) for each type. The first two admit no decaying
subspace and correspond to \textbf{(a)} unitary evolution and \textbf{(b)}
evolution due to a \textit{dephasing Lindbladian} (see Sec.~\ref{sec31}),
where all coherences are absent in the \(\infty\)-time limit. The
last three admit a decaying subspace \(^{\ytableaushort{ {} {} , {} {*(black)} }}\)
and correspond to Lindbladian evolution admitting \textbf{(c)} a unique
pure steady state, \textbf{(d)} a unique mixed steady state of rank
two, \textbf{(e)} a qubit decoherence-free subspace, and \textbf{(f)}
a qubit noiseless subsystem with a two-dimensional auxiliary subspace.}
\end{figure}

\selectlanguage{american}%

\section{Summary of applications to various \foreignlanguage{english}{$\text{As}(\textsf{H})$}}

\selectlanguage{english}%
\label{sec:Examples}

\subsection{Hamiltonian case}

\selectlanguage{american}%
As a sanity check, let us review how the structures we are studying
(trivially) simplify when $\L=-i[H,\cdot]$ \foreignlanguage{english}{{[}Fig.~\ref{fig:specific-ash}(a){]}}.
Then, there is no decaying subspace ($\pp=I$) and there is no extra
dephasing either {[}$\ash=\oph${]}. Literally everything is preserved
in the long-time limit, so $\ppp=\id$.
\selectlanguage{english}%

\subsection{Unique state case}

The most common case occurs when $\L$ admits only one steady state
{[}Fig.~\ref{fig:specific-ash}(c,d){]}, i.e., there exists a unique
state $\varrho\in\ash$ such that
\begin{equation}
\L(\varrho)=0\,.
\end{equation}
Recall that, according to the duality from Thm.~\ref{thm:dual},
there is a unique conserved quantity whose expectation value is preserved
in the infinite-time limit. Due to trace preservation of $e^{t\L}$,
that quantity is just the identity $I$, $\L^{\dgt}(I)=0$, and the
only quantity preserved is the trace of the initial state. In the
double-ket notation, the asymptotic projection can be written as $\ppp=|\varrho\kk\bb I|$.
In the case when $\varrho$ is not full-rank, $\varrho\in\ulbig$
and there is a decaying subspace $\lrbig$. Then, the projection $\pp$
on the support of $\varrho$ is the conserved quantity of $\L_{\ul}$
and the minimal projection is $\ps=|\varrho\kk\bb\pp|$.

This case is relatively trivial when examining its perturbation theory:
since \acs{ASH} is one-dimensional, there is nowhere to move \textit{within}
\acs{ASH}. Indeed, it is easy to show that for any trace-preserving
perturbation $\oo$,
\begin{equation}
\ppp\oo=|\varrho\kk\bb I|\oo=|\varrho\kk\bb\oo^{\dgt}(I)|=0\,.\label{eq:easy}
\end{equation}
The unique case does nevertheless admit a nontrivial \ac{QGT} and
corresponding metric, 
\begin{equation}
\met_{\a\b}=\tr\left\{ \p_{\sl\a}\pp\p_{\b\sr}\varrho\right\} \,,
\end{equation}
where $A_{\sl\a}B_{\b\sr}=A_{\a}B_{\b}+A_{\b}B_{\a}$. This metric
is distinct from the Hilbert-Schmidt metric $\tr\{\p_{(\a}\varrho\p_{\b)}\varrho\}$
for mixed $\varrho$ and is nonzero only when $\varrho$ is not full-rank.
For pure steady states, both metrics reduce to the Fubini-Study metric
\cite{provost1980}.
\selectlanguage{american}%

\subsection{DFS case}

\selectlanguage{english}%
Recall that the simplest multi-dimensional \acs{ASH} which stores
quantum information is a decoherence-free subspace or \foreignlanguage{american}{\acs{DFS}}
{[}Fig.~\ref{fig:specific-ash}(e){]}. A $d^{2}$-dimensional \foreignlanguage{american}{\acs{DFS}}\textit{
block} 
\begin{equation}
\ash=\ulbig
\end{equation}
is spanned by matrices $\{|\psi_{k}\ket\bra\psi_{l}|\}_{k,l=0}^{d-1}$,
where $\{|\psi_{k}\ket\}_{k=0}^{d-1}$ is a basis for a subspace of
the $d\leq N$-dimensional system space. The decaying block $\lrbig$
is then spanned by $\{|\psi_{k}\ket\bra\psi_{l}|\}_{k,l=d}^{N-1}$.
Evolution of the \foreignlanguage{american}{\acs{DFS}} under $\L$
is exclusively unitary, 
\begin{equation}
\p_{t}(|\psi_{k}\ket\bra\psi_{l}|)=\L(|\psi_{k}\ket\bra\psi_{l}|)=-i\left[\hout,|\psi_{k}\ket\bra\psi_{l}|\right]\,,\label{eq:eom-1}
\end{equation}
where $\hout=H_{\ul}$ is the asymptotic Hamiltonian and $k,l\leq d-1$. 

Since the entire upper-left block is preserved for a \foreignlanguage{american}{\acs{DFS}},
\begin{equation}
\ps(\rin)=\R_{\ul}(\rin)=\pp\rin\pp\,.
\end{equation}
We can thus deduce from (\ref{eq:clean-leak}) that the effect of
Hamiltonian perturbations $V$ within \acs{ASH} is $\hpert_{\ul}=\pp\hpert\pp$
\textemdash{} the Hamiltonian projected onto the \foreignlanguage{american}{\acs{DFS}}
(see Sec.~\ref{subsec:hams}). Likewise, if $\oo=\p_{\a}$, then
the Lindbladian adiabatic connection can be shown to reduce to $\p_{\a}\pp\cdot\pp$,
the adiabatic connection of the \foreignlanguage{american}{\acs{DFS}}
(see Sec.~\ref{subsec:holonomy}). Naturally, the \ac{QGT} and its
corresponding metric also reduces to that of the \foreignlanguage{american}{\acs{DFS}}
states. In other words, all such results are the same regardless of
whether the states form a \foreignlanguage{american}{\acs{DFS}} of
a Lindbladian or a degenerate subspace of a Hamiltonian. We apply
this powerful formula to coherent state quantum information processing
schemes in Chs.~\ref{ch:7}-\ref{ch:8}. Applications of this formula
to waveguide QED quantum computation schemes can be found in Ref.~\cite{Paulisch2015}.
We note that this extends a previous related result (Ref.~\cite{Zanardi2014},
footnote {[}23{]}).

Since all states in $\ulbig$ are asymptotic, the steady-state basis
elements and conserved quantities of $\L_{\ul}=\sout=-i[\hout,\cdot]$
are equal. In the notation of Thm.~\ref{prop:3}, $\J_{\ul}^{\la\m}=\St_{\la\m}$.
The structure of the piece $J_{\lr}^{\la\m}$ can be determined by
Thm.~\ref{prop:3} and depends on the $\L_{\lr}$.

\subsection{NS case}

This important case is a combination of the \foreignlanguage{american}{\acs{DFS}}
and unique steady-state cases {[}Fig.~\ref{fig:specific-ash}(f){]}.
In this case, the non-decaying portion of the system Hilbert space
($\pp\h$) factors into a $d$-dimensional subspace $\hhdfs$ spanned
by \foreignlanguage{american}{\acs{DFS}} states and a $\da$-dimensional
auxiliary (also, gauge) subspace $\hha$ which is the range of some
unique steady state $\as$,
\begin{equation}
\da=\dim(\hha)=\text{rank}(\as)\,.
\end{equation}
This combination of a \foreignlanguage{american}{\acs{DFS}} tensored
with the \textit{auxiliary state} $\as$ is called a \textit{noiseless
subsystem} (\foreignlanguage{american}{\acs{NS}}) \cite{Knill2000}.
For one \foreignlanguage{american}{\acs{NS}} block, $\h$ and \acs{OPH}
decompose as\begin{subequations}
\begin{eqnarray}
\h & = & \pp\h\oplus\qq\h=\left(\hhdfs\ot\hha\right)\oplus\qq\h\\
\oph & = & \ulbig\oplus\thobig=\left[\opdfs\ot\opa\right]\oplus\thobig\,.
\end{eqnarray}
\end{subequations}An \foreignlanguage{american}{\acs{NS}} block
is possible if $\L$ respects this decomposition and does not cause
any decoherence within the \foreignlanguage{american}{\acs{DFS}}
part. The Lindbladian in $\ulbig$ is then 
\begin{equation}
\L_{\ul}=\sdfs\ot\asi+\idfs\ot\aso=\sout+\idfs\ot\aso\,,
\end{equation}
where $\idfs$ ($\asi$) is the superoperator identity on the \foreignlanguage{american}{\acs{DFS}}
(auxiliary) space. The auxiliary state $\as$ is the unique steady
state of the Lindbladian $\aso$.

Note that the auxiliary factor becomes trivial when $\as$ is a pure
state ($\da=1$), reducing the \foreignlanguage{american}{\acs{NS}}
to a \foreignlanguage{american}{\acs{DFS}}. This means that the \foreignlanguage{american}{\acs{NS}}
case is distinct from the \foreignlanguage{american}{\acs{DFS}} case
only when $\as$ is mixed ($\da\neq1$). Similarly, if the dimension
of the \foreignlanguage{american}{\acs{DFS}} $d^{2}=d=1$, the \foreignlanguage{american}{\acs{NS}}
reduces to the unique steady-state case. The \foreignlanguage{american}{\acs{NS}}
case thus encapsulates both the \foreignlanguage{american}{\acs{DFS}}
and unique state cases.

The \foreignlanguage{american}{\acs{DFS}} basis elements $|\psi_{k}\ket\bra\psi_{l}|$
from eq.~(\ref{eq:eom-1}) generalize to $|\psi_{k}\ket\bra\psi_{l}|\ot\as$.
Let us now focus on a stationary \foreignlanguage{american}{\acs{NS}}
($\hout=0$) and construct a basis that will be used throughout the
rest of the thesis. Denote the respective \acs{ASH} basis elements
and conserved quantities as $|\St_{\m}\kk\equiv|\St_{\la=0,\m}\kk$
and $|\J^{\m}\kk\equiv|\J^{\la=0,\m}\kk$. Since \acs{ASH} is stationary,
we can construct a Hermitian matrix basis for both \acs{ASH} and
the corresponding conserved quantities that uses one index and is
orthonormal (see Sec.~\ref{app:Preliminaries}). For the \foreignlanguage{american}{\acs{DFS}}
part of the \foreignlanguage{american}{\acs{NS}}, we define the Hermitian
matrix basis $\{|\stdfs_{\m}\kk\}_{\m=0}^{d^{2}-1}$. In this new
notation, the basis elements for one \foreignlanguage{american}{\acs{NS}}
block are then
\begin{equation}
|\St_{\m}\kk=\frac{1}{\an}\begin{pmatrix}|\stdfs_{\m}\kk\ot|\as\kk & 0\,\,\\
0 & 0\,\,
\end{pmatrix}\,.
\end{equation}
We have normalized the states using the auxiliary state norm (purity)
\begin{equation}
\an\equiv\sqrt{\bb\as|\as\kk}=\sqrt{\tr\{\as^{2}\}}
\end{equation}
to ensure that $\bb\St_{\m}|\St_{\n}\kk=\d_{\m\n}$. Since an \foreignlanguage{american}{\acs{NS}}
block is a combination of the unique and \foreignlanguage{american}{\acs{DFS}}
cases, the conserved quantities of $\ulbig$ (i.e., of $\L_{\ul}$)
are direct products of the \foreignlanguage{american}{\acs{DFS}}
and auxiliary conserved quantities \cite{baum2,robin}. The unique
auxiliary conserved quantity is $\ai$, the identity on the auxiliary
subspace $\hha$. Combining this with the result above and multiplying
by $\an$ so that $\St_{\m}$ and $\J^{\m}$ are biorthogonal {[}see
eq.~(\ref{eq:norms}){]}, we obtain 
\begin{equation}
\bb\J^{\m}|=\an\begin{pmatrix}\,\bb\stdfs_{\m}|\ot\bb\ai| & 0\,\,\\
0 & \bb J_{\lr}^{\m}|\,\,
\end{pmatrix}\,.
\end{equation}
We use the \foreignlanguage{american}{\acs{NS}} block basis of the
above form throughout the thesis. The two projections are then\begin{subequations}
\begin{align}
\ppp & =\sum_{\m}|\St_{\m}\kk\bb\J^{\m}|=\ps-\ps\L\L_{\lr}^{-1}\label{eq:proj}\\
\ps & =\sum_{\m}|\St_{\m}\kk\bb\J_{\ul}^{\m}|\equiv\idfs\ot|\as\kk\bb\ai|\,,\label{eq:nsproj-1}
\end{align}
\end{subequations}where the superoperator projection onto the \foreignlanguage{american}{\acs{DFS}}
is explicitly 
\begin{equation}
\idfs(\cdot)=\sum_{\m}|\stdfs_{\m}\kk\bb\stdfs_{\m}|\cdot\kk=\iidfs\cdot\iidfs\,.
\end{equation}
Therefore, states in $\ulbig$ are not perfectly preserved, but are
instead partially traced over the auxiliary subspace:
\begin{equation}
\ps(\rin)=\tr_{\textsf{ax}}\left\{ \pp\rin\pp\right\} \ot\as\,,\label{eq:nsproj}
\end{equation}
where $\pp=\iidfs\ot\ai$ and $\iidfs$ ($\ai$) is the identity on
$\hhdfs$ ($\hha$). 

For this case, the effect of perturbations $\spert$ on \acs{ASH}
is more subtle due to the auxiliary factor, but the induced time evolution
on the \foreignlanguage{american}{\acs{DFS}} is still unitary. The
effective \foreignlanguage{american}{\acs{DFS}} Hamiltonian can be
extracted from $\ps\spert\ps$ (see Sec.~\ref{subsec:hams}) and
is 
\begin{equation}
W=\tr_{\textsf{ax}}\left\{ \as V_{\ul}\right\} \,.\label{eq:nspt}
\end{equation}
Similarly, if we define generators of motion $G_{\a}$ in the $\xx_{\a}$-direction
in parameter space (i.e., such that $\p_{\a}\rout=-i[G_{\a},\rout]$;
see Sec.~\ref{subsec:holonomy}), then the corresponding holonomy
(Berry phase) after a closed path is the path-ordered integral of
the various \foreignlanguage{american}{\acs{DFS}} adiabatic connections
\begin{equation}
\adfst_{\a}=\tr_{\textsf{ax}}\left\{ \as(G_{\a})_{\ul}\right\} \,.\label{eq:nsad}
\end{equation}
In both cases, the effect of the perturbation on the \foreignlanguage{american}{\acs{DFS}}
part depends on $\as$, meaning that $\as$ can be used to modulate
both Hamiltonian-based or holonomic quantum gates. Making contact
with adiabatic evolution, Zanardi and Campos Venuti showed that first-order
Hamiltonian evolution within \acs{ASH} can be thought of as a holonomy.
The results (\ref{eq:nspt}-\ref{eq:nsad}) further develop a connection
between holonomies and first-order perturbative effects within \acs{ASH}
(\cite{Zanardi2015}, Prop. 1) by showing that, for both processes,
evolution is generated by the same type of effective Hamiltonian ($W$
and $\adfst_{\a}$, respectively).

The \ac{QGT} for this case is rather complicated due to the $\as$-assisted
adiabatic evolution. However, we devote Ch.~\ref{ch:6} to showing
that the \ac{QGT} \textit{does} endow us with a metric on the parameter
space for one \foreignlanguage{american}{\acs{NS}} block.

\subsection{Multi-block case}

The noiseless subsystem is the most general form of one block of asymptotic
states of $\L$, and the most general \acs{ASH} is a direct sum of
such \foreignlanguage{american}{\acs{NS}} blocks \cite{baum2,Ticozzi2008,Deschamps2014}
(also called ``superselection sectors'' \cite{Pastawski2016}) {[}see
Fig.~\ref{fig:decomp}(a){]} with corresponding minimal projection
$\ps$. This important result applies to all quantum channels \cite{Lindblad1999,BlumeKohout2008,robin,baumr,Carbone2015}
and stems from a well-known algebra decomposition theorem (see \cite{wolf2010}
for a technical introduction). More technically, \acs{ASH} is a \textit{matrix}
or \textit{von Neumann algebra}, but with each block $\varkappa$
\textit{distorted} \cite{robin,Johnson2015} (see also Corr.~6.7
in Ref.~\cite{wolf2010}) by its corresponding fixed auxiliary steady
state $\as^{(\varkappa)}$:
\begin{equation}
\ash=\bigoplus_{\varkappa}\textnormal{Op\ensuremath{\left(\mathsf{H_{dfs}^{(\varkappa)}}\right)}}\ot\as^{(\varkappa)}\,.
\end{equation}
Of course, the above reduces to the \foreignlanguage{american}{\acs{NS}}
case when there is only one block. Throughout the text, we explicitly
calculate properties of one \foreignlanguage{american}{\acs{NS}}
block and sketch any straightforward generalizations to the multi-block
case. To reiterate, the subtleties of $\ash\subseteq\ulbig$ are independent
of the presence of a decaying subspace $\lrbig$.

If there are two \foreignlanguage{american}{\acs{NS}} blocks (characterized
by projections $\iidfs^{(\varkappa)}\ot\ai^{(\varkappa)}$ with $\varkappa\in\{1,2\}$)
and no decaying subspace, then the conserved quantities $\J^{\varkappa,\m}=\stdfs_{\varkappa,\m}\ot\ai^{(\varkappa)}$
do not have presence in the subspace of coherences between the blocks.
The blocks in which $\J^{\m}\neq0$ are shaded in gray in Fig.~\ref{fig:decomp}(b),
dual to $\St_{\m}$ in Fig.~\ref{fig:decomp}(a).

Both eqs.~(\ref{eq:nspt}) and (\ref{eq:nsad}) extend straightforwardly
to the multi-block case, provided that the blocks maintain their shape
during adiabatic evolution. We make the same connection between ordinary
and adiabatic perturbations to jump operators of $\L$; the latter
were first studied in Avron \textit{et al.} \cite{Avron2012a}. If
we add in jump operator perturbations $F\rightarrow F+f$, then the
generalization of the effective Hamiltonian $V_{\ul}$ (\ref{eq:nspt})
is derived in Sec.~\ref{subsec:linds} to be
\begin{equation}
X_{\ul}\equiv\hpert_{\ul}+\frac{i}{2}\lind\left(F^{\dg}f-f^{\dg}F\right)_{\ul}
\end{equation}
projected onto all of the \foreignlanguage{american}{\acs{NS}} blocks.
This quantity has previously been introduced (\cite{Avron2012a},
Thm.~5) as the operator resulting from joint adiabatic variation
of the Hamiltonian and jump operators of $\L$. It is thus not surprising
that the effect of perturbations to the Hamiltonian and jump operators
on $\rout$ is $X$ projected onto \acs{ASH}.

\section{Relation of conserved quantities to symmetries\label{sec:Relation-of-conserved}}

For $\hout=0$, all $J^{\la\m}=J^{\la=0,\m}$ are conserved quantities,
and a natural question is whether they always commute with the Hamiltonian
$H$ and the jump operators $F^{\ell}$. It turns out that they do
not always commute \cite{baum2,pub011}, and so various generalizations
of Noether's theorem have to be considered \cite{Avron2012a,Gough2015}.

\subsection{A Noether's theorem for Lindbladians?\label{sec12}}
\selectlanguage{american}%

\subsubsection*{Hamiltonian case}

\selectlanguage{english}%
In a unitary system, an (explicitly time-independent) observable $J=J^{\dg}$
is a conserved quantity (i.e. constant of motion) if and only if it
commutes with the Hamiltonian (e.g., angular momentum of the hydrogen
atom). In the spirit of Noether's theorem, one can then generate a
continuous symmetry $U=\exp(i\phi J)$ (for real $\phi$) that leaves
the Hamiltonian invariant. There is thus the following set of equivalent
statements for continuous symmetries in unitary evolution (with one-sided
arrows depicting an ``if-then'' statement, two-sided arrows depicting
``iff,'' and the dot being total time derivative): 
\begin{equation}
\begin{array}{rcl}
 & [J,H]=0\\
\lrff &  & \lrtf\\
\dot{J}=0~ & \Leftrightarrow & ~U^{\dg}HU=H
\end{array}\label{symreg}
\end{equation}
We call the above triple iff relationship \textit{Noether's theorem}.
\selectlanguage{american}%

\subsubsection*{Lindbladian case}

\selectlanguage{english}%
A conserved quantity in Lindbladian systems is one where $\L^{\dagger}(J)=0$
(\ref{eq:leig}). Let us define the superoperator corresponding to
$J$,
\begin{equation}
\mathcal{J}(\cdot)\equiv+i[J,\cdot],\,\text{and its corresponding unitary }\mathcal{U}=\exp(\phi\mathcal{J})\,.
\end{equation}
Just like with Hamiltonians $H$ and their superoperators $\H$, 
\begin{equation}
\mathcal{U}^{\dgt}(F^{\ell})\equiv U^{\dg}F^{\ell}U.
\end{equation}
We consider only superoperators $\mathcal{U}$ which can be written
in terms of a $J$ on the operator level. Using this notation, one
produces an analogous set of statements for $\L$: 
\begin{equation}
\begin{array}{c}
~~[J,H]=[J,F^{\ell}]=0~~\forall\ell\\
\begin{array}{rcl}
\lff~~~~~~~ &  & ~~~~~\rtf\\
\dot{J}=\L^{\dgt}(J)=0~~~~~~ &  & ~~~~\mathcal{U}^{\dgt}\L\mathcal{U}=\L
\end{array}
\end{array}\,.\label{dissym}
\end{equation}
In comparison to the original theorem (\ref{symreg}), four arrows
are lost! The two arrows emanating from $\mathcal{U}^{\dgt}\L\mathcal{U}=\L$
are lost because there exist operators which leave $\L$ invariant
but are neither conserved nor commute with everything. A simple example
of a symmetry that neither commutes with $F^{\ell}$ nor is conserved
is $J=\aa^{\dg}\aa\equiv\ph$ for a bosonic $\L$ with the lowering
operator $F=\aa$, $H=0$; we discuss this case in detail in Sec.~\ref{sec22}.
The loss of the two arrows emanating from $\dot{J}=\L^{\dgt}(J)=0$
is due to the decaying subspace, which guarantees that conserved quantities
$\J$ generally contain the piece $\J_{\lr}$. However, it turns out
that certain conditions can restore these arrows: (1) these arrows
are restored in $\ulbig$ and (2) they are restored for any $J$ which
squares to the identity. We discuss these cases below, noting that
the above theorem was motivated by Ref.~\cite{baum2}.

\subsubsection*{Noether's theorem partially restored in non-decaying subspace}

The only restriction on the Lindbladian $\L_{\ul}$ is that it contains
a full-rank steady state. It turns out that this restriction is sufficient
to partially restore Noether's theorem. Recall that each conserved
quantity $J^{\m}=\J_{\ul}^{\m}+\J_{\lr}^{\m}$. Using the dissipation
function from the proof of Thm.~\ref{prop:3}, one can show that
$[F_{\ul}^{\ell},\J_{\ul}^{\m}]=0$ for all $\ell$ and $\m$. This
restores some of the arrows and produces
\begin{equation}
\begin{array}{c}
~~[J_{\ul},H_{\ul}]=[J_{\ul},F_{\ul}^{\ell}]=0~~\forall\ell\\
\begin{array}{rcl}
\lrff~~~~~~~ &  & ~~~~~\rtf\\
\dot{J}_{\ul}=\L_{\ul}^{\dgt}(J_{\ul})=0~~~~~~ & \Rightarrow & ~~~~\mathcal{U}_{\ul}^{\dgt}\L_{\ul}\mathcal{U}_{\ul}=\L_{\ul}
\end{array}
\end{array}\,,\label{dissym-1}
\end{equation}
a partially restored Noether's theorem that holds in $\ulbig$. While
being in $\ulbig$ is thus a sufficient condition for the above to
hold, it is not a necessary one (see Sec.~\ref{subsec:Clean-case}
for an example).

\subsubsection*{Parity \& discrete rotations}

Having omitted discrete symmetries ($U$ where $\phi$ takes discrete
values), we expound on parity since it is the simplest discrete symmetry
and it is a good starting point for the further examples in Chs.~\ref{ch:3}
and \ref{ch:7}. Equation (\ref{dissym}) shows that if one can find
a non-trivial operator that commutes with everything in $\L$, then
one is lucky to have found both a symmetry and a conserved quantity.
It turns out that systems with parity conservation necessarily have
such an operator and parity can be thought of as a symmetry almost
in the unitary sense of eq.~(\ref{symreg}): 

\begin{equation}
\begin{array}{c}
~~~~~~[P,H]=[P,F^{\ell}]=0~~~\forall l\\
\begin{array}{rcl}
\lrff~~ &  & ~~~\rtf\\
\dot{P}=\L^{\dagger}(P)=0~~ & \Rightarrow & ~~~[\mathcal{P},\L]=0
\end{array}
\end{array}\,.\label{eq:noether-proj}
\end{equation}
In the above, $\mathcal{P}(F^{\ell})=PF^{\ell}P$. The proof is simple.
Assuming $\L^{\dgt}(P)=0$, it is possible to construct conserved
positive- and negative-parity projections $\Pi_{\pm}=\half(I\pm P)$,
which in turn must commute with all operators in $\L$ (\citep{baum2},
Lemma 7). Therefore, $P$ must commute with everything as well.

In general, any set of $d$ conserved projection operators partitions
\acs{OPH} into $d^{2}$ subspaces which evolve independently under
$\L$ (\citep{baum2}, Thm.~3), with at least $d$ of the subspaces
having their own steady state. In each of the $d$ subspaces, the
steady state basis element is a population {[}e.g., $|n\ket\bra n|$
with $n\in\{0,1,\cdots d-1\}$; see Fig.~\ref{fig:specific-ash}(b){]}
and an \acs{ASH} such as this corresponds to a classical $d$it (see
Sec. \ref{sec31}). However, such a symmetry is neither necessary
nor sufficient for the existence of a \foreignlanguage{american}{\acs{DFS}},
i.e., for cases where all $d^{2}$ subspaces each have one steady-state
basis element. Regarding the necessary condition, we study examples
with symmetries which \textit{do} admit a \foreignlanguage{american}{\acs{DFS}}
in Sec.~\ref{sec:Many-photon-absorption} and there exist other examples
which do not \cite{prozen}. Regarding sufficiency, steady-state coherences
$\St_{\m,\n\neq\m}$ can exist with or without a discrete symmetry
(e.g., Sec.~5.3 in Ref.~\citep{baum2}). Both of these cases are
demonstrated pictorially via the two types of $\rout$ below: 
\begin{equation}
\left(\begin{array}{cccc}
\vspace{0.025in}\St_{00} & \St_{01} & \leftarrow & ~~\\
\vspace{0.025in}\St_{10} & \St_{11} & \leftarrow\\
\vspace{0.025in}\uparrow & \uparrow & \nwarrow\\
\vspace{0.025in}
\end{array}\right),\left(\begin{array}{cc|cc}
\St_{00} & \leftarrow & \St_{01} & \leftarrow\\
\uparrow & \nwarrow & \uparrow & \nwarrow\\
\hline \St_{10} & \leftarrow & \St_{11} & \leftarrow\\
\uparrow & \nwarrow & \uparrow & \nwarrow
\end{array}\right)\,.
\end{equation}
In the above list, arrows represent parts of the space which converge
to $\St_{\m\n}\in\ash$, which form a qubit \foreignlanguage{american}{\acs{DFS}}.
The left example symbolizes a system with no parity symmetry. In the
right example of a system with parity symmetry, the full space is
``cut-up'' into four independent subspaces, each of which converges
to a steady state/coherence.

\subsection{Symmetries\label{sec22}}

We now partially extend the parity discussion in the previous Subsection
to more general symmetries. As mentioned in eq.~(\ref{dissym}),
a continuous symmetry $U$ is a unitary operator whose corresponding
superoperator $\mathcal{U}=e^{\phi\mathcal{J}}$ is such that $\mathcal{U}^{\dgt}\L\mathcal{U}=\L$,
or equivalently $[\mathcal{J},\L]=0$. It is therefore easy to see
that both $\mathcal{U}$ and $\mathcal{U}^{\dgt}$ are symmetries
of both $\L$ and $\L^{\dgt}$. To state in a different way, $\mathcal{U}$
commutes with time-evolution generated by $\L$, 
\begin{equation}
e^{t\L}\U^{\dgt}(\rin)=e^{t\L}(U^{\dg}\rin U)=U^{\dg}e^{t\L}(\rin)U=\U^{\dgt}e^{t\L}(\rin)
\end{equation}
for any $\rin\in\oph$. Examples of symmetries include any $U$ such
that $UHU^{\dg}=H$ and $UF^{\ell}U^{\dg}=e^{i\phi_{l}}F^{\ell}$
\citep{bardyn} or any permutations among the jump operators $F^{\ell}$
that leave $\L$ invariant \citep{prozen}. Note that the former case
provides an example of a symmetry whose generator doesn't commute
with $F^{\ell}$. The Lindbladian can be block-diagonalized by $\mathcal{U}$
(with each block corresponding to an eigenvalue of $\U$) in the same
way that a Hamiltonian can be block-diagonalized by $U$ (with each
block corresponding to an eigenvalue of $U$). Symmetries can thus
significantly reduce computational cost, with the additional complication
that the blocks of $\L$ may not be further diagonalizable. However,
symmetries by themselves \textit{do not} determine the dimension of
\acs{ASH} because some blocks may contain only decaying subspaces
and no steady states. Diagonal parts of $\rin$ will always be in
blocks with steady states since the trace is preserved. For a unitary
$U$ such that $[U,H]=[U,F^{\ell}]=0$, $\dim\{\ash\}$ will be at
least as much as the number of distinct eigenvalues of $U$ (\citep{prozen},
Thm.~A.1). One can see this by decomposing $U$ into a superposition
of projections on its eigenspaces and applying eq.~(\ref{eq:noether-proj})
to each projection. However, such a result once again does not say
anything about whether \acs{ASH} will be a quantum memory or a classical
one.

An example of a symmetry is invariance of the zero-temperature cavity
\textemdash{} $\L$ with $F=\aa$ and $H=0$ \textemdash{} under bosonic
rotations $R_{\phi}\equiv e^{i\phi\ph}$ (with $\ph=\aa^{\dg}\aa$).
This is an example of a continuous symmetry which does not stem from
a conserved quantity in \acs{OPH}. Instead, this symmetry stems from
the generator $\mathcal{N}$ of the corresponding $\mathcal{R}_{\phi}\equiv e^{i\phi\mathcal{N}}$,
which commutes with $\L$. The generator acts as $\mathcal{N}(\rin)=\ph\rin-\rin\ph$
and its commutation with $\L$ can be easily checked. The block diagonalization
of $\L$ stemming from this symmetry corresponds to equations of motion
for matrix elements $\bra n|\rin|m\ket$ with $m-n=r$ being decoupled
from those with $m-n\neq r$ \{\citep{zoller_book}, eq.~(6.1.6)\}.
This will be used to calculate conserved quantities in Sec.~\ref{sec:Many-photon-absorption}.
In this way, symmetries can help compartmentalize evolution of both
states and operators.

Any symmetry of $\L$ leaves the asymptotic subspace invariant,
\begin{equation}
\U^{\dgt}\L\U=\L\,\,\,\,\,\,\,\Rightarrow\,\,\,\,\,\,\,\U^{\dgt}\ppp\U=\ppp\,.
\end{equation}
This is because one can create a basis for \acs{ASH} which consists
of eigenstates of $\U$. Symmetries can thus classify \citep{bardyn}
unique steady states and/or constrain their properties \cite{popkov2013,Rivas2016}.
For the example of the previous paragraph, the vacuum $|0\ket\bra0|$
is rotationally invariant under $R_{\phi}$. When \acs{ASH} is not
one dimensional, symmetries will rotate \acs{ASH} into itself and
so can act nontrivially on any given state $\rout\in\ash$. Symmetries
can thus be used to perform unitary rotations on the steady-state
subspace. We briefly mention the existence of anti-commuting symmetries
such as chiral \citep{bardyn} or parity-time \citep{prosen2012,prosen2012a}
for dissipative dynamics. These can reveal symmetries in the spectrum
of $\L$ and $\L^{\dgt}$ \citep{prosen2012}, similar to the spectrum
of a chirally-symmetric Hamiltonian being symmetric around zero. A
brute-force approach of finding all symmetries of $\L$ is to find
the null space of the commutator super-superoperator $[\L,\cdot]$
(\cite{schirmer_symmetries}, Appx. A; see also \cite{Zeier2011}).\selectlanguage{english}%

\inputencoding{latin9}\newpage{}\foreignlanguage{english}{}%
\begin{minipage}[t]{0.5\textwidth}%
\selectlanguage{english}%
\begin{flushleft}
\begin{singlespace}\textit{``Never do a calculation until you know
the answer.''}\end{singlespace}
\par\end{flushleft}
\begin{flushleft}
\hfill{}\textendash{} Steven M. Girvin
\par\end{flushleft}\selectlanguage{english}%
\end{minipage}

\chapter{Examples\label{ch:3}}

\selectlanguage{english}%
Here, we present examples of $\L$ which do not have a unique steady
state. The first three Sections focus on calculating conserved quantities
of single-body systems such as qubits and oscillators \cite{pub011,ABFJ}.
The last section takes the systems theory perspective and reviews
standard techniques for stabilizing ground state subspaces of many-body
frustration-free Hamiltonians.

\section{Single-qubit dephasing\label{sec31}}

In this example, \acs{ASH} is two-dimensional: all steady states
can be written as convex combinations of two orthogonal basis states
$\St_{0}$ and $\St_{1}$. In other words, such an \acs{ASH} stores
one classical bit (i.e., one probability's worth of information).
Such an \acs{ASH} can also be thought of as a one-dimensional simplex
\cite{Macieszczak2015} and is the simplest version of an information-preserving
structure \cite{robin}. In the many-body case, such a system is called
bistable \cite{Letscher2016}. We also initially assume that there
is no decaying subspace ($\pp=I$).

Consider one qubit undergoing dephasing on two of the three axes of
its Bloch sphere, thereby stabilizing the Bloch vector onto the third
axis. In this case, there is one jump operator $F=Z$ (with $X,Y,Z$
the usual Pauli matrices) and no Hamiltonian. The master equation
simplifies to the \textit{Poisson semigroup generator} (i.e., a Lindbladian
with a unitary jump operator \cite{Lindblad1976})
\begin{equation}
\L(\r)=Z\r Z-\half\{I,\r\}=Z\r Z-\r=-\half[Z,[Z,\r]]\,.
\end{equation}
Picking the eigenbasis of $Z$, $Z|\m\ket=(-)^{\m}|\m\ket$ with $\m=0,1$,
one can see that the states $\St_{\m}=|\m\ket\bra\m|$ will be steady
but the coherence $|0\ket\bra1|$ will not survive. The steady-state
density matrix is then
\begin{equation}
\rout=\lim_{t\rightarrow\infty}e^{t\L}(\rin)=\ppp(\rout)=c_{0}|0\ket\bra0|+c_{1}|1\ket\bra1|\,.
\end{equation}
Naturally, one expects the system to record the initial $Z$-component
of $\rin$. One can see that $\L^{\dgt}=\L$ since the jump operator
is Hermitian, so the conserved quantities $J^{\m}=\St_{\m}$. Letting
$c_{Z}=\text{Tr}\{Z\rin\}$, one indeed determines that the $Z$-component
is preserved and $c_{\m}=\half[1+(-)^{\m}c_{Z}]$. 

This example can be straightforwardly generalized to an $N$-dimensional
system whose \acs{ASH} is spanned by all diagonal populations $|\m\ket\bra\m|$
(for $\m=0,1,\cdots,N-1$) and where there is still no decaying subspace.
Lindbladians with such \acs{ASH} include \textit{dephasing Lindbladians}
\cite{Avron2011} {[}see Fig.~\ref{fig:specific-ash}(b){]}. These
systems can be used to model what happens in a measurement \cite{weinberg2016},
with $\St_{\m}$ interpreted as \textit{pointer states} of the system
\cite{Zurek2003}.

One can also extend this qubit example to include a decaying subspace
$\lrbig$. While the structure of \acs{ASH} would remain the same,
$c_{\m}$ would additionally store information about the populations
(and possibly the coherences) of states initially in $\lrbig$. For
example, consider adding a third level $|2\ket$ to the Hilbert space
and adding another jump operator $F^{\prime}=|0\ket\bra2|$ which
decays that new level to $|0\ket$. Then, the conserved quantity associated
with $|0\ket$ gains an additional term which ``catches'' the decayed
population: $J^{0}=\St_{0}+|2\ket\bra2|$.

Finally, this example can be extended to cases where the two (or more)
$\St_{\m}$ are mixed states. One can imagine such a case in an thermal
equilibrium system with a parity symmetry: $\rout=c_{0}\St_{0}+c_{1}\St_{1}$,
where $\St_{\m}$ are the unique steady states in each parity sector.
We direct the interested reader to further examples of such systems
in spin chains \cite{prozen,Ilievski2014,Medvedyeva2016,Monthus2017,Monthus2017a},
fermionic systems (\cite{prosen2010}, Example 5.2), and quantum transport
in models of energy harvesting \cite{Manzano2013}.

\section{Two-qubit dissipation\label{sec32}}

In this example, \acs{ASH} is initially a one-qubit \foreignlanguage{american}{\acs{DFS}}
with $\hout=0$ and there is a two-dimensional decaying subspace $\lrbig$.
This example is taken from recent experimental work that stabilizes
Bell states using trapped ions \cite{zoller_stabilizers} and is closely
related to stabilizer generators of qubit codes \cite{sarma2013}.
We study this case in detail by adding different Hamiltonians and
jump operators \cite{pub011,ABFJ} and seeing how the structure of
\acs{ASH} and the conserved quantities changes.

\subsection{Clean case\label{subsec:Clean-case}}

Let \acs{OPH} be the space of matrices acting on the Hilbert space
$\h$ of two qubits. Let $\L$ have one jump operator ($c$ in Box
1 of \cite{zoller_stabilizers}) 
\begin{equation}
F=\half\left(I-Z_{1}Z_{2}\right)X_{2}\,,
\end{equation}
where the subscript labels the qubit. Intuitively, \acs{ASH} is equivalent
to the space spanned by $|\psi_{k}\ket\bra\psi_{l}|$, where $k,l\in\{0,1\}$
and the Bell states $|\psi_{k}\ket\equiv\frac{1}{\sqrt{2}}[|01\ket+(-)^{k}|10\ket]$.
While we can operate using the original Pauli matrices, let us instead
call the other two Bell states $|\psi_{k}^{\perp}\ket\equiv\frac{1}{\sqrt{2}}[|00\ket+(-)^{k}|11\ket]$
and re-write the jump using this basis:
\begin{equation}
F=\sum_{k=0}^{1}|\psi_{k}\ket\bra\psi_{k}^{\perp}|=\left(\begin{array}{cc|cc}
0 & 0 & 1 & 0\\
0 & 0 & 0 & 1\\
\hline 0 & 0 & 0 & 0\\
0 & 0 & 0 & 0
\end{array}\right)=F_{\ur}\,.\label{eq:jumporig}
\end{equation}
This allows us to conform to the $\empbig$ structure, which is delineated
using the two lines in the matrix above: the first two Bell states
$\{|\psi_{0}\ket,|\psi_{1}\ket\}$ form the \foreignlanguage{american}{\acs{DFS}}
and the latter two $\{|\psi_{0}^{\perp}\ket,|\psi_{1}^{\perp}\ket\}$
decay into the \foreignlanguage{american}{\acs{DFS}}. Obviously,
\acs{ASH} is spanned by 
\begin{equation}
\Psi_{kl}=|\psi_{k}\ket\bra\psi_{l}|\,.
\end{equation}
The conserved quantities can be determined by the formula in Thm.~\ref{prop:3}:
\begin{equation}
J^{kl}=|\psi_{k}\ket\bra\psi_{l}|+|\psi_{k}^{\perp}\ket\bra\psi_{l}^{\perp}|\,.
\end{equation}
One would think that since there is a decaying subspace, the conserved
quantities do not commute with $F$ (see Sec.~\ref{sec12}). However,
the presence of $\lrbig$ is not a sufficient condition and we do
in fact have $[J^{kl},F]=0$ because there is no extra decoherence
in $\lrbig$ ($F_{\lr}=0$) and because the decay of states is one-to-one
($|\psi_{k}^{\perp}\ket\rightarrow|\psi_{k}\ket$ for $k=0,1$). A
few sanity checks: the steady diagonal state basis elements add up
to the projection onto the \foreignlanguage{american}{\acs{DFS}},
$P=\Psi_{00}+\Psi_{11}$, and the diagonal conserved quantities correspondingly
add up the identity, $I=J^{00}+J^{11}$. Both $J^{kl}$ and $\St_{kl}$
form the Lie algebra $\mathfrak{u}(2)$. The steady state $\rout\in\ash$
for initial state $\rin\in\oph$ can be expressed as 
\begin{align}
\rout & =\sum_{k,l=0}^{1}\tr\{J^{kl\dg}\rin\}|\psi_{k}\ket\bra\psi_{l}|=P\rin P+\sum_{k,l=0}^{1}\bra\psi_{k}^{\perp}|\rin|\psi_{l}^{\perp}\ket|\psi_{k}\ket\bra\psi_{l}|\,.
\end{align}

Notice that $\Pi=J^{00}-J^{11}$ is a parity operator, meaning that
the analysis from Sec.~\ref{sec12} holds and we can partition the
evolution into invariant blocks. Everything in $\L$ commutes with
$\Pi$, so \acs{OPH} can be partitioned into four blocks of matrices
that are built out of the two subspaces of $\h$ of positive and negative
parity. Each block (indexed by $k,l$) is of the form $\{|\psi_{k}\ket\bra\psi_{l}|,|\psi_{k}\ket\bra\psi_{l}^{\perp}|,|\psi_{k}^{\perp}\ket\bra\psi_{l}|,|\psi_{k}^{\perp}\ket\bra\psi_{l}^{\perp}|\}$,
and there are four of them since $k,l\in\{0,1\}$. Each block has
its own steady state basis element $\St_{kl}$ and conserved quantity
$J^{kl}$ and evolves independently from the other blocks.

\subsection{Coherence suppressed by a jump}

Continuing from the previous \foreignlanguage{american}{\acs{DFS}}
example, let us add a term to $F_{\lr}$:
\begin{equation}
F=\sum_{k=0}^{1}|\psi_{k}\ket\bra\psi_{k}^{\perp}|+\a\sum_{k=0}^{1}(-1)^{k}|\psi_{k}^{\perp}\ket\bra\psi_{k}^{\perp}|=\left(\begin{array}{cc|cc}
0 & 0 & 1 & 0\\
0 & 0 & 0 & 1\\
\hline 0 & 0 & \a & 0\\
0 & 0 & 0 & -\a
\end{array}\right)\,,
\end{equation}
where the $\a$-dependent term $F_{\lr}$ now dephases the non-DFS
Bloch vector (with $\a\in\mathbb{R}$). The steady-state basis elements
are still $\St_{kl}=|\psi_{k}\ket\bra\psi_{l}|$ since $F_{\ul}=0$.
To determine the corresponding $J^{kl}$, we use eq.~(\ref{eq:main}).
Acting on $\St_{kl}$ with $\L^{\dgt}$ (\ref{eq:adjl}) and then
the adjoint of $\L_{\lr}^{-1}$ (\ref{eq:llr}) yields the corresponding
conserved quantities
\begin{equation}
\J^{kl}=|\psi_{k}\ket\bra\psi_{l}|+\frac{|\psi_{k}^{\perp}\ket\bra\psi_{l}^{\perp}|}{1+2\a^{2}(1-\d_{kl})}\,.
\end{equation}
The only non-trivial feature of the steady state is due to $F_{\lr}$
and $\L_{\lr}^{-1}$. Namely, an initial nonzero coherence $\bra\psi_{0}^{\perp}|\rin|\psi_{1}^{\perp}\ket$
leads necessarily to a mixed steady state due to coherence suppression
of order $O(\a^{-2})$.

\subsection{Coherence suppressed by a Hamiltonian ($H_{{\scriptscriptstyle \infty}}\protect\neq0$)\label{subsec:Coherence-suppressed-by}}

A similar coherence suppression can be achieved by adding the Hamiltonian
\begin{equation}
H=\half\b(|\psi_{0}\ket\bra\psi_{0}|-|\psi_{1}\ket\bra\psi_{1}|)=H_{\ul}
\end{equation}
 (with $\b\in\mathbb{R}$) to the original jump $F$ from eq.~(\ref{eq:jumporig}).
Now the \foreignlanguage{american}{\acs{DFS}} is non-stationary (with
$\hout=H$) and the off-diagonal \foreignlanguage{american}{\acs{DFS}}
elements $\St_{k\neq l}$ rotate. Abusing notation by omitting the
corresponding eigenvalue $\la=\b$, the left asymptotic eigenmatrices
become
\begin{equation}
\J^{kl}=|\psi_{k}\ket\bra\psi_{l}|+\frac{|\psi_{k}^{\perp}\ket\bra\psi_{l}^{\perp}|}{1+i\b(-)^{l}(1-\d_{kl})}\,.
\end{equation}
Despite the fact that $F_{\lr}=0$, the inverse $(\L-i\la)_{\lr}$
from Thm.~\ref{prop:3} still inflicts damage to the initial state
due to $\hout$ (for nonzero $\b$), but now the coherence suppression
is of order $O(\b^{-1})$. 

The coherence suppression shown above is due to $\ulbig$ rotating
while states from $\lrbig$ flow into it. It turns out that one can
cancel that suppression by also rotating $\lrbig$ in the same direction.
If we add another term to the Hamiltonian, 
\begin{equation}
H^{\pr}=\half\a(|\psi_{0}^{\perp}\ket\bra\psi_{0}^{\perp}|-|\psi_{1}^{\perp}\ket\bra\psi_{1}^{\perp}|)=H_{\lr}^{\pr}\,,
\end{equation}
then, once again due to the inverse piece $(\L-i\la)_{\lr}$ in determining
$J^{kl}$, 
\begin{equation}
J^{01}=J^{10\dg}=|\psi_{0}\ket\bra\psi_{1}|+\frac{|\psi_{0}^{\perp}\ket\bra\psi_{1}^{\perp}|}{1+i\left(\alpha-\beta\right)}\,.
\end{equation}
Setting $\a=\b$ means that the coherence suppression can be canceled
by a proper rotation in the decaying subspace $\lrbig$. This effect
will be studied in a future work.

\subsection{Driven case\label{sec33}}

As a final example, let us take this \foreignlanguage{american}{\acs{DFS}}
case and convert it into an \foreignlanguage{american}{\acs{NS}}
by changing the original jump from eq.~(\ref{eq:jumporig}) to (assuming
$\g\in\mathbb{R}$)
\begin{equation}
F=\sum_{k=0}^{1}|\psi_{k}\ket\bra\psi_{k}^{\perp}|-\g I\,.
\end{equation}
Note that we could have equivalently added the Hamiltonian 
\begin{equation}
H=-i\g\sum_{k=0}^{1}|\psi_{k}\ket\bra\psi_{k}^{\perp}|+H.c.
\end{equation}
due to the ``gauge'' transformation (\ref{eq:gauge}). This driving
expands the \foreignlanguage{american}{\acs{DFS}} into a qubit \foreignlanguage{american}{\acs{NS}}
{[}Fig.~\ref{fig:specific-ash}(f){]}, absorbing the decaying subspace
$\lrbig$. Since the new jump still commutes with everything the original
($\g=0$) jump commuted with, we still have a parity symmetry and
therefore can study each individual block of states $\{|\psi_{k}\ket\bra\psi_{l}|,|\psi_{k}\ket\bra\psi_{l}^{\perp}|,|\psi_{k}^{\perp}\ket\bra\psi_{l}|,|\psi_{k}^{\perp}\ket\bra\psi_{l}^{\perp}|\}$.
Now, the steady state basis element in each block is
\begin{equation}
\St_{kl}=\frac{\left(1+\g^{2}\right)|\psi_{k}\ket\bra\psi_{l}|+\g|\psi_{k}\ket\bra\psi_{l}^{\perp}|+\g|\psi_{k}^{\perp}\ket\bra\psi_{l}|+\g^{2}|\psi_{k}^{\perp}\ket\bra\psi_{l}^{\perp}|}{\sqrt{1+4\g^{2}+2\g^{4}}}\,,
\end{equation}
where we are adding $\g$-dependent factors so that $\bb\St_{kl}|\St_{pq}\kk=\d_{kp}\d_{lq}$.
The previous conserved quantities are carried over since they commute
with $F$ for all values of $\g$, but now we have to add extra factors
in order to make sure $\bb J^{kl}|\St_{pq}\kk=\d_{kp}\d_{lq}$:
\begin{equation}
J^{kl}=\frac{\sqrt{1+4\g^{2}+2\g^{4}}}{1+2\g^{2}}\left(|\psi_{k}\ket\bra\psi_{l}|+|\psi_{k}^{\perp}\ket\bra\psi_{l}^{\perp}|\right)\,.
\end{equation}
Let us now organize \acs{OPH} in a different way in order to reveal
the \foreignlanguage{american}{\acs{NS}} structure and its associated
quantities, introduced in Sec.~\ref{sec:Examples}. Both the $\St$'s
and $J$'s can be put into the following factored form:
\begin{align}
\St_{kl} & =|k\ket\bra l|\ot\frac{\as}{\an}\,\,\,\,\,\,\,\,\,\,\,\,\,\,\text{and}\,\,\,\,\,\,\,\,\,\,\,\,\,\,J^{kl}=|k\ket\bra l|\ot\an\ai\,,
\end{align}
where the auxiliary steady state, identity, and square-root of purity
($\an\equiv\sqrt{\tr\{\as^{2}\}}$) are 
\begin{equation}
\as\equiv\frac{1}{1+2\g^{2}}\begin{pmatrix}1+\g^{2} & \g\\
\g & \g^{2}
\end{pmatrix}\,,\,\,\,\,\,\,\,\,\,\,\,\,\,\,\,\,\ai\equiv\begin{pmatrix}1 & 0\\
0 & 1
\end{pmatrix}\,,\,\,\,\,\,\,\,\,\,\,\,\,\,\,\,\,\an\equiv\frac{\sqrt{1+4\g^{2}+2\g^{4}}}{1+2\g^{2}}
\end{equation}
and $|k\ket\bra l|$ is the basis for the \foreignlanguage{american}{\acs{DFS}}
part of the \foreignlanguage{american}{\acs{NS}}. Naturally, the
\foreignlanguage{american}{\acs{NS}} reduces to a \foreignlanguage{american}{\acs{DFS}}
as $\g\rightarrow0$, whereas the dissipation and driving balance
out and produce a maximally mixed $\as$ when $\g\rightarrow\infty$.
\selectlanguage{american}%

\section{Many-photon absorption\label{sec:Many-photon-absorption}}

\selectlanguage{english}%
In this example, \acs{ASH} is initially a \foreignlanguage{american}{\acs{DFS}}
with $\hout=0$ and there is an infinite-dimensional decaying subspace
$\lrbig$. This family of examples includes single-mode two-photon
\cite{loudon1975,simaan1975,agarwal1987,gilles1993,simaan1978} and
$d>2$-photon \cite{voigt1980,zubairy1980,klimov2003} absorption.
Their dynamics has been analytically solved for all time in the aforementioned
references, but studying \acs{ASH} does not require that tedious
algebra. These and related quantum optical systems (see Ref.~\cite{dodonov1997}
for a brief review of the older literature) have been recently gaining
interest from the quantum information (see Ch.~\ref{ch:7}) and optomechanics
\cite{nunnenkamp2012,borkje2013} communities. 

\subsection{Two-photon case\label{sec34}}

Consider bosonic systems with jump operator $F=\aa^{2}$ with $[\aa,\aa^{\dg}]=I$
and no Hamiltonian. While this system is infinite, one can successfully
analyze it for finite energy using a large finite Fock space spanned
by $\{|n\ket\}_{n=0}^{N}$ (where $N\gg1$) \cite{schirmer}. Here,
\acs{ASH} has the same qubit \foreignlanguage{american}{\acs{DFS}}
structure as in Example~\ref{sec32}, with basis $\St_{kl}=|k\ket\bra l|$
in Fock space (with $k,l=0,1$), as well as a similar parity symmetry.
The diagonal conserved quantities $J^{kk}$ correspond to projectors
on the even and odd subspaces respectively:
\begin{equation}
J^{kk}\equiv\sum_{n=0}^{\infty}|2n+k\ket\bra2n+k|\equiv\Pi_{k}\,,\label{eq:catproj2}
\end{equation}
and we can once again build a parity operator $J^{00}-J^{11}$ that
commutes with $F$. Due to this symmetry, \acs{OPH} is once again
split into four independent subspaces: the four blocks $\{|2n+k\ket\bra2m+l|\}_{n,m=0}^{\infty}$
(labeled by $k,l\in\{0,1\}$) evolve independently of each other.

The conserved quantity for the off-diagonal subspace,
\begin{equation}
J^{01}=\frac{(\ph-1)!!}{\ph!!}\Pi_{0}\aa,\label{eq:j01}
\end{equation}
where $(m+2)!!\equiv(m+2)\cdot m!!$ is the double factorial \cite{dfac},
does not commute with the jump operator $F$. This conserved quantity
(\ref{eq:j01}), derived first in Ref.~\cite{simaan1978}, has provided
the primary motivation for the initial portions of this body of work.
One can obtain this quantity by first using the parity symmetry to
isolate the subspace where it exists and then solving the equation
$\L^{\dgt}(J)=0$ in that subspace. Due to the parity structure, we
know that $J^{01}$ is off-diagonal in the sense that $J^{01}=\Pi_{0}J^{01}\Pi_{1}$.
Furthermore, since $\bb J^{01}|\St_{01}\kk=1$, $J^{01}$ has to overlap
with its corresponding steady-state coherence $\St_{01}=|0\ket\bra1|$.
With those two constraints and symmetry of $\L$ under $V=e^{i\phi\ph}$
(see Sec.~\ref{sec22}), $J^{01}$ must consist only of elements
from the sector $\{|2n\ket\bra2n+1|\}_{n=0}^{\infty}$. Assuming a
solution of the form $J^{01}=j(\ph)\Pi_{0}\aa$ and plugging into
$\L^{\dgt}(J^{01})=0$ yields a recursion relation for $j(\ph)$,
whose solution is eq.~(\ref{eq:j01}).

Physically, $J^{01}$ represents how the environment distinguishes
components of $\rin$. It partially preserves information only from
elements $\{|2n\ket\bra2n+1|\}_{n=0}^{\infty}$ since, in that case,
the same number of photon pairs is lost in relaxing to $|0\ket\bra1|$.
In all other even-odd basis sectors, e.g., $\{|2n\ket\bra2n-1|\}_{n=0}^{\infty}$,
different numbers of photon pairs are lost ($n$ vs. $n-1$ pairs
for the example). While $J^{01}$ tells us exactly which sector contributes
to the asymptotic state, it does not tell us how coherences decay
as photons are lost. To determine intermediate-time behavior, one
should consider the eigenvalues of the other eigenmatrices of $\L$
in each sector {[}see eq.~(\ref{eq:expansion}){]}. In fact, not
all information is preserved even in the $\{|2n\ket\bra2n+1|\}_{n=0}^{\infty}$
sector! The remaining left eigenmatrices $\{L^{01,m}\}_{m=0}^{\infty}$
in that sector are of the same form as $J^{01}$, namely $L^{01,m}=l^{01,m}\left(\ph\right)\Pi_{0}a$,
but with the double factorials generalized to operators whose nonzero
entries are
\begin{equation}
\bra2n|l^{01,m}\left(\ph\right)|2n\ket=\frac{1}{\sqrt{2m+1}}\prod_{k=m+1}^{n}\frac{k\left(2k-1\right)}{2\left(k^{2}-m^{2}\right)}
\end{equation}
and with the convention that the sum $\prod_{k=j}^{l}f\left(k\right)=1$
when $j=l+1$ and zero when $j>l+1$ \cite{zubairy1980}. Their corresponding
eigenvalues are $\l^{01,m}=-4m^{2}$ and one can see that $L^{01,0}=J^{01}$.
The smallest Fock state in the support of $l^{01,m}\left(\ph\right)$
is $|2m\ket$, implying that the coherences $|2m\ket\bra2m+1|$ decay
no faster than $e^{-4m^{2}t}$.

\subsection{$d$-photon case\label{sec35}}

Let us now generalize the $d=2$ case to all $d>0$ and consider the
jump operator $F=\aa^{d}$. Note that the $d=1$ case is simply single-photon
loss, which has the vacuum Fock state as its unique steady state.
For the general case, let
\begin{equation}
\Pi_{k}=\sum_{n=0}^{\infty}|dn+k\ket\bra dn+k|=\frac{1}{d}\sum_{l=0}^{d-1}e^{i\frac{2\pi}{d}(\ph-k)l}\label{eq:oscosc}
\end{equation}
be $d$ different projections with $k,l\in\{0,1,\cdots,d-1\}$. Noting
the cyclic relationship among projection operators,
\begin{equation}
\Pi_{k}\aa=\aa\Pi_{(k+1)\text{mod}\,{d}}=\Pi_{k}\aa\Pi_{(k+1)\text{mod}\,{d}},\label{eq:catrelproj}
\end{equation}
one can see that $[\Pi_{k},\aa^{d}]=0$. According to Sec.~\ref{sec12},
the Fock space is then partitioned into $d^{2}$ subspaces, each evolving
independently. We can thus write
\begin{equation}
\rout=\sum_{k,l=0}^{d-1}c_{kl}|k\ket\bra l|
\end{equation}
with $c_{kl}=\text{Tr}\{J_{kl}^{\dg}\rin\}$. Extending the recipe
of the $d=2$ case, there are $d^{2}$ conserved quantities 
\begin{equation}
J_{kl}=\frac{j_{kl}\left(\ph\right)}{\sqrt{(l)_{l-k}}}\Pi_{k}\aa^{l-k}\,,\label{eq:cctwophot}
\end{equation}
where the square-root is to satisfy the biorthogonality condition
(\ref{eq:norms}), $J_{lk}=J_{kl}^{\dg}$, 
\begin{equation}
\bra dn+k|j_{kl}\left(\ph\right)|dn+k\ket=\prod_{p=0}^{n-1}\frac{2\left(dp+l\right)_{l-k}}{\left(dp+l\right)_{l-k}+\left(dp+l+d\right)_{l-k}}
\end{equation}
(for all $n\in\mathbb{N}$) and zero elsewhere, and the falling factorial
$\left(x\right)_{n}=x\left(x-1\right)...\left(x-n+1\right)$. Since
$(x)_{0}=1$, the diagonal conserved quantities simplify to $J_{kk}=\Pi_{k}$.
Since $\sum_{k=0}^{d-1}J_{kk}=I$, only $d^{2}-1$ quantities are
independent. The off-diagonal quantity simplifies to eq.~(\ref{eq:j01})
for $d=2$ and only the identity remains for $d=1$. The $J_{kl}$
are reducible into a direct sum of $\mathfrak{u}(d)$ Lie algebras.
In other words, $\sum_{k,l=0}^{d-1}J_{kl}$ forms an infinite block-diagonal
matrix with blocks of length $d$, diagonal entries of 1, and off-diagonal
entries depending on $j_{kl}(\ph)$. 

\subsection{Steady state for an initial coherent state\label{subsec:Steady-state-for}}

As an example calculation, we determine $\rout$ when $\rin=|\b\ket\bra\b|$,
a coherent state $\aa|\b\ket=\b|\b\ket$ with $\b\in\mathbb{C}$.
Note that only the piece of $\rin$ that initially lives in a given
subspace, $\Pi_{k}\rin\Pi_{l}$, contributes to the corresponding
$c_{kl}$ in $\rout$. Since a coherent state fills the entire Fock
space, all subspaces evolve non-trivially and equilibrate to 

\begin{equation}
c_{kl}=\frac{\b^{\star l-k}e^{-\left|\b\right|^{2}}}{\sqrt{(l)_{l-k}}}\sum_{n=0}^{\infty}\frac{j_{kl}(dn+k)}{(dn+k)!}(\left|\b\right|^{2})^{dn+k}.\label{eq:sum}
\end{equation}
Since the factors $P_{n}=j_{kl}(dn+k)/(dn+k)!$ are polynomials in
$n$, $c_{kl}$ are generalized hypergeometric functions whose arguments
are be roots of $P_{n+1}/P_{n}$ \cite{zeilberger}. The diagonal
elements simplify if instead we express $\Pi_{k}$ using the right-hand
side of eq.~(\ref{eq:oscosc}), 
\begin{equation}
c_{kk}=\frac{1}{d}\sum_{l=0}^{d-1}e^{-i\frac{2\pi}{d}kl}\exp\left[\left|\b\right|^{2}\left(e^{i\frac{2\pi}{d}l}-1\right)\right]\,.\label{eq:diagqtys}
\end{equation}
In the large $|\b|$ limit, $c_{kk}\rightarrow1/d$, distributing
populations equally among the diagonal steady states. For $k\neq l$
in this limit, $c_{kl}$ converges to a constant times $e^{-i\theta(k-l)}$,
thus storing the phase $\theta\equiv\arg(\b)$ of the initial coherent
state for any $d$. Taking a look at specific cases, for $d=1$, eq.~(\ref{eq:sum})
is just $c_{00}=1$. For $d=2$, expressing in the $|k\ket\bra l|$
basis, 
\begin{equation}
\rout=\begin{pmatrix}\,\half(1+e^{-2\left|\b\right|^{2}})\,\,\,\,\,\,\, & \b^{\star}e^{-\left|\b\right|^{2}}I_{0}(\left|\b\right|^{2})\,\\
\text{c.c.} & \half(1-e^{-2\left|\b\right|^{2}})
\end{pmatrix},
\end{equation}
where $I_{0}$ is the modified Bessel function of the first kind \cite{dlmf}.
In the large $|\b|$ limit, $c_{01}\rightarrow e^{-i\theta}/\sqrt{2\pi}$.
We consider a generalized version of this case in Ch.~\ref{ch:7}.

\section{Ground state subspaces of frustration-free Hamiltonians\label{sec:Ground-state-subspaces}}

This final case is discussed primarily to present ``many-body''
Lindbladian examples and discuss some related work.\footnote{However, this is by no means an attempt to review the disconnected
literature.} The focus here is not on conserved quantities, but on a recipe for
jumps that guarantees stabilization of the ground state subspace of
any given frustration-free Hamiltonian. The presentation is standalone
for convenience, and these results can be reconstructed from, e.g.,
Thm.~1 of Ref.~\cite{Ticozzi2009} and Corr.~1 of Ref.~\cite{Ticozzi2012}.

These many-body examples feature a \textit{uniformly factorizable}
Hilbert space $\h=\h_{0}^{\otimes M}$ (a term borrowed from \cite{Pastawski2016}),
where $\h_{0}$ is the Hilbert space of some ``local'' site (such
as a spin) and $M$ is the number of such sites. A Hamiltonian $H$
\foreignlanguage{american}{is }\textit{frustration-free} \foreignlanguage{american}{if
all ground states of $H$ are also ground states of all individual
terms used to construct $H$}. The ground states of $H$ thus form
the $\ash=\ulbig$ of the Lindbladian constructed out of said jumps,
making this a \foreignlanguage{american}{\acs{DFS}} case. To make
things concrete, we state the following theorem and provide an expository
proof.
\selectlanguage{american}%
\begin{thm}
[Stabilizing frustration-free Hamiltonian ground states  \cite{Ticozzi2009,Verstraete2009,Ticozzi2012}]\label{thm:ff}Let
the Hilbert space \foreignlanguage{english}{$\h=\h_{0}^{\otimes M}$,
$H=\sum_{k}H_{k}$}, where $H_{k}$ is a Hamiltonian acting nontrivially
on a subset of the $M$ sites, and assume that $H$ is frustration-free.
Then, for each $H_{k}$, there exist jump operators $\{F^{k,\ell}\}_{\ell}$
such that $F^{k,\ell}$ act nontrivially on the same subset of sites
as $H_{k}$ and the ground states of $H$ form the \acs{DFS} ($\ulbig$)
of the Lindbladian constructed out of $\{F^{k,\ell}\}_{k,\ell}$.
Moreover, the jump operators satisfy
\begin{align}
F^{k,\ell} & =F_{\rmt}^{k,\ell}\\
\sum_{k,\ell}F_{\ur}^{k,\ell\dg}F_{\ur}^{k,\ell} & >0\,.
\end{align}
\end{thm}
\begin{proof}
The strategy is two-fold: first construct $\{F^{k,\ell}\}_{\ell}$
whose Lindbladian $\L_{k}$ stabilizes the ground states of $H_{k}$
and then show that the full Lindbladian $\L=\sum_{k}\L_{k}$ stabilizes
\textit{only} the ground states of $H$.

We can shift $H_{k}$ by a constant such that $H_{k}\geq0$ for all
$k$, meaning that all ground states of $H_{k}$ have eigenvalue zero.
\foreignlanguage{english}{For each $k$, define the projection $Q_{k}$
to be on the excited states (also, range or support; see footnote
\ref{fn:The-support-(kernel)} in Ch.~\ref{ch:1}) of $H_{k}$: $H_{k}=Q_{k}H_{k}Q_{k}$.
Define $P_{k}=I-Q_{k}$ to be the projection on the kernel of $H_{k}$,
i.e., $H_{k}P_{k}=P_{k}H_{k}P_{k}=0$. }The strategy is to \foreignlanguage{english}{create
jumps $F^{k,\ell}$ which take us from the range to the kernel. If
the range and kernel are the same dimension, we can do this by just
having one jump $F^{k}$ which is the \textit{isometry} ($F^{k\dg}F^{k}=Q_{k}$)
acting on states in $\ran Q_{k}$ and taking them to $\ker Q_{k}$.
If the range is bigger than the kernel, then we can have isometries
$F^{k,\ell}$ between distinct and non-overlapping subspaces of the
range and the kernel. In other words, if $\dim\ker Q_{k}=d$ and $\dim\ran Q_{k}=D$,
we can have $\left\lceil D/d\right\rceil $ jumps $\{F^{k,\ell}\}_{\ell=1}^{\left\lceil D/d\right\rceil }$
with $d$ ones on a diagonal in the upper right sector such that together
they form the matrix
\begin{equation}
\left(\begin{array}{c|cccc}
0 & F^{k,1} & F^{k,2} & \cdots & F^{k,\left\lceil D/d\right\rceil }\\
\hline  & 0\\
 &  & 0\\
 &  &  & \ddots\\
 &  &  &  & 0
\end{array}\right)\,,\label{eq:shape}
\end{equation}
where the upper left block is the $d$-dimensional $\ker Q_{k}$ and
lower right block is $D$-dimensional $\ran Q_{k}$. One can verify
by simple block matrix multiplication that the Hamiltonian formed
by our jumps, 
\begin{equation}
K^{k}=\half\sum_{\ell}F^{k,\ell\dg}F^{k,\ell}=\left(\begin{array}{c|c}
0 & 0\\
\hline 0 & \half\sum_{\ell}F^{k,\ell\dg}F^{k,\ell}>0
\end{array}\right)\,,\label{eq:cond2}
\end{equation}
is positive definite in the block corresponding to $\ran Q_{k}$.
Recall from Sec.~\ref{sec:What-is-a} that one can think of evolution
due to $\L_{k}$ (constructed out of $\{F^{k,\ell}\}_{\ell}$) as
being composed of a deterministic part, the anticommutator generated
by $-K^{k}$, and a recycling part, generated by applications of the
recycling term $F^{k,\ell}\cdot F^{k,\ell\dg}$. Due to the structure
of $F^{k,\ell}$, the recycling term always takes states out of $\ran Q_{k}$
and maps them into $\ker Q_{k}$. Due to $K^{k}>0$, the deterministic
part always decays anything in $\ran Q_{k}$. To show this, one can
consider the change in population on $\ran Q_{k}$,
\begin{align}
\tr\{Q_{k}\dot{\r}\} & =\tr\{Q_{k}\L_{k}(\r)\}=-2\tr\{K\r\}<0\,,
\end{align}
where we have used the structure of $F^{k,\ell}$ and the fact that
$K>0$ on $\ran Q_{k}$. A similar calculation, now using the four-corners
($\empbig$) decomposition, is now going to be done for the full Lindbladian
$\L=\sum_{k}\L_{k}$.}

\selectlanguage{english}%
Each dissipator $\L_{k}$ drives states to the ground states subspace
of its corresponding $H_{k}$ according to the above procedure. The
full generator $\L=\sum_{k}\L_{k}$ should then drive states into
$\ulbig$ \textemdash{} the intersection of the ground state spaces
of all $H_{k}$, i.e., the ground state subspace of $H$. To prove
this, we show that all states initially in $\lrbig$ decay to zero
in $\lrbig$. The change in population in $\lrbig$ is $\tr_{\lr}\left\{ \dot{\r}\right\} =\bb I|\R_{\lr}\L|\r\kk$,
where $\R_{\lr}$ is the superoperator projection onto $\lrbig$ .
We need to examine how the $F^{k,\ell}$'s decompose under the new
block structure. Since $\ulbig$ is their joint kernel, all jumps
annihilate states in $\ulbig$ ($F_{\ul}^{k,\ell}=0$) and no jump
can take states out ($F_{\ll}^{k,\ell}=0$). Therefore, $F^{k,\ell}=F_{\rmt}^{k,\ell}$.
Since there is no Hamiltonian, this implies that the jumps in $\L$
satisfy Thm.~\ref{prop:2}, meaning that $\L$ cannot take states
in $\thubig$ to $\lrbig$. Applying this yields
\begin{align}
\tr_{\lr}\left\{ \dot{\r}\right\}  & =\bb I|\R_{\lr}\L|\r\kk=\bb I|\R_{\lr}\L\R_{\lr}|\r_{\lr}\kk\equiv\bb I|\L_{\lr}|\r_{\lr}\kk\,.
\end{align}
The relevant piece $\L_{\lr}$ (\ref{eq:llr}) is
\begin{equation}
\L_{\lr}\left(\r_{\lr}\right)=\sum_{k,\ell}F_{\lr}^{k,\ell}\r_{\lr}F_{\lr}^{k,\ell\dg}-\sum_{k}\left\{ K_{\lr}^{k},\r_{\lr}\right\} \,.
\end{equation}
The Hamiltonian $K^{k}$ consists of two pieces,
\begin{equation}
K_{\lr}^{k}=\half\sum_{\ell}\left(F^{k,\ell\dg}F^{k,\ell}\right)_{\lr}=\half\sum_{\ell}F_{\lr}^{k,\ell\dg}F_{\lr}^{k,\ell}+F_{\ur}^{k,\ell\dg}F_{\ur}^{k,\ell}\,.
\end{equation}
The $F_{\lr}$ piece of $K^{k}$ conspires with the recycling term
$F_{\lr}^{k,\ell}\cdot F_{\lr}^{k,\ell\dg}$ and creates a bona fide
Lindbladian with jump operator $F_{\lr}$. Since Lindbladians are
trace-preserving, the $F_{\lr}$ pieces do not contribute to $\tr_{\lr}(\dot{\r})$.
The $F_{\ur}$ has no corresponding recycling term, so that part is
\textit{not} of Lindblad form. The $F_{\ur}$ parts instead give us
a potential decrease in trace:
\begin{align}
\tr_{\lr}\left\{ \dot{\r}\right\}  & =-\sum_{k,\ell}\tr\left\{ F_{\ur}^{k,\ell\dg}F_{\ur}^{k,\ell}\r_{\lr}\right\} \,.
\end{align}
Since $F_{\ur}^{k,\ell\dg}F_{\ur}^{k,\ell}\geq0$ and $\r_{\lr}\geq0$,
each term in the above sum is $\leq0$. We show that it is $<0$,
meaning that everything in $\lrbig$ decays. Since $\ulbig$ is the
intersection of the kernels of all $F^{k,\ell}$ and since we are
in $\lrbig$, there exists at least one $F^{j,m}$ for which $\tr\left\{ F^{j,m\dg}F^{j,m}\r_{\lr}\right\} >0$.
This provides a lower bound on the decay,
\begin{align}
\tr_{\lr}\left\{ \dot{\r}\right\}  & \leq-\tr\left\{ F^{j,m\dg}F^{j,m}\r_{\lr}\right\} <0\,,
\end{align}
ensuring that all states initially in $\lrbig$ decay into $\ulbig$.
Since the above is true for all $\r_{\lr}$, $\sum_{k,\ell}F_{\ur}^{k,\ell\dg}F_{\ur}^{k,\ell}>0$
on $\lrbig$ and all steady states of $\L$ are in $\ulbig$.
\end{proof}
\selectlanguage{english}%
The above procedure has been used in several specific cases throughout
the literature, for example in obtaining stabilizer quantum error-correcting
code states \cite{Kraus2008,Muller2011,Dengis2014}or ground states
of the AKLT model \cite{Kraus2008,Zhou2017}. However, the above recipe
for $F^{k,\ell}$ is not unique. Instead of having each jump be an
isometry from part of the range to the kernel of $H_{k}$, forming
the shape in eq.~(\ref{eq:shape}), one can instead have $F^{k}$
(for each $k$) act like a ladder operator. In other words, 
\begin{equation}
F^{k}=\left(\begin{array}{c|cccc}
0 & f^{k,1} & 0 & \cdots & 0\\
\hline  & 0 & f^{k,2} & 0 & \vdots\\
 &  & 0 & \ddots & 0\\
 &  &  & \ddots & f^{k,\left\lceil D/d\right\rceil }\\
 &  &  &  & 0
\end{array}\right)\,,
\end{equation}
where the upper left block is the $d$-dimensional $\ker Q_{k}$ and
lower right block is the $D$-dimensional $\ran Q_{k}$. To ensure
that $f^{k,p}$ transfers all states in the block below it into the
block to the left, we need to have $f^{k,p\dg}f^{k,p}>0$ when restricted
to the block below $f^{k,p}$. If done this way, only one jump per
each $k$ is sufficient. We note that the conserved quantities of
$\L=\sum_{k}\L_{k}$ are complicated since $\L_{\lr}\neq0$ for both
recipes, but they can nevertheless be determined by Thm.~\ref{prop:3}.
One nice example of these types of jumps was used to stabilize the
ground states of the Kitaev Majorana wire Hamiltonian \cite{Diehl2011}
(see also \cite{bardyn}). Conversely, there exists an algorithm \cite{Ticozzi2012a}
(see also \cite{ying2013}) which, given a subspace $\ulbig$, tries
to decompose a jump operator into a structure similar to the above
in order to check whether all states converge to $\ulbig$, failing
if the $f$ in the lower right corner is not positive definite.

Of course, other stabilization schemes exist besides those described
above {[}e.g., \cite{Verstraete2009}, eq.~(6){]}. A quite elegant
family of schemes is based on the idea that, given a quantum channel
$\E$, the Lindbladian
\begin{equation}
\L\equiv\E-\id\label{eq:stab}
\end{equation}
is one whose semigroup $e^{t\L}$ has the same fixed points as that
of $\E$ \cite{Wolf2008,Pastawski2011}. Therefore, given a cleverly
chosen $\E$ which stabilizes some desirable states, the asymptotic
projection $\ppp=\lim_{t\rightarrow\infty}e^{t\L}$ generated by the
above $\L$ will also stabilize those states. In another work, it
is shown that only one jump operator is required to stabilize any
state \cite{Ticozzi2010}. Note that Hamiltonian-based feedback control
can also be used to make sure that the states of interest are stabilized
(\cite{Ticozzi2009}, Thm.~2). Since all local gapped Hamiltonians
can be with approximated with ones that are frustration-free \cite{Hastings2006},
the above recipes allow for stabilization of states close to any phase
of matter that can be generated by such Hamiltonians. Extensions to
stabilization of mixed states using frustration-free Lindbladians
can be found in Ref.~\cite{Johnson2015}. 

The bad news regarding all of these preparation schemes is that, for
``exotic'' states such as 2D topological phases and assuming some
notion of locality for the jumps, the speed of convergence (i.e.,
inverse of the dissipative gap $\dgg$) \textit{increases} with the
length scale $L$ associated with the system size. For example, an
optimal toric code stabilizer \cite{Dengis2014} has gap $\dgg=O(1/L)$,
meaning that the system has arbitrarily small excitations above the
steady state in the thermodynamic limit.\selectlanguage{english}%

\inputencoding{latin9}\newpage{}\foreignlanguage{english}{}%
\begin{minipage}[t]{0.5\textwidth}%
\selectlanguage{english}%
\begin{flushleft}
\begin{singlespace}\textit{``I don\textquoteright t always integrate,
but when I do, I integrate by parts.''}\end{singlespace}
\par\end{flushleft}
\begin{flushleft}
\hfill{}\textendash{} Nicholas Read
\par\end{flushleft}\selectlanguage{english}%
\end{minipage}

\chapter{Time-dependent perturbation theory\label{ch:4}}

In this chapter, we apply the four-corners decomposition to the first-order
terms in ordinary time-dependent perturbation theory \cite{ABFJ}.
In Sec.~\ref{sec:Decomposing-the-Kubo}, we determine that the first-order
correction within \acs{ASH} is of Hamiltonian form and the energy
scale of the first-order term causing leakage out of \acs{ASH} is
governed by the dissipative gap of $\L_{\thu}$. We extend these conclusions
to jump operator perturbations $F^{\ell}\rightarrow F^{\ell}+f^{\ell}$.
In Sec.~\ref{sec:Exact-all-order-Dyson}, we determine the full Dyson
expansion to all orders exactly, given a perturbation which slowly
ramps up to a constant and an initial state that is a steady state
of the unperturbed $\L$. We conclude in Sec.~\ref{subsec:Relation-to-previous}
by making contact with previously studied topics: dark states, geometric
linear response, the Dyson series for the case of an unperturbed $\L$
with a unique steady state, quantum Zeno dynamics, and the effective
operator formalism.
\selectlanguage{english}%

\section{Decomposing the Kubo formula\label{sec:Decomposing-the-Kubo}}

Let us assume that time evolution is governed by a time-independent
Lindbladian $\L$ and that the system is perturbed as
\begin{equation}
\L\rightarrow\L+g(t)\oo\,,
\end{equation}
where the perturbation superoperator $\oo$ is multiplied by a slowly
ramping up time-dependent factor $g(t)$ from time $-\infty$ to a
time $t$. The Lindbladian-based Kubo formula \cite{Bernad2008,Bernad2010,mukamel,Wei2011,Jaksic2013,Shen2014,Ban2015,Venuti2015a,VillegasMartinez2016,Ban2017}
is derived analogously to the Hamiltonian formula, i.e., it is a leading-order
Dyson expansion of the full evolution.\footnote{We note that there exists an adiabatic derivation as well \cite{Chetrite2012},
which is not addressed here.} The main difference is that the derivation is performed in the superoperator
formalism. However, the superoperator formalism lends a natural interpretation
of the terms in the superoperator Dyson series. As a result, we use
the intuitiveness of the terms to justify the expansion, omitting
the quite standard technical modifications needed to obtain them.

The first term in such a series acts on a state $\r$ as \cite{VillegasMartinez2016}
\begin{equation}
\T_{t}^{\left(1\right)}|\r\kk=\intt 0td\tau g\left(\tau\right)e^{\left(t-\tau\right)\L}\oo e^{\tau\L}|\r\kk\,.\label{eq:dyson1}
\end{equation}
We remind the reader that we use vectorized notation for matrices
and the Hilbert-Schmidt inner product $\bb A|\r\left(t\right)\kk\equiv\tr\left\{ A^{\dg}\r\left(t\right)\right\} $
(see Ch.~\ref{app:Preliminaries}). This term offers an intuitive
interpretation if one thinks of the system as evolving from the right
side of the expression to the left. Reading the integrand from right
to left, the initial state $\r$ evolves under the unperturbed Lindbladian
$\L$ to time $\tau$, is perturbed by $\oo$, and then evolves under
$\L$ from $\tau$ to $t$. The integral represents a sum over all
possible acting times $\tau$ of the perturbation. Applying an observable
$\bb A|$ from the left is equivalent to evaluating said observable
at time $t$. If we now also make the assumption that we are in an
initially steady state $\r=\rout$, the right-most exponential $e^{\tau\L}$
is removed since $\L|\rout\kk=0$. Since we do not have any evolution
until the time of the perturbation with such an assumption, we can
extend the initial time from $0$ to $-\infty$. These manipulations
then produce the Kubo formula \cite{Kubo1957}
\begin{equation}
\bb A|\T_{t}^{\left(1\right)}|\rout\kk=\intt{-\infty}td\tau g\left(\tau\right)\bb A|e^{\left(t-\tau\right)\L}\oo|\rout\kk\,.\label{eq:kubo}
\end{equation}
We proceed to apply the four-corners decomposition to this formula.
However, before doing to, let us show that this is indeed the original
Kubo formula.

\paragraph{Hamiltonian case}

Let us set $\L=\H=-i\left[H,\cdot\right]$, $\oo=-i[V,\cdot]$ for
a Hamiltonian $V$, and massage eq.~(\ref{eq:kubo}) into standard
form. For that, define $O(t)\equiv e^{iHt}Oe^{-iHt}=e^{-t\H}\left(O\right)$
and recall that $\left[H,\rout\right]=0$ since $\rout$ is generically
a superposition of projections on eigenstates of $H$. We can then
commute $e^{iHt}$ with $\rout$ and cyclically permute under the
trace to obtain 
\begin{equation}
\bb A|\T_{t}^{\left(1\right)}|\rout\kk=\frac{1}{i}\intt{-\infty}t\dd\tau g\left(\tau\right)\tr\left\{ \left[A\left(t-\tau\right),\hpert\right]\rout\right\} \,,
\end{equation}
recovering the usual time-ordered commutator expression.

The perturbations considered here are Hamiltonian and jump operator
perturbations of $\L$ (\ref{eq:def-1}), respectively\begin{subequations}
\begin{eqnarray}
H & \rightarrow & H+g\left(t\right)\hpert\\
F^{\ell} & \rightarrow & F^{\ell}+g\left(t\right)f^{\ell}\label{eq:pertjump}
\end{eqnarray}
\end{subequations}{[}for $\hpert,f^{\ell}\in\oph$ and $\hpert^{\dg}=\hpert${]}.
It will be shown that both generate unitary evolution within all \acs{ASH}
and leakage caused by both does not take states into $\lrbig$. We
first handle the Hamiltonian case first for simplicity, 
\begin{equation}
\oo=-i\left[\hpert,\cdot\right]\equiv\spert\,,\label{eq:pert}
\end{equation}
returning to the jump case in Sec.~\ref{subsec:linds}. 

We now use four-corners projections $\R_{\emp}$ to partition eq.~(\ref{eq:kubo}).
Due to the no-leak property (\ref{eq:no-leak}), we have $\R_{\lr}\spert\R_{\ul}=0$.
Remembering that the Lindbladian is block upper-triangular in the
four-corners partition {[}see eq.~(\ref{eq:gen}){]}, it follows
that $e^{t\L}$ is also block upper-triangular. We do not make any
assumptions on the observable: $A=A_{\ul}+A_{\of}+A_{\lr}$. Further
decomposing the first term using the asymptotic projection $\ppp$
from eq.~(\ref{eq:proj}) and its complement $\qqq\equiv\id-\ppp$
yields\stepcounter{equation}
\begin{align}
\bb A|\T_{t}^{\left(1\right)}|\rout\kk&=\intt{-\infty}t d\tau g\left(\tau\right)\bb A_{\ul}|e^{\left(t-\tau\right)\sout}\ppp\spert\ppp|\rout\kk   \tag{{\theequation\textbf{A}}}\label{eq:ka}\\
&+\intt{-\infty}t d\tau g\left(\tau\right)\bb A_{\ul}|e^{\left(t-\tau\right)\L}\qqq\R_{\thu}\spert|\rout\kk    \tag{{\theequation\textbf{B}}}\label{eq:kb} \\
&+\intt{-\infty}t d\tau g\left(\tau\right)\bb A_{\of}|e^{\left(t-\tau\right)\L}\R_{\of}\spert|\rout\kk   \tag{{\theequation\textbf{C}}}\label{eq:kc} \,.
\end{align}The terms differ by which parts of $\spert$ perturb $\rout$ and
also which parts of $A$ ``capture'' the evolved result. The three
relevant parts of $A$ correspond to the three labels in Fig.~\ref{fig:lres}.
One can readily see that $A_{\lr}$ is irrelevant to this order due
to (\ref{eq:no-leak}). The term (\ref{eq:ka}) consists of perturbing
and evolving \textit{within} the asymptotic subspace \textbf{A}, shaded
gray in the Figure. The effect of the perturbation within \acs{ASH}
is $\ppp\spert\ppp$ (shown in Sec.~\ref{subsec:hams} to be of Hamiltonian
form), and $\sout$ is the part of the unperturbed $\L$ that generates
unitary evolution within \acs{ASH}. The term (\ref{eq:ka}) therefore
most closely resembles the traditional Hamiltonian-based Kubo formula.
The remaining two terms quantify \textit{leakage out of} \acs{ASH}
and contain non-Hamiltonian contributions. The term (\ref{eq:kb})
consists of perturbing into regions \textbf{B} and \textbf{C} in Fig.~\ref{fig:lres},
but then evolving under $\R_{\ul}e^{t\L}\R_{\thu}$ strictly into
region \textbf{B} (since $\ppp e^{t\L}\qqq=0$). The term (\ref{eq:kc})
consists of perturbing into region \textbf{C} and remaining there
after evolution due to $\R_{\of}e^{t\L}\R_{\of}$. This term is eliminated
if $A_{\of}=0$, i.e., if the observable is strictly in $\ulbig$.

\begin{figure}
\begin{centering}
\includegraphics[width=0.35\columnwidth]{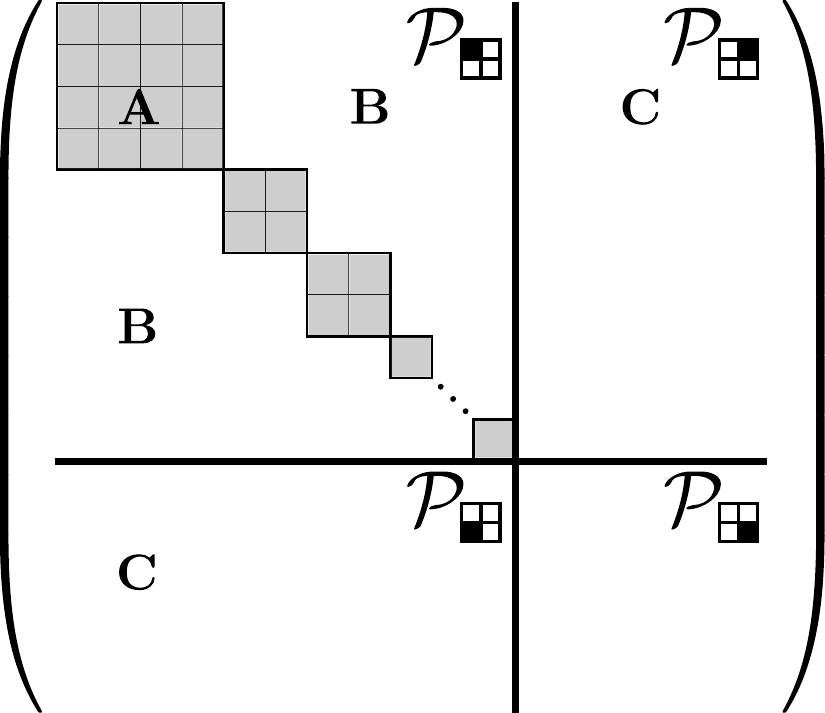}
\par\end{centering}
\caption{\label{fig:lres} Sketch of regions of linear response of the asymptotic
subspace \(\ash\) (gray) to a Hamiltonian perturbation. Each of three
regions \textbf{A}, \textbf{B}, and \textbf{C} corresponds to the
respective response term (\ref{eq:ka}), (\ref{eq:kb}), and (\ref{eq:kc})
in the text.}
\end{figure}

\paragraph{DFS case}

Recall that in this case $\ulbig$ is a \foreignlanguage{american}{\acs{DFS}}
($\ppp\R_{\ul}=\R_{\ul}$), and we do not assume it is stationary
($\hout\neq0$). From eq.~(\ref{eq:gen-1}), we can see that $\L$
cannot take any coherences in $\ofbig$ back into the \foreignlanguage{american}{\acs{DFS}}
($\R_{\ul}\L\R_{\of}=0$). Therefore, the interference term (\ref{eq:kb})
is eliminated and the response formula reduces to\begin{subequations}
\label{eq:kdfsmain}
\begin{align}
\bb A|\T_{t}^{\left(1\right)}|\rout\kk&=\intt{-\infty}t\dd\tau g\left(\tau\right)\bb A_{\ul}|e^{\left(t-\tau\right)\sout}\ppp\spert\ppp|\rout\kk \label{eq:kdfsa}\tag{{\theequation\textbf{A}}}\\
&+\intt{-\infty}t\dd\tau g\left(\tau\right)\bb A_{\of}|e^{\left(t-\tau\right)\L}\R_{\of}\spert|\rout\kk. \label{eq:kdfsb}\tag{{\theequation\textbf{C}}} 
\end{align}
\end{subequations}If furthermore $A_{\of}=0$, there are no interference terms coming
from outside of the \foreignlanguage{american}{\acs{DFS}} and the
Lindbladian linear response reduces to the purely Hamiltonian-based
term (\ref{eq:kdfsa}). Such a simplification can also be achieved
when $\hpert_{\of}=0$, which implies that the Hamiltonian does not
take $\rout$ out of the \foreignlanguage{american}{\acs{DFS}} to
begin with ($\R_{\of}\spert\R_{\ul}=0$).

For the rest of this chapter, we set $\hout=0$ and have the time-dependent
part $g\left(t\right)$ of our perturbation ramp up to a constant
at $t=0$:
\begin{equation}
g\left(t\right)=\lim_{\eta\rightarrow0}e^{\eta t\Theta\left(-t\right)}=\begin{cases}
\lim_{\eta\rightarrow0}e^{\eta t} & t<0\\
1 & t\geq0
\end{cases}\,,\label{eq:ramp}
\end{equation}
where $\Theta(t)$ is the Heaviside step function. As a result, the
integrals in the Kubo formula can be performed exactly (see Sec.~\ref{subsec:Proof}),
simplifying the formula to\footnote{\label{fn:ringdown}We also note that one can have a sudden perturbation
$g\left(t\right)=\Theta\left(t\right)$ \cite{Zanardi2014}. In that
case, one does not obtain the $\nicefrac{1}{\eta}$ term, but the
leakage term now contains a ``ringdown'' contribution due to the
sudden onset of the perturbation: $\L^{-1}\spert\rightarrow(e^{t\L}-1)\L^{-1}\spert$.
None of the results in Sec.~\ref{sec:Decomposing-the-Kubo} depend
on which $g\left(t\right)$ one picks, but the all-order Dyson series
in Thm.~\ref{thm:Dyson} relies on using a slowly ramping up perturbation
in order to avoid the ringdown terms.}
\begin{equation}
\bb A|\T_{t}^{\left(1\right)}|\rout\kk=\left(t+\frac{1}{\eta}\right)\bb A|\ppp\spert\ppp|\rout\kk-\bb A|\L^{-1}\spert|\rout\kk\,.\label{eq:kubo-simplified}
\end{equation}
This version, which is true for more general perturbations in Lindblad
form ($\spert\rightarrow\oo$), invites an analogy with degenerate
perturbation theory, in which $\ppp\spert\ppp$ is the superoperator
analogue of the perturbation projected onto the subspace of interest
and $\L^{-1}\spert$ is an analogue of the term governing corrections
to the wavefunction and including the famous ``energy denominator''.
The $\nicefrac{1}{\eta}$ factor, an ``infinity'', is the (unfortunate)
consequence of the perturbation acting on the steady-state subspace
for an infinite amount of time during the time interval $(-\infty,0]$
and in the $\eta\rightarrow0$ limit.\footnote{Note that, when evaluating the response in frequency space, a more
careful treatment of $\eta$ may be necessary \cite{Bradlyn2012}.} While this choice of $g\left(t\right)$ creates this uncomfortable,
but explainable, infinity within \acs{ASH}, it allows us to write
the leakage term strictly in terms of
\begin{align}
\L^{-1} & =\qqq\L^{-1}\qqq\equiv-\intt{-\infty}td\tau g\left(\tau\right)e^{\left(t-\tau\right)\L}\qqq=-\intt 0{\infty}d\tau e^{\tau\L}\qqq\,.\label{eq:inver}
\end{align}
This pseudo-inverse ($\L^{-1}\L=\L\L^{-1}=\qqq$) is also the inverse
of all invertible parts in the Jordan normal form of $\L$ (\cite{Zanardi2014},
Appx.~D). In the context of finite matrices, it is called the Drazin
pseudoinverse {[}\cite{grevillebook}, eq.~(41) for $k=1${]}. In
the Hamiltonian context, this is the familiar Green's function in
\acs{OPH} (i.e., Liouville space \cite{mukamel}). In the context
of linear operators, this is simply the resolvent of $\L$ at $z=0$,
i.e.,
\begin{equation}
\L^{-1}=-\frac{1}{2\pi i}\ointop_{\G}\frac{dz}{z}\left(\L-z\right)^{-1}\,,
\end{equation}
where $\Gamma$ is the contour which encircles zero and no other points
in the spectrum of $\L$ {[}see \cite{Avron2012a}, eq.~(70) or \cite{katobook},
Ch.~3, eq.~(6.23){]}. Note that $\L^{-1}$ is \textit{not} the Moore-Penrose
pseudoinverse; while $\L^{-1}$ inverts the Jordan normal form of
$\L$, the Moore-Pensore inverse inverts the diagonal matrix in the
singular-value decomposition of $\L$. While $\L^{-1}$ appears naturally
in the above formulation, the Moore-Penrose inverse can be used to
study time-independent Lindbladian perturbation theory \cite{Li2014}.
Since both pseudoinverses are basically identical for diagonalizable
$\L$, we anticipate that differences between the formalisms (if any)
should arise only in those parts of $\L$ which are not diagonalizable.

In the next Subsections, we use the no-leak and clean-leak properties
to determine that evolution within \acs{ASH} is of Hamiltonian form
and to quantify the leakage scale of the second term in eq.~(\ref{eq:kubo-simplified}).

\subsection{Evolution within $\textnormal{As(\ensuremath{\mathsf{H}})}$\label{subsec:hams}}

Let us focus on the term $\ppp\spert\ppp$ from eq.~(\ref{eq:kubo-simplified}),
which quantifies the effect of the perturbation \textit{within} \acs{ASH}.
A swift application of the no-leak and clean-leak properties (\ref{eq:no-leak}-\ref{eq:clean-leak})
allows us to substitute $\ps\equiv\ppp\R_{\ul}$ for $\ppp$. Recalling
that $\R_{\lr}\spert\R_{\ul}=0$ and that $\ppp\spert\ppp$ is strictly
acting on states $\rout\in\ash$ yields
\begin{equation}
\ppp\spert\ppp=\ppp\R_{\ul}\spert\ppp=\ps\spert\ps\,.
\end{equation}
It turns out that this first-order effect of the perturbation within
\acs{ASH} is always of Hamiltonian form, for some effective Hamiltonian
that we determine now.

\paragraph{DFS case}

Here, we can immediately read off the effective Hamiltonian. Since
$\ps=\R_{\ul}$ for the \foreignlanguage{american}{\acs{DFS}} case,
\begin{equation}
\ppp\spert\ppp=-i\left[V_{\ul},\cdot\right]
\end{equation}
with $V_{\ul}$ the perturbation projected onto the \foreignlanguage{american}{\acs{DFS}}. 

\paragraph{NS case}

In this case, we have to use the formula for $\ps$ from eq.~(\ref{eq:nsproj-1}),
re-stated below:
\begin{equation}
\ps=\idfs\ot|\as\kk\bb\ai|\,,\label{eq:nsproj3}
\end{equation}
with $\idfs(\cdot)=\iidfs\cdot\iidfs$ being the superoperator projection
on the \foreignlanguage{american}{\acs{DFS}} part, $\ai$ being the
operator projection on the auxiliary part, and $\pp=\iidfs\ot\ai$.
Direct multiplication yields 
\begin{equation}
\ppp\spert\ppp=\ps\spert\ps=\bb\ai|\spert|\as\kk\ot|\as\kk\bb\ai|\,,
\end{equation}
where the evolution within the auxiliary part is trivial and evolution
within the \foreignlanguage{american}{\acs{DFS}} part is generated
by the effective \foreignlanguage{american}{\acs{DFS}} Hamiltonian
$W$:
\begin{equation}
\bb\ai|\spert|\as\kk=-i\left[\tr_{\textsf{ax}}\{\as\hpert_{\ul}\},\cdot\right]\equiv-i\left[W,\cdot\right]\,.\label{eq:effham-ham}
\end{equation}
To better reveal the effect of $\as$, it is worthwhile to express
$\hpert_{\ul}$ as a sum of tensor products of various \foreignlanguage{american}{\acs{DFS}}
and auxiliary Hamiltonians: $V_{\ul}=\sum_{\iota}\hpertdfs^{\iota}\ot\apert^{\iota}$.
The effective Hamiltonian then becomes 
\begin{equation}
W=\sum_{\iota}\tr_{\textsf{ax}}\{\as\apert^{\iota}\}\hpertdfs^{\iota}\,.\label{eq:unit}
\end{equation}
In words, $\ppp\spert\ppp$ is a linear combination of Hamiltonian
perturbations $\hpertdfs^{\iota}$ on the \foreignlanguage{american}{\acs{DFS}},
with each perturbation weighted by the expectation value of the corresponding
auxiliary operator $\apert^{\iota}$ in the state $\as$. 

\subsection{Leakage out of $\textnormal{As(\ensuremath{\mathsf{H}})}$\label{subsec:Leakage-out-of}}

Now we can apply the clean-leak property (\ref{eq:clean-leak}) to
narrow down those eigenvalues of $\L$ which are relevant in characterizing
the scale of the leakage term $\L^{-1}\spert$ from the simplified
Kubo formula (\ref{eq:kubo-simplified}). By definition (\ref{eq:inver}),
$\L^{-1}$ has the same block upper-triangular structure as $\L$
from eq.~(\ref{eq:gen}). This fact conspires with $\R_{\lr}\spert\R_{\ul}=0$
to allow us to ignore $\L_{\lr}$ and write
\begin{equation}
\L^{-1}\spert|\rout\kk=\L_{\thu}^{-1}\spert|\rout\kk\,.\label{eq:leakagekubo}
\end{equation}
Therefore, the relevant gap is the nonzero eigenvalue of $\L_{\thu}$
with the smallest absolute value. However, we now show how the spectrum
of $\L_{\thu}$ is actually contained in the spectrum of $\L_{\ul}+\L_{\ur}$.
Recalling the block upper-triangular structure of $\L$ from eq.~(\ref{eq:gen}),
one can establish that its eigenvalues must consist of eigenvalues
of $\L_{\ul}$, $\L_{\of}$, and $\L_{\lr}$. However, evolution of
the two coherence blocks is decoupled, $\L_{\of}=\L_{\ur}+\L_{\ll}$
(see Sec.~\ref{app:decomp}), and eigenvalues of $\L_{\of}$ come
in pairs. Therefore, one can then define the \textit{effective dissipative
gap} $\adg$ to be the nonzero eigenvalue of $\L_{\ul}+\L_{\ur}$
with the smallest absolute value. If we want leakage to be suppressed,
we want $\adg$ to be as \textit{large} as possible.

\paragraph{DFS case}

Assume that we have a \foreignlanguage{american}{\acs{DFS}} case:
all of $\ulbig$ evolves unitarily, so $\L_{\ul}=\sout$ does not
have a dissipative gap. In that case, we can omit $\ulbig$ from eq.~(\ref{eq:leakagekubo})
and simplify it to
\begin{equation}
\L^{-1}\spert|\rout\kk=\L_{\of}^{-1}\spert|\rout\kk\,.\label{eq:simplified-dfs-leakage}
\end{equation}
Therefore, the effective dissipative gap $\adg$ is just the dissipative
gap of $\L_{\ur}$.

\subsection{Jump operator perturbations\label{subsec:linds}}

Having covered Hamiltonian perturbations, let us return to jump operator
perturbations of the Lindbladian (\ref{eq:def-1}). Recall from eq.~(\ref{eq:pertjump})
that
\begin{equation}
F\rightarrow F+g\left(t\right)f
\end{equation}
with $g\left(t\right)$ a ramping function and $f\in\oph$, not necessarily
Hermitian. It was first shown in Ref.~\cite{Zanardi2015} that such
perturbations actually induce unitary evolution on \foreignlanguage{american}{\acs{NS}}
blocks of those Lindbladians which do not possess a nontrivial decaying
space ($\pp=I$). Here we extend this interesting result to cases
where $P\neq I$, thereby covering all $\L$. Namely, just like Hamiltonian
perturbations $\spert$, jump operator perturbations induce unitary
evolution within \acs{ASH} and the leakage scale associated with
them is still $\adg$.

Returning to eq.~(\ref{eq:pert}), the action of the perturbation
to first order in $g$ is 
\begin{equation}
\oo(\r)\equiv\Y(\r)=\lind\left(F\r f^{\dg}+H.c.-{\textstyle \half}\left\{ f^{\dg}F+F^{\dg}f,\r\right\} \right)\,,\label{eq:problem}
\end{equation}
where $\lind$ is the rate corresponding to the jump operator $F$
(we ignore the index $\ell$ for clarity). We hope to invoke the clean-leak
property (\ref{eq:clean-leak}) once again, but the first term on
the right-hand side of the above acts simultaneously and non-trivially
on \textit{both} sides of $\r$. There is thus a possibility that
one can reach $\lrbig$ when acting with $\Y$ on a steady state.
However, the condition $F_{\ll}=0$ from Thm.~\ref{prop:2} implies
that $\R_{\lr}(F\r f^{\dg})$ is zero for all $f$, so one can \textit{still}
substitute $\ps$ for $\ppp$:
\begin{equation}
\ppp\Y\ppp|\rout\kk=\ps\Y\ps|\rout\kk\,.
\end{equation}
Furthermore, the fact that $\R_{\lr}\Y\R_{\ul}=0$ allows us to ignore
$\lrbig$ in determining the leakage energy scale associated with
these jump operator perturbations. We finish with calculating the
corresponding effective Hamiltonian for the most general cases.

\paragraph{NS case}

Having eliminated the influence of the decaying subspace $\lrbig$,
we can now repeat the calculation done for Hamiltonian perturbations
using the \foreignlanguage{american}{\acs{NS}} projection (\ref{eq:nsproj3}),
yielding 
\begin{equation}
\ps\Y\ps=\bb\ai|\Y|\as\kk\ot|\as\kk\bb\ai|\,.
\end{equation}
After some algebra, the \foreignlanguage{american}{\acs{DFS}} part
reduces to Hamiltonian form \cite{Zanardi2015}: $\bb\ai|\Y|\as\kk=-i[Y,\cdot]$
where
\begin{equation}
Y\equiv\frac{i}{2}\lind\tr_{\textsf{ax}}\left\{ \as\left(F_{\ul}^{\dg}f_{\ul}-f_{\ul}^{\dg}F_{\ul}\right)\right\} \,.\label{eq:effham}
\end{equation}

\paragraph{Multi-block case}

We now sketch the calculation of both Hamiltonian and jump operator
perturbations, $\oo=\spert+\Y$, for the most general case of $\ulbig$
housing multiple \foreignlanguage{american}{\acs{NS}} blocks. Once
again, we can get rid of the decaying subspace and substitute $\ps$
for $\ppp$. In addition, since $\ps$ does not have any presence
except within the (gray) \foreignlanguage{american}{\acs{NS}} blocks
of $\ulbig$ {[}see Fig.~\ref{fig:decomp}{]}, $\ps$ does not project
onto any coherences between the \foreignlanguage{american}{\acs{NS}}
blocks. The contributing part of $\ppp\oo\ppp$ thus consists of the
Hamiltonian and jump operator perturbations projected to each \foreignlanguage{american}{\acs{NS}}
block. Combining the effective Hamiltonians arising from $\spert$
and $\Y$ {[}respectively eqs.~(\ref{eq:effham-ham}) and (\ref{eq:effham}){]},
the effective evolution within the \foreignlanguage{american}{\acs{DFS}}
part of each \foreignlanguage{american}{\acs{NS}} block (indexed
by $\varkappa$) is generated by the Hamiltonian
\begin{equation}
X^{(\varkappa)}\equiv\tr_{\textsf{ax}}^{(\varkappa)}\left\{ \as^{(\varkappa)}\left(\hpert_{\ul}+\frac{i}{2}\lind(F_{\ul}^{\dg}f_{\ul}-f_{\ul}^{\dg}F_{\ul})\right)\right\} \,.
\end{equation}
The unprojected Hamiltonian $X\equiv\hpert+\frac{i}{2}\lind\left(F^{\dg}f-f^{\dg}F\right)$
is exactly the operator resulting from joint variation of the Hamiltonian
and jump operators of $\L$ (\cite{Avron2012a}, Thm.~5).
\selectlanguage{american}%

\section{Exact all-order Dyson expansion\label{sec:Exact-all-order-Dyson}}

\selectlanguage{english}%
Let us now return to general Lindbladian perturbations $\oo$ and
study how higher-order terms in the Dyson series are also naturally
interpreted from right to left. For example, the second-order term
acts on a states in \acs{ASH} as \cite{VillegasMartinez2016}
\begin{equation}
\T_{t}^{\left(2\right)}\ppp\equiv\intop_{-\infty}^{t}g(\tau_{2})d\tau_{2}e^{\left(t-\tau_{2}\right)\L}\oo\intop_{-\infty}^{\tau_{2}}g(\tau_{1})d\tau_{1}e^{\left(\tau_{2}-\tau_{1}\right)\L}\oo\ppp\,,\label{eq:dyson2}
\end{equation}
and one can see that it is a sum over all possible pairs of times
$\tau_{1}\leq\tau_{2}$ at which the perturbation can be applied.
The $\ppp$ on the left means that the initial state is necessarily
in \acs{ASH}, and we study the full Dyson series with this restriction
from now on. The full time-ordered ($\mathbb{T}$) evolution operator
is expanded as
\begin{equation}
\mathbb{T}e^{\int_{\tau=-\infty}^{t}dt\left(\L+g\left(\tau\right)\oo\right)}\ppp=\sum_{N=0}^{\infty}\T_{t}^{\left(N\right)}\ppp=\ppp+\T_{t}^{\left(1\right)}\ppp+\T_{t}^{\left(2\right)}\ppp+\cdots\,,\label{eq:dyson}
\end{equation}
where we have already seen the first two terms $\T^{\left(1\right)}$
(\ref{eq:dyson1}) and $\T^{\left(2\right)}$ (\ref{eq:dyson2}).
The $N$th order term $\T_{t}^{\left(N\right)}\ppp$ consists of $N$
applications of the perturbation $g\left(\tau\right)\oo$ at times
$\tau_{1}\leq\tau_{2}\leq\cdots\leq\tau_{N}$ with evolution generated
by the unperturbed term $\L$ between those times. Let us define the
operator which acts with the perturbation $g\left(\tau\right)\oo$
at time $\tau_{n}$, evolves with the unperturbed $\L$ from $\tau_{n}$
to $\tau_{m}$, and sums up over all possible $\tau_{n}$:
\begin{equation}
\S_{m,n}\equiv\intop_{-\infty}^{\tau_{m}}g\left(\tau_{n}\right)d\tau_{n}e^{\left(\tau_{m}-\tau_{n}\right)\L}\oo\,.
\end{equation}
Also, let $\tau_{t}\equiv t$ so that $\S_{t,n}$ is an integral over
$\tau_{n}$ from $-\infty$ to $t$. Then, the $N$th order term can
be expressed as a convolution of $\S$'s,
\begin{equation}
\T_{t}^{\left(N\right)}=\S_{t,N-1}\S_{N-1,N-2}\cdots\S_{3,2}\S_{2,1}\,.
\end{equation}
By convolution, we mean that $\S_{n,n-1}$ is a function of the variable
$\tau_{n}$ which is integrated out by $\S_{n+1,n}$. This way, the
$\T$'s can be defined recursively:
\begin{equation}
\T_{t}^{\left(N+1\right)}=\S_{t,N+1}\T_{N+1}^{\left(N\right)}\,.
\end{equation}

We return to the case of $g\left(\tau\right)$ slowly ramping up to
a constant, as in eq.~(\ref{eq:ramp}), and continue analyzing terms
for $N>1$. While the $\nicefrac{1}{\eta}$ infinity is ever-present
in the entire expansion, eq.~(\ref{eq:ramp}) allows us to compute
all of the integrals in the series (\ref{eq:dyson}) \textit{exactly}.
We state the result first and prove it in the next Subsection.
\begin{thm}
[Exact all-order Dyson expansion]\label{thm:Dyson}The $N$th order
term in the Dyson series (\ref{eq:dyson}), given a slowly ramping
up perturbation (\ref{eq:ramp}) and an initial state in $\ash$,
is
\begin{align}
{\cal T}_{t}^{\left(N\right)} & =\left(1+\frac{\p_{t}}{\eta}\right)\sum_{M=0}^{N}\frac{t^{N-M}}{\left(N-M\right)!}\sum_{\l\in\cat_{M}^{N}}\prod_{n=1}^{N}{\cal X}(\l_{n})\,,\label{eq:dyson-1}
\end{align}
where $\l=\left(\l_{1},\l_{2},\cdots,\l_{N}\right)$ is a sequence
of $N$ nonnegative integers, $\cat_{M}^{N}$ is a particular set
of such sequences,
\begin{equation}
\cat_{M}^{N}=\left\{ \l\,\text{such that}\,\sum_{n=1}^{l}\l_{N+1-n}\begin{cases}
\leq\min\left\{ l,M\right\}  & 1\leq l<N\\
=M & l=N
\end{cases}\right\} \,,\label{eq:rules}
\end{equation}
and the operator that is put into the product for each element $\l_{n}$
is
\begin{equation}
{\cal X}(\l_{n})=\begin{cases}
-\L^{-\l_{n}}\oo & \l_{n}>0\\
\ppp\oo & \l_{n}=0
\end{cases}\,.\label{eq:rules2}
\end{equation}
\end{thm}
The pseudoinverse $\L^{-1}$ is defined in eq.~(\ref{eq:inver})
and it is implied that $\L^{-1}$ and its powers act on the range
of $\L$, i.e., $\L^{-1}=\L^{-1}\qqq$. The sequences and their corresponding
terms are presented in Tab.~\ref{tab:list} up to $N=4$. Recall
that \((1+\frac{1}{\eta}\p_{t})t^{N-M}/\left(N-M\right)!\) prepends
and \(\ppp\) appends each term.

The sequences readily lend themselves to a diagrammatic interpretation.
Consider all paths on a two-dimensional grid from a point $\left(N,N\right)$
down to the point $\left(0,M\right)$ where one can only perform only
the following types of steps: one step to the left ($\oo$), no steps
at all ($\ppp$), or $\l_{n}$ steps down ($-\L^{-\l_{n}}$ for $\l_{n}>0$).
Using these rules for the steps allows one to associate each sequence
$\l\in\cat_{M}^{N}$ (\ref{eq:rules}) with its own path from $\left(N,N\right)$
to $\left(0,M\right)$. It should be clear that each path contains
exactly $N$ steps to the left (i.e., $N$ instances of $\oo$) and
$M$ steps down ($\sum_{n}\l_{n}=M$). Such a diagrammatic interpretation,
along with its close connection to well-established perturbative methods
\cite{Bloch1958,Klein1974,Brouder2010}, will be considered in a future
publication \cite{pert}.

\begin{sidewaystable}
\begin{tabular}{cccrrcccrrcccr}
\cline{1-4} \cline{6-9} \cline{11-14} 
$N$ & $M$ & $\cat_{M}^{N}$ & $\phantom{{\displaystyle ^{^{1}}}}$Term & \multirow{22}{*}{~} & $N$ & $M$ & $\cat_{M}^{N}$ & $\phantom{{\displaystyle ^{^{1}}}}$Term & \multirow{22}{*}{~} & $N$ & $M$ & $\cat_{M}^{N}$ & $\phantom{{\displaystyle ^{^{1}}}}$Term\tabularnewline
\cline{1-4} \cline{6-9} \cline{11-14} 
\multirow{2}{*}{1} & 0 & 0 & $\phantom{{\displaystyle ^{^{1}}}}\ppp\oo$ &  & \multirow{21}{*}{4} & 0 & 0000 & $\ppp\oo\ppp\oo\ppp\oo\ppp\oo$ &  & \multirow{21}{*}{4} & \multirow{7}{*}{3} & 1020 & $\L^{-1}\oo\ppp\oo\L^{-2}\oo\ppp\oo$\tabularnewline
\cline{2-4} \cline{7-9} 
 & 1 & 1 & $\phantom{{\displaystyle ^{^{1}}}}-\L^{-1}\oo$ &  &  & \multirow{4}{*}{1} & 0001 & $-\ppp\oo\ppp\oo\ppp\oo\L^{-1}\oo$ &  &  &  & 1200 & $\L^{-1}\oo\L^{-2}\oo\ppp\oo\ppp\oo$\tabularnewline
\cline{1-4} 
\multirow{5}{*}{2} & 0 & 00 & $\phantom{{\displaystyle ^{^{1}}}}\ppp\oo\ppp\oo$ &  &  &  & 0010 & $-\ppp\oo\ppp\oo\L^{-1}\oo\ppp\oo$ &  &  &  & 2001 & $\L^{-2}\oo\ppp\oo\ppp\oo\L^{-1}\oo$\tabularnewline
\cline{2-4} 
 & \multirow{2}{*}{1} & 01 & $\phantom{{\displaystyle ^{^{1}}}}-\ppp\oo\L^{-1}\oo$ &  &  &  & 0100 & $-\ppp\oo\L^{-1}\oo\ppp\oo\ppp\oo$ &  &  &  & 2010 & $\L^{-2}\oo\ppp\oo\L^{-1}\oo\ppp\oo$\tabularnewline
 &  & 10 & $\phantom{{\displaystyle ^{^{1}}}}-\L^{-1}\oo\ppp\oo$ &  &  &  & 1000 & $-\L^{-1}\oo\ppp\oo\ppp\oo\ppp\oo$ &  &  &  & 2100 & $\L^{-2}\oo\L^{-1}\oo\ppp\oo\ppp\oo$\tabularnewline
\cline{2-4} \cline{7-9} 
 & \multirow{2}{*}{2} & 11 & $\phantom{{\displaystyle ^{^{1}}}}\L^{-1}\oo\L^{-1}\oo$ &  &  & \multirow{9}{*}{2} & 0011 & $\ppp\oo\ppp\oo\L^{-1}\oo\L^{-1}\oo$ &  &  &  & 0300 & $-\ppp\oo\L^{-3}\oo\ppp\oo\ppp\oo$\tabularnewline
 &  & 20 & $\phantom{{\displaystyle ^{^{1}}}}-\L^{-2}\oo\ppp\oo$ &  &  &  & 0101 & $\ppp\oo\L^{-1}\oo\ppp\oo\L^{-1}\oo$ &  &  &  & 3000 & $-\L^{-3}\oo\ppp\oo\ppp\oo\ppp\oo$\tabularnewline
\cline{1-4} \cline{12-14} 
\multirow{14}{*}{3} & 0 & 000 & $\phantom{{\displaystyle ^{^{1}}}}\ppp\oo\ppp\oo\ppp\oo$ &  &  &  & 0110 & $\ppp\oo\L^{-1}\oo\L^{-1}\oo\ppp\oo$ &  &  & \multirow{14}{*}{4} & 1111 & $\L^{-1}\oo\L^{-1}\oo\L^{-1}\oo\L^{-1}\oo$\tabularnewline
\cline{2-4} 
 & \multirow{3}{*}{1} & 001 & $\phantom{{\displaystyle ^{^{1}}}}-\ppp\oo\ppp\oo\L^{-1}\oo$ &  &  &  & 1001 & $\L^{-1}\oo\ppp\oo\ppp\oo\L^{-1}\oo$ &  &  &  & 1120 & $-\L^{-1}\oo\L^{-1}\oo\L^{-2}\oo\ppp\oo$\tabularnewline
 &  & 010 & $\phantom{{\displaystyle ^{^{1}}}}-\ppp\oo\L^{-1}\oo\ppp\oo$ &  &  &  & 1010 & $\L^{-1}\oo\ppp\oo\L^{-1}\oo\ppp\oo$ &  &  &  & 1201 & $-\L^{-1}\oo\L^{-2}\oo\ppp\oo\L^{-1}\oo$\tabularnewline
 &  & 100 & $\phantom{{\displaystyle ^{^{1}}}}-\L^{-1}\oo\ppp\oo\ppp\oo$ &  &  &  & 1100 & $\L^{-1}\oo\L^{-1}\oo\ppp\oo\ppp\oo$ &  &  &  & 1210 & $-\L^{-1}\oo\L^{-2}\oo\L^{-1}\oo\ppp\oo$\tabularnewline
\cline{2-4} 
 & \multirow{5}{*}{2} & 011 & $\phantom{{\displaystyle ^{^{1}}}}\ppp\oo\L^{-1}\oo\L^{-1}\oo$ &  &  &  & 0020 & $-\ppp\oo\ppp\oo\L^{-2}\oo\ppp\oo$ &  &  &  & 2011 & $-\L^{-2}\oo\ppp\oo\L^{-1}\oo\L^{-1}\oo$\tabularnewline
 &  & 101 & $\phantom{{\displaystyle ^{^{1}}}}\L^{-1}\oo\ppp\oo\L^{-1}\oo$ &  &  &  & 0200 & $-\ppp\oo\L^{-2}\oo\ppp\oo\ppp\oo$ &  &  &  & 2101 & $-\L^{-2}\oo\L^{-1}\oo\ppp\oo\L^{-1}\oo$\tabularnewline
 &  & 110 & $\phantom{{\displaystyle ^{^{1}}}}\L^{-1}\oo\L^{-1}\oo\ppp\oo$ &  &  &  & 2000 & $-\L^{-2}\oo\ppp\oo\ppp\oo\ppp\oo$ &  &  &  & 2110 & $-\L^{-2}\oo\L^{-1}\oo\L^{-1}\oo\ppp\oo$\tabularnewline
\cline{7-9} 
 &  & 020 & $\phantom{{\displaystyle ^{^{1}}}}-\ppp\oo\L^{-2}\oo\ppp\oo$ &  &  & \multirow{7}{*}{3} & 0111 & $-\ppp\oo\L^{-1}\oo\L^{-1}\oo\L^{-1}\oo$ &  &  &  & 1300 & $\L^{-1}\oo\L^{-3}\oo\ppp\oo\ppp\oo$\tabularnewline
 &  & 200 & $\phantom{{\displaystyle ^{^{1}}}}-\L^{-2}\oo\ppp\oo\ppp\oo$ &  &  &  & 1011 & $-\L^{-1}\oo\ppp\oo\L^{-1}\oo\L^{-1}\oo$ &  &  &  & 2020 & $\L^{-2}\oo\ppp\oo\L^{-2}\oo\ppp\oo$\tabularnewline
\cline{2-4} 
 & \multirow{5}{*}{3} & 111 & $-\L^{-1}\oo\L^{-1}\oo\L^{-1}\oo$ &  &  &  & 1101 & $-\L^{-1}\oo\L^{-1}\oo\ppp\oo\L^{-1}\oo$ &  &  &  & 2200 & $\phantom{{\displaystyle ^{^{1}}}}\L^{-2}\oo\L^{-2}\oo\ppp\oo\ppp\oo$\tabularnewline
 &  & 120 & $\phantom{{\displaystyle ^{^{1}}}}\L^{-1}\oo\L^{-2}\oo\ppp\oo$ &  &  &  & 1110 & $-\L^{-1}\oo\L^{-1}\oo\L^{-1}\oo\ppp\oo$ &  &  &  & 3001 & $\L^{-3}\oo\ppp\oo\ppp\oo\L^{-1}\oo$\tabularnewline
 &  & 201 & $\phantom{{\displaystyle ^{^{1}}}}\L^{-2}\oo\ppp\oo\L^{-1}\oo$ &  &  &  & 0120 & $\ppp\oo\L^{-1}\oo\L^{-2}\oo\ppp\oo$ &  &  &  & 3010 & $\L^{-3}\oo\ppp\oo\L^{-1}\oo\ppp\oo$\tabularnewline
 &  & 210 & $\phantom{{\displaystyle ^{^{1}}}}\L^{-2}\oo\L^{-1}\oo\ppp\oo$ &  &  &  & 0201 & $\ppp\oo\L^{-2}\oo\ppp\oo\L^{-1}\oo$ &  &  &  & 3100 & $\L^{-3}\oo\L^{-1}\oo\ppp\oo\ppp\oo$\tabularnewline
 &  & 300 & $\phantom{{\displaystyle ^{^{1}}}}-\L^{-3}\oo\ppp\oo\ppp\oo$ &  &  &  & 0210 & $\ppp\oo\L^{-2}\oo\L^{-1}\oo\ppp\oo$ &  &  &  & 4000 & $-\L^{-4}\oo\ppp\oo\ppp\oo\ppp\oo$\tabularnewline
\cline{1-4} \cline{6-9} \cline{11-14} 
\end{tabular}

\caption{\label{tab:list}List of sequences \(\l=(\l_{1},\l_{2},\cdots,\l_{N})\)
and their corresponding terms for the Dyson series in eq.~(\ref{eq:dyson})
up to \(N=4\). Recall that \((1+\frac{1}{\eta}\p_{t})t^{N-M}/\left(N-M\right)!\)
prepends and \(\ppp\) appends each term.}
\end{sidewaystable}

Counting the sequences $\l\in\cat_{M}^{N}$ and the corresponding
set of all sequences required to construct $\T^{\left(N\right)}$,
\begin{equation}
\cat^{N}\equiv\bigcup_{M=0}^{N}\cat_{M}^{N}\,,
\end{equation}
reveals a quite interesting connection to the Catalan numbers \textemdash{}
a sequence of numbers which has 215 different combinatorial interpretations
\cite{catalan}! The number of possible sequences $\l$ of length
$N$ that sum up to $M$ is
\begin{equation}
|\cat_{M}^{N}|=\frac{\left(N+M\right)!\left(N+1-M\right)}{M!\left(N+1\right)!}\equiv C\left(N,M\right)\,,\label{eq:catalan}
\end{equation}
where $C\left(N,M\right)$ is an entry in Catalan's triangle \cite{catalan2}
and we define $C\left(N\right)\equiv C\left(N,N\right)$. As a result
of the properties of Catalan's triangle (see Tab.~\ref{tab:catalan}),
the total number of terms for order $N$ is the $N+1$st Catalan number:
\begin{equation}
\sum_{M=0}^{N}C\left(N,M\right)=\frac{\left(2N+2\right)!}{\left(N+1\right)!\left(N+2\right)!}=C\left(N+1\right)\,.
\end{equation}
The first few entries in Catalan's triangle are reproduced in Tab.~\ref{tab:catalan}.
Thus, the number of terms in the $N$th order piece $\T^{\left(N\right)}$
scales as
\begin{equation}
C\left(N+1\right)\sim\frac{4^{N+1}}{\left(N+1\right)^{3/2}\sqrt{\pi}}\,,
\end{equation}
indicating an exponentially increasing number of terms.

\begin{table}
\selectlanguage{american}%
\begin{centering}
\begin{tabular}{cccccc}
\toprule 
$N\backslash M$ & $0$ & $1$ & $2$ & $3$ & $4$\tabularnewline
\midrule
$0$ & $1$ & \selectlanguage{english}%
\selectlanguage{english}%
 & \selectlanguage{english}%
\selectlanguage{english}%
 & \selectlanguage{english}%
\selectlanguage{english}%
 & \selectlanguage{english}%
\selectlanguage{english}%
\tabularnewline
$1$ & $1$ & $1$ & \selectlanguage{english}%
\selectlanguage{english}%
 & \selectlanguage{english}%
\selectlanguage{english}%
 & \selectlanguage{english}%
\selectlanguage{english}%
\tabularnewline
$2$ & $1$ & $2$ & $2$ & \selectlanguage{english}%
\selectlanguage{english}%
 & \selectlanguage{english}%
\selectlanguage{english}%
\tabularnewline
$3$ & $1$ & $3$ & $5$ & $5$ & \selectlanguage{english}%
\selectlanguage{english}%
\tabularnewline
$4$ & $1$ & $4$ & $9$ & $14$ & $14$\tabularnewline
\bottomrule
\end{tabular}
\par\end{centering}
\selectlanguage{english}%
\caption{\label{tab:catalan}Catalan's triangle. Note how each term is a sum
of entries above and to the left.}
\end{table}

\subsection{Proof\label{subsec:Proof}}

Since $\T^{\left(N\right)}$'s can be defined recursively, we prove
eq.~(\ref{eq:dyson}) by induction on $N$. 

\paragraph{Base case }

This corresponds to $N=1$. Explicitly,
\begin{equation}
\T_{t}^{\left(1\right)}\ppp=\intop_{-\infty}^{t}g(\tau_{1})d\tau_{1}e^{\left(t-\tau_{1}\right)\L}\oo\ppp\,.
\end{equation}
Inserting the decomposition $\id=\ppp+\qqq$ to the left of $\oo$
and simplifying yields
\begin{equation}
\T_{t}^{\left(1\right)}\ppp=\intop_{-\infty}^{t}g(\tau_{1})d\tau_{1}\ppp\oo\ppp+\intop_{-\infty}^{t}d\tau_{1}e^{\left(t-\tau_{1}\right)\L}\qqq\oo\ppp\,,
\end{equation}
where the $\eta\rightarrow0$ limit can readily be taken in the second
integral. Performing both integrals and omitting $\qqq$ for conciseness
yields
\begin{equation}
\T_{t}^{\left(1\right)}\ppp=\left[\left(t+\frac{1}{\eta}\right)\ppp\oo+\L^{-1}\oo\right]\ppp=\left(1+\frac{\p_{t}}{\eta}\right)\left(t\ppp\oo+\L^{-1}\oo\right)\ppp\,.
\end{equation}
For $N=1$, the two sequences $\l$ which satisfy the rules from eq.~(\ref{eq:rules})
are $\l=\left(0\right)$ and $\l=\left(1\right)$. These respectively
correspond to the two terms above. The corresponding elements in the
second line of Catalan's triangle are $C\left(1,0\right)=1$ and $C\left(1,1\right)=1$,
which sum up to the Catalan number $C\left(2,2\right)=2$.

\paragraph{Inductive case}

Assume that eq.~(\ref{eq:dyson}) is true for $N$ and show that
it is true for $N+1$. Explicitly,\begin{subequations}
\begin{align}
\T_{t}^{\left(N+1\right)} & =\S_{t,N+1}\T_{N+1}^{\left(N\right)}\\
 & =\intop_{-\infty}^{t}g(\tau_{N+1})d\tau_{N+1}e^{\left(t-\tau_{N+1}\right)\L}\sum_{M=0}^{N}\frac{\left(1+\frac{\p_{\tau_{N+1}}}{\eta}\right)\tau_{N+1}^{N-M}}{\left(N-M\right)!}\sum_{\l\in\cat_{M}^{N}}\oo\prod_{n=1}^{N}{\cal X}(\l_{n})\,.
\end{align}
\end{subequations}We once again split the integrals using $\id=\ppp+\qqq$:\begin{subequations}
\begin{align}
\T_{t}^{\left(N+1\right)} & =\sum_{M=0}^{N}\intop_{-\infty}^{t}g(\tau_{N+1})d\tau_{N+1}\frac{\left(1+\frac{\p_{\tau_{N+1}}}{\eta}\right)\tau_{N+1}^{N-M}}{\left(N-M\right)!}\ppp\sum_{\l\in\cat_{M}^{N}}\oo\prod_{n=1}^{N}{\cal X}(\l_{n})\label{eq:firstpart}\\
 & +\sum_{M=0}^{N}\intop_{-\infty}^{t}d\tau_{N+1}\frac{\left(1+\frac{\p_{\tau_{N+1}}}{\eta}\right)\tau_{N+1}^{N-M}}{\left(N-M\right)!}e^{\left(t-\tau_{N+1}\right)\L}\qqq\sum_{\l\in\cat_{M}^{N}}\oo\prod_{n=1}^{N}{\cal X}(\l_{n})\label{eq:secondpart}
\end{align}
\end{subequations}We now simplify the integrals.
\begin{itemize}
\item Let us first simplify those integrals in $\T_{t}^{\left(N+1\right)}$
which involve $g(\tau_{N+1})$, i.e., eq.~(\ref{eq:firstpart}).
For each $M$, split each of the two integrals into an integral over
$(-\infty,0]$ and $[0,t]$, just like in the base case. All $(-\infty,0]$
integrals are zero, which we prove by using the following indefinite
integral for some $a\in\mathbb{R}$ (proven using integration by parts):
\begin{equation}
\int_{T_{1}}^{T_{2}}d\tau e^{a\tau}\frac{\tau^{M}}{M!}=\left.\left(-\right)^{M}e^{a\tau}\sum_{K=0}^{M}\frac{\left(-\tau\right)^{K}}{K!}a^{-\left(M+1-K\right)}\right|_{T_{1}}^{T_{2}}\,.\label{eq:help}
\end{equation}
For the $(-\infty,0]$ integrals, $a=\eta$. The above is zero when
evaluated at $T_{1}=-\infty$ due to the exponential $e^{\eta\tau}$.
It produces the ``infinity'' $\left(-\right)^{M}/\eta^{M+1}$ at
$T_{2}=0$. The only remaining integrals are those over $[0,t]$,
which are trivially evaluated. The two infinities coming from the
integrals of $\tau_{N+1}^{N-M}$ and $\tau_{N+1}^{N-M-1}$ over $t\in(-\infty,0]$
cancel due to having different signs, and, after simplification, all
of this yields 
\begin{equation}
\intop_{-\infty}^{t}g(\tau_{N+1})d\tau_{N+1}\frac{\left(1+\frac{\p_{\tau_{N+1}}}{\eta}\right)\tau_{N+1}^{N-M}}{\left(N-M\right)!}=\left(1+\frac{\p_{t}}{\eta}\right)\frac{t^{N+1-M}}{\left(N+1-M\right)!}\,.\label{eq:gt}
\end{equation}
\item Now let us apply eq.~(\ref{eq:help}) to the integrals in eq.~(\ref{eq:secondpart}).
This procedure is significantly simplified by observing that
\begin{equation}
\int d\tau e^{a\tau}\p_{\tau}\frac{\tau^{M}}{M!}=\left(-\right)^{M}e^{a\tau}\p_{\tau}\sum_{K=0}^{M}\frac{\left(-\tau\right)^{K}}{K!}a^{-\left(M+1-K\right)}\,.
\end{equation}
When performing integrals of $e^{\tau\L}$, we can heuristically pretend
$\L$ is a scalar since we are working only on its range. Performing
these manipulations and simplifying signs yields 
\begin{equation}
\intop_{-\infty}^{t}d\tau_{N+1}\frac{\left(1+\frac{\p_{\tau_{N+1}}}{\eta}\right)\tau_{N+1}^{N-M}}{\left(N-M\right)!}e^{\left(t-\tau_{N+1}\right)\L}\qqq=\left(1+\frac{\p_{t}}{\eta}\right)\sum_{K=0}^{N-M}\frac{t^{K}}{K!}\left[-\L^{-\left(N+1-M-K\right)}\right]\qqq\,.\label{eq:nogt}
\end{equation}
\end{itemize}
Plugging eqs.~(\ref{eq:gt}) and (\ref{eq:nogt}) and into $\T_{t}^{\left(N+1\right)}$
(\ref{eq:firstpart}) yields
\begin{align}
\T_{t}^{\left(N+1\right)} & =\left(1+\frac{\p_{t}}{\eta}\right)\sum_{M=0}^{N}\sum_{K=0}^{N+1-M}\frac{t^{K}}{K!}{\cal X}\left(N+1-M-K\right)\sum_{\l\in\cat_{M}^{N}}\prod_{n=2}^{N+1}{\cal X}(\l_{n})\,,
\end{align}
where we are now using the rule for ${\cal X}$ from eq.~(\ref{eq:rules2})
and have also shifted the labeling of the sequences $\l$ by one ($\l_{n}\rightarrow\l_{n+1}$)
for later convenience. Now, let us switch the order of the sums, sum
over $K$ backwards ($\sum_{k=0}^{N+1}a_{k}=\sum_{k=0}^{N+1}a_{N+1-k}$),
and rename indices as $M\leftrightarrow K$, yielding
\begin{align}
\T_{t}^{\left(N+1\right)} & =\left(1+\frac{\p_{t}}{\eta}\right)\sum_{M=0}^{N+1}\frac{t^{\left(N+1-M\right)}}{\left(N+1-M\right)!}\sum_{K=0}^{M-\d_{M,N+1}}{\cal X}\left(M-K\right)\sum_{\l\in\cat_{K}^{N}}\prod_{n=2}^{N+1}{\cal X}(\l_{n})\,.\label{eq:almost}
\end{align}
The Kronecker delta $\d_{M,N+1}$ is there because the sum over $K$
for $M=N+1$ has only $N+1$ (and not $N+2$) terms. For each $M$,
we proceed to rewrite the sum over $K$ in terms of new sequences
$\l^{\prime}$ and show that those sequences are elements of $\cat_{M}^{N+1}$. 
\begin{itemize}
\item According to eq.~(\ref{eq:rules}), the elements in each sequence
$\l\in\cat_{K}^{N}$ sum to $K$. Therefore, if we prepend these sequences
with $\l_{1}=M-K$, we obtain new sequences 
\begin{equation}
\l^{\prime}=(\l_{1},\l_{2},\cdots,\l_{N+1})
\end{equation}
whose elements sum to $M$. Since all elements $\l_{n\geq2}$ satisfy
eq.~(\ref{eq:rules}) and since $0\leq\l_{1}\leq N+1$, the new sequences
satisfy 
\begin{equation}
\sum_{n=1}^{l}\l_{N+2-n}\begin{cases}
\leq\min\left\{ l,M\right\}  & 1\leq l<N+1\\
=M & l=N+1
\end{cases}\,.\label{eq:rules-1}
\end{equation}
Therefore, for each $M$, all of the new sequences $\l^{\prime}\in\cat_{M}^{N+1}$.
\item According to eq.~(\ref{eq:catalan}), the total number of the old
sequences $\l$ is $|\cat_{K}^{N}|=C\left(N,K\right)$ for each $K$.
Therefore, the total number of new sequences $\l^{\prime}$ for each
$M$ is
\begin{equation}
\sum_{K=0}^{M-\d_{M,N+1}}C\left(N,K\right)=C(N+1,M-\d_{M,N+1})\,,
\end{equation}
where the right-hand side is true because of a property of Catalan's
triangle, namely, each term is a sum of entries above and to the left
\cite{catalan2}. Therefore, the total number of sequences $\l^{\prime}$
for each $M$ is $C(N+1,M)=|\cat_{M}^{N+1}|$ {[}for this formula,
we do not need the Kronecker delta because $C(N+1,N)=C(N+1)${]}.
Using once again properties of Catalan numbers, the total number of
terms for $N+1$ is
\begin{equation}
\sum_{M=0}^{N+1}C\left(N+1,M\right)=C\left(N+2\right)=|\cat^{N+1}|\,.
\end{equation}
\end{itemize}
Having shown that, for each $M,$ the sum over $K$ can be rearranged
as a sum over terms corresponding to all of the sequences in $\cat_{M}^{N+1}$,
we can rewrite eq.~(\ref{eq:almost}) as
\begin{align}
\T_{t}^{\left(N+1\right)} & =\left(1+\frac{\p_{t}}{\eta}\right)\sum_{M=0}^{N+1}\frac{t^{\left(N+1-M\right)}}{\left(N+1-M\right)!}\sum_{\l^{\pr}\in\cat_{M}^{N+1}}\prod_{n=1}^{N+1}{\cal X}(\l_{n})\,,
\end{align}
thereby completing the inductive step and the proof.\hfill$\boxempty$

\section{Relation to previous work\label{subsec:Relation-to-previous}}

\selectlanguage{american}%
We now mention six connections of the above general derivations to
previous works studying more specific cases. The first two deal with
first-order perturbation theory while the last five make contact with
higher-order effects.
\selectlanguage{english}%

\subsection{Decoherence Hamiltonian and dark states\label{subsec:Example:-decoherence-Hamiltonian}}

Focusing on the \foreignlanguage{american}{\acs{DFS}} case, we have
$\L^{-1}\spert|\rout\kk=\L_{\of}^{-1}\spert|\rout\kk$ (\ref{eq:simplified-dfs-leakage})
and the leakage rate $\adg$ is the dissipative gap of $\L_{\of}$.
However, we can show something more with a few minor assumptions.
Moreover, we show that for the semisimple \foreignlanguage{american}{\acs{DFS}}
cases of Sec.~\ref{subsec:Semisimple-DFS-case}, the dissipative
gap of $\L_{\ur}$ is the excitation gap of a related Hamiltonian.
We assume that $\L$ (\ref{eq:def}) can be written without a Hamiltonian
part,
\begin{equation}
\L(\r)=\half\sum_{\ell}\lind_{\ell}(2F^{\ell}\r F^{\ell\dg}-F^{\ell\dg}F^{\ell}\r-\r F^{\ell\dg}F^{\ell})\,,
\end{equation}
and that \foreignlanguage{american}{\acs{DFS}} states are annihilated
by the jump operators, $F^{\ell}|\psi_{k}\ket=0$ (if $|\psi_{k}\ket$
are also eigenstates of $H$, they are called \textit{dark states}
\cite{Kraus2008}). This implies that $F_{\ul}^{\ell}=\pp F^{\ell}\pp=0$
(with $\pp=\sum_{k=0}^{d-1}|\psi_{k}\ket\bra\psi_{k}|$). We now determine
$\adg$ for such systems. Borrowing from Sec.~\ref{app:decomp} and
using the above assumptions, $\L_{\ur}(\r)=-\half\sum_{\ell}\lind_{\ell}\pp\r\left(F^{\ell\dg}F^{\ell}\right)_{\lr}$
(\ref{eq:disgap}). From this, we can extract the decoherence \cite{Karasik2008}
or parent \cite{Iemini2015} Hamiltonian
\begin{equation}
\hdg\equiv\half\sum_{\ell}\lind_{\ell}F^{\ell\dg}F^{\ell}\,.
\end{equation}
The (zero-energy) ground states of $\hdg$ are exactly the \foreignlanguage{american}{\acs{DFS}}
states $|\psi_{k}\ket$ \cite{Karasik2008,Iemini2015} and the excitation
gap of $\hdg$ is $\adg$. We come back to this case in Ch.~\ref{ch:7}.

\subsection{Geometric linear response\label{subsec:Geometric-linear-response}}

We can avoid having to calculate the Green's function $\L^{-1}$ in
the Kubo formula (\ref{eq:kubo-simplified}) by cleverly choosing
an observable to measure. In what is essentially a linear response
version of the adiabatic response calculation of Ref.~\cite{Avron2012a},
Thm.~9, let us define the \textit{flux of an operator $X$}
\begin{equation}
A\equiv\dot{X}=\L^{\dgt}(X)\,.
\end{equation}
If we measure said flux $A$, a simple manipulation of the Kubo formula
(\ref{eq:kubo-simplified}) yields\begin{subequations}
\begin{align}
\bb A|\T_{t}^{\left(1\right)}|\rout\kk & =\bb\L^{\dgt}(X)|[({\textstyle t+\frac{1}{\eta}})\ppp-\L^{-1}\spert]|\rout\kk\\
 & =\bb X|\L[({\textstyle t+\frac{1}{\eta}})\ppp-\L^{-1}\spert]\spert|\rout\kk\\
 & =-\bb X|\L\L^{-1}\spert|\rout\kk\\
 & =-\bb X|\qqq\spert|\rout\kk\,.
\end{align}
\end{subequations}One can see that $\L$ is not present in the above
result. Therefore, if one perturbs with a \textit{current} $V=i\p_{\a}$
(i.e., $-i[V,\r]=\p_{\a}\rho$ for all $\rho$) and measures the flux
of $X=i\p_{\b}$ (for some parameters $\a,\b$), then one obtains
a Berry curvature in what can be called \textit{geometric linear response}.
This is a generalization of the geometric linear response of Hamiltonian
systems (\cite{Gritsev2012}, Appx.~C of Ref.~\cite{Bradlyn2012})
to Lindbladians. A more detailed linear response calculation can be
found in Sec.~IV.A.1 of Ref.~\cite{ABFJ}, complementing the earlier
adiabatic response calculation in Sec.~7 of \cite{Avron2012a}.

\subsection{Dyson series for unique state case}

Assume that the steady state is unique, so $\ppp=|\varrho\kk\bb I|$.
Then, assuming a trace-preserving perturbation, $\ppp\oo=|\varrho\kk\bb\oo^{\ddagger}(I)|=0$
(\ref{eq:easy}), all sequences $\l$ with $\l_{n}=0$ for some $n$
in the Dyson expansion in Thm.~\ref{thm:Dyson} vanish. The only
sequences that remain are of the form $\l=\left(1,1,\cdots,1\right)$
and the $N$th order term reduces to
\begin{align}
{\cal T}_{t}^{\left(N\right)} & =(-\L^{-1}\oo)^{N}\,.\label{eq:dyson-1-1}
\end{align}
This matches the time-independent perturbation theory calculation
from Ref.~\cite{Li2014}. Contrary to the exponentially increasing
number of terms when there is more than one steady state, the number
of terms in the $N$th order Dyson series term for an unperturbed
$\L$ with a unique steady state is just... one!

\subsection{Quantum Zeno dynamics}

Recall that the power of $t$ prepending each term is $N-M$, meaning
that the highest power of $t$ always prepends the sole $M=0$ term
\begin{equation}
\frac{t^{N}}{N!}(\ppp\oo\ppp)^{N}\,,\label{eq:dominant}
\end{equation}
associated with the sequence $\l=(0,0,\cdots,0)$. Therefore, if we
rescale our perturbation as $\oo\rightarrow\frac{1}{T}\oo$ and evolve
to time $T\gg1$, then, for each $N$, the term (\ref{eq:dominant})
is dominant and of order $O(1)$ while the remaining $M>0$ terms
are of order $O(1/T^{M})$. Since this dominant term acts within \acs{ASH}
and does not cause any leakage out of \acs{ASH}, it is often said
to generate quantum Zeno dynamics (\cite{Facchi2002,Schafer2014,Arenz2016};
see also \cite{Anandan1988,Beige2000a}). This effect has already
been mentioned in Appx.~D of Ref.~\cite{Zanardi2014} and derived
to within first order using related methods \cite{Paulisch2015,Azouit2016}.
Since we have shown in Sec.~\ref{subsec:hams} that $\ppp\oo\ppp$
is unitary for $\oo$ being a Hamiltonian ($\spert$) and/or jump
operator ($\Y$) perturbation, we can see that \textit{Zeno dynamics
is always unitary for all perturbations of those forms for any order
$N$ in the $T\rightarrow\infty$ limit}.

\subsection{Second-order terms\label{subsec:Second-order-terms}}

\selectlanguage{american}%
Omitting the $\nicefrac{1}{\eta}$ infinity, the $N=2$ terms are
explicitly
\begin{align}
\T_{t}^{\left(2\right)}\ppp & =\frac{(t\ppp\oo\ppp)^{2}}{2!}-t(\ppp\oo\L^{-1}+\L^{-1}\oo\ppp)\oo\ppp+\L^{-1}(\oo\L^{-1}-\L^{-1}\oo\ppp)\oo\ppp\,.
\end{align}
\foreignlanguage{english}{Being the first nonunitary correction to
\acs{ASH}, the term
\begin{equation}
\LE\equiv-\ppp\oo\L^{-1}\oo\ppp\label{eq:uls}
\end{equation}
is relevant in a variety of many-body contexts \cite{Garcia-Ripoll2009,prozen,Cai2013,Lesanovsky2013,Sciolla2015,Znidaric2015,Monthus2017,Monthus2017a,Medvedyeva2016}
and quantum optical scenarios (see next Subsection). This term can
also generate universal Lindbladian evolution on \acs{ASH} using}
the following clever Zeno-like scheme\foreignlanguage{english}{ }\cite{Zanardi2015a}\foreignlanguage{english}{.
Assume that we have a rescaled Hamiltonian perturbation $\oo=\frac{1}{\sqrt{T}}\spert$
and that $\ppp\spert\ppp=0$. Then, we can see that only $\LE$ and
$\frac{1}{T}\L^{-1}\spert\L^{-1}\spert\ppp$ remain, but the former
is dominant in the $T\rightarrow\infty$ limit. The first order leakage
term $\frac{1}{\sqrt{T}}\L^{-1}\spert\ppp$ is then the dominant correction,
being of order $O(\nicefrac{1}{\sqrt{T}}$). We note that, in general,
this term does not have to be in Lindblad form }\cite{Zanardi2015a}\foreignlanguage{english}{
and therefore can generate continuous-time but \textit{non-Markovian}
evolution.}

Just like for $N=1$, we can continue to apply the four-corners decomposition
$\empbig$ to the $N\geq2$ terms. \foreignlanguage{english}{For the
case of Hamiltonian ($\oo=\spert$) and/or jump operator ($\oo=\Y$)
perturbations, the piece $\L_{\lr}$ is also not relevant in $\LE$.
Since $\R_{\lr}\oo\R_{\ul}=0$ (see Sec.~\ref{subsec:hams}), one
has
\begin{equation}
\LE=-\ppp\oo\L_{\thu}^{-1}\oo\ps\,.
\end{equation}
However, we cannot replace the remaining $\ppp$ with $\ps$ since
two actions of $\spert$ \textit{can} take the state from $\ulbig$
to $\lrbig$. Moreover, this is not always true when the perturbation
$\oo$ is of general Lindblad form.}
\selectlanguage{english}%

\subsection{Effective operator formalism\label{subsec:Effective-Operator-Formalism}}

Let us return to the type of $\L$ studied in Sec.~\ref{subsec:non-Hermitian}
and assume that $F=F_{\ur}$, $H=H_{\lr}$, and the perturbation $\oo=\spert=-i[\hpert,\cdot]$
is Hamiltonian with $\hpert=\hpert_{\of}$. Then, $\LE$ (\ref{eq:uls})
is exactly that from the effective operator formalism of Ref.~\cite{Reiter2012}
(with $H_{\ul}=0$ for simplicity). Recall that now $\L_{\tho}\equiv\K$,
where $\K(\r)\equiv-i[K\r-\r K^{\dg}]$ and the ``non-Hermitian Hamiltonian''
\begin{equation}
K=K_{\lr}\equiv H_{\lr}-\frac{i}{2}\sum_{\ell}F^{\ell\dg}F^{\ell}\label{eq:nheh}
\end{equation}
(with $K_{\lr}>0$ on $\lrbig$). We show that $\LE$ is a Lindbladian,
\begin{equation}
\LE(\r)=-i[\HE,\r]+\FE\r\FE^{\dg}-\half\{\FE^{\dg}\FE,\r\}\,,\label{eq:leff-1}
\end{equation}
with Hamiltonian $\HE=-\half V(K^{-1}+K^{-1\dg})V$ and jump operators
$\FE^{\ell}=F^{\ell}K^{-1}\hpert_{\ll}$.

Recall that $\ppp$ splits into two terms, and that, for this \foreignlanguage{american}{\acs{DFS}}
case, the terms simplify to 
\begin{equation}
\ppp=\R_{\ul}+\ppp\R_{\lr}=\R_{\ul}-\R_{\ul}\L\L_{\lr}^{-1}=\R_{\ul}-\R_{\ul}\L\K_{\lr}^{-1}\,.
\end{equation}
Given right and left eigenstates $|n\ket$ and $\bra\tilde{m}|$ of
$K$ and their respective eigenvalues $\l_{n}$ and $\l_{m}$, 
\begin{equation}
\K^{-1}(|n\ket\bra\tilde{m}|)=i\frac{|n\ket\bra\tilde{m}|}{\l_{n}-\l_{m}^{\star}}\,.\label{eq:invk}
\end{equation}
Using the formula for $\ppp$, recalling (\ref{eq:simplified-dfs-leakage}),
and remembering that $\LE$ acts on $\rout\in\ulbig$ yields
\begin{align}
\LE & =\left(-\R_{\ul}+\R_{\ul}\L\K_{\lr}^{-1}\right)\spert\K_{\of}^{-1}\spert\R_{\ul}\,.\label{eq:leff}
\end{align}
Let us first calculate the piece $\spert\K_{\of}^{-1}\spert\R_{\ul}$
above:\begin{subequations}
\begin{align}
\spert\K_{\of}^{-1}\spert\R_{\ul}(\rout) & =-i\spert\K_{\of}^{-1}(\hpert_{\ll}\rout-\rout\hpert_{\ur})\\
 & =\spert(K^{-1}\hpert_{\ll}\rout+\rout\hpert_{\ur}K^{-1\dg})\\
 & =-i(\hpert_{\ll}\rout\hpert_{\ur}K^{-1\dg}-K^{-1}\hpert_{\ll}\rout\hpert_{\ur}+\hpert_{\ur}K^{-1}\hpert_{\ll}\rout-\rout\hpert_{\ur}K^{-1\dg}\hpert_{\ll})\,.\label{eq:fourfour}
\end{align}
\end{subequations}Here, $\K_{\of}^{-1}$ can be expressed in terms
of $K^{-1}$ by using eq.~(\ref{eq:invk}) and observing that either
the bra or ket in each outer product of operators in $\of$ is in
the kernel of $K$. We now examine the first and second pair of terms
to obtain the recycling and deterministic terms in $\LE$.
\begin{itemize}
\item The first two terms in eq.~(\ref{eq:fourfour}) involve $\hpert_{\ll}\rout\hpert_{\ur}\in\lrbig$
and are the terms seen by the second term of $\LE$ (\ref{eq:leff}).
By decomposing $\hpert_{\ll}\rout\hpert_{\ur}$ into outer products
$|n\ket\bra\tilde{m}|\in\lrbig$ of eigenstates of $K$, we can simply
study each outer product.\footnote{This assumes that $K$ is diagonalizable, although we are confident
this treatment can be extended.} For each $|n\ket\bra\tilde{m}|$,\begin{subequations}
\begin{align}
\K_{\lr}^{-1}(|n\ket\bra\tilde{m}|K^{-1\dg}-K^{-1}|n\ket\bra\tilde{m}|) & =\left(\frac{1}{\l_{m}^{\star}}-\frac{1}{\l_{n}}\right)\K_{\lr}^{-1}\left(|n\ket\bra\tilde{m}|\right)\\
 & =\frac{\l_{n}-\l_{m}^{\star}}{\l_{n}\l_{m}^{\star}}i\frac{|n\ket\bra\tilde{m}|}{\l_{n}-\l_{m}^{\star}}=i\frac{|n\ket\bra\tilde{m}|}{\l_{n}\l_{m}^{\star}}\\
 & =iK^{-1}|n\ket\bra\tilde{m}|K^{-1\dg}\,.
\end{align}
\end{subequations}Plugging this in and using $\R_{\ul}\L\R_{\lr}(\cdot)=\sum_{\ell}F^{\ell}\cdot F^{\ell\dg}$
(\ref{eq:transfer}) yields the second term in $\LE$,
\begin{align}
\R_{\ul}\L\K_{\lr}^{-1}\spert\K_{\of}^{-1}\spert\R_{\ul}(\rout) & =\sum_{\ell}F^{\ell}K^{-1}\hpert_{\ll}\rout\hpert_{\ur}K^{-1\dg}F^{\ell\dg}\equiv\sum_{\ell}\FE\rout\FE^{\dg}\,.
\end{align}
\item The second two terms in eq.~(\ref{eq:fourfour}) involve $\hpert_{\ur}K^{-1}\hpert_{\ll}\rout-H.c.\in\ofbig$
and are the terms seen by the second term of $\LE$ (\ref{eq:leff}).
Here, it is useful to decompose $K^{-1}=K_{+}^{-1}+K_{-}^{-1}$, where
$K_{\pm}=\half(K^{-1}\pm K^{-1\dg})$. The $K_{+}$ term produces
an effective Hamiltonian $\HE$ part of $\hpert_{\ur}K^{-1}\hpert_{\ll}$
while the $K_{-}$ term can be combined with eq.~(\ref{eq:nheh})
to relate $\hpert_{\ur}K^{-1}\hpert_{\ll}$ to the anti-commutator
piece consisting of $\sum_{\ell}\FE^{\dg}\FE$:\begin{subequations}
\begin{align}
-i\hpert_{\ur}K_{-}^{-1}\hpert_{\ll} & =-i\hpert_{\ur}\half(K^{-1}-K^{-1\dg})\hpert_{\ll}=-\frac{i}{2}\hpert_{\ur}K^{-1\dg}(K-K^{\dg})K^{-1}\hpert_{\ll}\\
 & =-\frac{i}{2}\hpert_{\ur}K^{-1\dg}\left(-i\sum_{\ell}F^{\ell\dg}F^{\ell}\right)K^{-1}\hpert_{\ll}=-\frac{1}{2}\sum_{\ell}\FE^{\dg}\FE\,.
\end{align}
\end{subequations}
\end{itemize}
Thus, we have constructed both the jump and deterministic terms of
eq.~(\ref{eq:leff-1}) and concisely linked the effective operator
formalism to ordinary second-order perturbation theory.\selectlanguage{english}%

\inputencoding{latin9}\newpage{}\foreignlanguage{english}{}%
\begin{minipage}[t]{0.5\textwidth}%
\selectlanguage{english}%
\begin{flushleft}
\begin{singlespace}\textit{``In general one may expect such effects
whenever an isolated system is considered as being divided into two
interacting parts, each slaved to a different aspect of the other.
The systems considered {[}...{]} might be regarded as a special case,
in which the coupling is with 'the rest of the Universe' (including
us as observers). The only role of the rest of the Universe is to
provide a Hamiltonian with slowly-varying parameters, thus forcing
the system to evolve adiabatically with phase continuation governed
by the time-dependent Schrödinger equation.''}\end{singlespace}
\par\end{flushleft}
\begin{flushleft}
\hfill{}\textendash{} Michael V. Berry
\par\end{flushleft}\selectlanguage{english}%
\end{minipage}

\chapter{Adiabatic perturbation theory\label{ch:5}}

\selectlanguage{english}%
We now apply the four-corners decomposition to adiabatic perturbation
theory. The leading order term governs adiabatic evolution within
\acs{ASH} while all other terms are non-adiabatic corrections. We
show that for a cyclic adiabatic deformation of steady \acs{ASH},
the \textit{holonomy} is unitary \cite{ABFJ}. We also determine that
the energy scale governing non-adiabatic corrections is once again
governed by the effective dissipative gap $\adg$. We begin by reviewing
the adiabatic/Berry connection for a Hamiltonian system in Sec.~\ref{subsec:Connection:-non-degenerate-Hamil}
and the \acs{DFS} case in Sec.~\ref{subsec:Unitary-case:-degenerate}.
We then continue to do full adiabatic perturbation theory for the
Lindbladian case in Sec.~\ref{subsec:adiabatic-response}.

The adiabatic limit has been generalized to Lindbladians \cite{AbouSalem2007,Joye2007,Avron2012a,Avron2012b,Schmid2013,Venuti2015}
and all orders of corrections to adiabatic evolution have been derived
(e.g., \cite{Avron2012b}, Thm.~6). By ``\textit{the} adiabatic
limit'', we mean that dominated by the steady states of $\L$. Another
adiabatic limit exists which is dominated by eigenstates of the Hamiltonian
part of $\L$ \cite{davies1978,Thunstrom2005,pekola2010}, which we
do not address here. Unlike adiabatic evolution of \textquotedblleft non-Hermitian
Hamiltonian\textquotedblright{} systems, Lindbladian adiabatic evolution
always obeys the rules of quantum mechanics (i.e., is completely-positive
and trace-preserving). In this work, we do not make the adiabatic
approximation to Hamiltonians (and later to Lindbladians \cite{Sarandy2005})
since it is not sufficient for the adiabatic theorem to hold. \foreignlanguage{american}{In
the adiabatic approximation, one assumes that certain (seemingly reasonable)
quantitative requirements on eigenstates and their derivatives w.r.t.
parameters are sufficient for the system to be approximately adiabatic.
However, those conditions have }been shown to be insufficient (see
\cite{Ortigoso2012} and refs. therein), so here we work strictly
in the adiabatic limit and do not assume any of the conditions of
the adiabatic approximation. We assume that \acs{ASH} is steady ($\hout=0$),
but note that this analysis can be extended to non-steady \acs{ASH}
by carefully including a \textquotedblleft dynamical phase\textquotedblright{}
contribution from $\hout$. 

The first work to make contact between adiabatic/Berry connections
and Lindbladians was Ref.~\cite{Sarandy2006}. Avron \textit{et al.}
(\cite{Avron2012b}, Prop. 3) showed that the corresponding holonomy
is trace-preserving and completely positive. Reference \cite{Carollo2006}
(see also \cite{Carollo2003}) showed that the holonomy is unitary
for Lindbladians possessing one \foreignlanguage{american}{\acs{DFS}}
block. Reference \cite{Oreshkov2010} proposed a theory of adiabaticity
which extended that result to the multi-block case and arrived at
eq.~(\ref{eq:nsad2}). They showed that corrections to their result
were $O\left(\nicefrac{1}{\sqrt{T}}\right)$ (with $T$ being the
traversal time), as opposed to $O\left(\nicefrac{1}{T}\right)$ as
in a proper adiabatic limit. By explicitly calculating the adiabatic
connections below, we connect the result of Ref.~\cite{Oreshkov2010}
with the formulation of Ref.~\cite{Sarandy2006}, showing that non-adiabatic
corrections are actually $O\left(\nicefrac{1}{T}\right)$. We also
extend Ref.~\cite{Oreshkov2010} to \foreignlanguage{american}{\acs{NS}}
cases where the dimension of the auxiliary subspace (i.e., the rank
of $\as$) can change. Regarding leakage out of \acs{ASH}, the idea
that $\lrbig$ is not relevant to first-order non-adiabatic corrections
is mentioned in the Supplemental Material of Ref.~\cite{Oreshkov2010}.

\section{Hamiltonian case\label{subsec:Connection:-non-degenerate-Hamil}}

First, let us review two important consequences of the (Hamiltonian)
quantum mechanical adiabatic theorem. Namely, adiabatic evolution
can be thought of as either (1) being generated by an effective operator
\cite{Kato1950} or (2) generating transport of vectors in parameter
space, leading to Abelian \cite{Vinitskii1990,Pancharatnam,berry1984,Aharonov1987}
or non-Abelian \cite{Wilczek1984} holonomies. We loosely follow the
excellent expositions in Ch.~2.1.2 of Ref.~\cite{phasebook} and
Sec.~9 of Ref.~\cite{avronleshouches}. We conclude with a summary
of \textit{four} different ways (\ref{eq:hol1}-\ref{eq:uzan-1})
of writing holonomies for the non-degenerate case and outline the
generalizations done in the following Sections.

Let $|\psi_{0}^{(t)}\ket$ be the instantaneous unique (up to a phase)
zero-energy ground state of a Hamiltonian $H(t)$. We assume that
the ground state is separated from all other eigenstates of $H(t)$
by a nonzero excitation gap for all times of interest. Let us also
rescale time ($s=t/T$) such that the exact state $|\psi(s)\ket$
evolves according to 
\begin{equation}
\frac{1}{T}\p_{s}|\psi(s)\ket=-iH(s)|\psi(s)\ket\,.\label{eq:schro-1}
\end{equation}
The adiabatic theorem states that $|\psi(s)\ket$ (with $|\psi(0)\ket=|\psi_{0}^{(s=0)}\ket$)
remains an instantaneous eigenstate of $H(s)$ (up to a phase $\t$)
in the limit as $T\rightarrow\infty$, with corrections of order $O(\nicefrac{1}{T})$.
Let $\pad^{(s)}=|\hol_{0}^{(s)}\ket\bra\hol_{0}^{(s)}|$ be the projection
onto the instantaneous ground state. In the adiabatic limit,
\begin{equation}
|\psi(s)\ket=e^{i\t(s)}|\psi_{0}^{(s)}\ket\label{eq:coord}
\end{equation}
and the initial projection $\pad^{(0)}$ evolves into
\begin{equation}
\pad^{(s)}=\uad(s)\pad^{(0)}\uad^{\dg}(s)\label{eq:adia}
\end{equation}
(with $\uad$ generating purely adiabatic evolution). The adiabatic
evolution operator $\uad$ is determined by the Kato equation 
\begin{equation}
\p_{s}\uad=-i\hk\uad\,,\label{eq:diffeq}
\end{equation}
with so-called Kato Hamiltonian \cite{Kato1950} ($\padd\equiv\p_{s}\pad$)
\begin{equation}
\hk=i[\padd,\pad]\,.\label{eq:kato}
\end{equation}
Such an adiabatic operator $\uad$ can be shown to satisfy eq.~(\ref{eq:adia})
(see \cite{phasebook}, Prop. 2.1.1) using 
\begin{equation}
\pad\padd\pad=\qad\padd\qad=0\,.\label{eq:projid}
\end{equation}
The $\pad\padd\pad=0$ is a key consequence of the idempotence of
projections while $\qad\padd\qad=0$ is obtained by application of
the no-leak property (\ref{eq:no-leak}); both are used throughout
the text. The conventional adiabatic evolution operator is then a
product of exponentials of $-i\hk$ ordered along the path $s^{\pr}\in[0,s]$
(with path-ordering denoted by $\path$):
\begin{equation}
\uad(s)=\path\exp\left(\int_{0}^{s}[\padd,\pad]\dd s^{\prime}\right)\,.
\end{equation}

Due to the \textit{intertwining property} (\ref{eq:adia}), $\uad(s)$
simultaneously transfers states in $\pad^{(0)}\h$ to $\pad^{(s)}\h$
and states in $\qad^{(0)}\h$ to $\qad^{(s)}\h$ (with $\qad\equiv I-\pad$)
without mixing the two subspaces during the evolution. The term $\padd\pad$
in eq.~(\ref{eq:kato}) is responsible for generating the adiabatic
evolution of $\pad\h$ while the term $\pad\padd$ generates adiabatic
evolution of $\qad\h$. To see this, observe that the adiabatically
evolving state $|\psi(s)\ket=\uad(s)|\psi_{0}^{(s=0)}\ket\in\pad^{(s)}\h$
obeys the Schrödinger equation
\begin{equation}
\p_{s}|\psi(s)\ket=[\padd,\pad]|\psi(s)\ket\,.\label{eq:comm-1-1}
\end{equation}
Applying property (\ref{eq:projid}), the second term in the commutator
can be removed without changing the evolution. Since we are interested
only in adiabatic evolution of the zero-eigenvalue subspace $\pad\h$
(and not its complement), we can simplify $\uad$ by removing the
second term in the Kato Hamiltonian. This results in the adiabatic
equation 
\begin{equation}
\p_{s}|\hol(s)\ket=\padd\pad|\hol(s)\ket\label{eq:schro-2}
\end{equation}
and effective adiabatic evolution operator
\begin{equation}
\uadd^{(s)}=\path\exp\left(\int_{0}^{s}\padd\pad\dd s^{\prime}\right)\,.
\end{equation}

We now assume that $s$ parameterizes a path in the parameter space
$\psp$ of some external time-dependent parameters of $H(s)$. For
simplicity, we assume that $\psp$ is simply-connected.\footnote{\label{fn:sc}If $\psp$ is not simply connected (i.e., has holes),
then the Berry phase may contain ``topological'' contributions.
Such effects are responsible for anyonic statistics (e.g., \cite{Read2009},
Sec.~I.B) and can produce Berry phases even for a one-dimensional
parameter space \cite{Zak1989}.} By writing $\pad$ and $\padd$ in terms of $|\psi_{0}\ket$ and
explicitly differentiating, the adiabatic Schrödinger equation (\ref{eq:schro-2})
becomes
\begin{equation}
\p_{s}|\hol\ket=(I-\pad)\p_{s}|\hol\ket.\label{eq:schro}
\end{equation}
This implies a \textit{parallel transport condition}
\begin{equation}
0=\pad\p_{s}|\hol\ket=\bra\hol|\p_{s}\hol\ket|\hol\ket\,,\label{eq:pratran}
\end{equation}
which describes how to move the state vector from one point in $\psp$
to another. The particular condition resulting from adiabatic evolution
eliminates any first-order deviation from the unit overlap between
nearby adiabatically evolving states \cite{simon1983}: 
\begin{equation}
\bra\hol(s+\d s)|\hol(s)\ket=1+O(\d s^{2})\,.
\end{equation}
Therefore, we have shown two interpretations stemming from the adiabatic
theorem. The first is that adiabatic evolution of $|\hol(s)\ket$
(with $|\psi(0)\ket=|\psi_{0}^{(s=0)}\ket$) is generated (in the
ordinary quantum mechanical sense) by the $\padd\pad$ piece of the
Kato Hamiltonian $\hk$. The second is that adiabatic evolution realizes
parallel transport of $|\hol(s)\ket$ along a curve in parameter space.
As we show now, either framework can be used to determine the adiabatically
evolved state and the resulting Berry phase.

We now define a coordinate basis $\{\xx_{\a}\}$ for the parameter
space $\psp$. In other words, 
\begin{equation}
\p_{t}=\frac{1}{T}\p_{s}=\frac{1}{T}\sum_{\a}\vv_{\a}\p_{\a}\,,\label{eq:param}
\end{equation}
where $\p_{s}$ is the derivative along the path, $\p_{\a}\equiv\p/\p\xx_{\a}$
are derivatives in various directions in parameter space, and $\dot{\xx}_{\a}\equiv\frac{\dd\xx_{\a}}{\dd s}$
are (dimensionless) parameter velocities. Combining eqs.~(\ref{eq:coord})
and (\ref{eq:param}) with the parallel transport condition (\ref{eq:pratran})
gives
\begin{equation}
0=\pad\p_{s}|\psi\ket=i\sum_{\a}\dot{\xx}_{\a}(\p_{\a}\t-A_{\a,00})|\psi\ket\,,\label{eq:partran2}
\end{equation}
where the \textit{adiabatic/Berry connection} $A_{\a,00}=i\bra\psi_{0}|\p_{\a}\psi_{0}\ket$
is a vector/gauge potential in parameter space. The reason we can
think of $A_{\a,00}$ as a gauge potential is because it transforms
as one under \textit{gauge transformations} $|\psi_{0}\ket\rightarrow e^{i\vartheta}|\psi_{0}\ket$
where $\vartheta\in\mathbb{R}$:
\begin{equation}
A_{\a,00}\rightarrow A_{\a,00}-\p_{\a}\vartheta\,.
\end{equation}
These structures arise because the adiabatic theorem has furnished
for us a vector bundle over the parameter-space manifold $\psp$ \cite{simon1983,avronleshouches}.
More formally, given the \textit{trivial bundle} $\psp\times\h$ (where
at each point in $\psp$ we have a copy of the full Hilbert space
$\h$), the projection $\pad$ defines a (possibly nontrivial) sub-bundle
of $\psp\times\h$ (in this case, a line bundle, since $\pad$ is
rank one). The trivial bundle has a covariant derivative $\nabla_{\alpha}\equiv\partial_{\alpha}$
with an associated connection that can be taken to vanish. The Berry
connection $A_{\alpha,00}$ is then simply the connection associated
with the covariant derivative $\pad\nabla_{\alpha}$ induced on the
sub-bundle defined by $\pad$.

The Berry connection describes what happens to the initial state vector
as it is parallel transported. It may happen that the vector does
not return to itself after transport around a closed path in parameter
space (due to e.g., curvature or non-simple connectedness of $\psp$).
Given an initial condition $\t(0)=0$, the parallel transport condition
(\ref{eq:partran2}) uniquely determines how $\t$ changes during
adiabatic traversal of a path $C$ parameterized by $s\in[0,1]$,
i.e., from a point $\xx_{\a}^{(s=0)}\in\psp$ to $\xx_{\a}^{(1)}$.
For a closed path ($\xx_{\a}^{(1)}=\xx_{\a}^{(0)}$) and assuming
$A_{\a,00}$ is defined uniquely for the whole path \cite{Read2009},
the state transforms as $|\psi(0)\ket\rightarrow B|\psi(0)\ket$ with
resulting gauge-invariant \textit{holonomy} (here, Berry phase)
\begin{equation}
B\equiv\exp\left(i\sum_{\a}\oint_{C}A_{\a,00}\dd\xx_{\a}\right)\,.\tag{Hol1}\label{eq:hol1}
\end{equation}
Alternatively, we can use (\ref{eq:param}) and the Schrödinger equation
(\ref{eq:schro-2}): $|\psi(0)\ket\rightarrow\uadd|\psi(0)\ket$ with
holonomy 
\begin{equation}
\uadd\equiv\path\exp\left(\sum_{\a}\oint_{C}\p_{\a}\pad\pad\dd\xx_{\a}\right)\,.\tag{Hol2}\label{eq:hol2}
\end{equation}
Since the geometric and Kato Hamiltonian formulations of adiabatic
evolution are equivalent, eqs.~(\ref{eq:hol1}-\ref{eq:hol2}) offer
two ways to get to the same answer. They reveal two representations
of the Berry connection and holonomy: the \textit{coordinate representation}
$\{iA_{\a,00},B\}$, which determines evolution of $\t$ from eq.~(\ref{eq:coord}),
and the \textit{operator representation} $\{\p_{\a}\pad\pad,\uadd\}$,
which determines evolution of $|\psi_{0}\ket$ \{Prop. 1.2 of \cite{Avron1989},
eq.~(5) of \cite{Avron2011}\}. Despite the latter being a path-ordered
product of matrices, it simplifies to the Berry phase in the case
of closed paths.

For completeness, we also state an alternative form for each holonomy
representation (\ref{eq:hol1}-\ref{eq:hol2}). If there are two or
more parameters, the coordinate representation can be expressed in
terms of the (here, Abelian) Berry curvature $F_{\a\b,00}\equiv\p_{\a}A_{\b,00}-\p_{\b}A_{\a,00}$
using Stokes' theorem:
\begin{equation}
B=\exp\left(\frac{i}{2}\sum_{\a,\b}\iint_{S}F_{\a\b,00}\dd\xx_{\a}\dd\xx_{\b}\right)\,,\tag{Hol3}\label{eq:bercurv-1}
\end{equation}
where $S$ is a surface whose boundary is the contour $C$. The operator
representation can also be written as a product of the path-dependent
projections $\pad$: 
\begin{equation}
\uadd=\path\prod_{s\in C}\pad^{(s)}\,,\tag{Hol4}\label{eq:uzan-1}
\end{equation}
where $\path\prod$ denotes a continuous product ordered from right
to left along the path $C$ \{eq.~(47) of \cite{Avron2012b}, Prop.
1 of \cite{Zanardi2015}\}. This form of the holonomy should be reminiscent
of the Pancharatnam phase \cite{Pancharatnam,phasebook} and, more
generally, of a dynamical quantum Zeno effect (\cite{Facchi2002,Schafer2014,Arenz2016};
see also \cite{Anandan1988,Beige2000a}).

In the following Sections, we generalize both representations to multi-dimensional
subspaces in operator form, superoperator form, and to Lindblad \foreignlanguage{american}{\acs{NS}}
blocks, so (in that order) the relevant quantities generalize to 
\begin{eqnarray*}
\pad & \rightarrow & \{\iidfs,\idfs,\ppp\}\\
A & \rightarrow & \{\adfs,\aadfs,\A\}\\
B & \rightarrow & \{\bdfs,\bbdfs,\ber\}\\
\uadd & \rightarrow & \{\uk,\ukdfs,\U\}\,.
\end{eqnarray*}

\section{DFS case\label{subsec:Unitary-case:-degenerate}}

We now generalize the above to a degenerate $d$-dimensional instantaneous
ground state eigenspace spanned by orthonormal basis states $\{|\psi_{k}^{(s)}\ket\}_{k=0}^{d-1}$
of a Hamiltonian $H(s)$. Due to the clean-leak condition (\ref{eq:clean-leak}),
this analysis also holds if that same eigenspace comprises the \foreignlanguage{american}{\acs{DFS}}
of a Lindbladian. For notational convenience, we indicate the s-dependence
as a superscript. We denote the respective operator and superoperator
projections as 
\begin{equation}
\begin{aligned}\iidfs^{(s)} & =\sum_{k=0}^{d-1}|\psi_{k}^{(s)}\ket\bra\psi_{k}^{(s)}|\\
\idfs^{(s)}(\r) & =\sum_{\m=0}^{d^{2}-1}|\stdfs_{\m}(s)\kk\bb\stdfs_{\m}(s)|\r\kk=\iidfs^{(s)}\r\iidfs^{(s)}
\end{aligned}
\label{eq:projs}
\end{equation}
{[}where $\stdfs_{\m}\in\text{span}\{|\psi_{k}\ket\bra\psi_{l}|\}$
is a Hermitian matrix basis for the \foreignlanguage{american}{\acs{DFS}},
$\iidfs^{\dg}=\iidfs$, $\idfs^{\dgt}(\r)=\iidfs^{\dg}\r\iidfs^{\dg}=\idfs$,
and $\r\in\text{\acs{OPH}}${]} in order to make contact with the
next Section, where such a set is the \foreignlanguage{american}{\acs{DFS}}
part of an \foreignlanguage{american}{\acs{NS}} block. For now however,
adiabatic evolution of $\{|\psi_{k}\ket\}$ occurs under the ordinary
Hamiltonian Schrödinger eq.~(\ref{eq:schro-1}).

Adiabatic evolution generalizes straightforwardly from the previous
Section by letting $\pad^{(s)}\rightarrow\iidfs^{(s)}$. The resulting
Wilczek-Zee adiabatic connection \cite{Wilczek1984} becomes a Hermitian
matrix (for each $\a$) with elements
\begin{equation}
\adfs_{\a,kl}\equiv i\bra\psi_{k}|\p_{\a}\psi_{l}\ket\,.\label{eq:bercondfs-1}
\end{equation}
This connection transforms as a gauge potential under transformations
$|\psi_{k}\ket\rightarrow|\psi_{l}\ket R_{lk}$, where $R\in U(d)$
is a unitary rotation of the \foreignlanguage{american}{\acs{DFS}}
states:
\begin{equation}
\adfs_{\a}\rightarrow R^{\dg}\adfs_{\a}R+iR^{\dg}\p_{\a}R\,.
\end{equation}
The holonomy (here, Wilson loop) is given by the matrix
\begin{equation}
\bdfs=\path\exp\left(i\sum_{\a}\oint_{C}\adfs_{\a}\dd\xx_{\a}\right)\label{eq:hol-1}
\end{equation}
acting on the vector of coefficients $c_{k}$. Generalizing eq.~(\ref{eq:hol1}),
$|\psi\ket$ transforms under the holonomy as
\begin{equation}
|\psi(0)\ket=\sum_{k=0}^{d-1}c_{k}|\psi_{k}^{(0)}\ket\rightarrow\sum_{k,l=0}^{d-1}\bdfs_{kl}c_{l}|\psi_{k}^{(0)}\ket\,.
\end{equation}

We now express the above structures in superoperator form in order
to bridge the gap between unitary and Lindblad systems. To do this,
we upgrade the state basis $\{|\psi_{k}\ket\}_{k=0}^{d-1}$ to the
matrix basis $\{|\stdfs_{\m}\kk\}_{\m=0}^{d^{2}-1}$. The adiabatic
Schrödinger equation can equivalently be expressed in operator and
superoperator form using the superoperator projection from eq.~(\ref{eq:projs}):
\begin{equation}
\p_{s}\r=|\p_{s}\psi\ket\bra\psi|+|\psi\ket\bra\p_{s}\psi|=\idfsd\idfs(\r)\label{eq:schrodfs}
\end{equation}
since $|\p_{s}\psi\ket=\iidfsd\iidfs|\psi\ket$, $\idfsd(\r)=\iidfsd\r\iidfs+\iidfs\r\iidfsd$,
and $\iidfsd=\p_{s}\iidfs$. The operator representation of the holonomy
is then the path-ordered product of exponentials of the generator
$\idfsd\idfs$.

The superoperator version of the coordinate form of the \foreignlanguage{american}{\acs{DFS}}
connection is then 
\begin{equation}
\aadfs_{\a,\m\n}=\bb\stdfs_{\m}|\p_{\a}\stdfs_{\n}\kk=\tr\{\stdfs_{\m}\p_{\a}\stdfs_{\n}\}\,.\label{eq:bercondfs}
\end{equation}
Sticking with the convention that $\stdfs_{0}\equiv\frac{1}{\sqrt{d}}\iidfs$
is the only traceful element and using property (\ref{eq:projid}),
$\stdfs_{0}\p_{\a}\stdfs_{0}\stdfs_{0}=0$ and we can see that $\aadfs_{\a,\m0}=\aadfs_{\a,0\m}=0$
for all $\m$. Thus, $\aadfs_{\a}$ consists of a direct sum of zero
with a $(d^{2}-1)$-dimensional anti-symmetric matrix acting on the
Bloch vector components $\{|\stdfs_{\m\neq0}\kk\}$. Since the latter
is anti-symmetric, the holonomy is unitary. Formally, letting $\text{\acs{OPH}}^{\star}$
be the space of traceless $d$-dimensional Hermitian matrices, $\idfs$
defines a sub-bundle of the trivial bundle $\psp\times\text{\acs{OPH}}^{\star}$
and $\aadfs_{\a}$ is the connection associated with the covariant
derivative $\idfs\p_{\a}$ induced on that sub-bundle.

\section{Lindbladian case\label{subsec:adiabatic-response}}

Throughout this entire Section, we assume that \acs{ASH} is steady
($\hout=0$). Recall that a system evolves in a rescaled time $s\equiv t/T\in\left[0,1\right]$
according to a time-dependent Lindbladian $\L\left(s\right)$, where
the end time $T$ is infinite in the adiabatic limit. For all $s$,
we define a continuous and differentiable family of asymptotic projections
\begin{equation}
\ppp^{(s)}=\sum_{\m}|\St_{\m}^{(s)}\kk\bb\J^{\m}(s)|\,,
\end{equation}
steady-state basis elements $\St_{\m}^{(s)}$ (such that $\L(s)|\St_{\m}^{(s)}\kk=0$),
and conserved quantities $\J^{\m}(s)$ (such that $\bb\J^{\m}(s)|\L(s)=0$).
Each projection therefore is associated with its own instantaneous
asymptotic subspace, $\ppp^{(s)}\oph$. The dimension of the instantaneous
subspaces (i.e., the rank of $\ppp^{(s)}$) is assumed to stay constant
during this evolution. In other words, the zero eigenvalue of $\L\left(s\right)$
is isolated from all other eigenvalues at all points $s$ by the\textit{
dissipative gap} $\dgg$ {[}analogous to the excitation gap in Hamiltonian
systems; see eq.~(\ref{eq:dissipative-gap-def}){]}. We once again
assume that $s\in\left[0,1\right]$ parameterizes a path in a space
of control parameters $\psp$, whose coordinate basis is $\{\xx_{\a}\}$,
and use the parameterization from the Hamiltonian case in eq.~(\ref{eq:param}):
\begin{equation}
\p_{t}=\frac{1}{T}\p_{s}=\frac{1}{T}\sum_{\a}\vv_{\a}\p_{\a}\,,\label{eq:param-1}
\end{equation}
where $\p_{s}$ is the derivative along the path, $\p_{\a}\equiv\p/\p\xx_{\a}$
are derivatives in various directions in parameter space, and $\dot{\xx}_{\a}\equiv\frac{\dd\xx_{\a}}{\dd s}$
are (unitless) parameter velocities. 

Following Ref.~\cite{Avron2012b}, starting with an initially steady
state $|\r(0)\kk\in\text{\acs{ASH}}$, adiabatic perturbation theory
is an expansion of the equation of motion
\begin{equation}
\frac{1}{T}\p_{s}|\r(s)\kk=\L(s)|\r(s)\kk
\end{equation}
in a series in $\nicefrac{1}{T}$. Each term in the expansion is further
divided using the decomposition $\id=\ppp+\qqq$ into terms inside
and outside the instantaneous \acs{ASH}. This allows one to derive
both the adiabatic limit (when $T\rightarrow\infty$) and all corrections.
The $O(\nicefrac{1}{T})$ expansion for the final state from Thm.~6
of Ref.~\cite{Avron2012b} reads\begin{align}
|\r\left(s\right)\kk &=\U^{(s,0)}|\r(0)\kk+\frac{1}{T}\L^{-1}\left(s\right)\pppd^{(s)}\U^{(s,0)}|\r(0)\kk \nonumber\\
&\!\!\!\!\!\!\!\!+\frac{1}{T}\int_{0}^{s}\dd r\U^{(s,r)}\{\pppd\L^{-1}\pppd\}^{\left(r\right)}\U^{(r,0)}|\r(0)\kk\,,\label{eq:adiaaaa}
\end{align}where all quantities in curly brackets are functions of $r$, $\pppd\equiv\p_{s}\ppp$,
$\qqq\equiv\id-\ppp$, and $\L^{-1}$ is the instantaneous inverse
(\ref{eq:inver}). The superoperator 
\begin{equation}
\U^{(s,s^{\pr})}=\path\exp\left(\int_{s^{\pr}}^{s}\pppd^{(r)}\ppp^{(r)}\dd r\right)
\end{equation}
parallel transports states in $\ppp^{(s^{\pr})}\text{\acs{OPH}}$
to states in $\ppp^{(s)}\text{\acs{OPH}}$ and is a path-ordered product
of exponentials of the \textit{adiabatic connection} $\pppd\ppp$,
the generator of Lindbladian adiabatic evolution.

Like the Kubo formula, all terms can be interpreted when read from
right to left. The first term in eq.~(\ref{eq:adiaaaa}) represents
adiabatic evolution of \acs{ASH}, the (second) \textit{leakage term}
quantifies leakage of $|\r(0)\kk$ out of \acs{ASH}, and the (last)
\textit{tunneling term} represents interference coming back into \acs{ASH}
from outside. This term is a continuous sum of adiabatically evolved
steady states which are perturbed by $\pppd\L^{-1}\pppd$ at all points
$r\in[0,s]$ during evolution. Due to its dependence on the spectrum
of $\L$, this term needs to be minimized to determine the optimal
adiabatic path through \acs{ASH} \cite{Avron2010}. Notice also the
similarity between the leakage term and the leakage term (\ref{eq:leakagekubo})
of the Kubo formula. Motivated by this, we proceed to apply the four-corners
decomposition to all three terms.

\subsection{Evolution within $\textnormal{As(\ensuremath{\mathsf{H}})}$\label{subsec:holonomy}}

Let us now assume a \textit{closed path} {[}$\L\left(s\right)=\L\left(0\right)${]}.
In the adiabatic limit {[}according to eq.~(\ref{eq:adiaaaa}){]},
an initial steady state evolves in closed path $C$ as 
\begin{equation}
|\r(0)\kk\rightarrow\U|\r(0)\kk\label{eq:hollind}
\end{equation}
and acquires a \textit{holonomy} 
\begin{equation}
\U\equiv\U^{(1,0)}=\path\exp\left(\oint_{C}\pppd\ppp\dd s\right)\,.\label{eq:udfs}
\end{equation}
This is operator version of the holonomy and connection ($\pppd\ppp$)
since the above expression acts on the steady-state basis elements
$\St_{\m}^{(s=0)}$ used to express the initial steady state 
\begin{equation}
|\r\left(0\right)\kk=\sum_{\m}c_{\m}|\St_{\m}^{(0)}\kk\,.\label{eq:state}
\end{equation}

Let us now study the coordinate representation of the holonomy. This
can be done by a straightforward generalization of the Hamiltonian
analysis of Sec.~\ref{subsec:Unitary-case:-degenerate} to Lindbladians
\cite{Sarandy2006,Avron2012b}, which produces a parallel transport
condition
\begin{equation}
\ppp\p_{s}|\r\kk=0\label{eq:partranmain}
\end{equation}
characterizing the Lindbladian adiabatic limit. After expressing $\p_{s}$
in terms of the various $\p_{\a}$'s (\ref{eq:param-1}), this condition
provides an equation of motion for the coordinate vector $c_{\m}$
from eq.~(\ref{eq:state}). Solving this equation yields the coordinate
representation of the holonomy
\begin{equation}
\ber=\path\exp\left(-\sum_{\a}\oint_{C}\A_{\a}\dd\xx_{\a}\right)\label{eq:reshol}
\end{equation}
and corresponding adiabatic connection 
\begin{equation}
\A_{\a,\m\n}\equiv\bb\J^{\m}|\p_{\a}\St_{\n}\kk\,.\label{eq:bercon}
\end{equation}
Note that $\A_{\a}$ is a real matrix since $\{J^{\m},\St_{\m}\}$
are Hermitian. The connection transforms as a gauge potential under
$|\St_{\m}\kk\rightarrow|\St_{\n}\kk\rest_{\n\m}$ and $\bb\J^{\m}|\rightarrow\rest_{\m\n}^{-1}\bb\J^{\n}|$
for any $\rest\in GL\left[\dim\text{\acs{ASH}},\mathbb{R}\right]$:
\begin{equation}
\A_{\a}\rightarrow\rest^{-1}\A_{\a}\rest+\rest^{-1}\p_{\a}\rest\,.\label{eq:trans}
\end{equation}
Upon evolution in the closed path, the density matrix transforms as
\begin{equation}
|\r(0)\kk=\sum_{\m=0}^{d^{2}-1}c_{\m}|\St_{\m}^{(0)}\kk\rightarrow\sum_{\m,\n=0}^{d^{2}-1}\ber_{\m\n}c_{\n}|\St_{\m}^{(0)}\kk\,,
\end{equation}
equivalent to the operator representation (\ref{eq:hollind}). We
study both representations below, showing that the holonomy is unitary
for all \acs{ASH}.

First, let us remove the decaying subspace from both representations
of the connection by applying the clean-leak property (\ref{eq:clean-leak}).
Simplifying $\A_{\a}$ turns out to be similar to calculating the
effective Hamiltonian perturbation $\rrreac$ within \acs{ASH} in
Sec.~\ref{subsec:hams}. By (\ref{eq:clean-leak}), 
\begin{equation}
\A_{\a,\m\n}\equiv\bb\J^{\m}|\p_{\a}\St_{\n}\kk=\bb\J_{\ul}^{\m}|\p_{\a}\St_{\n}\kk\,.\label{eq:bercon-1}
\end{equation}
For the operator representation, one first applies (\ref{eq:clean-leak})
to the parallel transport condition (\ref{eq:partranmain}):
\begin{equation}
0=\ppp|\p_{s}\r\kk=\ps|\p_{s}\r\kk\,.
\end{equation}
Then, one uses this condition to obtain an equation of motion for
$\r$: 
\begin{equation}
|\p_{s}\r\kk=\left(\id-\ps\right)|\p_{s}\r\kk=\psd\ps|\r\kk\,.
\end{equation}
The last equality above can be checked by expressing both sides in
terms of the steady-state basis elements $\St_{\m}$ and conserved
quantities $\J^{\m}$. For a closed path, the solution to this equation
of motion is then the same holonomy, but now with the minimal projection
$\ps$ instead of the asymptotic projection $\ppp$: 
\begin{equation}
\U=\path\exp\left(\oint_{C}\psd\ps\dd s\right)\,.\label{eq:udfsalt-1}
\end{equation}
The holonomy $\U$ thus does not depend on the piece $\ppp\R_{\lr}$
associated with the decaying subspace.

\subsubsection*{Unique state case}

Now the only conserved quantity is the identity $J=I$, so it is easy
to show that
\begin{equation}
\A_{\a}=\bb I|\p_{\a}\varrho\kk=\tr\left\{ \p_{\a}\varrho\right\} =0\,.
\end{equation}
The unique steady state can never acquire a Berry phase. This may
clash with the reader's memories from introductory quantum mechanics,
where a nonzero Berry phase was calculated for a Hamiltonian with
a unique ground state. Such a phase is undetectable since it is an
overall phase of the ground-state wavefunction. This phase disappears
when the state is written as density matrix. Since the Lindbladian
formalism deals with density matrices directly, one never encounters
such an overall phase. This should not be confused with interferometry
experiments used to detect Berry phases. Such experiments implicitly
assume that the adiabatically evolving subspace is more than one-dimensional.
In that case, a phase gained by one basis component and not the others
is then \textit{not} an overall phase, but a relative (and thus observable)
phase.

\subsubsection*{NS case}

For this case, the \foreignlanguage{american}{\acs{NS}} factors into
a \foreignlanguage{american}{\acs{DFS}} and an auxiliary part \textit{for
each }$s\in[0,1]$. The \foreignlanguage{american}{\acs{DFS}} part
is mapped into a reference \foreignlanguage{american}{\acs{DFS}}
space spanned by a (parameter-independent Hermitian matrix) basis
$\{|\stb_{\m}\kk\}_{\m=0}^{d^{2}-1}$.\footnote{Note that in general $|\stb_{\m}\kk\neq|\stdfs_{\m}(s=0)\kk$ since
$s$ parameterizes a particular path in $\psp$ while $\{|\stb_{\m}\kk\}$
is fixed.} We let $S(s)$ (with $\dist(\r)\equiv S\r S^{\dg}$) be the unitary
operator which simultaneously maps the instantaneous basis elements
$|\St_{\m}^{(s)}\kk$ into the reference \foreignlanguage{american}{\acs{DFS}}
basis and diagonalizes $\as^{(s)}$. Similarly, this $S(s)$ factors
the instantaneous conserved quantities $\bb\J_{\ul}^{\m}(s)|$ into
a \foreignlanguage{american}{\acs{DFS}} part and the identity $\ai^{(s)}$
on the auxiliary space. Therefore, we define the family of instantaneous
minimal projections as 
\begin{equation}
\ps^{(s)}=\dist(s)\left(\idfsb\ot|\as^{(s)}\kk\bb\ai^{(s)}|\right)\dist^{\dgt}(s)\,,\label{eq:states-1}
\end{equation}
where $\idfsb(\cdot)=\sum_{\m=0}^{d^{2}-1}|\stb_{\m}\kk\bb\stb_{\m}|\cdot\kk=\iidfsb\cdot\iidfsb$
is the superoperator projection onto the $\xx_{\a}$-independent \foreignlanguage{american}{\acs{DFS}}
reference basis. The generators of motion
\begin{equation}
G_{\a}\equiv iS^{\dg}\p_{\a}S\,\,\,\,\,\,\,\,\text{and}\,\,\,\,\,\,\,\,\GG_{\a}\equiv-i[G_{\a},\cdot]
\end{equation}
can \textit{mix up} the \foreignlanguage{american}{\acs{DFS}} with
the auxiliary part, generating novel dissipation-assisted adiabatic
dynamics.

We note that $\as^{(s)}$ (and therefore $\ai^{(s)}$) can change
rank ($\da^{(s)}$), provided that $\ps^{(s)}$ remains differentiable.
For example, one can imagine $\as^{(s)}$ to be a thermal state associated
with some Hamiltonian on $\hha$ whose rank jumps from one to $\da$
as the temperature is turned up from zero. This implies that $\pp^{(s)}$
and thus $\R_{\emp}^{(s)}$ can change rank also. However, such deformations
do not change the dimension $d^{2}$ of the steady-state subspace
and thus do not close the dissipative gap. To account for such deformations
in the one \foreignlanguage{american}{\acs{NS}} block case, the path
can be partitioned into segments of constant rank$\{\pp\}$ and the
connection calculation below can be applied to each segment.

Simplifying eq.~(\ref{eq:bercon-1}) by invoking the reference basis
structure of $\{\J,\St\}$ from eq.~(\ref{eq:states-1}) gives
\begin{equation}
\begin{aligned}\A_{\a} & =\aadfst_{\a}+\aax_{\a}=-i[\adfst_{\a},\cdot]\ot|\as\kk\bb\ai|+\aax_{\a}\,,\end{aligned}
\label{eq:bercons}
\end{equation}
where the \foreignlanguage{american}{\acs{DFS}} effective Hamiltonian
is \cite{Oreshkov2010} 
\begin{equation}
\adfst_{\a}\equiv\tr_{\textsf{ax}}\left\{ \left(\iidfsb\ot\as^{(s)}\right)G_{\a}\right\} \label{eq:nsad2}
\end{equation}
and the second term is the $\an$-dependent constant 
\begin{equation}
\aax_{\a,\m\n}=-\p_{\a}\ln\an^{(s)}\d_{\m\n}\,.
\end{equation}
The first term clearly leaves the auxiliary part invariant and generates
unitary evolution within the \foreignlanguage{american}{\acs{DFS}}
part of the \foreignlanguage{american}{\acs{NS}}. We can thus see
that \foreignlanguage{american}{\acs{DFS}} holonomies can be influenced
by $\as^{(s)}$. We will see that the second term's only role is to
preserve the trace for open paths.

Sticking with the convention that $\St_{0}^{(s)}$ is traceful and
the traceless $\St_{\m\neq0}^{(s)}$ carry the \foreignlanguage{american}{\acs{DFS}}
Bloch vector, we notice that $\A_{\a}$ transforms as a gauge potential
under orthogonal Bloch vector rotations $\rest\in SO(d^{2}-1)$: 
\begin{equation}
|\St_{\m\neq0}\kk\rightarrow|\St_{\n\neq0}\kk\rest_{\n\m}\,\,\,\,\,\,\,\text{and}\,\,\,\,\,\,\,|\J_{\ul}^{\m\neq0}\kk\rightarrow|\J_{\ul}^{\n\neq0}\kk\rest_{\n\m}\,.\label{eq:transnscase}
\end{equation}
In addition, one has the freedom to internally rotate $\as$ without
mixing $\St_{\m}$ with $\St_{\n\neq\m}$. Under such a transformation
$\ddax$, 
\begin{equation}
|\St_{\m}\kk\rightarrow\ddax|\St_{\m}\kk=\dist\kkk{\stb_{\m}\ot\frac{\uax\as\uax^{\dg}}{\an}}
\end{equation}
for some $\uax\in U(\da)$ and the connection transforms as an Abelian
gauge potential: 
\begin{equation}
\A_{\a,\m\n}\rightarrow\A_{\a,\m\n}+\bb\J_{\ul}^{\m}|\ddax^{\dgt}\p_{\a}\ddax|\St_{\n}\kk\,.
\end{equation}

Plugging in eq.~(\ref{eq:bercons}) into the Lindblad holonomy (\ref{eq:reshol}),
we can see that $\aax_{\a}$ is proportional to the identity matrix
(of the space of coefficients $c_{\m}$) and thus can be factored
out. Therefore,
\begin{equation}
\ber=\exp\left(\sum_{\a}\oint_{C}\p_{\a}\ln\an\dd\xx_{\a}\right)\bbdfs\,,\label{eq:holcoord}
\end{equation}
where $\bbdfs$ is the unitary $\as$-influenced holonomy associated
with $\adfst$. The first term in the above product for an open path
$s\in[0,1]$ is simply $\an^{(1)}/\an^{(0)}$, providing the proper
re-scaling of the coefficients $c_{\m}$ to preserve the trace of
$|\r(0)\kk$.\footnote{For open paths, $\ber$ is related to non-cyclic geometric phases
in other dissipative systems \{e.g., \cite{Sinitsyn2009}, eq.~(47)\}
and non-Hermitian systems \cite{viennot2012}.} For a closed path, this term vanishes (since $\an$ is real and positive)
and $\ber=\bbdfs$. Thus, we have shown that the holonomy after a
closed-loop traversal of one \foreignlanguage{american}{\acs{NS}}
block is unitary.

\subsubsection*{Multi-block case}

The generalization to multiple \foreignlanguage{american}{\acs{NS}}
blocks is straightforward: the reference basis now consists of multiple
blocks. Recall that $\J^{\m}$ do not have presence in the off-diagonal
parts neighboring the \foreignlanguage{american}{\acs{NS}} blocks
{[}Fig.~\ref{fig:decomp}(b){]} and that the only \foreignlanguage{american}{\acs{NS}}
block that $\p_{\a}\St_{\m}$ has presence in is that of $\St_{\m}$.
Therefore, each \foreignlanguage{american}{\acs{NS}} block is imparted
with its own unitary holonomy.

\subsection{Adiabatic curvature\label{subsec:Berry-curvature}}

The adiabatic connection $\A_{\a}$ (\ref{eq:bercon}) can be used
to define an adiabatic curvature defined on the parameter space induced
by the steady states. For simply-connected parameter spaces $\psp$
(see footnote \ref{fn:sc}), the adiabatic curvature can be shown
to generate the corresponding holonomy. More precisely, the Ambrose-Singer
theorem (\cite{nakahara}, Thm.~10.4) implies that the holonomy for
an infinitesimal closed path $C$ with basepoint $\xx_{\a}^{(0)}$
is the adiabatic curvature at $\xx_{\a}^{(0)}$. One can alternatively
use a generalization of Stokes' theorem to non-Abelian connections
\cite{Arefeva1980} to express the holonomy in terms of a ``surface-ordered''
integral of the corresponding adiabatic curvature, generalizing the
Abelian case (\ref{eq:bercurv-1}). Letting $\p_{\al\a}A_{\b\ar}=\p_{\a}A_{\b}-\p_{\b}A_{\a}$,
the curvature is
\begin{equation}
\F_{\a\b,\m\n}\equiv\p_{\al\a}\A_{\b\ar,\m\n}+[\A_{\a},\A_{\b}]_{\m\n}\,.\label{eq:curves-1}
\end{equation}
Using the \foreignlanguage{american}{\acs{NS}} adiabatic connection
(\ref{eq:bercons}) and remembering that $\p_{\a}\aax_{\b}$ is symmetric
in $\a,\b$, the adiabatic curvature for one \foreignlanguage{american}{\acs{NS}}
block,
\begin{equation}
\F_{\a\b,\m\n}=\p_{\al\a}\aadfst_{\b\ar,\m\n}+[\aadfst_{\a},\aadfst_{\b}]_{\m\n}\,,\label{eq:curvres}
\end{equation}
is just the curvature associated with the connection $\adfst$.

\subsection{Leakage out of the asymptotic subspace\label{subsec:gap}}

We now return to the adiabatic response formula (\ref{eq:adiaaaa})
to apply the four-corners decomposition to the $O\left(\nicefrac{1}{T}\right)$
non-adiabatic corrections. By definition (\ref{eq:inver}), $\L^{-1}$
has the same block upper-triangular structure as $\L$ from eq.~(\ref{eq:gen}).
The derivative of the asymptotic projection has partition 
\begin{equation}
\pppd=\left[\begin{array}{ccc}
\,(\psd)_{\ul}\, & \,\R_{\ul}\pppd\R_{\of}\, & \,\R_{\ul}\pppd\R_{\lr}\\
\R_{\of}\psd\R_{\ul} & 0 & \,\R_{\of}\pppd\R_{\lr}\\
0 & 0 & 0
\end{array}\right]\,.\label{eq:derproj}
\end{equation}
One can interpret $\pppd$ as a perturbation, analogous to $\spert$
from Ch.~\ref{ch:4}, and observe from the above partition that $\pppd$
does not connect block diagonal spaces: $\R_{\lr}\pppd\R_{\ul}=0$.
In addition, whenever $\pppd^{(r)}$ acts on a parallel transported
state living in $\R_{\ul}^{(r)}\text{\acs{OPH}}$, only the first
column in the above partition ($\pppd\R_{\ul}$) is relevant. These
observations result in $\L^{-1}\rightarrow\L_{\thu}^{-1}$ and the
replacement of two factors of $\pppd$ with $\psd$ in eq.~(\ref{eq:adiaaaa}).
Interestingly, we cannot replace the remaining $\pppd$ since $\R_{\ul}\pppd\R_{\of}$
contains contributions from $|\St^{\m}\kk\bb\p_{s}\J_{\lr}^{\m}|\R_{\of}$:\begin{align}
|\r\left(s\right)\kk &=\U^{(s,0)}|\r (0)\kk+\frac{1}{T}\L^{-1}_{\thu}\left(s\right)\psd^{(s)}\U^{(s,0)}|\r(0)\kk \nonumber\\
&\!\!\!\!\!\!\!\!+\frac{1}{T}\int_{0}^{s}\dd r\U^{(s,r)}\{\pppd\L^{-1}_{\thu}\psd\}^{\left(r\right)}\U^{(r,0)}|\r(0)\kk\,.\label{eq:dissss}
\end{align}Using the results of Sec.~\ref{subsec:Leakage-out-of}, the energy
scale governing the leading-order non-adiabatic corrections is once
again the effective dissipative gap $\adg$ \textemdash{} the nonzero
eigenvalue of $\L_{\ul}+\L_{\ur}$ with the smallest real part. A
similar result is shown for the leakage term in the Supplement of
Ref.~\cite{Oreshkov2010}. In addition, the tunneling term, which
is similar to the second-order perturbative correction $\ppp\spert\L^{-1}\spert\ppp$
discussed in Sec.~\ref{subsec:Leakage-out-of}, does not contain
contributions from $\L_{\lr}$.

\begin{sidewaystable*}
\begin{tabular}{llll}
\toprule 
\addlinespace
 & Hamiltonian/\foreignlanguage{american}{\acs{DFS}} systems: & Hamiltonian/\foreignlanguage{american}{\acs{DFS}} systems: & \multirow{2}{*}{Lindbladians: one \foreignlanguage{american}{\acs{NS}} block}\tabularnewline
 & operator notation & superoperator notation & \tabularnewline\addlinespace
\midrule
\addlinespace[0.2cm]
State basis & $|\psi_{k}\ket=$ \foreignlanguage{american}{\acs{DFS}} states & $\stdfs_{\m}=(\stdfs_{\m})^{\dg}\in\text{span}\{|\psi_{k}\ket\bra\psi_{l}|\}$ & ${\displaystyle |\St_{\m}\kk=|\stdfs_{\m}\kk\ot\kkk{\frac{\as}{\an}}}$\tabularnewline
\addlinespace
 & ${\displaystyle \iidfs=\sum_{k=0}^{d-1}|\psi_{k}\ket\bra\psi_{k}|}$ & ${\displaystyle \idfs=\sum_{\m=0}^{d^{2}-1}|\stdfs_{\m}\kk\bb\stdfs_{\m}|}$ & ${\displaystyle \ppp=\sum_{\m=0}^{d^{2}-1}|\St_{\m}\kk\bb\J^{\m}|}$\tabularnewline
\addlinespace
 &  &  & $\phantom{\ppp}=\ps+\ppp\R_{\lr}$\tabularnewline\addlinespace
\midrule
\addlinespace
Connection & $\adfs_{\a,kl}=i\bra\psi_{k}|\p_{\a}\psi_{l}\ket$ & $\aadfs_{\a,\m\n}=\bb\stdfs_{\m}|\p_{\a}\stdfs_{\n}\kk$ & $\A_{\a,\m\n}=\bb\J^{\m}|\p_{\a}\St_{\n}\kk$\tabularnewline
\addlinespace
 &  &  & \hyperref[eq:bercons]{$\phantom{\A_{\a,\m\n}}=\aadfst_{\a,\m\n}+\aax_{\a,\m\n}$}\tabularnewline\addlinespace
\midrule 
\addlinespace
Curvature & $\fdfs_{\a\b}=\p_{\al\a}\adfs_{\b\ar}-i[\adfs_{\a},\adfs_{\b}]$ & $\ffdfs_{\a\b}=\p_{\al\a}\aadfs_{\b\ar}+[\aadfs_{\a},\aadfs_{\b}]$ & $\F_{\a\b}=\p_{\al\a}\A_{\b\ar}+[\A_{\a},\A_{\b}]$\tabularnewline
\addlinespace
 &  &  & \hyperref[eq:curvres]{$\phantom{\mathcal{F}_{\a\b}}=\p_{\al\a}\aadfst_{\b\ar}+[\aadfst_{\a},\aadfst_{\b}]$}\tabularnewline\addlinespace
\addlinespace
 & $\fdfs_{\a\b,kl}=\bra\psi_{k}|\p_{\al\a}\iidfs\p_{\b\ar}\iidfs|\psi_{l}\ket$ & $\ffdfs_{\a\b,\m\n}=\bb\stdfs_{\m}|\p_{\al\a}\idfs\p_{\b\ar}\idfs|\stdfs_{\n}\kk$~~~~ & $\F_{\a\b,\m\n}=\bb\J_{\ul}^{\m}|\p_{\al\a}\ps\p_{\b\ar}\ps|\St_{\n}\kk$\tabularnewline\addlinespace
\midrule 
\addlinespace
\acs{QGT} & $\qdfs_{\a\b,kl}=\bra\psi_{k}|\p_{\a}\iidfs\p_{\b}\iidfs|\psi_{l}\ket$ & $\qqdfs_{\a\b,\m\n}=\bb\stdfs_{\m}|\p_{\a}\idfs\p_{\b}\idfs|\stdfs_{\n}\kk$ & $\geom_{\a\b,\m\n}=\bb\J_{\ul}^{\m}|\p_{\a}\ps\p_{\b}\ps|\St_{\n}\kk$\tabularnewline
\addlinespace
 & $\phantom{\qdfs_{\a\b,kl}}=-i\p_{\a}\adfs_{\b,kl}-(\adfs_{\a}\adfs_{\b})_{kl}$~~~~ & $\phantom{\qqdfs_{\a\b,\m\n}}=\p_{\a}\aadfs_{\b,\m\n}+(\aadfs_{\a}\aadfs_{\b})_{\m\n}$ & \hyperref[eq:qgt3]{$\phantom{\geom_{\a\b,\m\n}}=\p_{\a}\A_{\b,\m\n}+(\A_{\a}\A_{\b})_{\m\n}$}\tabularnewline
\addlinespace
 & $\phantom{\qdfs_{\a\b,kl}=}\,\,\,\,\,\,\,\,\,\,\,\,-\bra\psi_{k}|\p_{\a}\p_{\b}\psi_{l}\ket$ & $\phantom{\qqdfs_{\a\b,\m\n}=}\,\,\,\,\,\,\,\,\,-\bb\stdfs_{\m}|\p_{\a}\p_{\b}\stdfs_{\n}\kk$ & \hyperref[eq:qgt3]{$\phantom{\geom_{\a\b,\m\n}=}\,\,\,\,\,\,\,\,\,-\bb\J_{\ul}^{\m}|\p_{\a}\p_{\b}\St_{\n}\kk$}\tabularnewline\addlinespace
\midrule
\addlinespace
Metric tensor~~~~ & $\gdfs_{\a\b}=\tr\{\iidfs\p_{\sl\a}\iidfs\p_{\b\sr}\iidfs\}$ & $\ggdfs_{\a\b}=\Tr\{\idfs\p_{\sl\a}\idfs\p_{\b\sr}\idfs\}$ & $\met_{\a\b}=\Tr\{\ps\p_{(\a}\ps\p_{\b)}\ps\}$\tabularnewline\addlinespace
\bottomrule
\end{tabular}

\caption{\label{tab:Summary-of-adiabatic-structures}Summary of quantities
defined in Chs.~\ref{ch:5} and \ref{ch:6}.}
\end{sidewaystable*}
\selectlanguage{english}%

\inputencoding{latin9}\newpage{}\foreignlanguage{english}{}%
\begin{minipage}[t]{0.5\textwidth}%
\selectlanguage{english}%
\begin{flushleft}
\begin{singlespace}\textit{``The main added value of the paper is
that of providing results contained in 2-3 papers in a single one.''}\end{singlespace}
\par\end{flushleft}
\begin{flushleft}
\hfill{}\textendash{} Anonymous Referee
\par\end{flushleft}\selectlanguage{english}%
\end{minipage}

\chapter{Quantum geometric tensor\label{ch:6}}

\selectlanguage{english}%
Here, we introduce the Lindbladian \ac{QGT} $\geom$ and explicitly
calculate it for the unique state and \foreignlanguage{american}{\acs{NS}}
block cases \cite{ABFJ}. The anti-symmetric part of the \ac{QGT}
is equal to the curvature $\F$ generated by the connection $\A$
(see Sec.~\ref{subsec:Berry-curvature}). We show here that the symmetric
part of the \ac{QGT} produces a generalized metric tensor $\met$
for parameter spaces associated with Lindbladian steady-state subspaces.
We first review the Hamiltonian \ac{QGT} for a single state in Sec.~\ref{sec:Background:-non-degenerate-Hamil}
and then extend to the \foreignlanguage{american}{\acs{DFS}} case
in Sec.~\ref{sec:DFS-case}. The Lindbladian \ac{QGT} is calculated
in Sec.~\ref{sec:quantum-geometry}. We introduce other geometric
quantities in Sec.~\ref{sec:alt-geom-tensor}, including an alternative
geometric tensor $\geom^{\textsf{alt}}$ whose curvature is different
from the adiabatic curvature, but whose metric appears in the Lindbladian
adiabatic path length. Most of the relevant quantities for the Hamiltonian,
\foreignlanguage{american}{\acs{DFS}}, and Lindbladian cases are
summarized in Table \ref{tab:Summary-of-adiabatic-structures}.

The original geometric quantity, later called the \ac{QGT} by Berry
\cite{BerryQGT}, is introduced for Hamiltonian systems in Ref.~\cite{provost1980}.
This quantity encodes both a metric for measuring distances \cite{Anandan1990}
and the adiabatic curvature. The \ac{QGT}  is experimentally probeable
(e.g., via current noise measurements \cite{Neupert2013}). The Berry
curvature can be obtained from adiabatic transport in Hamiltonian
\cite{Avron1985,Xiao2010,Read2011} and Lindbladian \cite{Avron2011,Avron2012a}
systems and even ordinary linear response (see Sec.~\ref{subsec:Geometric-linear-response}).
Singularities and scaling behavior of the metric are in correspondence
with quantum phase transitions \cite{CamposVenuti2007,Zanardi2007,Kolodrubetz2013}.
Conversely, flatness of the metric and curvature may be used to quantify
stability of a given phase \cite{Roy2014,Dobardzic2013,Jackson2015,Bauer2015},
a topic of particular interest due to its applications in engineering
exotic topological phases. Regarding generalization of the \ac{QGT}
 to Lindbladians, to our knowledge there has been no introduction
of a tensor including both the adiabatic curvature and a metric associated
with \acs{ASH}. However, Refs.~\cite{Banchi2014,Marzolino2014}
did apply various known metrics to study distinguishability within
families of Gaussian fermionic and spin-chain steady states, respectively.

\section{Hamiltonian case\label{sec:Background:-non-degenerate-Hamil}}

First let us review the non-degenerate Hamiltonian case before generalizing
to degenerate Hamiltonians in operator/superoperator form. We recommend
Refs.~\cite{avronleshouches,Kolodrubetz2016} for detailed expositions.
Continuing from Sec.~\ref{subsec:Connection:-non-degenerate-Hamil},
we begin with an instantaneous zero-energy state $|\psi_{0}\ket$
and projection $\pad=|\psi_{0}\ket\bra\psi_{0}|$ which are functions
of a vector of control parameters $\{\xx_{\a}\}$. The distance between
the projections $\pad^{(s)}$ and $\pad^{(s+\d s)}$ along a path
parameterized by $s\in[0,1]$ (with parameter vectors $\xx_{\a}^{(s)}$
at each $s$) is governed by the \ac{QGT} \begin{subequations}
\begin{eqnarray}
Q_{\a\b,00} & = & \bra\psi_{0}|\p_{\a}\pad\p_{\b}\pad|\psi_{0}\ket\\
 & = & \bra\p_{\a}\psi_{0}|(I-\pad)|\p_{\b}\psi_{0}\ket\,.
\end{eqnarray}
\end{subequations}The second form can be obtained from the former
by explicit differentiation of $\pad$ and $\p_{\a}\pad\p_{\b}\pad=(\p_{\a}\pad)(\p_{\b}\pad)$
by convention. The $I-\pad$ term makes $Q_{\a\b,00}$ invariant upon
the gauge transformations $|\psi_{0}\ket\rightarrow e^{i\vartheta}|\psi_{0}\ket$.
The tensor can be split into symmetric and anti-symmetric parts,
\begin{equation}
2Q_{\a\b,00}=M_{\a\b,00}-iF_{\a\b,00}\,,
\end{equation}
which coincide with its real and imaginary parts. The anti-symmetric
part is none other than the adiabatic/Berry curvature from eq.~(\ref{eq:bercurv-1}).
The symmetric part is the quantum Fubini-Study metric tensor \cite{provost1980}
\begin{equation}
M_{\a\b,00}=\tr\{\pad\p_{\sl\a}\pad\p_{\b\sr}\pad\}=\tr\{\p_{\a}\pad\p_{\b}\pad\}\,,\label{eq:fsm}
\end{equation}
where $A_{\sl\a}B_{\b\sr}=A_{\a}B_{\b}+A_{\b}B_{\a}$ and the latter
form can be obtained using $\pad\p_{\a}\pad\pad=0$. This quantity
is manifestly symmetric in $\a,\b$ and real; it is also non-negative
when evaluated in parameter space (see \cite{Rezakhani2010}, Appx.~D).

\section{DFS case\label{sec:DFS-case}}

For degenerate Hamiltonian systems \cite{Ma2010,Rezakhani2010} and
in the \foreignlanguage{american}{\acs{DFS}} case, the \ac{QGT}
$\qdfs$ is a tensor in both parameter ($\a,\b$) and state ($k,l$)
indices and can be written as\begin{subequations}
\begin{eqnarray}
\qdfs_{\a\b,kl} & = & \bra\psi_{k}|\p_{\a}\iidfs\p_{\b}\iidfs|\psi_{l}\ket\label{eq:conv}\\
 & = & \bra\p_{\a}\psi_{k}|(I-\iidfs)|\p_{\b}\psi_{l}\ket\,,\label{eq:ouras}
\end{eqnarray}
where $\iidfs=\sum_{k=0}^{d-1}|\psi_{k}\ket\bra\psi_{k}|$ is the
projection onto the degenerate zero eigenspace of $H(s)$. Since
projections are invariant under changes of basis of their constituents,
it is easy to see that $\qdfs_{\a\b}\rightarrow R^{\dg}\qdfs_{\a\b}R$
under \foreignlanguage{american}{\acs{DFS}} changes of basis $|\psi_{k}\ket\rightarrow|\psi_{l}\ket R_{lk}$
for $R\in U(d)$. Notice that the \ac{QGT} in eq.~(\ref{eq:ouras})
consists of overlaps between states outside of the zero eigenspace.
For our applications, we write the \ac{QGT} in a third way such that
it consists of overlaps within the zero eigenspace only:
\begin{equation}
\qdfs_{\a\b,kl}=-i\p_{\a}\adfs_{\b,kl}-(\adfs_{\a}\adfs_{\b})_{kl}-\bra\psi_{k}|\p_{\a}\p_{\b}\psi_{l}\ket\,,\label{eq:ourresf}
\end{equation}
\end{subequations}where $\adfs_{\a}$ is the \foreignlanguage{american}{\acs{DFS}}
Berry connection and we used
\begin{equation}
\begin{aligned}0=\p_{\b}\bra\psi_{k}|\psi_{l}\ket & =\bra\p_{\b}\psi_{k}|\psi_{l}\ket+\bra\psi_{k}|\p_{\b}\psi_{l}\ket\\
\p_{\a}\bra\psi_{k}|\p_{\b}\psi_{l}\ket & =\bra\p_{\a}\psi_{k}|\p_{\b}\psi_{l}\ket+\bra\psi_{k}|\p_{\a}\p_{\b}\psi_{l}\ket\,.
\end{aligned}
\label{eq:trick}
\end{equation}
The Berry curvature is the part of the \ac{QGT} anti-symmetric in
$\a,\b$ (here, also the imaginary part of the \ac{QGT}): $\fdfs_{\a\b}=i\qdfs_{\al\a\b\ar}$.
From (\ref{eq:ourresf}) we easily recover the proper form of the
\foreignlanguage{american}{\acs{DFS}} Berry curvature listed in Table
\ref{tab:Summary-of-adiabatic-structures}.

The symmetric part of the \ac{QGT} appears in the infinitesimal distance
between nearby \textit{parallel transported} \textit{rays} (i.e.,
states of arbitrary phase) $\psi(s)$ and $\psi(s+\d s)$ in the degenerate
subspace:
\begin{equation}
\bra\p_{s}\psi|\p_{s}\psi\ket=\bra\p_{s}\psi|(I-\iidfs)|\p_{s}\psi\ket\,,
\end{equation}
where we used the parallel transport condition $\iidfs|\p_{s}\psi\ket=0$.
Expanding $\p_{s}$ into parameter derivatives using eq.~(\ref{eq:param})
and writing out $|\psi\ket=\sum_{k=0}^{d-1}c_{k}|\psi_{k}\ket$ yields
\begin{equation}
\bra\p_{s}\psi|\p_{s}\psi\ket=\half\sum_{\a,\b}\sum_{k,l=0}^{d-1}\qdfs_{\sl\a\b\sr,kl}\vv_{\a}\vv_{\b}c_{k}^{\star}c_{l}\,.
\end{equation}
The corresponding Fubini-Study metric on the parameter space $\psp$
is $\qdfs_{\sl\a\b\sr}$ traced over the degenerate subspace:
\begin{equation}
\gdfs_{\a\b}\equiv\sum_{k=0}^{d-1}\qdfs_{\sl\a\b\sr,kk}=\bb\iidfs|\p_{\sl\a}\iidfs\p_{\b\sr}\iidfs\kk\,.
\end{equation}

All of this reasoning easily extends to the superoperator formalism
($|\psi_{k}\ket\rightarrow|\stdfs_{\m}\kk$). The superoperator \ac{QGT}
corresponding to $\qdfs$ can be written as
\begin{eqnarray}
\qqdfs_{\a\b,\m\n} & = & \bb\stdfs_{\m}|\p_{\a}\idfs\p_{\b}\idfs|\stdfs_{\n}\kk\label{eq:qdfs}\\
 & = & \p_{\a}\aadfs_{\b,\m\n}+(\aadfs_{\a}\aadfs_{\b})_{\m\n}-\bb\stdfs_{\m}|\p_{\a}\p_{\b}\stdfs_{\n}\kk\,,\nonumber 
\end{eqnarray}
where $\aadfs_{\a}$ is the adiabatic connection (\ref{eq:bercondfs}).
The \ac{QGT} is a real matrix (since $\aadfs_{\a}$ is real) and
consists of parts symmetric ($\qqdfs_{\sl\a\b\sr}$) and antisymmetric
($\qqdfs_{\al\a\b\ar}$) in $\a,\b$. Observing the second line of
(\ref{eq:qdfs}), it should be easy to see that the Berry curvature
$\ffdfs_{\a\b}=\qqdfs_{\al\a\b\ar}$. The symmetric part of the superoperator
\ac{QGT} appears in the infinitesimal Hilbert-Schmidt distance (\cite{geombook},
Sec.~14.3) between nearby parallel transported \foreignlanguage{american}{\acs{DFS}}
states $\r(s)$ and $\r(s+\d s)$:
\begin{equation}
\bb\p_{s}\r|\p_{s}\r\kk=\bb\p_{s}\r|(\id-\idfs)|\p_{s}\r\kk\,,
\end{equation}
where we used the parallel transport condition $\idfs|\p_{s}\r\kk=0$.
Similar manipulations as with the operator \ac{QGT}, including the
expansion $|\r\kk=\sum_{\m=0}^{d^{2}-1}c_{\m}|\stdfs_{\n}\kk$, yield
\begin{equation}
\bb\p_{s}\rout|\p_{s}\rout\kk=\half\sum_{\a,\b}\sum_{\m,\n=0}^{d^{2}-1}\qqdfs_{\sl\a\b\sr,\m\n}\vv_{\a}\vv_{\b}c_{\m}c_{\n}\,.
\end{equation}
The corresponding superoperator metric 
\begin{equation}
\ggdfs_{\a\b}\equiv\Tr\{\idfs\p_{\sl\a}\idfs\p_{\b\sr}\idfs\}\,,
\end{equation}
where $\Tr$ is the trace in superoperator space, is the symmetric
part of the superoperator \ac{QGT} traced over the degenerate subspace.
Since $\text{\acs{OPH}}=\h\otimes\h^{\star}$, it is not surprising
that $\ggdfs_{\a\b}$ is proportional to the operator metric $\gdfs_{\a\b}$:
\begin{equation}
\ggdfs_{\a\b}=\sum_{\m=0}^{d^{2}-1}\qqdfs_{\sl\a\b\sr,\m\m}=2d\gdfs_{\a\b}\,.
\end{equation}

\section{Lindbladian case\label{sec:quantum-geometry}}

Now let us turn to the Lindbladian \ac{QGT} and show that its symmetric
part produces a generalized metric tensor for parameter spaces associated
with Lindbladian steady-state subspaces. In Ch.~\ref{ch:6}, we showed
using the operator representation of the adiabatic connection and
the conditions (\ref{eq:no-leak}-\ref{eq:clean-leak}) that the minimal
projection $\ps=\ppp\R_{\ul}$ (and not $\ppp$) generates adiabatic
evolution within \acs{ASH}. Following this, we define 
\begin{equation}
\geom_{\a\b}\equiv\ps\p_{\a}\ps\p_{\b}\ps\ps
\end{equation}
to be the associated \ac{QGT}. While $\ps=\sum_{\m}|\St_{\m}\kk\bb\J_{\ul}^{\m}|$
is not always Hermitian due to $J_{\ul}^{\m}\neq\St_{\m}$ (e.g.,
in the \foreignlanguage{american}{\acs{NS}} case), we show that the
\ac{QGT} nevertheless remains a meaningful geometric quantity. Looking
at the matrix elements of $\geom_{\a\b}$ and explicitly plugging
in the instantaneous $\ps$ (\ref{eq:states-1}) yields the following
three forms:\begin{subequations} 
\begin{eqnarray}
\!\!\!\!\!\!\!\!\!\!\!\!\!\!\!\!\!\!\!\!\!\geom_{\a\b,\m\n} & \equiv & \bb\J_{\ul}^{\m}|\p_{\a}\ps\p_{\b}\ps|\St_{\n}\kk\label{eq:qgt1}\\
 & = & \bb\p_{\a}\J_{\ul}^{\m}|\left(\id-\ps\right)|\p_{\b}\St_{\n}\kk\label{eq:qgt2}\\
 & = & \p_{\a}\A_{\b,\m\n}+(\A_{\a}\A_{\b})_{\m\n}-\bb\J_{\ul}^{\m}|\p_{\a}\p_{\b}\St_{\n}\kk\,,\label{eq:qgt3}
\end{eqnarray}
\end{subequations} with $\A_{\a}$ the Lindblad adiabatic connection
(\ref{eq:bercon}). Since $\A_{\a,\m\n}$ are real and $\{\J^{\m},\St_{\n}\}$
are Hermitian, the matrix elements are all real. From its second form,
one easily deduces that the \ac{QGT} transforms as $\geom_{\a\b}\rightarrow\rest^{-1}\geom_{\a\b}\rest$
for any basis transformation $\rest\in GL\left[\dim\text{\acs{ASH}},\mathbb{R}\right]$
{[}see eq.~(\ref{eq:trans}){]}. Each matrix $\geom_{\a\b}$ consists
of parts symmetric ($\geom_{\sl\a\b\sr}$) and antisymmetric ($\geom_{\al\a\b\ar}$)
in $\a,\b$. From the third form, it is evident that its anti-symmetric
part is exactly the adiabatic curvature $\F_{\a\b}$ from eq.~(\ref{eq:curves-1})
(cf. \cite{Avron2012a}, Prop. 13). The rest of this Section is devoted
to calculating the symmetric part and its corresponding metric on
$\psp$, which is defined as the trace of the symmetric part of the
\ac{QGT},
\begin{equation}
\met_{\a\b}\equiv\Tr\{\ps\p_{\sl\a}\ps\p_{\b\sr}\ps\}=\sum_{\m=0}^{d^{2}-1}\geom_{\sl\a\b\sr,\m\m}\,.
\end{equation}

Before proving that this is a metric for some of the relevant cases,
let us first reveal how such a structure corresponds to an infinitesimal
distance between adiabatically connected Lindbladian steady states
by adapting results from non-Hermitian Hamiltonian systems \cite{Nesterov2009,Brody2013,Brody2014}.
The zero eigenspace of $\L_{\ul}$ is diagonalized by right and left
eigenmatrices $|\St_{\m}\kk$ and $\bb\J_{\ul}^{\m}|$, respectively.
In accordance with this duality between $\St$ and $\J_{\ul}$, we
introduce an \textit{associated operator} $|\widehat{\rout}\kk$ \cite{Brody2013,Brody2014},
\begin{equation}
|\rout\kk=\sum_{\m=0}^{d^{2}-1}c_{\m}|\St_{\m}\kk\,\,\,\leftrightarrow\,\,\,|\widehat{\rout}\kk\equiv\sum_{\m=0}^{d^{2}-1}c_{\m}|\J_{\ul}^{\m}\kk\,,\label{eq:assoc}
\end{equation}
to every steady-state subspace operator $|\rout\kk$. This allows
us to define a modified inner product $\bb\widehat{A}|B\kk$ for matrices
$A$ and $B$ living in the steady-state subspace. Since $\St_{\m}$
and $\J_{\ul}^{\m}$ are biorthogonal ($\bb\J_{\ul}^{\m}|\St_{\n}\kk=\d_{\m\n}$),
this inner product is surprisingly equivalent to the Hilbert-Schmidt
inner product $\bb A|B\kk$. However, the infinitesimal distance is
not the same:
\begin{equation}
\bb\p_{s}\widehat{\rout}|\p_{s}\rout\kk\neq\bb\p_{s}\rout|\p_{s}\rout\kk\,.
\end{equation}
The symmetric part $\geom_{(\a\b)}$ shows up in precisely this modified
infinitesimal distance. Using eq.~(\ref{eq:assoc}), the parallel
transport condition (\ref{eq:partranmain}), and parameterizing $\p_{s}$
in terms of the $\p_{\a}$'s (\ref{eq:param-1}) yields
\begin{equation}
\bb\p_{s}\widehat{\rout}|\p_{s}\rout\kk=\half\sum_{\a,\b}\sum_{\m,\n=0}^{d^{2}-1}\geom_{\sl\a\b\sr,\m\n}\vv_{\a}\vv_{\b}c_{\m}c_{\n}\,,
\end{equation}
as evidenced by the second form (\ref{eq:qgt2}) of the Lindblad \ac{QGT}.
Tracing the symmetric part over the steady-state subspace gives the
metric $\met_{\a\b}$. 

\subsection{Unique state case}

Here things simplify significantly, yet the obtained metric turns
out to be novel nonetheless. The asymptotic projection is $\ps=|\varrho\kk\bb\pp|$
and a straightforward calculation using eq.~(\ref{eq:qgt2}) yields
\begin{equation}
\met_{\a\b}=\bb\p_{\sl\a}\pp|\p_{\b\sr}\varrho\kk\,.\label{eq:metricunique}
\end{equation}
Using the eigendecomposition $\varrho=\sum_{k=0}^{\du-1}\l_{k}|\psi_{k}\ket\bra\psi_{k}|$,
\begin{eqnarray}
\met_{\a\b} & = & 2\sum_{k=0}^{\du-1}\l_{k}\bra\p_{\sl\a}\psi_{k}|\qq|\p_{\b\sr}\psi_{k}\ket\,,
\end{eqnarray}
where $\qq=I-\pp$ and $\bra\p_{\sl\a}\psi_{k}|\qq|\p_{\b\sr}\psi_{k}\ket$
is the Fubini-Study metric corresponding to the eigenstate $|\psi_{k}\ket$.
In words, $\met_{\a\b}$ is the sum of the eigenstate Fubini-Study
metrics weighted by their respective eigenvalues/populations. If $\varrho$
is pure, then it is clear that $\met_{\a\b}$ reduces to the Fubini-Study
metric. Finally, if $\varrho$ is full rank, then $\pp=I$ and $\met_{\a\b}=0$.
This means that the metric is non-zero only for those $\varrho$ which
are not full rank.

\subsection{NS case}

Recall from eq.~(\ref{eq:states-1}) that adiabatic evolution on
the \foreignlanguage{american}{\acs{NS}} is parameterized by the
instantaneous minimal projections
\begin{eqnarray}
\ps^{(s)} & = & \dist(s)\left(\idfsb\ot|\as^{(s)}\kk\bb\ai^{(s)}|\right)\dist^{\dgt}(s)\,,\label{eq:proj-1}
\end{eqnarray}
where $\idfsb(\cdot)=\sum_{\m=0}^{d^{2}-1}|\stb_{\m}\kk\bb\stb_{\m}|\cdot\kk=\iidfsb\cdot\iidfsb$
is the superoperator projection onto the $\xx_{\a}$-independent \foreignlanguage{american}{\acs{DFS}}
reference basis. We remind the reader (see Sec.~\ref{subsec:holonomy})
that the only assumption of such a parameterization is that the state
$|\rout^{(s)}\kk$ is unitarily equivalent (via unitary $\dist$)
to a tensor product of a \foreignlanguage{american}{\acs{DFS}} state
and auxiliary part for all points $s\in[0,1]$ in the path.

We can simplify $\met_{\a\b}$ and show that it is indeed a metric
(more technically, a semi-metric). In the reference basis decomposition
of $\ps$ from eq.~(\ref{eq:proj-1}), the operators $G_{\a}\equiv iS^{\dg}\p_{\a}S$
(with $\dist(s)|\r\kk\equiv|S\r S^{\dg}\kk$) generate motion in parameter
space. After significant simplification, one can express $\met_{\a\b}$
in terms of these generators:
\begin{align}
\met_{\a\b} & =\met_{\a\b}^{(1)}+\met_{\a\b}^{(2)}\label{eq:met}\\
\met_{\a\b}^{(1)} & =2d\bb\iidfsb\ot\as|G_{\sl\a}(I-\iidfsb\ot\ai)G_{\b\sr}\kk\nonumber \\
\met_{\a\b}^{(2)} & =2d\bb G_{(\a}|\idfsb^{\star}\ot\ooo|G_{\b)}\kk\nonumber 
\end{align}
with projection $\idfsb^{\star}$ consisting of only traceless \foreignlanguage{american}{\acs{DFS}}
generators (we set $\stb_{0}=\frac{1}{\sqrt{d}}\iidfsb$),
\begin{equation}
\idfsb^{\star}\equiv\sum_{\m=1}^{d^{2}-1}|\stb_{\m}\kk\bb\stb_{\m}|=\idfsb-|\stb_{0}\kk\bb\stb_{0}|\,,
\end{equation}
and auxiliary superoperator defined (for all auxiliary operators $A$)
as $\ooo(A)\equiv(A-\bb\as|A\kk)\as$. 

The quantity $\met_{\a\b}$ is clearly real and symmetric in $\a,\b$,
so to show that it is a metric, we need to prove positivity ($\pt_{\a}\met_{\a\b}\pt_{\b}\geq0$,
with sum over $\a,\b$ implied, for all vectors $\pt$ in the tangent
space $\tang$ at a point $\xx\in\psp$ \cite{nakahara}). Since $\as$
is positive definite, one can show that the first term in (\ref{eq:met})
\begin{equation}
\pt_{\a}\met_{\a\b}^{(1)}\pt_{\b}=4d\bb O|O\kk\geq0
\end{equation}
with $O=(I-\iidfsb\ot\ai)(G_{\a}\pt_{\a})(\iidfsb\ot\sqrt{\as})$.
For the second term $\met_{\a\b}^{(2)}$, we can see that $\idfsb^{\star}$
is positive semidefinite since it is a projection. We show that $\ooo$
is positive semidefinite by utilizing yet another inner product associated
with open systems \cite{Alicki1976}. First note that
\begin{equation}
\bb A|\ooo|A\kk=\tr\{\as A^{\dg}A\}-\left|\tr\left\{ \as A\right\} \right|^{2}\,.\label{eq:oax}
\end{equation}
Since $\as$ is full-rank, $\bb A|B\kk_{\as}\equiv\tr\{\as A^{\dg}B\}$
is a valid inner product \cite{Alicki1976} and $\bb A|\ooo|A\kk\geq0$
is merely a statement of the Cauchy\textendash Schwarz inequality
associated with this inner product. For Hermitian $A$, (\ref{eq:oax})
reduces to the variance of $\bb A|\as\kk$.

Roughly speaking, the first term $\met_{\a\b}^{(1)}$ describes how
much the \foreignlanguage{american}{\acs{DFS}} and auxiliary parts
mix and the second term $\met_{\a\b}^{(2)}$ describes how much they
leave the $\ulbig$ block while moving in parameter space. For the
\foreignlanguage{american}{\acs{DFS}} case, $\met_{\a\b}^{(2)}=0$
(due to $\ooo=0$ for that case) and the metric reduces to the standard
\foreignlanguage{american}{\acs{DFS}} metric covered in Sec.~\ref{sec:DFS-case}.
For the unique state case, $\met_{\a\b}^{(2)}$ is also zero (due
to $\idfsb^{\star}$ not containing any traceful \foreignlanguage{american}{\acs{DFS}}
elements and thus reducing to zero when $\iidfsb=1$). The mixing
term $\met_{\a\b}^{(2)}$ is thus of course nonzero only in the \foreignlanguage{american}{\acs{NS}}
block case.
\selectlanguage{american}%

\section{Other geometric tensors\foreignlanguage{english}{\label{sec:alt-geom-tensor}}}

\selectlanguage{english}%
In the previous Section, we showed that the anti-symmetric part of
the \ac{QGT} 
\begin{equation}
\geom=\ps\p\ps\p\ps\ps
\end{equation}
corresponds to the curvature $\F$ associated with the adiabatic connection
$\A$ from Ch.~\ref{ch:6}. We thus postulate that this \ac{QGT}
and its corresponding symmetric part should be relevant in determining
distances between adiabatically connected Lindbladian steady states.
However, the story does not end there as there are \textit{two more}
tensorial quantities that can be defined using the steady-state subspace.
The first is an extension of the Fubini-Study metric to non/pseudo-Hermitian
Hamiltonians \cite{Mostafazadeh2007,Mostafazadeh2009,Brody2013,Brody2014}
(\textit{different} from \cite{Nesterov2009}) that can also be generalized
to Lindblad systems; we do not further comment on it here. The second
is the alternative geometric tensor
\begin{equation}
\geom^{\textsf{alt}}=\ps^{\dgt}\p\ps^{\dgt}\p\ps\ps\,,
\end{equation}
which is different from the \ac{QGT} due to $\ps$ not being Hermitian.
We show that $\geom^{\textsf{alt}}$ appears in a bound on the adiabatic
path length for Lindbladian systems, which has traditionally been
used to determine the shortest possible distance between states in
a parameter space $\psp$. Here we introduce the adiabatic path length,
generalize it to Lindbladians, and comment on $\geom^{\textsf{alt}}$.

\subsection{Hamiltonian case}

The \textit{adiabatic path length} for Hamiltonian systems quantifies
the distance between two adiabatically connected states $|\psi_{0}^{(s=0)}\ket$
and $|\psi_{0}^{(1)}\ket$. The adiabatic evolution operator (derived
in Sec.~\ref{subsec:Connection:-non-degenerate-Hamil}) for an arbitrary
path $s\in[0,1]$ and for initial zero-energy state $|\psi_{0}^{(0)}\ket$
is 
\begin{equation}
\uadd^{(1)}=\path\exp\left(\int_{0}^{1}\padd\pad\dd s\right)\,.
\end{equation}
Consider the Frobenius norm (\ref{eq:inprod}) of $\uadd^{(1)}$.
By expanding the definition of the path-ordered exponential, one can
show that $\Vert\uadd^{(1)}\Vert\leq\exp(L_{0})$ with path length
\begin{equation}
L_{0}\equiv\int_{0}^{1}\Vert\padd\pad\Vert\dd s\,.\label{eq:pathlength}
\end{equation}
Remembering that $\Vert A\Vert=\sqrt{\tr\{A^{\dg}A\}}$ and writing
$\p_{s}$ in terms of parameter derivatives, we see that the Fubini-Study
metric appears in the path length:
\begin{equation}
\Vert\padd\pad\Vert^{2}=\half\sum_{\a,\b}M_{\a\b,00}\vv_{\a}\vv_{\b}\,.\label{eq:hammet}
\end{equation}
Therefore, the shortest path between states in Hilbert space projects
to a \textit{geodesic} in parameter space satisfying the Euler-Lagrange
equations associated with the metric $M_{\a\b,00}$ and minimizing
the path length \{e.g., \cite{nakahara}, eq.~(7.58)\} (with sum
implied)
\begin{equation}
L_{0}=\int_{0}^{1}\sqrt{{\textstyle \half}G_{\a\b,00}\vv_{\a}\vv_{\b}}\dd s\,.
\end{equation}
In Hamiltonian systems, the adiabatic path length appears in bounds
on corrections to adiabatic evolution (\cite{Jansen2007}, Thm.~3;
see also \cite{Rezakhani2010}). This path length is also applicable
when one wants to simulate adiabatic evolution in a much shorter time
(\textit{counter-diabatic/superadiabatic dynamics} \cite{Demirplak2003,Lim1991,Vacanti2014}
or \textit{shortcuts to adiabaticity} \cite{Berry2009,Torrontegui2013})
by explicitly engineering the Kato Hamiltonian $i[\padd,\pad]$ from
eq.~(\ref{eq:kato}).

\subsection{Lindbladian case}

The tensor $\geom_{\a\b}^{\textsf{alt}}$ arises in the computation
of the corresponding Lindbladian adiabatic path length
\begin{equation}
L\equiv\int_{0}^{1}\Vert\psd\ps\Vert\dd s\,,\label{eq:path}
\end{equation}
where the superoperator norm of $\psd\ps$ is the analogue of the
operator Frobenius norm from eq.~(\ref{eq:inprod}): $\Vert\oo\Vert\equiv\sqrt{\Tr\{\oo^{\dgt}\oo\}}$
where $\oo$ is a superoperator. This path length provides an upper
bound on the norm of the Lindblad adiabatic evolution superoperator
(\ref{eq:udfs})
\begin{equation}
\U^{(1,0)}=\path\exp\left(\int_{0}^{1}\psd\ps\dd s\right)\,.\label{eq:udfsalt}
\end{equation}
Using properties of norms and assuming one \foreignlanguage{american}{\acs{NS}}
block, it is straightforward to show that 
\begin{equation}
\Vert\U^{(1,0)}\Vert\leq\exp(L)\,\,\,\,\,\,\,\text{with}\,\,\,\,\,\,\,L=\int_{0}^{1}\sqrt{{\textstyle \half}\da\met_{\a\b}^{\textsf{alt}}\vv_{\a}\vv_{\b}}\,\dd s
\end{equation}
(with sum over $\a,\b$ implied). The metric governing this path length
turns out to be 
\begin{equation}
\met_{\a\b}^{\textsf{alt}}=\bb\as|\as\kk\sum_{\a,\b}\geom_{\sl\a\b\sr,\m\m}^{\textsf{alt}}\,.
\end{equation}

For a unique steady state $\varrho$, this alternative metric reduces
to the Hilbert-Schmidt metric
\begin{equation}
\met_{\a\b}^{\textsf{alt}}=\bb\p_{\sl\a}\varrho|\p_{\b\sr}\varrho\kk\,.
\end{equation}
Note the subtle difference between this metric and the \ac{QGT} metric
$\met_{\a\b}=\bb\p_{\sl\a}\pp|\p_{\b\sr}\varrho\kk$ (\ref{eq:metricunique}).
This difference is precisely due to the absence of $\varrho$ in the
left eigenmatrices $J_{\ul}$. For the \ac{QGT} metric, $\varrho$
is never in the same trace twice while for the alternative metric,
the presence of $\ps^{\dgt}$ yields such terms. We note that for
a pure steady state $\varrho=\pp$ (with $\pp$ being rank one), both
metric tensors reduce to the Fubini-Study metric.

Another notable example is the \foreignlanguage{american}{\acs{DFS}}
case ($\as=1$). In that case, $\J_{\ul}^{\m}=\St_{\m}$ \textemdash{}
the \ac{QGT} and alternative tensor become equal ($\geom^{\textsf{alt}}=\geom$).
It is therefore the presence of $\as$ that allows for two different
metrics $\met_{\a\b}$ and $\met_{\a\b}^{\textsf{alt}}$. However,
for the \foreignlanguage{american}{\acs{NS}} case, the ``alternative''
curvature $\geom_{\al\a\b\ar,\m\n}^{\textsf{alt}}$ does not reduce
to the adiabatic curvature $\F_{\a\b,\m\n}$ associated with the connection
$\A_{\a}$ (unlike the \ac{QGT} curvature). How this subtle difference
between $\geom_{\a\b}$ and $\geom_{\a\b}^{\textsf{alt}}$ for the
\foreignlanguage{american}{\acs{NS}} and unique steady state cases
is relevant in determining distances between adiabatic steady states
of Lindbladians should be a subject of future investigation.\selectlanguage{english}%

\inputencoding{latin9}\newpage{}\foreignlanguage{english}{}%
\begin{minipage}[t]{0.5\textwidth}%
\selectlanguage{english}%
\begin{flushleft}
\begin{singlespace}\textit{``At the first of the 1960's Rochester
Coherence Conferences, I suggested that a license be required for
use of the word 'photon', and offered to give such a license to properly
qualified people. My records show that nobody working in Rochester,
and very few other people elsewhere, ever took out a license to use
the word 'photon'.''}\end{singlespace}
\par\end{flushleft}
\begin{flushleft}
\hfill{}\textendash{} Willis E. Lamb
\par\end{flushleft}\selectlanguage{english}%
\end{minipage}

\chapter{Application: driven two-photon absorption\label{ch:7}}

\selectlanguage{english}%
This chapter consists of a detailed investigation of a Hamiltonian-driven
version of the two-photon absorption process from Sec.~\ref{sec34}
(\cite{Hach1994}; \cite{puri}, Sec.~13.2.2). This is also the same
case we discussed in the overview of results in Sec.~\ref{sec:Questions-addressed-and}.
Variants of this case are also manifest in the degenerate parametric
oscillator (\cite{Wolinsky1988}; see also \cite{carmichael2}, eq.~12.10),
a laser-driven trapped ion (\cite{Poyatos1996}, Fig.~2d; see also
\cite{Garraway1998,carvalho2001}), nano-mechanical systems \cite{voje2013},
and superconducting qubit systems \cite{coolforcats,cats,Leghtas2014,Azouit2015,Albert2015,Minganti2016a,Bartolo2016,cohen2016,S.Touzard}
(where this case is colloquially known as the ``two-cat pump'').

\begin{figure}[t]
\centering{}\includegraphics[width=0.45\textwidth]{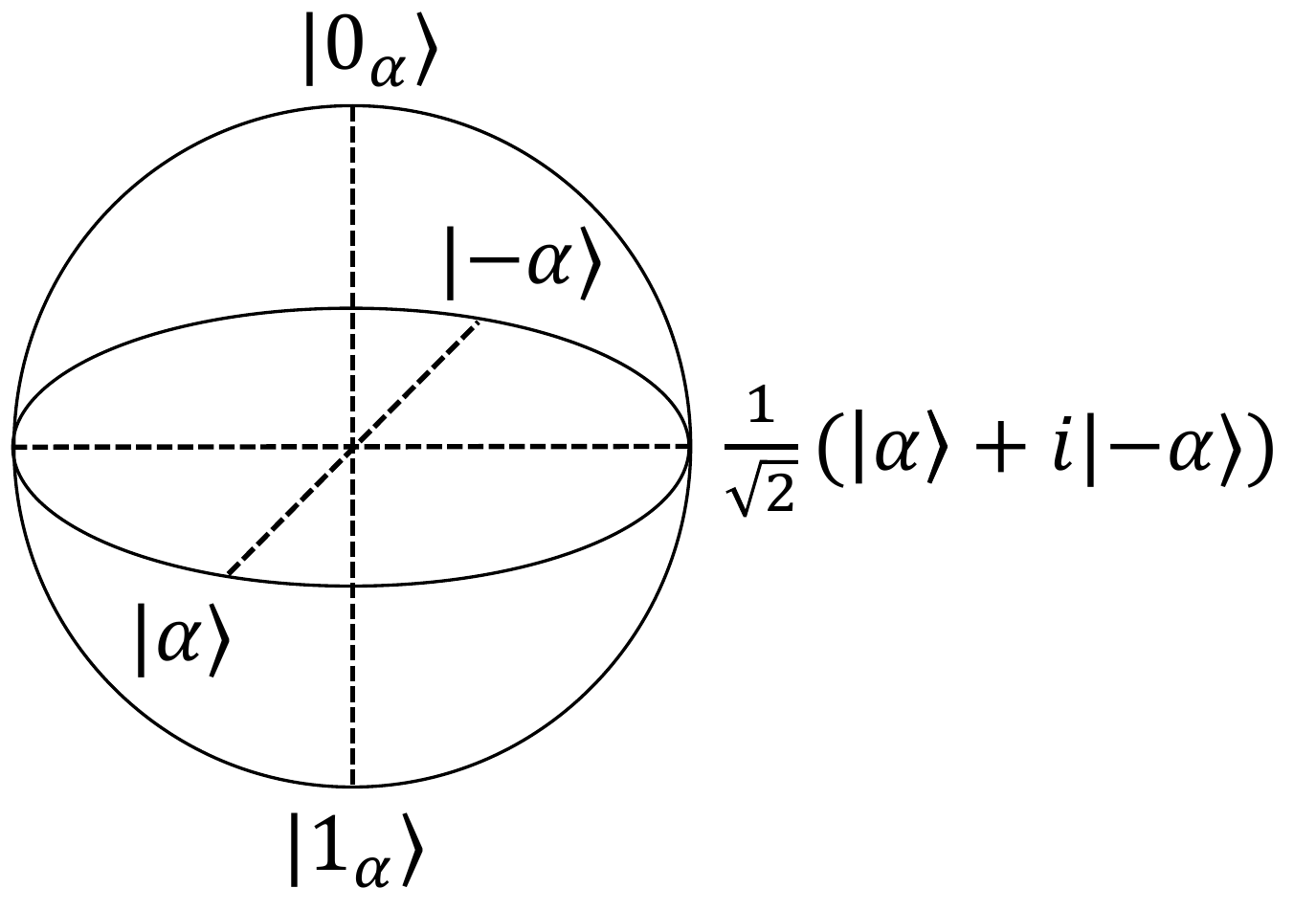}\caption{\foreignlanguage{american}{\label{fig:catschematic}Schematic of the cat-state Bloch sphere in
the large \(\a\) limit. Cat states \(|k_{\a}\ket\), \(k\in\{0,1\}\),
form the \(z\)-axis while coherent states \(\left|\pm\a\right\rangle \)
form the \(x\)-axis.}}
\end{figure}

\selectlanguage{american}%

\section{The Lindbladian and its steady states}

\selectlanguage{english}%
Consider a Lindbladian with a single jump operator 
\begin{equation}
F=\aa^{2}-\a^{2}=\left(\aa-\a\right)\left(\aa+\a\right)\,,
\end{equation}
where $\a\in\mathbb{R}$, $[\aa,\aa^{\dg}]=I$ and $\ph\equiv\aa^{\dg}\aa$.
Due to the ``gauge'' symmetry (\ref{eq:gauge}), this is equivalent
to adding a squeezing Hamiltonian $H=-i\a^{2}(\aa^{2}-\aa^{\dagger2})$
to a Lindbladian with the undriven two-photon absorption jump operator
$F=\aa^{2}$. Unlike the driven case of the two-qubit example from
Sec.~\ref{sec33}, in which driving takes the \foreignlanguage{american}{\acs{DFS}}
into an \foreignlanguage{american}{\acs{NS}}, here the undriven \foreignlanguage{american}{\acs{DFS}}
remains a \foreignlanguage{american}{\acs{DFS}} for all driving parameters.
Recall from Sec.~\ref{sec34} that, for $\a=0$, \acs{ASH} is a
qubit and consists of Fock states $|k\ket$, $k\in\{0,1\}$ (since
$F$ annihilates both). We have also already mentioned in Ch.~\ref{ch:1}
that, for large enough $\a$, \acs{ASH} remarkably retains its qubit
form, which this time is spanned by superpositions of coherent states
$\ct{\pm\a}$. Here, a treatment is given which is valid for all $\a$.
One may have noticed that both states $\ct{\pm\a}$ go to $|0\ket$
in the $\a\rightarrow0$ limit and do not reproduce the $\a=0$ steady
state basis. This issue is resolved by introducing the cat state basis
\cite{dodonov1974} 
\begin{equation}
|k_{\a}\ket\equiv\frac{e^{-\half\a^{2}}}{\sqrt{\nn_{k}}}\sum_{n=0}^{\infty}\frac{\a^{2n+k}}{\sqrt{(2n+k)!}}|2n+k\ket\sim\begin{cases}
|k\ket & \a\rightarrow0\\
\frac{1}{\sqrt{2}}(|\a\ket+(-)^{k}\ct{-\a}) & \a\rightarrow\infty
\end{cases}\label{eq:twocat}
\end{equation}
with normalization $\nn_{k}\equiv\half[1+(-)^{k}\exp(-2\a^{2})]$.
As $\a\rightarrow0$, cat states approach Fock states while for $\a\rightarrow\infty$,
the cat states (exponentially) quickly become ``macroscopic'' superpositions
of $\ct{\pm\a}$. This Lindbladian thus has \textit{only two distinct
parameter regimes}: one in which coherent states come together ($\a\approx1$)
and one in which they are well-separated ($\a\gg1,$ or more practically
$\a\apprge2$). Eq.~(\ref{eq:twocat}) shows that (for large enough
$\a$) cat states and coherent states become conjugate $z$- and $x$-bases
respectively, forming the \acs{ASH} qubit (see Fig.~\ref{fig:catschematic}).
Using projections $\Pi_{k}=\sum_{n=0}^{\infty}|2n+k\ket\bra2n+k|$
(\ref{eq:catproj2}), cat states can be concisely written as projected
(and normalized) coherent states:
\begin{equation}
|k_{\a}\ket\equiv\frac{\Pi_{k}|\a\ket}{\sqrt{\bra\a|\Pi_{k}|\a\ket}}\,\,\,\,\text{with normalization\,\,\,\,\ensuremath{\nn_{k}\equiv\bra\a|\Pi_{k}|\a\ket=\frac{1+(-)^{k}e^{-2\a^{2}}}{2}}}\,.\label{eq:catsintro}
\end{equation}
The projections are orthogonal: $\Pi_{k}\Pi_{l}=\d_{kl}^{\text{mod\,}2}\Pi_{k}$,
\foreignlanguage{american}{where }$\d_{qp}^{\text{mod\,}2}=1$ whenever
$q=p\text{\,mod\,}2$. Action of lowering or raising operators switches
subspaces \foreignlanguage{american}{{[}see eq.~(\ref{eq:catrelproj})}{]},
implying that
\begin{equation}
a\Pi_{k}=\Pi_{k+1\text{mod\,}2}a\,.\label{eq:catrelproj2}
\end{equation}

The cat state label $k\in\{0,1\}$ corresponds to the respective $\pm1$
eigenspace of the parity operator $\left(-\right)^{\ph}=\Pi_{0}-\Pi_{1}$.
This parity operator commutes with $F$ for all $\a$, so \acs{OPH}
is split into four blocks, 
\begin{equation}
\{|2n+k\ket\bra2m+l|\}_{n,m=0}^{\infty}\,\,\,\,\,\,\,\,\,\,\,\text{(labeled by \ensuremath{k,l\in\{0,1\}})}\,,\label{eq:blocks}
\end{equation}
which evolve independently of each other. The outer product $\St_{kl}\equiv|k_{\a}\ket\bra l_{\a}|$
is the unique steady-state basis element in the respective block $\{|2n+k\ket\bra2m+l|\}_{n,m=0}^{\infty}$,
and together the basis elements $\{\St_{kl}\}_{k,l=0}^{1}$ span $\ash=\ulbig$.
Outer products of all states orthogonal to $|k_{\a}\ket$ span the
decaying subspace $\lrbig$.

The cat state basis, unlike the coherent state basis, is orthonormal
for all values of $\a$ and simplifies most of the calculations done
here, with all of the complexity coming from the normalization factors
$\nn_{k}$. \foreignlanguage{american}{For example, using eqs.~(\ref{eq:catsintro}-\ref{eq:catrelproj2}),
orthogonality of projections, and the property of coherent states
$a|\a\ket=\a|\a\ket$}, the cat states have average occupation number
\begin{equation}
\bra k_{\a}|\ph|k_{\a}\ket=\a\frac{\bra\a|\Pi_{k}a^{\dg}\Pi_{k+1}|\a\ket}{\pi_{k}}=\a^{2}\frac{\pi_{k+1}}{\pi_{k}}=\begin{cases}
k+O(\a^{4}) & \a\rightarrow0\\
\a^{2}+O(\a^{2}e^{-2\a^{2}}) & \a\rightarrow\infty
\end{cases}\,.\label{eq:technical}
\end{equation}
This is sensible since Fock states have distinct average occupation
numbers while coherent states with the same magnitude $\a$ have the
same average occupation number.
\selectlanguage{american}%

\section{Conserved quantities}

\selectlanguage{english}%
We now search for the four conserved quantities corresponding to $\St_{kl}$.
By the correspondence from Thm.~\ref{thm:dual}, there exist four
$\{J^{kl}\}_{k,l=0}^{1}$ such that
\begin{equation}
\L^{\dgt}(J^{kl})=F^{\dg}J^{kl}F-\half\{F^{\dg}F,J^{kl}\}=0\,.
\end{equation}
Since parity symmetry is preserved, the diagonal $k=l$ conserved
quantities remain the same as for the $\a=0$ case: $J^{kk}=\Pi_{k}$.
One can use Thm.~\ref{prop:3} to determine the remaining conserved
quantity $J^{01}$. However, since $F_{\lr}\neq0$, inverting $\L_{\lr}$
is non-trivial. Fortunately, this inversion can be avoided and we
can use the $\a=0$ conserved quantity (\ref{eq:cctwophot}; now renamed
to $J^{01,q=0}$) to determine $J^{01}$. To do so, we apply $\L^{\dgt}$
to $J^{01,0}$, which yields nonzero terms only from the $\a$-dependent
part of $\L^{\dgt}$ (since $J^{01,0}$ is conserved under the $\a$-independent
part). These nonzero terms, which we call $J^{01,q=\pm1}$, can in
turn be plugged into $\L^{\dgt}$ themselves. Such recursive steps
produce a pattern: the quantities $J^{01,q}$ (labeled by $q\in\mathbb{Z}$)
turn out to be
\begin{equation}
J^{01,q}=\begin{cases}
\frac{\left(\ph-1\right)!!}{\left(\ph+2q\right)!!}\Pi_{0}a^{2q+1} & \,\,\,\,\,\,\,\,q\geq0\\
\Pi_{0}a^{\dg2\left|q\right|-1}\frac{\ph!!}{\left(\ph+2\left|q\right|-1\right)!!} & \,\,\,\,\,\,\,\,q<0
\end{cases}
\end{equation}
and the equation of motion they satisfy is
\begin{equation}
\L^{\dgt}\left(J^{01,q}\right)=\half\left(2q+1\right)\left[\a^{2}\left(J^{01,q-1}-J^{01,q+1}\right)-2qJ^{01,q}\right]\,.\label{eq:eomjq}
\end{equation}
Recall that we are looking for a conserved quantity $J^{01}$ such
that $\L^{\dgt}(J^{01})=0$. Since $J^{01}\rightarrow J^{01,q=0}$
for $\a\rightarrow0$ and since the set $\{J^{01,q}\}_{q\in\mathbb{Z}}$
is closed under application of $\L^{\dgt}$, $J^{01}$ for any $\a$
must be constructed out of the $J^{01,q}$'s:
\begin{equation}
J^{01}\propto\sum_{q\in\mathbb{Z}}a_{q}J^{01,q}\,,
\end{equation}
with some coefficients $a_{q}$. Determining these coefficients becomes
easy when one notices that the equations of motion for $J^{01,q}$
mimic the recurrence relation
\begin{equation}
\a^{2}\left[I_{q-1}\left(\a^{2}\right)-I_{q+1}\left(\a^{2}\right)\right]+2qI_{q}\left(\a^{2}\right)=0
\end{equation}
satisfied by the modified Bessel functions of the first kind $I_{q}$
\cite{dlmf}. Taking care of the factor of $2q+1$ in eq.~(\ref{eq:eomjq})
and an extra $q$-dependent sign yields
\begin{equation}
a_{q}=\frac{\left(-\right)^{q}}{2q+1}I_{q}\left(\a^{2}\right)\,.
\end{equation}
Now all that is left is to biorthogonalize the $J^{01}$ with its
corresponding \acs{ASH} basis element $\St_{01}$, i.e., make sure
that $\bra0_{\a}|J^{01\dg}|1_{\a}\ket=1$. Explicitly calculating
\begin{equation}
\bb\J^{01,q}|\St_{01}\kk=\bra0_{\a}|\J^{01,q\dg}|1_{\a}\ket=\sqrt{\frac{2\a^{2}}{\sinh2\a^{2}}}I_{q}\left(\a^{2}\right)\label{eq:normq}
\end{equation}
and using eq.~(5.8.6.2) from \cite{prudnikov},
\begin{equation}
\sum_{q\in\mathbb{Z}}\frac{\left(-\right)^{q}}{2q+1}I_{q}\left(\a^{2}\right)I_{q}\left(\a^{2}\right)=\frac{\sinh2\a^{2}}{2\a^{2}}\,,
\end{equation}
we obtain the properly normalized conserved quantity
\begin{equation}
J^{01}=\sqrt{\frac{2\a^{2}}{\sinh2\a^{2}}}\sum_{q\in\mathbb{Z}}\frac{\left(-\right)^{q}}{2q+1}I_{q}\left(\a^{2}\right)J^{01,q}\,.\label{eq:j01-1}
\end{equation}
One can check that $\L^{\dgt}(J^{01})=0$ as follows. First, use linearity
and the equation of motion (\ref{eq:eomjq}) for $J^{01,q}$. Then,
observe that each $J^{01,q}$ is supported on a different set of Fock
state outer products ($\{|2n\ket\bra2n+2q+1|\}_{n=0}^{\infty}$ for
$q\geq0$ and $\{|2n+2|q|\ket\bra2n+1|\}_{n=0}^{\infty}$ for $q<0$).
This means that the coefficient in front of each $J^{01,q}$ must
be zero for $J^{01}$ to be conserved. Rearranging the three infinite
sums (coming from $J^{01,q}$, $J^{01,q+1}$, and $J^{01,q-1}$) in
order to obtain that coefficient yields exactly the Bessel function
recursion relation above.
\selectlanguage{american}%

\section{State initialization\label{sec:State-initialization}}

\selectlanguage{english}%
We now determine the asymptotic state
\begin{equation}
\rout=\sum_{k,l=0}^{1}c_{kl}\St^{kl}=\sum_{k,l=0}^{1}c_{kl}|k_{\a}\ket\bra l_{\a}|\label{eq:purcats}
\end{equation}
starting from an initial coherent state $\rin=|\b\ket\bra\b|$. By
the correspondence from Thm.~\ref{thm:dual}, we know that $c_{kl}=\bb J^{kl}|\rin\kk$. 

\subsection{Steady state for an initial fixed-parity state}

Due to the decoupling of the blocks $\{|2n+k\ket\bra2m+l|\}_{n,m=0}^{\infty}$
(\ref{eq:blocks}), any state which starts exclusively in one of the
blocks evolves within that block into the block's fixed point $|k_{\a}\ket\bra l_{\a}|$.
Therefore, if we start in \textit{any} state of fixed parity $k\in\{0,1\}$
{[}i.e., $\left(-\right)^{\ph}=\left(-\right)^{k}$ for that state{]},
we necessarily converge to the \textit{pure} asymptotic state $\rout=|k_{\a}\ket\bra k_{\a}|$.
This holds true for mixed fixed-parity initial states as well, which
is an example of the environment (which induces this Lindbladian)
taking entropy out of the system.

\subsection{Steady state for an initial coherent state}

Now let $\rin=|\b\ket\bra\b|$ for some $\b\in\mathbb{C}$. The diagonal
quantities $c_{kk}=\bb J^{kk}|\rin\kk$ have already been determined
in eq.~(\ref{eq:diagqtys}). The tricky part is the off-diagonal
quantity, which simplifies to
\begin{equation}
c_{01}=\bra\b|J^{01\dg}|\b\ket=\frac{i\a\b^{\star}e^{-|\b|^{2}}}{\sqrt{2\sinh2\a^{2}}}\int_{\phi=0}^{\pi}d\phi e^{-i\phi}I_{0}\left(\left|\a^{2}-\b^{2}e^{2i\phi}\right|\right)\,.
\end{equation}
To derive this, we first apply eq.~(\ref{eq:j01-1}) to obtain the
sum
\begin{equation}
c_{01}=\frac{\sqrt{2}\a\b^{\star}e^{-|\b|^{2}}}{\sqrt{\sinh2\a^{2}}}\sum_{q\in\mathbb{Z}}\frac{\left(-\right)^{q}}{2q+1}I_{q}\left(\a^{2}\right)I_{q}\left(\left|\b\right|^{2}\right)e^{-i2q\t}\,,\label{eq:intermediate}
\end{equation}
where $\t=\arg\b$. This sum is convergent because the sum without
the $2q+1$ term is an addition theorem for $I_{q}$ {[}eq.~(5.8.7.2)
from \cite{prudnikov}{]}. To put the above into integral form, we
use the identity (derivable from the addition theorem) 
\begin{equation}
I_{q}\left(\a^{2}\right)I_{q}\left(\left|\b\right|^{2}\right)=\frac{1}{2\pi}\int_{\phi=0}^{2\pi}d\phi e^{iq\left(\phi+\pi\right)}I_{0}\left(\left|\a^{2}-\left|\b\right|^{2}e^{i\phi}\right|\right)\,.
\end{equation}
Plugging in the above identity into eq.~(\ref{eq:intermediate}),
interchanging the sum and integral (possible because of convergence),
evaluating the sum (which is a simple Fourier series), and performing
a change of variables yields the integral formula for $c_{01}$.

Using eq.~(5.8.1.15) from \cite{prudnikov}, one can calculate limits
for large $\left|\b\right|$ along the real and imaginary axes in
phase space of $\b$: 
\begin{equation}
\lim_{\b\rightarrow\infty}c_{01}=\frac{1}{2}\frac{\textrm{erf}(\sqrt{2}\a)}{\sqrt{1-e^{-4\a^{2}}}}\overset{_{\a\rightarrow\infty}}{\longrightarrow}\half\qquad\textrm{ and }\qquad\lim_{\b\rightarrow i\infty}c_{01}=-i\frac{1}{2}\frac{\textrm{erfi}(\sqrt{2}\a)}{\sqrt{e^{4\a^{2}}-1}}\overset{_{\a\rightarrow\infty}}{\longrightarrow}0\,,
\end{equation}
where erf and erfi are the error function and imaginary error function,
respectively. Recalling that $c_{kk}\rightarrow\half$ in both limits
(see Sec.~\ref{subsec:Steady-state-for}), we see that $\rout$ becomes
pure when $\b$ is real and large and that $\rout$ becomes maximally
mixed when $\b$ is pure imaginary and large. To study the remaining
sectors of $\b$ phase space, we numerically calculate the purity
for a lattice of $\b$'s in Fig.~\ref{f:catfinal} for $\a$ being
$0.001$, $\half$, $1$, and $5$. The rightmost panel shows the
behavior for large $\a$, showing that initial states $\b$ near the
respective steady states $\ct{\pm\a}$ converge to pure states. In
fact, one can show that those pure states are exactly $\ct{\pm\a}$.
In other words, the two-photon system is similar to a classical double-well
system in the combined large $\a,\b$ regime. However, starting in
the state $\propto|\b\ket+\ct{-\b}$ for any $\b$ guarantees a pure
asymptotic state by the symmetry arguments of the previous Subsection.
Therefore, while a ``classical'' initial state $|i\a\ket$ results
in a maximally mixed asymptotic state (in the large $\a$ limit),
the $\ct{-i\a}$ component in an initial cat state $\propto|i\a\ket+\ct{-i\a}$
cancels that effect and results in a pure asymptotic state!

\begin{figure}
\includegraphics[width=1\textwidth]{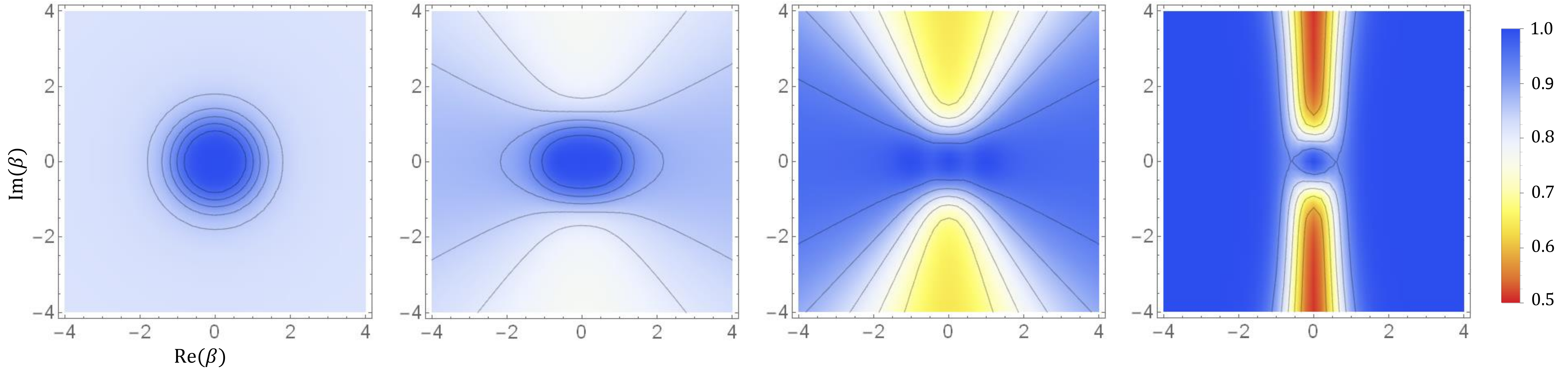}\caption{\foreignlanguage{american}{\label{f:catfinal}Purity of the asymptotic state of the driven two-photon
absorption Lindbladian. The four panels correspond to \foreignlanguage{english}{\(\a\)
being $0.001$, \(\half\), $1$, and $5$, respectively.} For each
panel, a point \(\b\) in phase space corresponds to the purity (\(\tr\{\rout^2\}\))
of the asymptotic state \(\rout\) (\ref{eq:purcats}) given an initial
coherent state \(|\beta\ket\). Besides \(\rout\) being pure away
from the vertical axis when \(\a\gg1\), one can observe that \(\rout\)
is also pure for initial states near the center of phase space. Indeed,
starting in the vacuum state (\(\beta=0\)), a fixed-parity state,
the system is driven to the pure Schrödinger cat state \(|0_\a\ket\)
(\ref{eq:catsintro}).}}
\end{figure}

\subsection{Steady state for an initial cat state}

Now let us briefly consider an initial state proportional to $|\a\ket+e^{i\t}|\b\ket$
with $\b\neq-\a$, i.e., a cat state in which one component is already
in \acs{ASH}. For simplicity, let us consider the large $\a$ limit,
meaning that all we say is true up to exponentially small corrections
due to the overlap between coherent states. For this case, it is useful
to consider the four-corners decomposition, in which $\ash=\ulbig$
is the cat-state subspace with projection 
\begin{equation}
P=|0_{\a}\ket\bra0_{\a}|+|1_{\a}\ket\bra1_{\a}|\overset{_{^{\a\rightarrow\infty}}}{\,\sim\,}|\a\ket\bra\a|+\ct{-\a}\cb{-\a}\,.\label{eq:catproj}
\end{equation}
We have seen above that initial states $|\b\ket$ which are much closer
to $|\a\ket$ than they are to $\ct{-\a}$ (i.e., $\left|\b+\a\right|\gg\left|\b-\a\right|$)
converges to $|\a\ket$. The same can be said of $\ct{-\a}$ , the
other ``well'' in this (approximately) double-well system. Let us
assume that $\b$ is much closer to $-\a$ so that both components
in the initial cat state do not converge to the same well. However,
we keep in mind that, in this approximation, $\bra\b|\a\ket\approx0$,
so $|\b\ket$ is still outside of both wells. This is a case in which
$\rin$ contains components in all four corners of \acs{OPH}:
\begin{equation}
\rin=\begin{pmatrix}(\rin)_{\ul} & (\rin)_{\ur}\\
(\rin)_{\ll} & (\rin)_{\lr}
\end{pmatrix}=\begin{pmatrix}|\a\ket\bra\a| & e^{i\t}|\a\ket\bra\b|\\
e^{-i\t}|\b\ket\bra\a| & |\b\ket\bra\b|
\end{pmatrix}\,.
\end{equation}
Due to Thm.~\ref{prop:3}, we know that the asymptotic projection
$\ppp=\ppp\R_{\di}$, meaning that coherences $\ofbig$ are not preserved
in the infinite-time limit. In the language of conserved quantities,
$J_{\of}^{kl}=0$. This means that $\t$ is not imprinted on $\rout$.
Moreover, since the component in $\lrbig$ converges to a different
location in $\ulbig$ than the component already in $\ulbig$, we
necessarily have a mixed asymptotic state ($\rout=\half P$).
\selectlanguage{american}%

\section{Ordinary perturbation theory\label{sec:Ordinary-perturbation-theory}}

Let us know apply the first-order perturbation theory developed in
Ch.~\ref{ch:4} to study the behavior of the cat-state \acs{ASH}
under both Hamiltonian and Lindbladian perturbations. We learn how
to induce induce unitary evolution within \acs{ASH} using Hamiltonians
and that the driven two-photon absorption Lindbladian suppresses the
effect of some (but not all) noise. In process, we apply Thm.~\ref{prop:3},
which greatly simplifies the calculations. All studied effects are
verified numerically in Ref.~\cite{cats}, and the Hamiltonian perturbation
calculations here offer another way to get to the same answers.

\subsection{A Hamiltonian-based gate\label{subsec:A-Hamiltonian-based-gate}}

Recall from eq.~(\ref{eq:kubo-simplified}) that first-order response
of a state in \acs{ASH} due to a slowly ramped-up perturbation $\oo$
is
\begin{equation}
\T_{t}^{\left(1\right)}|\rout\kk=t\ppp\oo\ppp|\rout\kk-\L^{-1}\oo|\rout\kk\,,\label{eq:pert-main-form}
\end{equation}
where we have omitted the ``infinity'' which occurs within \acs{ASH}
(due to the slow ramp-up of the perturbation) since it doesn't affect
our conclusions (see footnote \ref{fn:ringdown} in Ch.~\ref{ch:4}).
Recall also from Sec.~\ref{subsec:Relation-to-previous} and the
References therein that if we rescale the perturbation as $\spert\rightarrow\frac{1}{T}\spert$
and evolve to a time $t=T$, then the ``Zeno term'' $\ppp\spert\ppp$
is order $O\left(1\right)$ and dominates the $O(\nicefrac{1}{T})$
leakage term as $T\rightarrow\infty$. We use this effect to induce
a Hamiltonian-based gate on \acs{ASH}. 

\paragraph{Evolution within $\protect\ash$}

According to Sec.~\ref{subsec:hams}, $\ppp\spert\ppp$ is of Hamiltonian
form and, for \acs{DFS} cases, reduces to
\begin{equation}
\ppp\spert\ppp=-i[V_{\ul},\cdot]\equiv-i[\pp\hpert\pp,\cdot]\,,
\end{equation}
where $\pp$ (\ref{eq:catproj}) is the projection onto the cat subspace.
In other words, while $\hpert$ can in general drive states in $\ash=\ulbig$
out of \acs{ASH}, $V_{\lr}\neq0$, that part of $\hpert$ does \textit{not}
contribute to within first order in the perturbation. Consider the
perturbative Hamiltonian
\begin{equation}
V=i\b\left(a^{\dagger}e^{-i\theta}-ae^{i\theta}\right)
\end{equation}
with $\beta\in\mathbb{R}$ and an angle $\theta=[0,2\pi)$. After
writing out both projections in $V_{\ul}=\pp\hpert\pp$ in terms of
cat states, we need to calculate the matrix elements of $\hpert$
in \acs{ASH}, i.e., $\bra k_{\a}|V|l_{\a}\ket$. A calculation similar
to the one from eq.~(\ref{eq:technical}) yields
\begin{align}
\bra k_{\a}|V|l_{\a}\ket & =i\d_{l,k+1}^{\text{mod\,}2}\a\b\left(\sqrt{\frac{\nn_{k}}{\nn_{k+1}}}e^{-i\theta}-e^{i\theta}\sqrt{\frac{\nn_{k+1}}{\nn_{k}}}\right)\,.\label{eq:mat-elem-ham}
\end{align}
\foreignlanguage{english}{To make sense of these matrix elements,
we consider the small and large $\a$ limits. Plugging the expansions
from }eq.~(\ref{eq:technical}) \foreignlanguage{english}{into $V_{\ul}$
yields the two cases
\begin{equation}
V_{\ul}\sim\begin{cases}
i\a\b e^{-i\theta}|0\ket\bra1|+H.c. & \a\rightarrow0\\
2\a\b\sin\theta\left(|0_{\a}\ket\bra1_{\a}|+H.c.\right) & \a\rightarrow\infty
\end{cases}\,.
\end{equation}
For small $\a$, $V_{\ul}$ is a rotation on \acs{ASH}, which is
now spanned by outer products of Fock states $|k\ket$ ($k\in\{0,1\}$),
and the axis of the rotation is determined by $\theta$. For large
$\a$, $V_{\ul}$ is also a rotation, but its axis is fixed to be
the $x$-axis of the cat-qubit Bloch sphere from Fig.~\ref{fig:catschematic}
and only its \textit{strength} is dependent on $\theta$. For maximum
effect in this limit, $\theta$ needs to be $\nicefrac{\pi}{2}$,
which translates to driving perpendicular to the horizontal line connecting
$\a$ and $-\a$ in the phase space of the oscillator (\cite{cats},
Fig.~3). Graphically, such a gate shifts the fringes in the Wigner
function of a cat state and produces the same effect as the holonomic
gate from Fig.~\ref{f1}a. This gate was realized experimentally
in Ref. \cite{S.Touzard}.}

\selectlanguage{english}%
\begin{figure}
\begin{centering}
\includegraphics[width=0.5\columnwidth]{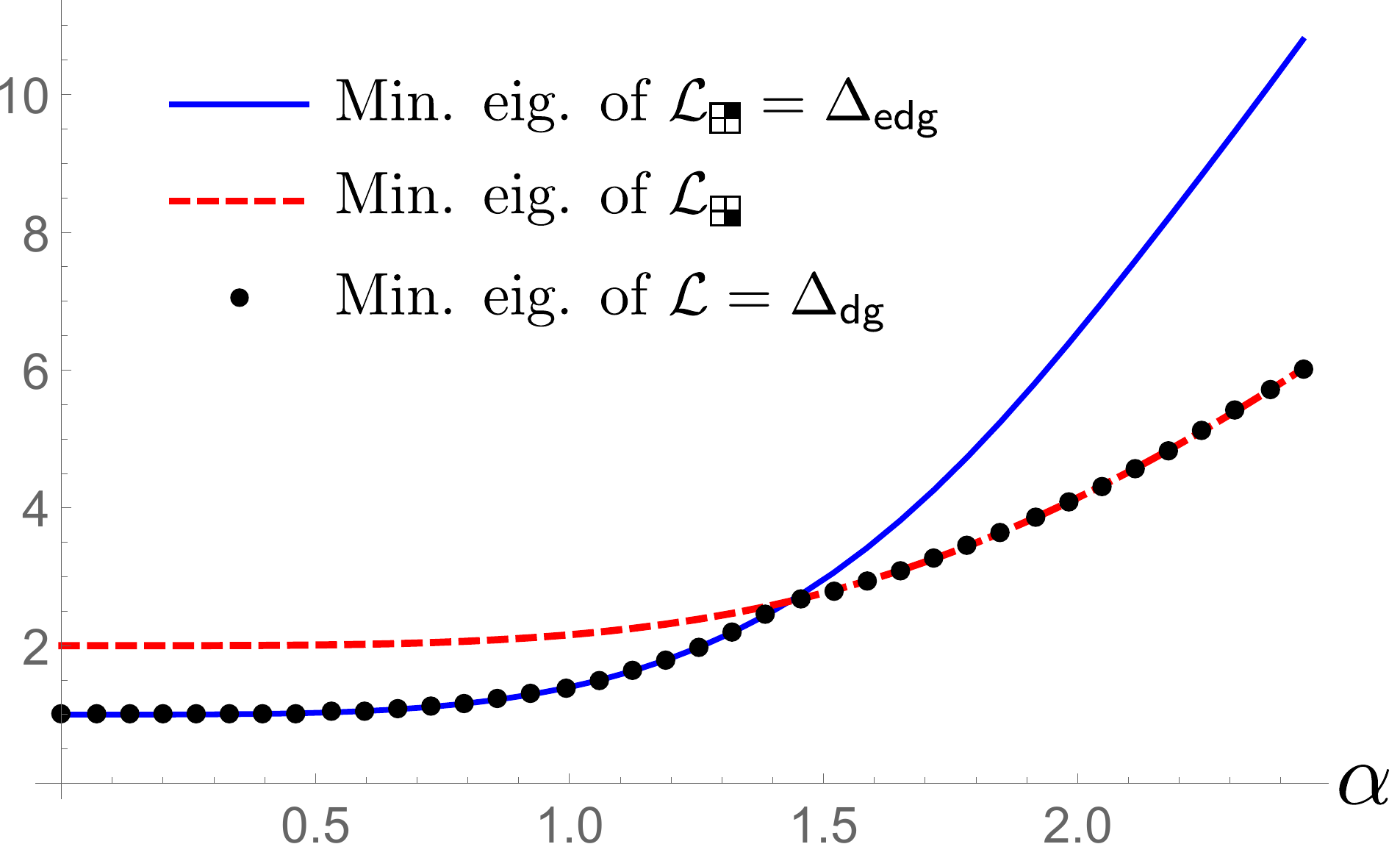}
\par\end{centering}
\caption{\label{fig:gap}\protect\ytableausetup{boxsize = 2pt}Plot of the
effective dissipative gap \(\Delta_{\textnormal{\textsf{edg}}}\),
the nonzero eigenvalue of \(\L\ytableaushort{ {*(black)} {*(black)} , {*(black)} {} }\)
with smallest real part, and the dissipative gap \(\Delta_{\textnormal{\textsf{dg}}}\)
(\ref{eq:dissipative-gap-def}) versus \(\alpha\) for the Lindbladian
with jump operator \(F=\aa^{2}-\a^{2}\). One can see that \(\Delta_{\textnormal{\textsf{edg}}}\geq\Delta_{\textnormal{\textsf{dg}}}\). }
\end{figure}

\selectlanguage{american}%

\paragraph{Leakage out of $\protect\ash$}

\selectlanguage{english}%
Let us now return to the leakage term $\L^{-1}\spert$. Since we have
a \foreignlanguage{american}{\acs{DFS}} case, we can apply the example
from Sec.~\ref{subsec:Example:-decoherence-Hamiltonian}. The leakage
caused by a Hamiltonian perturbation is then restricted to be in $\ofbig$,
$\L_{\of}(\r)=-\{(\hdg)_{\lr},\r_{\of}\}$, and the decoherence Hamiltonian
is
\begin{equation}
\hdg\equiv\half F^{\ell\dg}F^{\ell}=\half\left[\ph(\ph-1)-\a^{2}(\aa^{2}+\aa^{\dg2})+\a^{4}\right]\,.
\end{equation}
The ground states of $\hdg$ are exactly the cat states $|k_{\a}\ket$,
meaning that this Hamiltonian provides another way to stabilize such
states \cite{Cochrane1999,puri2016}. Moreover, $\adg$ is the excitation
gap of $\hdg$. It turns out that, for large enough $\a$, $\adg$
is larger than the dissipative gap $\dgg$ of the full $\L$, thereby
providing another layer of protection against leakage besides the
$T\rightarrow\infty$ Zeno limit. The excitation gap of $\hdg$ ($\adg$)
is plotted in Fig.~\ref{fig:gap} vs. $\a$, along with $\dgg$ and
the eigenvalue of $\L_{\lr}$ with smallest real part. One can see
that for $\a>1.5$, the dissipative gap of $\L$ is smaller and does
not coincide with the energy scale governing leakage.
\selectlanguage{american}%

\subsection{Passive protection against dephasing noise}

Let us now take a look at a perturbation of Lindblad form and consider
the $\ppp\oo\ppp$ (\ref{eq:pert-main-form}) for
\begin{equation}
\oo(\r)=\half\k\left(2\ph\r\ph-\left\{ \ph^{2},\r\right\} \right)\,.\label{eq:dephasing-noise}
\end{equation}
The above Lindbladian occurs when there are fluctuations in the frequency
parameter of the oscillator and is called the dephasing error channel
or, more colloquially, \textit{dephasing noise}. The term $\ppp\oo\ppp$
governs evolution within \acs{ASH}, which no longer has to be unitary
since $\oo$ is not in Hamiltonian form. Since we cannot reduce this
term to an operator like we did in the previous Subsection, we have
to consider the full superoperator and decompose the asymptotic projection
in terms of steady states and conserved quantities,
\begin{equation}
\ppp=\sum_{k,l=0}^{1}|\St_{kl}\kk\bb\J^{kl}|\,.
\end{equation}
That way, $\ppp\oo\ppp$ is determined by the 16 matrix elements of
$\oo$ within \acs{ASH}, $\bb\J^{kl}|\oo|\St_{pq}\kk$ for $k,l,p,q\in\{0,1\}$.
Luckily, dephasing noise \foreignlanguage{english}{preserves parity,
so it also does not couple the blocks $\{|2n+k\ket\bra2m+l|\}_{n,m=0}^{\infty}$
(\ref{eq:blocks}). Therefore, we can immediately say that $\bb\J^{kl}|\oo|\St_{pq}\kk\propto\d_{kp}\d_{lq}$.}
Moving $\oo$ to act on the $\J$'s and using $\St_{kl}^{\dg}=\St_{lk}$
(same for $J^{kl}$) and $\oo^{\dgt}=\oo$, we can instead consider
how $\oo$ acts on the conserved quantities:
\begin{equation}
\bb\J^{kl}|\oo|\St_{kl}\kk=\tr\{\J^{kl\dg}\oo(\St_{kl})\}=\tr\{\oo^{\dgt}(\J^{lk})\St_{kl}\}=\bb\St_{lk}|\oo|\J^{lk}\kk\,.
\end{equation}
\foreignlanguage{english}{Since $\oo(e^{i\pi\ph})=0$, the diagonal
conserved quantities $J^{kk}$ remain conserved.} Moreover, since
$J^{10}=J^{01\dg}$ and $\oo(J^{\dg})=[\oo(J)]^{\dg}$ for any $\oo$
in Lindblad form, we need only to determine the effect of $\oo$ on
$J^{01}$. Recall that $J^{01}$ (\ref{eq:j01-1}) is a superposition
of $J^{01,q}$'s for $q\in\mathbb{Z}$, which in turn are composed
of superpositions of \foreignlanguage{english}{$\{|2n\ket\bra2n+2q+1|\}_{n=0}^{\infty}$
for $q\geq0$ and $\{|2n+2|q|\ket\bra2n+1|\}_{n=0}^{\infty}$ for
$q<0$. Applying $\oo$ to each $J^{01,q}$ and simplifying yields
the simple equation
\begin{equation}
\oo|J^{01,q}\kk=-\half\k\left(2q+1\right)^{2}|J^{01,q}\kk\,,
\end{equation}
signaling that $J^{01,q}$ are actually eigenstates of $\oo$. One
of the $2q+1$ terms cancels the $2q+1$ in the denominator of the
sum used to write $J^{01}$ in terms of $J^{01,q}$'s. The matrix
element is then\begin{subequations}
\begin{eqnarray}
\bb\St_{01}|\oo|\J^{01}\kk & = & \sqrt{\frac{2\a^{2}}{\sinh2\a^{2}}}\sum_{q\in\mathbb{Z}}\frac{\left(-\right)^{q}}{2q+1}I_{q}\left(\a^{2}\right)\bb\St_{01}|\oo|J^{01,q}\kk\\
 & = & -\half\k\sqrt{\frac{2\a^{2}}{\sinh2\a^{2}}}\sum_{q\in\mathbb{Z}}\left(-\right)^{q}\left(2q+1\right)I_{q}\left(\a^{2}\right)\bb\St_{01}|J^{01,q}\kk\\
 & = & -\half\k\frac{2\a^{2}}{\sinh2\a^{2}}\sum_{q\in\mathbb{Z}}\left(-\right)^{q}\left(2q+1\right)I_{q}\left(\a^{2}\right)I_{q}\left(\a^{2}\right)\\
 & = & -\frac{\k\a^{2}}{\sinh2\a^{2}}\,,\label{eq:final-pert-cat}
\end{eqnarray}
\end{subequations}where in the last two steps we used eq.~(\ref{eq:normq})
and eq.~(5.8.7.2) from Ref.~\cite{prudnikov}, respectively. In
the Zeno limit discussed above, the quantity $c_{01}=\bb J^{01}|\rout\kk$
representing the coherence of the cat qubit decays exponentially at
the rate $\bb\St_{01}|\oo|\J^{01}\kk$. For small $\a$, the rate
reduces to the usual dephasing rate $\nicefrac{\k}{2}$ induced on
the Fock state outer product $|0\ket\bra1|$ by $\oo$. However, for
large $\a$, \textit{the rate itself} is exponentially suppressed
for large $\a$ since $\bb\St_{01}|\oo|\J^{01}\kk\sim-2\k\a^{2}e^{-2\a^{2}}$. }

\selectlanguage{english}%
\begin{figure}
\begin{centering}
\includegraphics[width=0.5\textwidth]{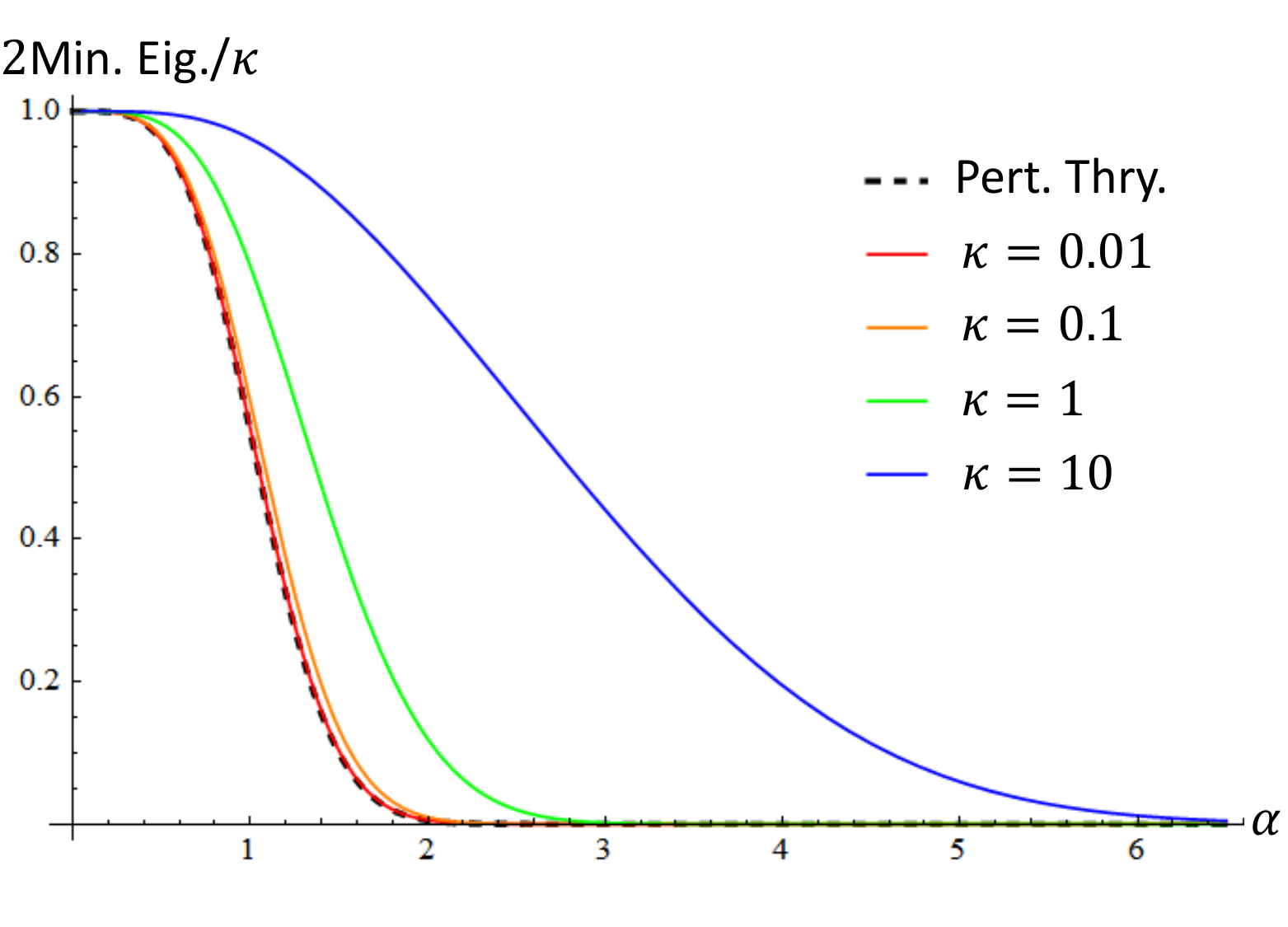} 
\par\end{centering}
\caption{\label{fig:cat-pert-final}Plot versus $\protect\a$ of the dissipative
gap (scaled by $-\kappa/2$) of \(\L+\oo\) (\ref{eq:dephasing-noise})
restricted to the ``off-diagonal'' block spanned by Fock state basis
elements \(\{|2n\ket\langle2m+1|\}_{n,m=0}^{\infty}\) for various
values of $\kappa$. The plot includes the result from the perturbative
calculation from eq.~(\ref{eq:final-pert-cat}) (dashed). One can
see that the effect of the perturbation quickly tends to zero with
increasing cat size $\alpha$.}
\end{figure}

We have numerically confirmed that eq.~(\ref{eq:final-pert-cat})
is indeed the first-order correction to \acs{ASH} due to $\oo$.
In Fig.~\ref{fig:cat-pert-final}, versus $\a$ and for various $\k$,
we plot the dissipative gap of $\L+\oo$ restricted to the block spanned
by $\{|2n\ket\bra2m+1|\}_{n,m=0}^{\infty}$. For small values of $\k$,
the numerical result approaches our above analytical estimate. In
fact, a similar trend holds for large values of $\kappa$, indicating
that higher-order terms (see Sec.~\ref{sec:Exact-all-order-Dyson})
should scale in a similar fashion. Since the Lindbladian itself is
preventing dephasing noise from acting within \acs{ASH}, we can say
that the cat pump \textit{passively protects} (in the sense of Ref.~\cite{Terhal2015})
from this error process.

It is illuminating to determine the degree to which the protection
from dephasing noise is coming from the driven two-photon Lindbladian
$\L$. To do so, we can split the perturbation into two terms,
\begin{equation}
\ppp\oo\ppp=\R_{\ul}\oo\R_{\ul}+\ppp\R_{\lr}\oo\R_{\ul}\,,\label{eq:pieces}
\end{equation}
and calculate the first term (which would be the only term if $\L$
had been Hermitian). The second term is purely a dissipative effect
and is due to $\L$ not being Hermitian and therefore not having the
same left and right eigenmatrices ($\J^{kl}=\St_{kl}+\J_{\lr}^{kl}\neq\St_{kl}$;
see Thm.~\ref{prop:3}). The first term has matrix elements $\bb\St_{kl}|\oo|\St_{kl}\kk$,
which are easily evaluated using techniques from eq.~(\ref{eq:technical}):
\begin{equation}
\bb\St_{kl}|\oo|\St_{kl}\kk=-\half\k\a^{2}\left(\frac{\pi_{k+1}}{\pi_{k}}+\frac{\pi_{l+1}}{\pi_{l}}-2\a^{2}\left[\frac{\pi_{k+1}\pi_{l+1}}{\pi_{k}\pi_{l}}-1\right]\right)\sim\begin{cases}
-\half\k\left(k-l\right)^{2} & \a\rightarrow0\\
-\k\a^{2} & \a\rightarrow\infty
\end{cases}\,.
\end{equation}
Therefore, for large $\a$, the piece $\R_{\ul}\oo\R_{\ul}\sim-\k\a^{2}\R_{\ul}$
is not trace-preserving. This shows the importance of using proper
Lindbladian perturbation theory instead of merely projecting perturbations
onto $\ulbig$.

It is worth noting that the leakage term $\L^{-1}\oo$ from eq.~(\ref{eq:pert-main-form})
dephases the cat-state basis elements that comprise the cat qubit,
reducing the purity of the full density matrix $\T_{t}^{\left(1\right)}|\rout\kk$
(\ref{eq:pert-main-form}). In phase space, this effect translates
to a slight smearing of the two Gaussian peaks that represent each
cat state. However, since this effect is due to leakage \textit{outside}
of \acs{ASH}, the quantum information that is stored \textit{within}
\acs{ASH} (and represented by $c_{kl}=\bb J^{kl}|\rout\kk$) is not
affected.
\selectlanguage{american}%

\subsection{Decoherence under single-photon loss}

While the cat pump is resilient to dephasing noise, it is unfortunately
incapable of protecting the quantum information in $\ulbig$ from
the most common type of error in photonic systems \textemdash{} amplitude
damping,
\begin{equation}
\oo(\r)=\half\k\left(2\aa\r\aa^{\dg}-\left\{ \ph,\r\right\} \right)\,.\label{eq:ampdamp}
\end{equation}
Let us show how the qubit in $\ulbig$ breaks down under this type
of error process by calculating $\ppp\oo\ppp$. We are interested
in large $\a$, so we work in the coherent state basis $\ct{\pm\a}$
(meaning that everything below is accurate up to exponential corrections
due to the overlap between the two states). Luckily, the recycling
term $\aa\cdot\aa^{\dg}$ keeps us in $\ulbig$ since $\aa|p\a\ket=p\a|p\a\ket$
with $p=\pm1$. In addition, the anti-commutator acts only from one
side at a time and so does not take $\ulbig$ into $\lrbig$ by (\ref{eq:no-leak}).
Therefore, $\ppp\R_{\lr}\oo\ppp=0$ and we luckily only need to calculate
$\R_{\ul}\oo\R_{\ul}$. Since $\ct{\pm\a}$ are approximately orthogonal,
$\R_{\ul}\oo\R_{\ul}$ is diagonal in the coherent state basis. Letting
$p,q\in\{\pm1\}$, a straightforward calculation in the large $\a$
limit yields
\begin{equation}
\bra p\a|\oo\left(|p\a\ket\bra q\a|\right)|q\a\ket\sim-\k\a^{2}\left(1-pq\right)\,.
\end{equation}
This perturbative result shows that the coherence $|\a\ket\cb{-\a}$
decays as $-2\k\a^{2}t$, in agreement with the small $\k$ limit
of the exact decay rate $-2\a^{2}(1-e^{-\k t})$ \{\cite{klimov_book},
below eq.~(9.11)\}.

\section{Holonomic quantum control\label{sec:Holonomic-quantum-control}}

Here, we apply the ideas learned in Ch.~\ref{ch:5} to study holonomies
on the cat-state \acs{ASH}. Recall that a slow (i.e., adiabatic)
variation of the parameters of a system in a closed loop returns the
system to its initial state, up to an operation (called a holonomy)
which is due to curvature and/or non-simple-connectedness of the parameter
space. Such holonomies can be used to perform quantum gates, either
on the ground states of a Hamiltonian or the \acs{ASH} of a Lindbladian,
in a process called \textit{holonomic quantum computation} \foreignlanguage{english}{\cite{Zanardi1999,pachos1999,lidarbook_zanardi}}.
We show how to perform such computation on the qubit spanned by \foreignlanguage{english}{the
cat states $|k_{\a}\ket$, $k\in\{0,1\}$}. 

In order to perform one of the gates, we need to introduce another
parameter into the previous Lindbladian, so from now we consider a
slightly more general $\L$ with jump operator
\begin{equation}
F=(\aa-\a_{0})(\aa-\a_{1})\,,\label{eq:f-general}
\end{equation}
\foreignlanguage{english}{where $\a_{l}$, $l\in\{0,1\}$, depend
on time. This jump operator stabilizes a two-dimensional \acs{ASH}
for \textit{all values of $\a_{0}$ and $\a_{1}$}. To see this, observe
that this new jump operator can be obtained by defining $\a_{\pm}=\half(\a_{0}\pm\a_{1})$
and conjugating $\aa^{2}-\a_{-}^{2}$ with the displacement operator
$D_{\a}$ (which acts on the vacuum state as $D_{\a}|0\ket=|\a\ket$):
\begin{equation}
D_{\a_{+}}\left(\aa^{2}-\a_{-}^{2}\right)D_{\a_{+}}^{\dg}=\left(\aa-\a_{+}\right)^{2}-\a_{-}^{2}=(\aa-\a_{0})(\aa-\a_{1})=F\,.
\end{equation}
Since the two jumps are related by a unitary conjugation, all spectral
properties of $\L$ with the original jump, including the existence
of a dissipative gap, hold for this case as well. Most importantly,
the two steady states are the displaced cat states $D_{\a_{+}}|k_{\a_{-}}\ket$.
While we can use this exact form of the steady states in the following
calculations, we instead use the coherent state basis and work in
the $|\a_{0}-\a_{1}|\gg1$ limit for some of the time in order to
simplify calculations. In this limit, \acs{ASH} is spanned by the
two coherent states $|\a_{l}\ket$, $l\in\{0,1\}$. }

\selectlanguage{english}%
The positions of the cat-qubit's two states $|\a_{l}(t)\ket$ in phase
space are now each controlled by a tunable parameter. We let $\a_{0}(0)=-\a_{1}(0)\equiv\a$,
meaning that the states $\ct{\pm\a}$ are the starting point of parameter
space evolution and the qubit defined by them (for large enough $\a$)
is shown in Fig.~\ref{fig:catschematic}. We now introduce two different
gates for this cat-qubit, the loop gate and the collision gate, which
together allow us to universally control said qubit. We work in the
adiabatic limit, meaning that the time $T$ used to perform the parameter
path is taken to infinity. Before proceeding, we want to briefly mention
that the leading-order $O(\nicefrac{1}{T})$ non-adiabatic correction
in the adiabatic perturbation theory expansion from Sec.~\ref{subsec:gap}
that causes leakage out of \acs{ASH} is still governed by the dissipative
gap $\adg$ of $\L_{\of}^{-1}$ and not that of $\L^{-1}$. This is
identical to the effect of the leakage term $\L^{-1}\spert=\L_{\of}^{-1}\spert$
in the ordinary perturbation theory calculations studied in Sec.~\ref{sec:Ordinary-perturbation-theory},
given a Hamiltonian perturbation $\spert$.

\begin{figure}[t]
\begin{centering}
\includegraphics[width=0.5\textwidth]{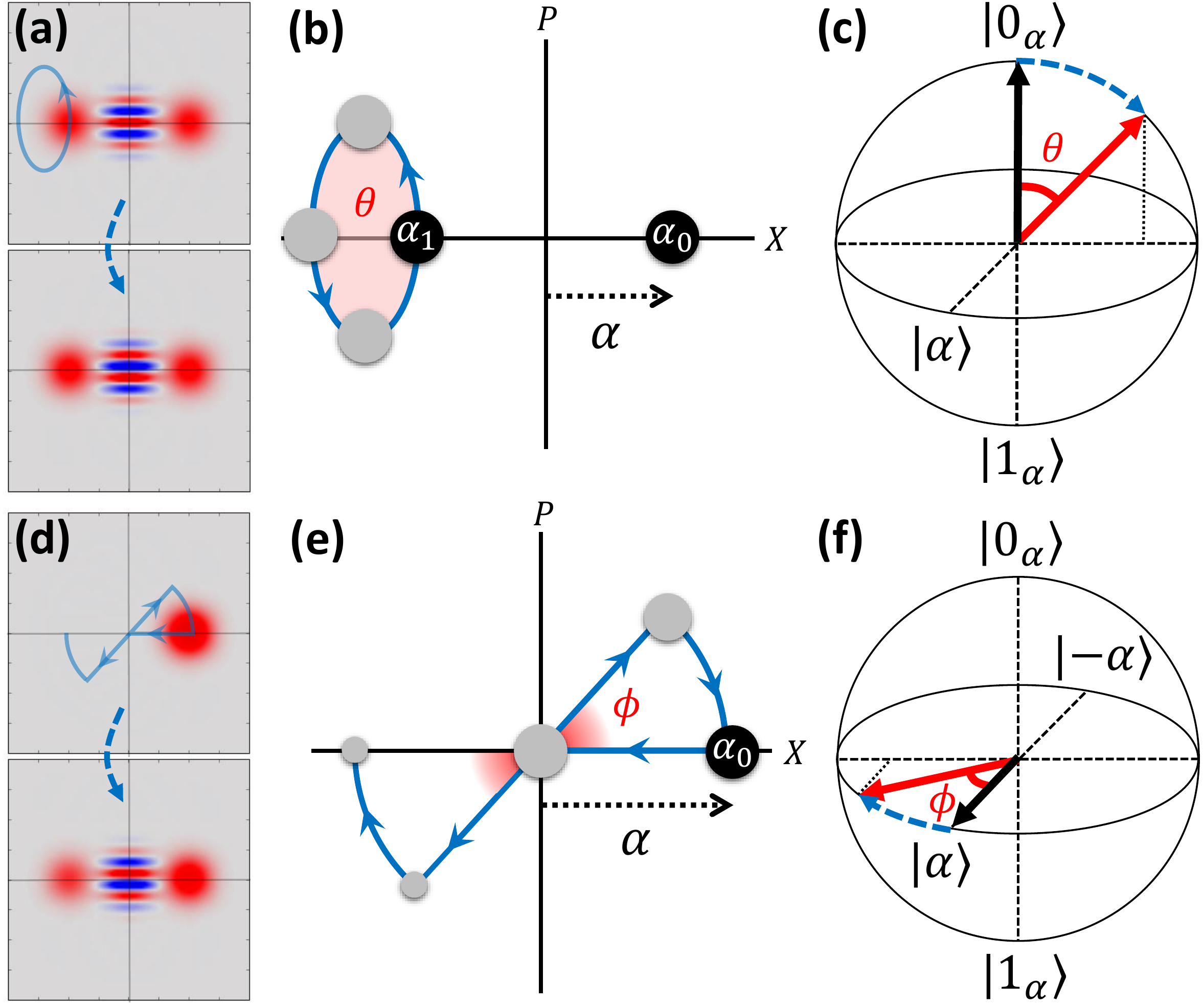} 
\par\end{centering}
\caption{\label{f1}\textbf{(a)} Wigner function sketch of the state before
(top) and after (bottom) a loop gate acting on \(\ct{-\a}\), depicting
the path of \(\ct{\a_{1}(t)}\) during the gate (blue) and a shift
in the fringes between \(\ct{\pm\a}\). \textbf{(b)} Phase space diagram
for the loop gate; \(X=\frac{1}{2}\bra\aa+\aa^{\dg}\ket\) and \(P=-\frac{i}{2}\bra\aa-\aa^{\dg}\ket\).
The parameter \(\a_{1}(t)\) is varied along a closed path (blue)
of area \(A\), after which the state \(\ct{-\a}\) gains a phase
\(\t=2A\) relative to \(|\a\ket\). \textbf{(c)} Effective Bloch
sphere of the \(\ct{\pm\a}\) qubit depicting the rotation caused
by the \(d=2\) loop gate. Black arrow depicts initial state while
red arrow is the state after application of the gate. The dotted blue
arrow does not represent the path traveled since the states leave
the logical space \(\ct{\pm\a}\) during the gate.\textbf{ (d-f)}
Analogous descriptions of the collision gate, which consists of reducing
\(\a\) to \(0\), driving back to \(\a\exp(i\phi)\), and rotating
back to \(\a\).}
\end{figure}

\subsection{Loop gate}

The loop gate involves an adiabatic variation of $\a_{1}(t)$ through
a closed path in phase space (see Fig.~\ref{f1}b). The state $|\a_{1}(t)\ket$
follows the path and, as long as the path is well separated from $|\a_{0}(t)\ket=|\a\ket$,
picks up a Berry phase of $\t=2A$, with $A$ being the area enclosed
by the path \cite{Chaturvedi1987}. It should be clear that initializing
the qubit in $\ct{-\a}$ produces only an irrelevant \textit{overall}
phase upon application of the gate. However, once the qubit is initialized
in a superposition of the two coherent states with coefficients $c_{\pm}$,
the gate imparts a \textit{relative} phase: 
\begin{equation}
c_{+}|\a\ket+c_{-}\ct{-\a}\longrightarrow c_{+}|\a\ket+c_{-}e^{i\theta}\ct{-\a}\,.
\end{equation}
Recalling that $|\a\ket$ lies on the $x$-axis of the cat-qubit Bloch
sphere from Fig.~\ref{fig:catschematic}, this gate can be thought
of as a rotation around that axis (depicted blue in Fig.~\ref{f1}c).
Similarly, adiabatically traversing a closed and isolated path with
the other state parameter $|\a_{0}(t)\ket$ induces a phase on $|\a\ket$.

\subsection{Collision gate}

For this gate, we utilize the small $\a$ regime to perform rotations
around the Bloch sphere $z$-axis (Fig.~\ref{f1}f), which effectively
induce a collision and population transfer between $|\a\ket$ and
$\ct{-\a}$. The procedure hinges on the following observation: applying
a bosonic rotation $R_{\phi}\equiv\exp(i\phi\ph)$ to well-separated
coherent or cat state superpositions \textit{does not} induce state-dependent
phases while applying $R_{\phi}$ to Fock state superpositions \textit{does}.
Only one tunable parameter $\a_{0}(t)=-\a_{1}(t)$ is necessary here,
so $F=\sqrt{\k}[\aa^{2}-\a_{0}(t)^{2}]$ with $|\a_{0}(0)|=\a$. The
collision gate consists of reducing $\a$ to $0$, driving back to
$\a\exp(i\phi)$, and rotating back to $\a$ (Fig.~\ref{f1}e). Recall
from Sec.~\ref{subsec:Unitary-case:-degenerate} that, for the \foreignlanguage{american}{\acs{DFS}}
case, the holonomy corresponding to \acs{ASH} is generated by $\p_{\l}\pp\pp$
(for some parameter $\l$). Using eq.~(\ref{eq:uzan-1}), we can
represent this holonomy as a path-ordered product of projections $\pp_{\a}=|0_{\a}\ket\bra0_{\a}|+|1_{\a}\ket\bra1_{\a}|$
onto the \foreignlanguage{american}{\acs{DFS}}. Therefore, the part
of the path which corresponds to the nonunitary driving from $0$
to $\a\exp(i\phi)$ can be approximated as $S_{\phi}\approx P_{\a e^{i\phi}}\cdots P_{\frac{2}{M}\a e^{i\phi}}P_{\frac{1}{M}\a e^{i\phi}}$
for integer $M\gg1$. Similarly, the part which ``deflates the cat''
from $\a$ to $0$ is approximately $S_{0}^{\dg}$. Combining these
with the rotation for the last segment of the path, the full ``pizza-slice''
path is\foreignlanguage{american}{ }represented by $R_{\phi}^{\dg}S_{\phi}S_{0}^{\dg}$.
Since 
\begin{equation}
R_{\phi}^{\dg}S_{\phi}R_{0}^{\dg}=R_{\phi}^{\dg}(R_{\phi}S_{0}R_{\phi}^{\dg})S_{0}^{\dg}=S_{0}R_{\phi}^{\dg}S_{0}^{\dg}\,,\label{eq:conj}
\end{equation}
the collision gate is equivalent to reducing $\a$, applying $R_{\phi}^{\dg}$
on the steady-state Fock basis $|k\ket=|k_{\a=0}\ket$, and driving
back to $\a$. The net result is thus a relative phase between the
states $\ct{k_{\a}}$:
\begin{equation}
c_{0}\ct{0_{\a}}+c_{1}\ct{1_{\a}}\longrightarrow c_{0}\ct{0_{\a}}+c_{1}e^{-i\phi}\ct{1_{\a}}\,.
\end{equation}
In the coherent state basis, this translates to a coherent population
transfer between $\ct{\pm\a}$.

\selectlanguage{american}%
Both gates can also be understood in terms of Berry connections of
the cat states, 
\begin{equation}
A_{kl}^{\l}\equiv i\bra k_{\a}|\p_{\l}l_{\a}\ket\,,
\end{equation}
where $\l$ is a parameter that is varied during the path. (Recall
from Ch.~\ref{ch:5} that $A^{\l}$ is the coordinate representation
of the connection $\p_{\l}\pp\pp$.) \foreignlanguage{english}{For
example, the collision gate arises from changes in the magnitude and
phase of the coherent state parameter $\a e^{i\varphi}$. Therefore,
$\lambda\in\{|\a|,\arg\a\equiv\varphi\}$ and a simple calculation
reveals 
\begin{align}
A_{kl}^{|\a|}=0\,\,\,\,\,\,\,\,\,\,\,\,\,\,\,\,\,\text{and}\,\,\,\,\,\,\,\,\,\,\,\,\,\,\,\,\,A_{kl}^{\varphi} & =-\d_{kl}\bra k_{\a}|\ph|k_{\a}\ket\,.
\end{align}
Recalling the small and large $\a$ limits of the average occupation
number from eq.~(\ref{eq:technical}), this confirms that the effective
operation induced by the collision gate is indeed caused by the rotation
induced on the Fock states at $\a=0$.}
\inputencoding{latin9}\newpage{}\foreignlanguage{english}{}%
\begin{minipage}[t]{0.5\textwidth}%
\selectlanguage{english}%
\begin{flushleft}
\begin{singlespace}\textit{``For today's electrical engineers worrying
about {[}Moore's Law{]}, quantum mechanics is a bug, but the hope
is that we can turn it into a feature.''}\end{singlespace}
\par\end{flushleft}
\begin{flushleft}
\hfill{}\textendash{} Robert J. Schoelkopf
\par\end{flushleft}\selectlanguage{english}%
\end{minipage}

\chapter{Application: single- and multi-mode cat codes\label{ch:8}}

We now proceed to state a series of extensions of the cat-state \acs{ASH}
stabilized by the ``two-cat pump'' of the previous \foreignlanguage{english}{chapter}.
The single-mode generalizations are called \textit{cat codes} \cite{Cochrane1999,Leghtas2013b,cats,Albert2015,paircat}
\textemdash{} \foreignlanguage{english}{quantum memories for coherent-state
quantum information processing \cite{catbook,cvbook_cats} which store
information in superpositions of well-separated coherent states evenly
distributed around the origin of phase space.} Here, we review the
single-mode generalizations and introduce an $M$-mode generalization
of cat-codes, making contact with the Lindbladians necessary to generate
these codes. We note that the states we consider have been studied
in a quantum optical context for $M=2$ \cite{Liu2001,Choi2008} and
$M=3$ \cite{an2003}.

\section{Single-mode cat codes}

\selectlanguage{english}%
In the previous chapter, we studied features of the Lindbladian generated
by the jump operator $F=\aa^{2}-\a^{2}$, which stabilized an \acs{ASH}
consisting of cat states $|k_{\a}\ket$, $k\in\{0,1\}$ (\ref{eq:catsintro}).
Through the lens of quantum information, this \acs{ASH} is part of
a \textit{quantum code} \textemdash{} a subspace that is used to store
an arbitrary quantum state in order to prevent its quantum information
from decohering or changing without notice. While we saw that such
an \acs{ASH} is passively protected from dephasing noise, it is not
protected from amplitude damping. We now double the size of this \acs{ASH}
in order to accommodate (and thus protect from) the effects of amplitude
damping.

Notice that the amplitude damping process (\ref{eq:ampdamp}) is generated
by the jump operator $\aa$, which decreases the occupation number
of a state by one, thereby flipping the occupation number parity.
If we had some way of storing information in a subspace of fixed (say,
even) occupation number parity which then could ``jump'' into an
error subspace of odd parity after being acted upon by $\aa$, then
we could in principle track such a jump and prevent the quantum information
from decohering. This is similar to more traditional multi-qubit stabilizer
codes \cite{nielsen_chuang}, which provide a large enough number
of error subspaces for a code such that the quantum information can,
after undergoing an error, ``jump'' from the code subspace into
an error subspace without overlapping with itself and decohering.
However, here we gain an additional advantage: we do not have to correct
the error and can simply track which subspace our quantum information
is in. In order to allow for the tracking of loss events $\aa$, all
we have to do is to make sure we have both an even- and an odd-parity
subspace in our \acs{ASH}, each of which are large enough to store
a qubit. This can be achieved by ``doubling'' the jump operator
to 
\begin{equation}
F=\aa^{4}-\a^{4}\,.\label{eq:newjump}
\end{equation}
For large enough $\a$, the four coherent states $|i^{k}\a\ket$,
$k\in\{0,1,2,3\}$, form the \acs{ASH} of the Lindbladian generated
by this jump operator. However, we once again would like to build
an orthonormal basis valid for all $\a$ whose states are eigenstates
of the parity operator $\left(-\right)^{\ph}$. Redefining projections
$\Pi_{k}=\sum_{n=0}^{\infty}|4n+k\ket\bra4n+k|$, the ``four-cat''
state basis consists of
\begin{equation}
|k_{\a}\ket\equiv\frac{\Pi_{k}|\a\ket}{\sqrt{\bra\a|\Pi_{k}|\a\ket}}\,\,\,\,\text{with normalization\,\,\,\,\ensuremath{\nn_{k}\equiv\bra\a|\Pi_{k}|\a\ket}}\,.\label{eq:catsintro-1}
\end{equation}
Similar to the two-cat pump states, these become Fock states $|k\ket$
for small $\a$ and equal superpositions of coherent states $\{|i^{k}\a\ket\}_{k=0}^{3}$
for large $\a$. The states $|0_{\a}\ket,|2_{\a}\ket$ are even parity
states \textemdash{} $\left(-\right)^{\ph}|0_{\a}\ket=|0_{\a}\ket$
and same for $|2_{\a}\ket$ \textemdash{} while the states $|1_{\a}\ket,|3_{\a}\ket$
are odd parity. Instead of using the entire four-dimensional space
to store a quartrit, we can use the even parity subspace as the (qubit)
code subspace and the odd parity subspace as the error subspace. That
way, it is possible to track loss events as they happen. Such tracking
has been experimentally realized in Ref.~\cite{Ofek2016}.

The new ``four-cat pump'' Lindbladian generated by $F$ from eq.~(\ref{eq:newjump})
enjoys many of the same features as the two-cat pump. A Hamiltonian-based
gate can be performed on the quantum information in either the even
or odd-parity subspace using the Hamiltonian $V=\aa^{2}+H.c.$, just
like $V=\aa+H.c.$ performed a gate on the two-dimensional \acs{ASH}
of the two-cat pump in Sec.~\ref{sec:Ordinary-perturbation-theory}.
The new jump operator is also of the type $F=F_{\ur}$, meaning that
the effective dissipative gap $\adg$ is the excitation gap of $\half F^{\dg}F$.
While an analytic representation for the 16 conserved quantities of
this case has not yet been found, it has been numerically determined
that dephasing noise is also suppressed (\cite{cats}, Fig.~A1b).
Holonomic quantum control can be performed on the entire \acs{ASH}
or only on its even/odd parity blocks \cite{Albert2015}.

We can continue along this line of reasoning and consider having $d-1$
error subspaces in order to track up to $d-1$ loss events (\cite{zaki},
Supplementary Material; \cite{Bergmann2016,li2016}). Such a scheme
is realized by the jump operator
\begin{equation}
F=\aa^{2d}-\a^{2d}\,.
\end{equation}
This operator annihilates the (unnormalized) states 
\begin{equation}
\Pi_{k}|\a\ket=e^{-\half\a^{2}}\sum_{n=0}^{\infty}\frac{\a^{2dn+k}}{\sqrt{(2dn+k)!}}|2dn+k\ket\,,\label{eq:gencats}
\end{equation}
 where $k\in\{0,1,\cdots,2d-1\}$ and the $2d$ projections from eq.~(\ref{eq:oscosc})
are
\begin{equation}
\Pi_{k}=\sum_{n=0}^{\infty}|2dn+k\ket\bra2dn+k|=\frac{1}{2d}\sum_{l=0}^{2d-1}e^{i\frac{\pi}{d}(\ph-k)l}\,.
\end{equation}
These states are eigenstates of the discrete rotation operator $e^{i\frac{\pi}{d}\ph}$,
just like the two-cat states are eigenstates of the parity operator
$e^{i\pi\ph}$. The power of $\aa$ is $2d$ (instead of $d$) because
this provides enough room for (a two-dimensional) code subspace and
the $d-1$ (two-dimensional) error subspaces that are required to
track up to $d-1$ loss events. If one wants to store a qu$D$it's
worth of information in the code subspace while still protecting from
$d-1$ loss events, then one can consider the jump operator $F=\aa^{Dd}-\a^{Dd}$.

The above schemes can be extended even further by considering jump
operators which are polynomials in $\aa$,
\begin{equation}
F=\prod_{k=0}^{d-1}(\aa-\a_{k})
\end{equation}
for some complex $\a_{k}$. Assuming that each $\a_{k}$ is well-separated
from the others in phase space, the kernel of the jump operator is
spanned by the $d$ coherent states $\{|\a_{k}\ket\}_{k=0}^{d-1}$.
Holonomic gates on such \acs{ASH} have been considered in Ref.~\cite{Albert2015}.
However, this general case (for $d>2$) is more complicated to work
with because, unlike the $d=2$ case (see Sec.~\ref{sec:Holonomic-quantum-control}),
it is \textit{not} unitarily related to $\aa^{d}-\a^{d}$ for some
$\a$.
\selectlanguage{american}%

\section{Two-mode cat codes\label{sec:Two-mode-cat-codes}}

Let us now consider a generalization of single-mode cat-codes to two
modes \cite{paircat} using the pair-coherent/Barut-Girardello states
\cite{Barut1971,Agarwal1986,Agarwal1988}. Recall that, in the single
mode case, our codes were eigenstates of powers of the lowering operator
$\aa$. In this case, our code states are eigenstates of powers of
$\aa\bl$, where $\bl$ is the lowering operator for another oscillator.
Naturally, $[\bl,\bl^{\dg}]=1$ and $\pb\equiv\bl^{\dg}\bl$. Recall
also that cat states were built by projecting (using $\Pi_{k}$) a
coherent state $|\a\ket$ onto eigenspaces of the rotation operator
$e^{i\frac{2\pi}{2d}\ph}$. We perform the same trick with $\aa\bl$.
However, while $\aa$ has only one type of eigenstate ($|\a\ket$),
$\aa\bl$ has a countable infinity of types, each of which is indexed
by a continuous parameter! Therefore, a careful analysis of the eigenstates
of $\aa\bl$ is required before we consider its higher powers.

An easy way to determine the eigenstates of $\aa$ is to simply plug
in a general state $|\psi\ket=\sum_{n=0}^{\infty}c_{n}|n\ket$ into
the eigenvalue relation 
\begin{equation}
\aa|\psi\ket=\a|\psi\ket
\end{equation}
and solve for the coefficients $c_{n}=\frac{\a^{n}}{\sqrt{n!}}$.
If we do this procedure with a two mode state $|\psi\ket=\sum_{n,m=0}^{\infty}c_{n,m}|n,m\ket$,
then we have two indices to consider. Instead, we can first use symmetries
to restrict what types of states $|\psi\ket$ we can plug in and thereby
avoid having to deal with two indices. Notice that $[\aa\bl,\pb-\ph]=0$,
meaning that $\aa\bl$ preserves the \textit{occupation number difference}
\begin{equation}
\de\equiv\pb-\ph\,.
\end{equation}
Therefore, any eigenstates of $\aa\bl$ are also eigenstates of $\de$.
The latter can be organized into subspaces of the same eigenvalue
$\Delta\in\mathbb{Z}$, namely $\{|n,n+\ra\ket\}_{n=0}^{\infty}$
for $\ra\geq0$ and $\{|n+\ra,n\ket\}_{n=0}^{\infty}$ for $\ra<0$.
That way, $\de|\psi\ket=\ra|\psi\ket$ for any state $|\psi\ket$
lying in a subspace of fixed $\ra$. 

For convenience, we can introduce the exchange operator 
\begin{equation}
X\equiv\exp\left[i\frac{\pi}{2}(\aa^{\dg}-\bl^{\dg})(\aa-\bl)\right]\,,\,\,\,\,\,\,\,\text{which acts as }\,\,\,\,\,\,\,X|n,m\ket=|m,n\ket\,,
\end{equation}
and write the subspaces for all negative $\ra$ as $\{X|n,n+\ra\ket\}_{n=0}^{\infty}$.
From now on, we assume that $\ra\geq0$, remembering that an application
of $X$ yields the corresponding results for $\ra<0$.

The projections onto each of the subspaces are
\begin{equation}
P_{\ra}\equiv\sum_{n=0}^{\infty}|n,n+\ra\ket\bra n,n+\ra|=\frac{1}{2\pi}\intop_{0}^{2\pi}d\t e^{i\left(\pb-\ph-\ra\right)\t}\label{eq:raproj}
\end{equation}
(for $\ra\geq0$ and $XP_{|\ra|}X$ for $\ra<0$). Since $\aa\bl$
is block diagonal in the decomposition of subspaces of fixed $\de$,
we have to only consider general states \textit{in each subspace}:
\begin{equation}
\aa\bl|\psi_{\ra}\ket\equiv\aa\bl\sum_{n=0}^{\infty}c_{n}|n,n+\ra\ket=\g^{2}|\psi_{\ra}\ket\,,
\end{equation}
with eigenvalue $\g^{2}$. Solving this equation by acting with $\aa\bl$
on each Fock state yields the solution $c_{n}=\frac{\g^{2n+\D}}{\sqrt{n!(n+\ra)!}}$,
and the resulting normalized state is the pair-coherent state
\begin{equation}
|\g_{\ra}\ket=\frac{1}{\sqrt{I_{\ra}(2|\g|^{2})}}\sum_{n=0}^{\infty}\frac{\g^{2n+\ra}}{\sqrt{n!\left(n+\ra\right)!}}|n,n+\ra\ket\,,\label{eq:genpair}
\end{equation}
where $I_{\ra}$ is the \foreignlanguage{english}{modified Bessel
function of the first kind \cite{dlmf}. These well-known states satisfy
several of the properties of ordinary coherent states: they are eigenstates
of a lowering operator ($\aa\bl$) and they are overcomplete (on each
subspace of fixed $\ra$). However, they are \textit{not} generated
by a displacement-like operator: $\exp(\g^{\star}\aa\bl-H.c.)|0,0\ket$
does not produce $|\g_{\ra}\ket$ but instead produces what is known
as a two-mode squeezed state \cite{scully}. Nevertheless, such states
\textit{can} be conveniently related to a two-mode coherent state
$|\g,\g\ket$ via the projections (\ref{eq:raproj}):
\begin{equation}
|\g_{\ra}\ket=\frac{P_{\ra}|\g,\g\ket}{\sqrt{\bra\g,\g|P_{\ra}|\g,\g\ket}}\,.
\end{equation}
}

\selectlanguage{english}%
Having introduced all of the eigenstates of $\aa\bl$, we can now
further apply projections $\Pi_{k}$ of the type discussed in the
previous Section in order to present two-mode cat codes. Notice that,
for the single-mode case in eq.~(\ref{eq:gencats}), applying $\Pi_{k}$
to a coherent state $|\a\ket$ is equivalent to having the index $n$
of the sum over Fock states of $|\a\ket$ transform as $n\rightarrow2dn+k$.
Here, we observe a similar pattern, but this time for the index $n$
in the sum of the pair coherent state (\ref{eq:genpair}). Let us
introduce projections onto eigenspaces of the joint rotation $e^{i\frac{2\pi}{4d}(\ph+\pb)}$
(where there is an extra factor of $2$ compared to the single-mode
case, $e^{i\frac{2\pi}{2d}\ph}$, corresponding to there being two
modes),
\begin{equation}
\Pi_{k}\equiv\frac{1}{4d}\sum_{l=0}^{4d-1}e^{i\frac{\pi}{2d}(\ph+\pb-k)l}=\sum_{n,m=0}^{\infty}|n,m\ket\bra n,m|\d_{n+m,k}^{\text{mod\,}4d}\,,\label{eq:tmtot}
\end{equation}
where $\d_{n+m,k}^{\text{mod\,}4d}=1$ whenever $n+m=k$ modulo $4d$.
Notice that $[\Pi_{k},P_{\ra}]=0$ since they are both functions of
$\ph,\pb$ only and that $\Pi_{k}\aa=\aa\Pi_{(k+1)\text{mod\,}4d}$
(and same for $\bl$). This implies that 
\begin{equation}
\Pi_{k}(\aa\bl)^{2d}=\aa^{2d}\Pi_{k+2d}\bl^{2d}=(\aa\bl)^{2d}\Pi_{k+4d}=(\aa\bl)^{2d}\Pi_{k}\,,
\end{equation}
meaning that any eigenstates of $(\aa\bl)^{2d}$ are also those of
$\Pi_{k}$. 

Applying $\Pi_{2k+\ra}$ to $|\g_{\ra}\ket$ produces the \textit{two-mode
cat code state}
\begin{equation}
|k_{\g,\ra}\ket=\frac{\Pi_{2k+\ra}P_{\ra}|\g,\g\ket}{\sqrt{\pi_{k,\ra}}}\label{eq:catsintro-1-1}
\end{equation}
with normalization $\nn_{k,\Delta}\equiv\bra\g,\g|\Pi_{2k+\ra}P_{\ra}|\g,\g\ket$
and $k\in\{0,1,\cdots,2d-1\}$. Fixing $\g$ to be real, the Fock
state representation of these states is 
\begin{equation}
|k_{\g,\ra}\ket=\frac{e^{-\g^{2}}}{\sqrt{\pi_{k,\ra}}}\sum_{n=0}^{\infty}\frac{\g^{2dn+k+\half\Delta}}{\sqrt{(2dn+k)!(2dn+k+\ra)!}}|2dn+k,2dn+k+\ra\ket\,.
\end{equation}
To show this, first observe that $P_{\ra}|\g,\g\ket$ consists of
Fock states from the subspace $\{|n,n+\ra\ket\}_{n=0}^{\infty}$.
Then, notice that $\Pi_{2k+\ra}$ projects those Fock states further
onto the subspace for which the total occupation number
\begin{equation}
2n+\ra=2k+\ra\mod4d\,.
\end{equation}
This implies that $n=k$ modulo $2d$. Therefore, for a given $k$,
the subspace of the states $\{|n,n+\ra\ket\}_{n=0}^{\infty}$ that
is preserved under $\Pi_{2k+\ra}$ is $\{|2dn+k,2dn+k+\ra\ket\}_{n=0}^{\infty}$.

One can check that $|k_{\g,\ra}\ket$, $k\in\{0,1,\cdots,2d-1\}$,
are eigenstates of $(\aa\bl)^{2d}$, meaning that the jump operator
used to stabilize an \acs{ASH} consisting of them is
\begin{equation}
F=(\aa\bl)^{2d}-\g^{2d}\,.
\end{equation}
As with the single-mode cat codes, we have verified numerically that
this Lindbladian suppresses dephasing noise in both modes for $d=1$.
In addition, this \acs{ASH} can store a qubit (say, in the $\ra=k=0$
subspace) that can be protected from arbitrary single-mode loss events
$\aa^{n}$ and $\bl^{m}$ for $n,m\in\{1,2,\cdots,\infty\}$ as well
as joint events $(\aa\bl)^{p}$ for $p\leq d-1$. This protection
can be understood by studying how these loss events interact with
the projections $P_{\ra}$ and $\Pi_{2k+\ra}$.
\begin{enumerate}
\item Single-mode loss events $\aa^{n}$ and $\bl^{m}$ shift the value
of $\ra$:
\begin{equation}
\begin{pmatrix}\aa\\
\bl
\end{pmatrix}P_{\ra}=\begin{pmatrix}P_{\ra+1}\aa\\
P_{\ra-1}\bl
\end{pmatrix}\,.
\end{equation}
Since the value of $\ra$ is shifted in different directions depending
on which mode incurred the losses, it is possible to track those events
by continuously monitoring the occupation number difference $\de=\pb-\ph$.
Since the eigenvalues of $\de$ are integers, an arbitrary amount
of single-mode loss events can be detected. Note that single-mode
events also do not commute with $\Pi_{2k+\ra}$, meaning that the
error subspace to which the qubit jumps to after such events have
different values of both $\ph+\pb$ (modulo $4d$) and $\de$.
\item For each $\ra$, there are $2d$ states of fixed photon number difference.
A joint loss event $\aa\bl$ commutes with $P_{\ra}$ but not with
$\Pi_{k}$, shifting $k\rightarrow k-1$:
\begin{equation}
\Pi_{2k+\ra}\aa\bl=\aa\bl\Pi_{2(k-1)+\ra}\,.
\end{equation}
Since the eigenspace of $\de$ doesn't change upon these errors, the
syndrome set is different from that associated with the single-mode
events above. Since there are $d-1$ error subspaces for each $\ra$,
it is possible to track up to $d-1$ such joint loss events.
\end{enumerate}
\selectlanguage{american}%
Some of the Hamiltonian-based and holonomic gates discussed in the
previous \foreignlanguage{english}{chapter} can also be extended to
these cases. For example, for $d=1$, the Hamiltonian $V=\aa\bl+H.c.$
performs a gate between the two states $|k_{\g,\ra}\ket$, $k\in\{0,1\}$,
for each $\ra$. A holonomic gate which consists of the path $\g\rightarrow0\rightarrow\g e^{i\phi}\rightarrow\g$,
a generalization of the single-mode collision gate from Sec.~\ref{sec:Holonomic-quantum-control},
induces a similar effect (again for each $\ra$).

\section{$M$-mode cat codes\label{sec:-mode-cat-codes}}

The two-mode generalization above can be naturally extended to $M$
modes, whose corresponding operators are labeled $\{\aa_{m},\aa_{m}^{\dg},\ph_{m}\}$
with $m\in\{1,2,\cdots,M\}$. Such codes for $M\geq3$ gain the additional
advantage of being able to correct for higher-weight products of losses
or for photon losses and gains at the same time (see Ref.~\cite{paircat}).
One way to characterize their code states is to use to a vector of
$M-1$ occupation number differences between neighboring modes,
\begin{equation}
\vec{\de}=\left\langle \ph_{2}-\ph_{1},\ph_{3}-\ph_{2},\cdots,\ph_{M}-\ph_{M-1}\right\rangle \,.
\end{equation}
Projections onto subspaces of fixed differences $\vec{\ra}=\left\langle \ra_{1},\ra_{2},\cdots,\ra_{M-1}\right\rangle $
are
\begin{equation}
P_{\vec{\ra}}=\sum_{n=0}^{\infty}\bigotimes_{m=1}^{M}\left|n+\sum_{p=1}^{m-1}\ra_{p}\right\rangle \left\langle n+\sum_{p=1}^{m-1}\ra_{p}\right|
\end{equation}
for $\Delta_{p}\geq0$ and projections on the total occupation number
$\hat{N}\equiv\sum_{m=1}^{M}\ph_{m}$ generalize straightforwardly
from eq.~(\ref{eq:tmtot}):
\begin{equation}
\Pi_{k}=\frac{1}{2dM}\sum_{l=0}^{2dM-1}e^{i\frac{\pi}{dM}\left(\hat{N}-k\right)l}\,.
\end{equation}

As an example, we write down only the $\vec{\ra}=\vec{0}$ and $d=1$
states, which are permutation symmetric. For $k\in\{0,1\}$,
\begin{equation}
|k_{\g,\vec{0}}\ket\equiv\frac{\Pi_{k}P_{\vec{0}}\cdot|\g\ket^{\otimes M}}{\sqrt{\pi_{k,\vec{0}}}}\propto\sum_{n=0}^{\infty}\g^{M(2n+k)}\bigotimes_{m=1}^{M}\frac{|2n+k\ket}{\sqrt{(2n+k)!}}\,,
\end{equation}
with normalization $\pi_{k,\vec{0}}\equiv\bra\g|^{\otimes M}\cdot\Pi_{k}P_{\vec{0}}\cdot|\g\ket^{\otimes M}$.
The jump operator which annihilates these states is 
\begin{equation}
F=\left(\prod_{m=1}^{M}\aa_{m}\right)^{2}-\g^{2M}\,.
\end{equation}
Once again, the analysis of the previous \foreignlanguage{english}{chapter}
is extendable to these codes.
\inputencoding{latin9}\newpage{}\foreignlanguage{english}{}%
\begin{minipage}[t]{0.5\textwidth}%
\selectlanguage{english}%
\begin{flushleft}
\begin{singlespace}\textit{``Rather than working on No-Go theorems,
I prefer to do Lego experiments.''}\end{singlespace}
\par\end{flushleft}
\begin{flushleft}
\hfill{}\textendash{} Michel H. Devoret
\par\end{flushleft}\selectlanguage{english}%
\end{minipage}

\chapter{Outlook\label{ch:9}}

\selectlanguage{english}%
This work is concerned with Lindbladians which admit more than one
steady state. The motivation for studying such Lindbladians is two-fold.
First, using a set of techniques typically characterized as\textit{
quantum} \textit{reservoir engineering}, such Lindbladians can be
used to stabilize exotic phases of matter (corresponding to possibly
degenerate ground states), quantum entanglement (for quantum communication
or metrology), or desirable subspaces (for quantum information processing).
Second, such Lindbladians can be used for \textit{autonomous/passive
quantum error correction} \cite{Terhal2015}, suppressing the effect
of errors on a steady-state subspace containing quantum information
and/or driving any leaked quantum information back into said subspace
after an error. We have reviewed and made contact with previous work,
detailed relevant manuscripts to which the author of this thesis contributed,
and presented previously unpublished results. A summary can be found
in Sec.~\ref{sec:Questions-addressed-and}.

\selectlanguage{american}%
One item that is anticipated to gain further application is the all-order
Dyson series for slowly ramping-up perturbations in Ch.~\ref{ch:4}.
Due to no restrictions on the unperturbed Lindbladian, the number
and type of its steady states, and the type of perturbation, that
analysis is just about as general as one could hope for while still
adhering to the laws of quantum mechanics. \foreignlanguage{english}{While
we focus on response to Hamiltonian perturbations within first-order
and evolution within the adiabatic limit, it would be of interest
to further study other Lindbladian perturbations \cite{Venuti2015a}
and their corresponding higher-order effects. While several elements
of this study consider asymptotic subspaces consisting of only one
block of steady states, it is not unreasonable to imagine that the
aforementioned second-order and/or non-adiabatic effects could produce
transfer of information between two or more blocks. Recently developed
diagrammatic series aimed for determining perturbed steady states
\cite{Li2015a} (see also \cite{Li2014}) may benefit from the four-corners
decomposition (whenever the unperturbed steady state is not full-rank).
The all-order Dyson series should provide insight into reservoir engineering
theory and experiments in which the ``good'' dissipation is stronger
than any ``bad'' noise.} Finally, while we have applied the general
result of Thm.~\ref{prop:3} to ordinary (time-dependent) and adiabatic
perturbation theory, future work could include an application to \textit{singular
perturbation theory} and \textit{adiabatic elimination} techniques
\cite{Mirrahimi2008,Azouit2016,Azouit2017} or \textit{quasi-degenerate
perturbation theory} (\cite{Winkler2003}, Appx. B).

A glaring item that is missing from this work, with the notable exception
of the frustration-free Lindbladians of Sec.~\ref{sec:Ground-state-subspaces},
is an application to many-body open systems \foreignlanguage{english}{with
\textit{non-equilibrium steady states} (NESS)}. This is not necessarily
exclusive to this work, as many concepts are only recently being extended
to \foreignlanguage{english}{NESS}. Examples include topological order
\cite{Rivas2013,Budich2015,Huang2014}, Thouless pumping \cite{Linzner2016},
and spontaneous symmetry breaking \cite{Wilming2016} to name a few.
\foreignlanguage{english}{It would also be of interest to determine
how the effective dissipative gap scales with system size vs. the
true dissipative gap in many-body systems \cite{Cai2013,Znidaric2015,Iemini2015,Wilming2016}.
}Other many-body concepts have yet to be extended to Lindbladians
with multiple \foreignlanguage{english}{NESS}. Matrix product methods
determining the steady state of a Lindbladian \cite{cui2015}\foreignlanguage{english}{
and current applications of the Keldysh formalism to Lindbladians
\cite{Sieberer2016} do not tackle degenerate cases. Stability of
NESS studied for the unique state case \cite{Cubitt2015a} should
also be extendable.}

\selectlanguage{english}%
Given a quantum channel $\E$, there exists a recipe (see Sec.~\ref{subsec:Quantum-channel-simulation}
and Ref.~\cite{ABFJ}) for a Lindbladian $\L$ whose time evolution
in the infinite-time limit contains one action of $\E$. Such an embedding
may prove useful in autonomous error correction and experimental quantum
channel simulation. It would be of interest to study the applicability
of this recipe in the broader context of previous efforts on channel
simulation, both theoretical \cite{Lloyd2001,Andersson2008,Wang2015,Shen2016,Iten2016}
and experimental \cite{Lu2015,McCutcheon2017}.

The metric stemming from the \ac{QGT}  will be examined in future
work, particularly to see whether it reveals information about bounds
on convergence rates \cite{delcampo2013,Rouchon2013,jing2015,Pires2016}.
It remains to be seen whether the scaling behavior of the metric is
correlated with phase stability \cite{Roy2014,Dobardzic2013,Jackson2015,Bauer2015}
and phase transitions \cite{CamposVenuti2007,Zanardi2007,Kolodrubetz2013}
for NESS. We do not derive a \ac{QGT} or metric for the case of multiple
\foreignlanguage{american}{\acs{NS}} blocks, so taking into account
any potential interaction of the blocks during adiabatic evolution
remains an open problem.

It has recently been postulated \cite{Macieszczak2015} that Lindbladian
meta-stable states also possess the same structure as the steady states.
This may mean that our results regarding conserved quantities (which
are dual to the steady states) also apply to the pseudo-conserved
quantities (dual to the meta-stable states).

Lastly, the properties of Lindbladian eigenmatrices should be extendable
to memory-kernel dynamics \cite{Janssen2017} and can be extended
to eigenmatrices of more general quantum channels \cite{BlumeKohout2008,robin,baumr,Carbone2015}.
Statements similar to Thm.~\ref{prop:3} exist for fixed points of
quantum channels \cite{robin,Cirillo2015} and their extension to
rotating points will be a subject of future work. These results may
also be useful in determining properties of asymptotic algebras of
observables \cite{Dhahri2010,alipour2015} and properties of quantum
jump trajectories when the Lindbladian is ``unraveled'' \cite{wisemanmilburn,Benoist2015}.\selectlanguage{english}%

\selectlanguage{english}%
\renewcommand*{\bibfont}{\footnotesize}

\selectlanguage{american}%
\inputencoding{latin9}\newpage{}

\manualmark \markboth{\spacedlowsmallcaps{\bibname}}{\spacedlowsmallcaps{\bibname}}
\refstepcounter{dummy}
\addtocontents{toc}{\protect\vspace{\beforebibskip}}
\addcontentsline{toc}{chapter}{\tocEntry{\bibname}}

\label{app:bibliography}

\begin{singlespace}
\printbibliography
\end{singlespace}

\endgroup\selectlanguage{english}%

\end{document}